\documentclass[11pt]{article}
\usepackage{hyperref}

\usepackage{authblk}
\usepackage[leqno]{amsmath}
\usepackage{amsthm}
\usepackage{amsfonts}
\usepackage{amssymb}
\usepackage{booktabs}
\usepackage{mathrsfs}

\usepackage{dsfont}
\usepackage{bbold}
\usepackage{stmaryrd}
\usepackage{eucal}
\usepackage{mathtools}
\usepackage{xcolor}
\usepackage[all]{xy}
\usepackage{tikz}
\usetikzlibrary{arrows,automata,decorations.markings,intersections, shapes.geometric,matrix,positioning,patterns}
\tikzset{
    state/.style={
           rectangle,
           rounded corners,
           draw=black, very thick,
           minimum height=2em,
           inner sep=2pt,
           text centered,
           },
}
\usepackage{tikz-cd}

\pdfoutput=1 

\usepackage[T1]{fontenc} 

\DeclareFontFamily{OT1}{rsfs}{}
\DeclareFontShape{OT1}{rsfs}{n}{it}{<-> rsfs10}{}
\DeclareMathAlphabet{\mathscr}{OT1}{rsfs}{n}{it}

\DeclareMathOperator{\Hom}{Hom}

\DeclareMathOperator{\Res}{Res}

\DeclareMathOperator{\GL}{GL}
\DeclareMathOperator{\gl}{\mathfrak{gl}}
\DeclareMathOperator{\SL}{SL}

\DeclareMathOperator{\Spec}{Spec}

\newcommand*{\N}{\ensuremath{\mathbf{N}}}                        

\newcommand*{\m}{\mathfrak{m}}                                   

\theoremstyle{plain}
  \newtheorem{theorem}{Theorem}
  \newtheorem{proposition}[subsubsection]{Proposition}
  \newtheorem{lemma}[subsubsection]{Lemma}
  \newtheorem{corollary}{Corollary}
  \newtheorem{conjecture}{Conjecture}

\theoremstyle{definition}
  \newtheorem{definition}[subsubsection]{Definition}

\theoremstyle{remark}
  \newtheorem{example}[subsubsection]{Example}
  \newtheorem{remark}[subsubsection]{Remark}

\numberwithin{equation}{section}

\DeclareMathOperator{\sym}{Sym}

\title{Deformed Double Current Algebras, Matrix Extended $\mathcal W_{\infty}$ Algebras, Coproducts, and Intertwiners from the M2-M5 Intersection}

\author[a]{Davide Gaiotto\thanks{dgaiotto@perimeterinstitute.ca}}
\author[b]{Miroslav Rapčák\thanks{miroslav.rapcek@gmail.com}}
\author[c]{Yehao Zhou\thanks{yehao.zhou@ipmu.jp}}

\affil[a]{Perimeter Institute for Theoretical Physics, Waterloo, Ontario, Canada N2L 2Y5}
\affil[b]{Theoretical Physics Department, CERN, 1211 Geneva 23, Switzerland}
\affil[c]{Kavli Institute for the Physics and Mathematics of the Universe (WPI), University of Tokyo, Kashiwa, Chiba 277-8583, Japan}

\date{}

\begin{document}

\maketitle

\begin{abstract}
We study the algebraic structures which govern the deformation of supersymmetric intersections of M2 and M5 branes. The universal algebras on M2 and M5 branes are deformed double current algebra of $\mathfrak{gl}_K$ and $\mathfrak{gl}_K$-extended $\mathcal{W}_{\infty}$-algebra respectively. We give a new presentation of the deformed double current algebra of $\mathfrak{gl}_K$, and we give a rigorous mathematical construction of the $\mathfrak{gl}_K$-extended $\mathcal{W}_{\infty}$-algebra. A new presentation of the affine Yangian of $\mathfrak{gl}_K$ is also obtained. We construct various coproducts of these algebras, which are expected to encode the fusions of defects in twisted M-theory. The matrix extended Miura operators are identified as intertwiners in certain bimodules of these algebras.
\end{abstract}

{\tableofcontents}

\section{Introduction}
The purpose of this paper is to describe certain algebraic structures which characterize the properties of twisted M-theory on the 
\begin{equation}\label{eqn: twisted M-theory background}
\mathbb{R} \times \mathbb{C}^2 \times \mathbb{R}^2_{\epsilon_1} \times \frac{\mathbb{R}^2_{K^{-1} \epsilon_2} \times \mathbb{R}^2_{K^{-1} \epsilon_3}}{\mathbb{Z}_K}
\end{equation}
background, where $K$ is a positive integer. This background produces a five-dimensional theory depending on the three $\epsilon_i$ parameters, which must satisfy a constraint 
\begin{equation}
K \epsilon_1 + \epsilon_2 + \epsilon_3 = 0
\end{equation}
Our results generalizes the $K=1$ results derived in \cite{gaiotto2022miura}. 

The properties of (resp. topological or holomorphic) defects in the twisted M-theory background are controlled by two algebras \cite{costello2017holography,gaiotto2019aspects}: the $\mathfrak{gl}(K)$-extended deformed double current algebra $\mathsf A^{(K)}$ and the $\mathfrak{gl}(K)$-extended ${\cal W}_\infty$ chiral algebra ${\cal W}^{(K)}_\infty$. Both are non-linear deformations of the universal enveloping algebra of the classical gauge algebra of the theory: the semi-direct product of the Lie algebra $\mathfrak{sl}_K\otimes\mathscr O(\mathbb C^2)$ of traceless $K \times K$ matrices valued in polynomials in two variables and the Lie algebra of symplectomorphisms $\mathfrak{po}(\mathbb C^2)$ of $\mathbb{C}^2$. 

Concretely, $\mathsf A^{(K)}$ is generated from a collection of generators $e^i_{j;a,b}$, with integer indices  $0 \leq i,j \leq K$ and $0 \leq a,b$,
such that $e^i_{i;a,b} = \epsilon_2 t_{a,b}$. The $e^i_{j;a,b}$ and $t_{a,b}$ are deformations of the generators of $\mathfrak{sl}_K\otimes\mathscr O(\mathbb C^2)$ and $\mathfrak{po}(\mathbb C^2)$ respectively.
On the other hand, ${\cal W}^{(K)}_\infty$ is generated by $K \times K$ matrices of vertex operators $U^i_{j;n}(z)$ of dimension $n$ for $n \geq 1$. 

The relation between these algebras and the defects in twisted M-theory is predicated by Koszul duality \cite{costello2017holography}: a defect is defined by coupling the M-theory fields to a (chiral) algebra which admits a (chiral) algebra homomorphism from $\mathsf A^{(K)}$ or ${\cal W}^{(K)}_\infty$. In \cite{oh2021feynman}, perturbative computation in 5d holomorphic-topological Chern-Simons theory up to first order in $\epsilon_1$ was done as an examination of such prediction.

The existence of gauge-invariant junctions of defects requires the algebras to have extra structures. The reference \cite{gaiotto2022miura} proposed to encode the extra structure in a collection of compatible coproducts. In the language of Koszul duality, these structures should be the image of an holomorphic-topological factorization algebra structure on $\mathbb{R} \times \mathbb{C}$ under Koszul duality along the topological or holomorphic direction. In \cite{oh2021twisted}, perturbative computation in 5d holomorphic-topological Chern-Simons theory with various defects up to first order in $\epsilon_1$ was done as an examination of such prediction.

We expect that the construction in this paper naturally extends when $\gl_K$ is replaced by the Lie super algebra $\gl_{K|M}$, see the relevant discussions in \cite{ueda2019affine,ueda2022affine,rapvcak2020extensions,gaberdiel2018twin},

In this paper we will provide rigorous mathematical proofs for the algebraic relations expected to follow from the existence of twisted M-theory. We will mention physical motivations only sparingly. 

\subsection{Main results of the paper}
The main objects in this paper are the following.
\begin{itemize}
    \item The associative algebra $\mathsf A^{(K)}$ is defined in Section \ref{sec: DDCA}. $\mathsf A^{(K)}$ is a version of the deformed double current algebra of type $\gl_K$. Its relations to other versions of DDCAs of type $\gl_K$ are summarized in the Figure \eqref{relations between DDCAs} \footnote{For the DDCA of simple Lie algebras, see \cite{guay2017deformed}.}.
    \item The vertex algebra $\mathcal W^{(K)}_{\infty}$ is defined in Section \ref{sec: W(infinity)}. $\mathcal W^{(K)}_{\infty}$ is the $\gl_K$-extended version of $\mathcal W_{\infty}$-algebra. It is the uniform-in-$L$ version of the rectangular W-algebra $\mathcal W^{(K)}_L$. This generalizes the $K=1$ case in \cite{linshaw2021universal}.
    \item The associative algebra $\mathsf Y^{(K)}$ is defined in Section \ref{sec: Y^K and coproducts}. When $K\neq 2$, $\mathsf Y^{(K)}$ is isomorphic to the affine Yangian of type $A_{K-1}$ after localizing $\epsilon_2\epsilon_3$, see Theorem \ref{thm: compare affine Yangian with Y}. This provides a new presentation of affine Yangian of type $A_{K-1}$ when $K\neq 2$, and we conjecture that such new presentation extends to hold for $K=2$.
\end{itemize}
All of these algebras are defined over the base ring $\mathbb C[\epsilon_1,\epsilon_2]$. They are related as follows:
\begin{align}
    \mathsf A^{(K)}\subset \mathsf Y^{(K)}\subset {U}(\mathcal W^{(K)}_{\infty})[\bar\alpha^{-1}]\subset \mathfrak{U}(\mathcal W^{(K)}_{\infty})[\bar\alpha^{-1}],
\end{align}
where ${U}(\mathcal W^{(K)}_{\infty})$ is the restricted mode algebra of $\mathcal W^{(K)}_{\infty}$ defined in Appendix \ref{sec: Restricted Mode Algebra}, and $\mathfrak{U}(\mathcal W^{(K)}_{\infty})$ is the usual mode algebra of $\mathcal W^{(K)}_{\infty}$ whose definition is recalled in Appendix \ref{sec: Restricted Mode Algebra}. All the inclusions are $\mathbb C[\epsilon_1,\epsilon_2]$-algebra embeddings. Moreover, 
\begin{enumerate}
    \item $\mathsf Y^{(K)}$ contains two subalgebras $\mathsf A^{(K)}_+$ and $\mathsf A^{(K)}_-$, both isomorphic to $\mathsf A^{(K)}$, and $\mathsf Y^{(K)}$ is generated by $\mathsf A^{(K)}_+$ and $\mathsf A^{(K)}_-$ and a central element $\mathbf c$ subject to a set of gluing relations \eqref{eqn: gluing relations of Y}.
    \item The subalgebra $\mathsf A^{(K)}_+\subset \mathsf Y^{(K)}$ is mapped to the positive mode subalgebra $\mathfrak{U}_+(\mathcal W^{(K)}_{\infty})[\bar\alpha^{-1}]$ such that the action on the vacuum $\mathsf A^{(K)}_+|0\rangle$ factors through an augmentation $\mathfrak C_{\mathsf A}:\mathsf A^{(K)}_+\to \mathbb C[\epsilon_1,\epsilon_2]$. See \eqref{eqn: map to W(infinity)} and the comments that follow.
    \item The subalgebra $\mathsf A^{(K)}_-$ is the image of $\mathsf A^{(K)}_+$ under an anti-involution of $\mathfrak U(\mathcal W^{(K)}_{\infty})$ which generalizes the anti-involution $J^a_{b,n}\mapsto J^b_{a,-n}$ of the affine Lie algebra $\widehat{\gl}_K$. See \eqref{eqn: Psi_infinity minus on generators} and the proof of Proposition \ref{prop: map Y^K to modes of W(infinity)} for details.
    \item The algebra $\mathsf L^{(K)}:=\mathsf Y^{(K)}/(\mathbf c)$ is generated by two subalgebras $\mathsf A^{(K)}_+$ and $\mathsf A^{(K)}_-$ subject to a set of gluing relations \eqref{eqn: gluing relations of L}. See Theorem \ref{thm: gluing construction of L}.
    \item There exists a $\mathbb C[\epsilon_1,\epsilon_2]$-algebra embedding $S(w):\mathsf L^{(K)}\hookrightarrow \mathsf A^{(K)}(\!(w^{-1})\!)$, where $w$ is treated as a formal parameter. In fact, we treat this embedding as the definition of $\mathsf L^{(K)}$, see Definition \ref{def: the algebra L^K}.
\end{enumerate}
In the physics setting, $\mathsf A^{(K)}$ is the algebra of gauge-invariant local observables on the M2 branes (topological defects) in the twisted M-theory background \eqref{eqn: twisted M-theory background}, and $\mathcal W^{(K)}_{\infty}$ is the vertex algebra of gauge-invariant local observables on the M5 branes (holomorphic defects) in the twisted M-theory background \eqref{eqn: twisted M-theory background}. 

\bigskip The physics setting also predicts a variety of coproduct maps between these algebras which control the properties of configurations of defects which lie within a $\mathbb{R} \times \mathbb{C}$ subspace of the 5d geometry \cite{gaiotto2022miura}. We write down the explicit formulae for the predicted coproducts and prove that they are algebra homomorphisms: 
\begin{enumerate}
\item A meromorphic algebra coproduct $\Delta_{\mathsf A}(w): \mathsf A^{(K)} \to \mathsf A^{(K)} \otimes \mathsf A^{(K)}(\!(w^{-1})\!)$ describing the fusion of lines, see Proposition \ref{prop: AA coproduct}.
\item A chiral algebra coproduct $\Delta_{{\cal W}}: {\cal W}^{(K)}_{\infty} \to{\cal W}^{(K)}_{\infty} \otimes {\cal W}^{(K)}_{\infty}$ describing the fusion of surfaces, see Theorem \ref{thm: matrix extended W(infty)}.
\item An algebra coproduct $\Delta_{\mathsf Y}:\mathsf Y^{(K)}\to \mathsf Y^{(K)}\widetilde{\otimes} \mathsf Y^{(K)}$, see Definition \ref{def: YY coproduct} and equation \eqref{eqn: YY coproduct}. Here $\widetilde{\otimes}$ is a completion of tensor product, defined in Appendix \ref{sec: Completion of Tensor Product}.
\item A mixed coproduct $\Delta_{\mathsf A,\mathsf Y}: \mathsf A^{(K)} \to \mathsf A^{(K)}\widetilde{\otimes} \mathsf Y^{(K)}$ controlling the gauge-invariance constraints of M2-M5 junctions \footnote{The actual physical prediction is a mixed coproduct $\Delta_{\mathsf A,\mathcal W}: \mathsf A^{(K)} \to \mathsf A^{(K)}\widetilde{\otimes} U({\cal W}^{(K)}_{\infty})[\bar\alpha^{-1}]$, and this follows from $\Delta_{\mathsf A,\mathsf Y}$ via the embedding $\mathsf Y^{(K)}\subset U({\cal W}^{(K)}_{\infty})[\bar\alpha^{-1}]$.}, see \eqref{eqn: AW(infinity) coproduct}.
\item The linear meromorphic coproduct $\Delta_{\mathcal W}(w): U({\cal W}^{(K)}_{\infty}) \to U({\cal W}^{(K)}_{\infty}) \otimes U_+({\cal W}^{(K)}_{\infty})(\!(w^{-1})\!)$ which exists for all chiral algebras, defined by the action on products of vertex operators. Given a current ${\cal O}_n(z)$ of dimension $n$ in ${\cal W}^{(K)}_{\infty}$ the coproduct is \footnote{To derive this, use the action of a mode $O_{n,m}$ on two vertex operators
\begin{align*}
\oint {\cal O}_n(z) z^{n+m-1} {\cal V}_1(0) {\cal V}_2(w) = \left[\oint_{z=0} {\cal O}_n(z) z^{n+m-1} {\cal V}_1(0)\right] {\cal V}_2(w) + {\cal V}_1(0) e^{w L_{-1}} \left[\oint_{z=0} {\cal O}_n(z) (z+w)^{n+m-1}  {\cal V}_2(0)\right]
\end{align*}}
\begin{equation}
{\mathcal O}_{n,m} \mapsto {\mathcal O}_{n,m} \otimes 1 + \sum_{s=0}^\infty \binom{n+m-1}{s}w^{n+m-1-s} \left(1 \otimes {\mathcal O}_{n,s-n+1} \right).
\end{equation}
The details of this meromorphic coproduct for the restricted mode algebra of a chiral algebra will be given in the Appendix \ref{subsec: meromorphic coproduct of mode algebra}.
\end{enumerate}
These coproducts satisfy the following compatibilities:
\begin{enumerate}
\item $\Delta_{\mathsf Y}$ is the restriction of $\Delta_{\cal W}$ to the subalgebra $\mathsf Y^{(K)}\subset  U(\mathcal W^{(K)}_{\infty})[\bar\alpha^{-1}]$. I.e. if we denote $\Psi_\infty:\mathsf Y^{(K)}\hookrightarrow  U(\mathcal W^{(K)}_{\infty})[\bar\alpha^{-1}]$ then
\begin{equation}\label{eqn: D_Y compatible with D_W}
    [\Psi_{\infty}\otimes\Psi_{\infty}]\circ\Delta_{\mathsf Y}=\Delta_{\mathcal W}\circ \Psi_{\infty}
\end{equation}
\item $\Delta_{\mathsf Y}$ is compatible with mixed coproduct $\Delta_{\mathsf A,\mathsf Y}$ in the sense that
\begin{equation}\label{eqn: D_Y compatible with mixed coproduct}
    \Delta_{\mathsf Y}\circ i=[i\otimes 1]\circ \Delta_{\mathsf A,\mathsf Y},
\end{equation}
where $i:\mathsf A^{(K)}\hookrightarrow \mathsf Y^{(K)}$ is the natural inclusion.
\item $\Delta_{\mathsf Y}$ is compatible with meromorphic coproduct $\Delta_{\mathsf A}(w)$ in the sense that
\begin{equation}
    [1\otimes S_{\mathsf Y}(w)]\circ \Delta_{\mathsf Y}\circ i= [i\otimes 1]\circ \Delta_{\mathsf A}(w).
\end{equation}
Here $S_{\mathsf Y}(w):\mathsf Y^{(K)}\to \mathsf A^{(K)}(\!(w^{-1})\!)$ is the composition of the quotient map $\mathsf Y^{(K)}\twoheadrightarrow \mathsf L^{(K)}$ and the embedding $S(w):\mathsf L^{(K)}\hookrightarrow \mathsf A^{(K)}(\!(w^{-1})\!)$.
\item $\Delta_{\mathsf Y}$ is compatible with meromorphic coproduct $\Delta_{\mathcal W}(w)$  in the sense that
\begin{equation}\label{eqn: D_Y compatible with meromorphic coproducts}
    [\Psi_{\infty}\otimes \Psi_{\infty}]\circ [1\otimes i]\circ[1\otimes S_{\mathsf Y}(w)]\circ \Delta_{\mathsf Y}=\Delta_{\mathcal W}(w)\circ \Psi_{\infty}.
\end{equation}
\end{enumerate}
These coproducts satisfy a variety of co-associativity relations. In particular
\begin{enumerate}
\item $\Delta_{{\cal W}}$ is co-associative, i.e. ${\cal W}^{(K)}_{\infty}$ is a coalgebra object in the category of chiral algebras, see Theorem \ref{thm: matrix extended W(infty)}..
\item $\Delta_{\mathsf A}(w)$ satisfies co-associativity relations analogous to associativity of an OPE. More precisely, $\Delta_{\mathsf A}(w)$ together with the augmentation $\mathfrak C_{\mathsf A}$ make $\mathsf A^{(K)}$ a vertex coalgebra object in the category of algebras, see Theorem \ref{thm: A^K is vertex coalgebra}. The composition 
\begin{align}\label{eqn: YA meormorphic coproduct}
    \Delta_{\mathsf Y}(w):\mathsf Y^{(K)}\overset{\Delta_{\mathsf Y}}{\longrightarrow} \mathsf Y^{(K)}\widetilde{\otimes} \mathsf Y^{(K)} \overset{1\otimes S_{\mathsf Y}(w)}{\longrightarrow} \mathsf Y^{(K)}\otimes \mathsf A^{(K)}(\!(w^{-1})\!)
\end{align}
makes $\mathsf Y^{(K)}$ a vertex comodule of $\mathsf A^{(K)}$ in the category of algebras, see Proposition \ref{prop: Y^K is vertex comodule}. We review the definition of vertex coalgebras and vertex comodules in Appendix \ref{sec: vertex coalgebras and vertex comodules}.
\item $\Delta_{\mathcal W}(w)$ also satisfies co-associativity relations analogous to associativity of an OPE. More precisely, $\Delta_{\mathcal W}(w)$ together with the augmentation $U_+(\mathcal W^{(K)}_\infty)\to \mathbb C[\epsilon_1,\epsilon_2]$ that maps all nontrivial operators to zero make $U_+(\mathcal W^{(K)}_\infty)$ a vertex coalgebra object in the category of algebras, and also make $U(\mathcal W^{(K)}_\infty)$ a vertex comodule of $U_+(\mathcal W^{(K)}_\infty)$ in the category of algebras. The co-associativity of linear meromorphic coproduct holds for all chiral algebras, see Proposition \ref{prop: vertex coalgebra from vertex algebra}.
\item $\Delta_{\mathsf Y}$ is co-associative. More precisely $\Delta_{\mathsf Y}$ together with an augmentation map $\mathfrak C_{\mathsf Y}:\mathsf Y^{(K)}\to \mathbb C[\epsilon_1,\epsilon_2]$ make $\mathsf Y^{(K)}$ a bialgebra. The co-associativity of $\Delta_{\mathsf Y}$ together with the compatibility \eqref{eqn: D_Y compatible with mixed coproduct} implies that $\mathsf A^{(K)}$ is a comodule of $\mathsf Y^{(K)}$, i.e. the two ways to map $\mathsf A^{(K)} \to \mathsf A^{(K)} \widetilde{\otimes}\mathsf Y^{(K)}\widetilde{\otimes}\mathsf Y^{(K)}$ agree:
\begin{equation}
\left[\Delta_{\mathsf A,\mathsf Y} \otimes 1 \right] \circ \Delta_{\mathsf A,\mathsf Y} = \left[1 \otimes \Delta_{\mathsf Y}\right] \circ \Delta_{\mathsf A,\mathsf Y}.
\end{equation}
See Proposition \ref{prop: Y^K is a bialgebra} for details.
\item The two ways to map $\mathsf A^{(K)} \to \mathsf A^{(K)}\otimes \mathsf A^{(K)}\otimes \mathsf Y^{(K)}\otimes \mathsf Y^{(K)}(\!(w^{-1})\!)$ should agree:
\begin{equation}\label{eqn: mixed compatible with meromorphic}
\left[\Delta_{\mathsf A,\mathsf Y} \otimes \Delta_{\mathsf A,\mathsf Y}\right] \circ \Delta_{\mathsf A}(w) = \left[\Delta_{\mathsf A}(w) \otimes \Delta_{\mathsf Y}(w)\right] \circ \Delta_{\mathsf A,\mathsf Y},
\end{equation}
where $\Delta_{\mathsf Y}(w)$ is the map in \eqref{eqn: YA meormorphic coproduct}. However, a direct computation using explicit formulae for the coproducts shows that two sides of \eqref{eqn: mixed compatible with meromorphic} suffer from divergence issue, so \eqref{eqn: mixed compatible with meromorphic} is not well-defined. To cure the divergence issue, we use the restricted mode algebra (resp. positive restricted mode algebra) of $\mathcal W^{(K)}_\infty$ instead of $\mathsf Y^{(K)}$ (resp. $\mathsf A^{(K)}$) and the correct statement is that 
\begin{equation}
    [\Delta_{\mathcal W}\otimes \Delta_{\mathcal W}]\circ \Delta_{\mathcal W}(w)=\Delta_{\mathcal W\otimes \mathcal W}(w)\circ \Delta_{\mathcal W}
\end{equation}
which follows from the the functoriality of the meromorphic coproduct of the restricted mode algebra, see Proposition \ref{prop: functoriality of meromorphic coproduct}.

\end{enumerate}

The existence of universal algebras $\mathsf A^{(K)}$ and ${\cal W}^{(K)}_\infty$ equipped with such coproducts should be 
thought of as a way to encode a holomorphic-topological factorization algebra \cite{costello2021factorization,butson2020equivariant}. With some help from dualities, one may predict the existence of gauge-invariant junctions, which are 
special elements in certain $\mathsf A^{(K)} \widetilde{\otimes} U({\cal W}^{(K)}_{\infty}) - \mathsf A^{(K)}$ bimodules which intertwine the right action of $\mathsf A^{(K)}$ and left action of $\mathsf A^{(K)}$ via the coproduct $\Delta_{\mathsf A,\mathcal W}:\mathsf A^{(K)} \to \mathsf A^{(K)} \widetilde{\otimes} U({\cal W}^{(K)}_{\infty})$. We discuss the bimodules in Section \ref{subsec: bimodule} and intertwiners (Miura operators) in Section \ref{sec: Miura Operators as Intertwiners}. We also derive some formulae for the correlators of Miura operators in Section \ref{sec: Miura Operators as Intertwiners}.

\section{The Algebra \texorpdfstring{$\mathsf A^{(K)}$}{Ak}}\label{sec: DDCA}
Let $D_{\epsilon_2}(\mathbb C):=\mathbb C[\epsilon_2]\langle x,y\rangle/([y,x]=\epsilon_2)$ be the Weyl algebra (the algebra of $\epsilon_2$-differential operators on $\mathbb C$), and let $\mathfrak{gl}_K$ be the associative algebra of $K\times K$ complex matrices. Consider the Lie algebra $D_{\epsilon_2}(\mathbb C)\otimes \mathfrak{gl}_K$, which is a free module over the base ring $\mathbb C[\epsilon_2]$ with a basis as follows
\begin{align}
    \mathsf T_{n,m}(X):=\sym(x^my^n)\otimes X,\; X\in \mathfrak{gl}_K,
\end{align}
for all $(n,m)\in \mathbb N^2$. Here $\sym(\cdots)$ means averaging over all permutations.

\begin{definition}
We define the subspace $D_{\epsilon_2}(\mathbb C)\otimes \mathfrak{gl}_K^{\sim}\subset D_{\epsilon_2}(\mathbb C)\otimes \mathfrak{gl}_K[\epsilon_2^{-1}]$ to be the $\mathbb C[\epsilon_2]$-submodule generated by $\mathsf T_{n,m}(X)$ for all $X\in \gl_K, (n,m)\in \mathbb N^2$ and 
\begin{align}
    \mathsf t_{n,m}:=\frac{1}{\epsilon_2}\mathsf T_{n,m}(1),
\end{align}
for all $(n,m)\in \mathbb N^2$. 
\end{definition}

\begin{lemma}
$D_{\epsilon_2}(\mathbb C)\otimes \mathfrak{gl}_K^{\sim}$ is a Lie subalgebra of $D_{\epsilon_2}(\mathbb C)\otimes \mathfrak{gl}_K[\epsilon_2^{-1}]$.
\end{lemma}

\begin{proof}
Note that the Lie bracket $[\mathsf T_{n,m}(1),\mathsf T_{p,q}(X)]$ has no $\mathcal O(\epsilon_2^0)$-term since its $\epsilon_2\to 0$ limit vanishes, so $\epsilon_2^{-1}[\mathsf T_{n,m}(1),\mathsf T_{p,q}(X)]$ is still inside $D_{\epsilon_2}(\mathbb C)\otimes \mathfrak{gl}_K$. Similarly $[\mathsf T_{n,m}(1),\mathsf T_{p,q}(1)]$ is a linear combination of $\mathsf T_{r,s}(1)$ for some $(r,s)$, and it is divisible by $\epsilon_2$, therefore $\epsilon_2^{-2}[\mathsf T_{n,m}(1),\mathsf T_{p,q}(1)]$ is still a linear combination of $\mathsf t_{r,s}$. This proves that $D_{\epsilon_2}(\mathbb C)\otimes \mathfrak{gl}_K^{\sim}$ is a Lie subalgebra.
\end{proof}

Note that $D_{\epsilon_2}(\mathbb C)\otimes \mathfrak{gl}_K^{\sim}$ is a free $\mathbb C[\epsilon_2]$ module with a basis $\{\mathsf T_{n,m}(X),\mathsf t_{n,m}\:|\: X\in\text{a basis of }\mathfrak{sl}_K,(n,m)\in \mathbb N^2\}$. 

\begin{definition}
We define the $\gl_K$-extended double current algebra to be universal enveloping algebra $U(D_{\epsilon_2}(\mathbb C)\otimes \mathfrak{gl}_K^{\sim})$ over the base ring $\mathbb C[\epsilon_2]$.
\end{definition}

A coordinate-free description of $D_{\epsilon_2}(\mathbb C)\otimes \mathfrak{gl}_K^{\sim}$ is as follows. Notice that $[D_{\epsilon_2}(\mathbb C),D_{\epsilon_2}(\mathbb C)\otimes \mathfrak{gl}_K]\subset \epsilon_2\cdot D_{\epsilon_2}(\mathbb C)\otimes \mathfrak{gl}_K$, therefore we can modify the Lie algebra by defining $D_{\epsilon_2}(\mathbb C)\otimes \mathfrak{gl}_K^{\sim}$ to be the $\mathbb C[\epsilon_2]$-submodule of $D_{\epsilon_2}(\mathbb C)\otimes \mathfrak{gl}_K[\epsilon_2^{-1}]$ generated by $D_{\epsilon_2}(\mathbb C)\otimes \mathfrak{sl}_K$ and $\frac{1}{\epsilon_2}\cdot D_{\epsilon_2}(\mathbb C)$, then $D_{\epsilon_2}(\mathbb C)\otimes \mathfrak{gl}_K^{\sim}$ is a Lie subalgebra of $D_{\epsilon_2}(\mathbb C)\otimes \mathfrak{gl}_K[\epsilon_2^{-1}]$, and it contains $D_{\epsilon_2}(\mathbb C)\otimes \mathfrak{gl}_K$ as a Lie subalgebra. Moreover its $\epsilon_2\to 0$ limit is 
\begin{align}
    D_{\epsilon_2}(\mathbb C)\otimes \mathfrak{gl}_K^{\sim}/(\epsilon_2=0)\cong \mathfrak{po}(\mathbb C^2)\ltimes \left(\mathscr O(\mathbb C\times\mathbb C)\otimes\mathfrak{sl}_K\right),
\end{align}
where $\mathfrak{po}(\mathbb C^2)$ is the Lie algebra of functions on $\mathbb C_x\times\mathbb C_y$ equipped with Poisson bracket $\{y,x\}=1$, and $\mathscr O(\mathbb C\times\mathbb C)$ is the function ring on $\mathbb C_x\times\mathbb C_y$ (considered as abelian Lie algebra), and $\mathfrak{po}(\mathbb C^2)$ naturally acts on the Lie algebra $\mathscr O(\mathbb C\times\mathbb C)\otimes\mathfrak{sl}_K$ via the Poisson bracket with the first tensor component. We shall define a deformation of $U(D_{\epsilon_2}(\mathbb C)\otimes \mathfrak{gl}_K^{\sim})$ in this section. 

\begin{remark}\label{rmk: modified current algebra in general}
More generally, one can define the Lie algebra $D_{\epsilon_2}(\mathcal M)\otimes \gl_K^{\sim}$ for any affine smooth variety $\mathcal M$ to be the $\mathbb C[\epsilon_2]$-submodule of $D_{\epsilon_2}(\mathcal M)\otimes \mathfrak{gl}_K[\epsilon_2^{-1}]$ generated by $D_{\epsilon_2}(\mathcal M)\otimes \mathfrak{sl}_K$ and $\frac{1}{\epsilon_2}\cdot D_{\epsilon_2}(\mathcal M)$. Then $D_{\epsilon_2}(\mathcal M)\otimes \gl_K^{\sim}$ contains $D_{\epsilon_2}(\mathcal M)\otimes \gl_K$ as a Lie subalgebra, and $D_{\epsilon_2}(\mathcal M)\otimes \mathfrak{gl}_K^{\sim}/(\epsilon_2=0)$ is isomorphic to $\mathfrak{po}(T^*\mathcal M)\ltimes \left(\mathscr O(T^*\mathcal M)\otimes\mathfrak{sl}_K\right)$, where $\mathfrak{po}(T^*\mathcal M)$ is the Lie algebra of functions on cotangent bundle $T^*\mathcal M$ equipped with standard Poisson bracket.
\end{remark}

\begin{definition}\label{def: A^K}
We define $\mathsf A^{(K)}$ to be the $\mathbb C[\epsilon_1,\epsilon_2]$-algebra generated by $\{\mathsf T_{n,m}(X),\mathsf t_{n,m}\:|\: X\in\mathfrak{gl}_K,(n,m)\in \mathbb N^2\}$ with the relations \eqref{eqn: A0}-\eqref{eqn: A4} as follows.
\begin{equation}\label{eqn: A0}
    \mathsf T_{n,m}(1)=\epsilon_2 \mathsf t_{n,m},\; \mathsf T_{n,m}(aX+bY)=a\mathsf T_{n,m}(X)+b\mathsf T_{n,m}(Y),\forall (a,b)\in \mathbb C^2,\tag{A0}
\end{equation}
\begin{equation}\label{eqn: A1}
    [\mathsf T_{0,0}(X),\mathsf T_{0,n}(Y)]=\mathsf T_{0,n}([X,Y]),\; [\mathsf T_{0,0}(X),\mathsf t_{0,n}]=0,\tag{A1}
\end{equation}
\begin{equation}\label{eqn: A2}
\text{ for }p+q\le 2,
\begin{cases}
&[\mathsf t_{p,q},\mathsf T_{n,m}(X)]=(mp-nq)\mathsf T_{p+n-1,q+m-1}(X),\\
&[\mathsf t_{p,q},\mathsf t_{n,m}]=(mp-nq)\mathsf t_{p+n-1,q+m-1},
\end{cases}
\tag{A2}
\end{equation}
To write down \eqref{eqn: A3}-\eqref{eqn: A4}, we introduce notation $\epsilon_3=-K\epsilon_1-\epsilon_2$, and $$\mathsf T_{u,r,t,s}(X\otimes Y):=\mathsf T_{u,r}(X)\mathsf T_{t,s}(Y)$$ for $X,Y\in \mathfrak{gl}_K$, and $\Omega:=E^a_b\otimes E^b_a\in \mathfrak{gl}_K^{\otimes 2}$, then
\begin{equation}\label{eqn: A3}
\begin{cases}
    &\begin{aligned}
        [\mathsf T_{1,0}(X),\mathsf T_{0,n}(Y)]=&\mathsf T_{1,n}([X,Y])-\frac{\epsilon_3 n}{2}\mathsf T_{0,n-1}(\{X,Y\})-n\epsilon_1\mathrm{tr}(Y) \mathsf T_{0,n-1}(X)\\
&+\epsilon_1\sum_{m=0}^{n-1}\frac{m+1}{n+1}\mathsf T_{0,m,0,n-1-m}(([X,Y]\otimes 1)\cdot \Omega)\\
&+\epsilon_1 \sum_{m=0}^{n-1}\mathsf T_{0,m,0,n-1-m}((X\otimes Y-XY\otimes 1)\cdot \Omega)
    \end{aligned}\\
    &[\mathsf T_{1,0}(X),\mathsf t_{0,n}]=n \mathsf T_{0,n-1}(X),
\end{cases}\tag{A3}
\end{equation}
\begin{equation}\label{eqn: A4}
    \begin{split}
        &[\mathsf t_{3,0},\mathsf t_{0,n}]=3n \mathsf t_{2,n-1}+\frac{n(n-1)(n-2)}{4}(\epsilon_1^2-\epsilon_2\epsilon_3) \mathsf t_{0,n-3}\\
&-\frac{3\epsilon_1}{2}\sum_{m=0}^{n-3}(m+1)(n-2-m)(\mathsf T_{0,m,0,n-3-m}(\Omega)+\epsilon_1\epsilon_2\mathsf t_{0,m}\mathsf t_{0,n-3-m}), \;(n\ge 3)
    \end{split}\tag{A4}
\end{equation}
\end{definition}

It is easy to see that there is a $\mathbb C[\epsilon_2]$-algebra homomorphism $\mathsf A^{(K)}/(\epsilon_1)\to U(D_{\epsilon_2}(\mathbb C)\otimes \mathfrak{gl}_K^{\sim})$ sending $\mathsf T_{n,m}(X),\mathsf t_{n,m}\in \mathsf A^{(K)}/(\epsilon_1)$ to the same symbols in $D_{\epsilon_2}(\mathbb C)\otimes \mathfrak{gl}_K^{\sim}$. In the later subsection we will show deduce from a PBW theorem for $\mathsf A^{(K)}$ that this is an isomorphism, see Corollary \ref{cor: A is DDCA}.

As a preliminary observation, \eqref{eqn: A2} implies that $\{\mathsf t_{2,0},\mathsf t_{1,1},\mathsf t_{0,2}\}$ forms an $\mathfrak{sl}_2$ triple and its adjoint action on $\mathsf A^{(K)}$ integrates to an $\SL_2$ action:
\begin{align}\label{eqn: sl2 action}
    \begin{pmatrix}
a&b\\c&d
    \end{pmatrix}
    \curvearrowright \mathsf A^{(K)}: \;
    \begin{split}
        \mathsf T(X;u,v)&\mapsto \mathsf T(X;au+cv,bu+dv),\\
        \mathsf t(u,v)&\mapsto \mathsf t(au+cv,bu+dv),
    \end{split} 
\end{align}
where $\mathsf T(X;u,v):=\sum_{m,n\in \mathbb N^2}\mathsf T_{m,n}(X)u^mv^n$ and $\mathsf t(u,v):=\sum_{m,n\in \mathbb N^2}\mathsf t_{m,n}u^mv^n$ are generating series. In particular the $\SL_2$ element 
\begin{align}\label{eqn: automorphism tau}
\tau=\begin{pmatrix}
0&1\\-1&0
    \end{pmatrix}
\end{align}
acts on generators by $\tau(\mathsf t_{n,m})=(-1)^m \mathsf t_{m,n},\;\tau(\mathsf T_{n,m}(X))=(-1)^m \mathsf T_{m,n}(X)$.

\begin{definition}\label{def: subalg of A}
The algebra $\mathsf B^{(K)}$ is defined as the $\mathbb C[\epsilon_1,\epsilon_2]$-subalgebra of $\mathsf A^{(K)}$ generated by $\{\mathsf T'_{n,m}(X):=\epsilon_1\mathsf T_{n,m}(X), \mathsf t'_{n,m}:=\epsilon_1\mathsf t_{n,m}\:|\: X\in \mathfrak{gl}_K,(n,m)\in \mathbb N^2\}$.
The algebra $\mathsf D^{(K)}$ is defined as the $\mathbb C[\epsilon_1,\epsilon_2]$-subalgebra of $\mathsf A^{(K)}$ generated by $\{\mathsf T_{n,m}(X)\:|\: X\in \mathfrak{gl}_K,(n,m)\in \mathbb N^2\}$. If $K>1$, then the algebra $\mathbb D^{(K)}$ is defined as the $\mathbb C[\epsilon_1,\epsilon_2]$-subalgebra of $\mathsf A^{(K)}$ generated by $\{\mathsf T_{n,m}(X)\:|\: X\in \mathfrak{sl}_K,(n,m)\in \{0,1\}^2\}$.
\end{definition}

It is obvious from the definition that $\mathsf D^{(K)}[\epsilon_2^{-1}]=\mathsf A^{(K)}[\epsilon_2^{-1}]$, and that $\mathsf B^{(K)}[\epsilon_1^{-1}]=\mathsf A^{(K)}[\epsilon_1^{-1}]$. Later we will show that $\mathbb D^{(K)}[\epsilon_3^{-1}]=\mathsf D^{(K)}[\epsilon_3^{-1}]$, see Corollary \ref{cor: compare two DDCAs}.

It is also obvious from the definition that the $\SL_2$ action in \eqref{eqn: sl2 action} preserves the subalgebras $\mathsf B^{(K)},\mathsf D^{(K)}$, and $\mathbb D^{(K)}$.\\

There is a natural grading on $\mathsf A^{(K)}$ induced by the Cartan of $\SL_2$ action mentioned above, which is equivalent to setting
\begin{align}\label{eqn: grading on A}
    \deg \mathsf T_{n,m}(E^a_b)=\deg \mathsf t_{n,m}=m-n,\quad\deg\epsilon_1=\deg\epsilon_2=0,
\end{align}
so $\mathsf A^{(K)}$ is a $\mathbb Z$-graded algebra and $\mathsf D^{(K)}$ is a homogeneous subalgebra.\\

Before we proceed to the discussions on the properties of $\mathsf A^{(K)}$, let us write down the generators and relations of $\mathsf A^{(1)}$ in a more compact form.
\begin{lemma}\label{lem: relations A^1}
$\mathsf A^{(1)}$ is the $\mathbb C[\epsilon_1,\epsilon_2]$-algebra generated by $\{\mathsf t_{n,m}\:|\: (n,m)\in \mathbb N^2\}$ with the following relations
\begin{equation}\label{eqn: A1, K=1}
    \text{ for }p+q\le 2,\; [\mathsf t_{p,q},\mathsf t_{n,m}]=(mp-nq)\mathsf t_{p+n-1,q+m-1}
\end{equation}
\begin{equation}\label{eqn: A2, K=1}
\begin{split}
        [\mathsf t_{3,0},\mathsf t_{0,n}]=~&3n \mathsf t_{2,n-1}-\frac{n(n-1)(n-2)}{4}\sigma_2 \mathsf t_{0,n-3}\\
&+\frac{3\sigma_3}{2}\sum_{m=0}^{n-3}(m+1)(n-2-m)\mathsf t_{0,m}\mathsf t_{0,n-3-m}, \;(n\ge 3)
    \end{split}
\end{equation}
where we set $\sigma_2=\epsilon_1\epsilon_2+\epsilon_2\epsilon_3+\epsilon_3\epsilon_1$ and $\sigma_3=\epsilon_1\epsilon_2\epsilon_3$.
\end{lemma}

\begin{proof}
In the case $K=1$, the relation \eqref{eqn: A0} in Definition \ref{def: A^K} simply says that $\mathsf T_{n,m}(x),x\in \mathbb C$ is redundant and it equals to $x\epsilon_2\mathsf t_{n,m}$. Then relations \eqref{eqn: A1} and \eqref{eqn: A3} and the first line of \eqref{eqn: A2} are special cases of the second line of \eqref{eqn: A2}.
\end{proof}

To conclude the definition of $\mathsf A^{(K)}$, let us record a commutation relation which will be useful in the later subsections.
\begin{lemma}
The following relations hold in $\mathsf A^{(K)}$
\begin{equation}\label{eqn: A5}
    [\mathsf T_{n,0}(X),\mathsf t_{2,1}]=n\mathsf T_{n+1,0}(X),\quad [\mathsf t_{n,0},\mathsf t_{2,1}]=n\mathsf t_{n+1,0}.\tag{A5}
\end{equation}
\end{lemma}

\begin{proof}
For the first equation, let us first prove two related equations: 
\begin{align}\label{eqn: A5 related}
    [\mathsf t_{3,0},\mathsf T_{n-1,1}(X)]=3\mathsf T_{n+1,0}(X),\quad  [\mathsf t_{3,0},\mathsf T_{n,0}(X)]=0.
\end{align}
By \eqref{eqn: A2} and \eqref{eqn: A3}, we have $[\mathsf T_{1,0}(X),\mathsf t_{n,0}]=0$, so we have 
\begin{equation*}
\begin{split}
 &[\mathsf t_{3,0},\mathsf T_{n-1,1}(X)]=\frac{1}{2^{n-1}n!}\mathrm{ad}_{\mathsf t_{2,0}}^{n-1}([\mathsf t_{3,0},\mathsf T_{0,n}(X)])=\frac{1}{2^{n-1}(n+1)!}\mathrm{ad}_{\mathsf t_{2,0}}^{n-1}(\mathrm{ad}_{\mathsf T_{1,0}(X)}([\mathsf t_{3,0},\mathsf t_{0,n+1}]))\\
&=\frac{1}{2^{n-1}(n+1)!}\mathrm{ad}_{\mathsf T_{1,0}(X)}(\mathrm{ad}_{\mathsf t_{2,0}}^{n-1}([\mathsf t_{3,0},\mathsf t_{0,n+1}]))=\frac{1}{2^{n-1}(n+1)!}\mathrm{ad}_{\mathsf T_{1,0}(X)}(\mathrm{ad}_{\mathsf t_{2,0}}^{n-1}(3(n+1)\mathsf t_{2,n}))\\
&=\frac{3}{2^{n+1}(n+2)!}\mathrm{ad}_{\mathsf T_{1,0}(X)}(\mathrm{ad}_{\mathsf t_{2,0}}^{n+1}(\mathsf t_{0,n+2}))=\frac{3}{2^{n+1}(n+2)!}\mathrm{ad}_{\mathsf t_{2,0}}^{n+1}(\mathrm{ad}_{\mathsf T_{1,0}(X)}(\mathsf t_{0,n+2}))=3\mathsf T_{n+1,0}(X),
\end{split}
\end{equation*}
and similarly
\begin{equation*}
\begin{split}
 &[\mathsf t_{3,0},\mathsf T_{n,0}(X)]=\frac{1}{2^{n}n!}\mathrm{ad}_{\mathsf t_{2,0}}^{n}([\mathsf t_{3,0},\mathsf T_{0,n}(X)])=\frac{3}{2^{n+2}(n+2)!}\mathrm{ad}_{\mathsf t_{2,0}}^{n+2}(\mathrm{ad}_{\mathsf T_{1,0}(X)}(\mathsf t_{0,n+2}))=0.
\end{split}
\end{equation*}
Using \eqref{eqn: A5 related}, we derive
\begin{equation*}
\begin{split}
&[\mathsf T_{n,0}(X),\mathsf t_{2,1}]=-\frac{1}{6}[\mathsf T_{n,0}(X),\mathrm{ad}_{\mathsf t_{0,2}}(\mathsf t_{3,0})]=\frac{1}{6}([\mathrm{ad}_{\mathsf t_{0,2}}(\mathsf T_{n,0}(X)),\mathsf t_{3,0}]-\mathrm{ad}_{\mathsf t_{0,2}}([\mathsf T_{n,0}(X),\mathsf t_{3,0}]))\\
&=-\frac{n}{3}[\mathsf T_{n-1,1}(X),\mathsf t_{3,0}]=n\mathsf T_{n+1,0}(X).
\end{split}
\end{equation*}
The other equation $[\mathsf t_{n,0},\mathsf t_{2,1}]=n\mathsf t_{n+1,0}$ can be derived in the similar way and we omit the detail.
\end{proof}

\subsection{A filtration on \texorpdfstring{$\mathsf A^{(K)}$}{Ak}}\label{subsec: flitration on A}
We define an increasing filtration $0=F_{-1}\mathsf A^{(K)}\subset F_0\mathsf A^{(K)}\subset F_1\mathsf A^{(K)}\cdots$ as follows. Define the degree on generators as 
\begin{align}
    \deg_F \epsilon_1=\deg_F \epsilon_2=0,\; \deg_F \mathsf T_{n,m}(X)=\deg_F \mathsf t_{n,m}=n+m+1,
\end{align}
and this gives rise to a grading on the tensor algebra $\mathbb C[\epsilon_1,\epsilon_2]\langle \mathsf T_{n,m}(X), \mathsf t_{n,m}\:|\: X\in \mathfrak{gl}_K,(n,m)\in \mathbb N^2\rangle$. We define $F_{i}\mathsf A^{(K)}$ to be the image of the span of homogeneous elements in the tensor algebra of degrees $\le i$. We shall call it the diagonal filtration on $\mathsf A^{(K)}$

\begin{proposition}\label{prop: filtration}
For all $(n,m,p,q)\in \mathbb N^4$ and $X,Y$ chosen from a basis of $\mathfrak{gl}_K$, there exists 
\begin{equation}
    \begin{split}
        f_{n,m,p,q}^{X,Y}&\in F_{n+m+p+q}\mathbb C[\epsilon_1,\epsilon_2]\langle \mathsf T_{i,j}(Z)\:|\: Z\in \mathfrak{gl}_K, (i,j)\in \mathbb N^2\rangle,\\
        g_{n,m,p,q}^{X}&\in F_{n+m+p+q-2}\mathbb C[\epsilon_1,\epsilon_2]\langle \mathsf T_{i,j}(Z)\:|\: Z\in \mathfrak{gl}_K, (i,j)\in \mathbb N^2\rangle,\\
        h_{n,m,p,q}&\in F_{n+m+p+q-2}\mathbb C[\epsilon_1,\epsilon_2]\langle \mathsf T_{i,j}(Z), \mathsf t_{i,j}\:|\: Z\in \mathfrak{gl}_K, (i,j)\in \mathbb N^2\rangle,\\
    \end{split}
\end{equation}
such that following equations hold in $\mathsf A^{(K)}$
\begin{align}\label{eqn_schematic commutators 1}
    [\mathsf T_{n,m}(X),\mathsf T_{p,q}(Y)]=\mathsf T_{n+p,m+q}([X,Y])+\bar f_{n,m,p,q}^{X,Y},
\end{align}
\begin{align}\label{eqn_schematic commutators 2}
    [\mathsf t_{n,m},\mathsf T_{p,q}(X)]=(nq-mp)\mathsf T_{n+p-1,m+q-1}(X)+\bar g_{n,m,p,q}^{X},
\end{align}
\begin{align}\label{eqn_schematic commutators 3}
    [\mathsf t_{n,m},\mathsf t_{p,q}]=(nq-mp)\mathsf t_{n+p-1,m+q-1}+\bar h_{n,m,p,q},
\end{align}
where $\bar f_{n,m,p,q}^{X,Y}$ (resp. $\bar g_{n,m,p,q}^{X}$, resp. $\bar h_{n,m,p,q}$) is the image of $f_{n,m,p,q}^{X,Y}$ (resp. $g_{n,m,p,q}^{X}$, resp. $h_{n,m,p,q}$) in $\mathsf A^{(K)}$.
\end{proposition}

\begin{proof}
We construct $f_{n,m,p,q}^{X,Y}, g_{n,m,p,q}^X,h_{n,m,p,q}$ inductively. First of all, it is evident from \eqref{eqn: A0}-\eqref{eqn: A4} that we can set $g_{n,m,p,q}^X=h_{n,m,p,q}=0$ for all $n+m\le 2$, and set $f^{X,Y}_{0,0,0,n}=0$, and set $\mathsf T_{1,n}([X,Y])+f^{X,Y}_{1,0,0,n}$ to be the right-hand-side of the first equation of \eqref{eqn: A3}. Then we set $f^{X,Y}_{1,0,i+1,n-i-1}$ to be the lift of $\frac{1}{2n-2i}[\mathsf t_{2,0},\bar f^{X,Y}_{1,0,i,n-i}]$, inductively for all $i<n$. Next we set $f^{X,Y}_{0,1,p,q}$ to be the lift of $\frac{1}{2}[\bar f^{X,Y}_{1,0,p,q},\mathsf t_{0,2}]-p\bar f^{X,Y}_{1,0,p-1,q+1}$.

By \eqref{eqn: A2} and \eqref{eqn: A3}, we have $[\mathsf T_{1,0}(X),\mathsf t_{3,0}]=0$, thus 
\begin{equation}
    \begin{split}
        &[\mathsf t_{3,0},\mathsf T_{0,n}(X)]=\frac{1}{n+1}[\mathsf t_{3,0},[\mathsf T_{1,0}(X), \mathsf t_{0,n+1}]]=\frac{1}{n+1}[\mathsf T_{1,0}(X),[\mathsf t_{3,0}, \mathsf t_{0,n+1}]]\\
    &=3n \mathsf T_{2,n-1}(X)+\frac{n(n-1)(n-2)}{4}(\epsilon_1^2-\epsilon_2\epsilon_3)\mathsf T_{0,n-3}(X)+\text{quadratic+cubic}.
    \end{split}
\end{equation}
We will not need the exact form the of the quadratic or cubic terms, but only point out that they are all polynomials of $\mathsf T_{i,j}(Z)$ of total degree $\le n+1$. We set $g^X_{3,0,0,n}$ using the above equation. Then we set $g^X_{3,0,i+1,n-i-1}$ to be the lift of $\frac{1}{2n-2i}[\mathsf t_{2,0},\bar g^X_{3,0,i,n-i}]$, inductively for all $i<n$. Next, we set $g^X_{i-1,4-i,p,q}$ to be the lift of $\frac{1}{2i}[\bar g^X_{i,3-i,p,q},\mathsf t_{0,2}]-\frac{p}{i}\bar g^X_{i,3-i,p-1,q+1}$, inductively for all $0<i\le 3$.

For general cases, we proceed by induction. Assume that $f^{X,Y}_{n,m,p,q}$ have been constructed for all $(n+m)\le s$, then 
\begin{align*}
    &[\mathsf T_{s+1,0}(X),\mathsf T_{p,q}(Y)]=-\frac{1}{s}[[\mathsf t_{2,1},\mathsf T_{s,0}(X)],\mathsf T_{p,q}(Y)]\\
    &=-\frac{1}{s}([\mathsf t_{2,1},[\mathsf T_{s,0}(X),\mathsf T_{p,q}(Y)]]-[\mathsf T_{s,0}(X),[\mathsf t_{2,1},\mathsf T_{p,q}(Y)]])\\
    &= \mathsf T_{s+p+1,q}([X,Y])-\frac{1}{s}(\bar g^{[X,Y]}_{2,1,s+p,q}+[\mathsf t_{2,1},\bar f^{X,Y}_{s,0,p,q}]-(2q-p)\bar f^{X,Y}_{s,0,p+1,q}-[\mathsf T_{s,0}(X),\bar g^Y_{2,1,p,q}]),
\end{align*}
so we set $f^{X,Y}_{s+1,0,p,q}$ using the above equation, and note that $\deg_F f^{X,Y}_{s+1,0,p,q}\le p+s+q+1$. Next, we set $f^{X,Y}_{i-1,s+2-i,p,q}$ to be the lift of $\frac{1}{2i}[\bar g^X_{i,s+1-i,p,q},\mathsf t_{0,2}]-\frac{p}{i}\bar g^X_{i,s+1-i,p-1,q+1}$, inductively for all $0<i\le s$. This finishes the construction of $f^{X,Y}_{n,m,p,q}$. The construction of $g^X_{n,m,p,q}$ and $h_{n,m,p,q}$ are similar.
\end{proof}

\begin{proposition}\label{prop: filtration'}
Let $\mathsf T'_{n,m}(X),\mathsf t'_{n,m}$ be in the Definition \ref{def: subalg of A}, then for all $(n,m,p,q)\in \mathbb N^4$ and $X,Y$ chosen from a basis of $\mathfrak{gl}_K$, there exists 
\begin{equation}
    \begin{split}
        f_{n,m,p,q}^{'X,Y}&\in F_{n+m+p+q}\mathbb C[\epsilon_1,\epsilon_2]\langle \mathsf T'_{i,j}(Z)\:|\: Z\in \mathfrak{gl}_K, (i,j)\in \mathbb N^2\rangle,\\
        g_{n,m,p,q}^{'X}&\in F_{n+m+p+q-2}\mathbb C[\epsilon_1,\epsilon_2]\langle \mathsf T'_{i,j}(Z)\:|\: Z\in \mathfrak{gl}_K, (i,j)\in \mathbb N^2\rangle,\\
        h'_{n,m,p,q}&\in F_{n+m+p+q-2}\mathbb C[\epsilon_1,\epsilon_2]\langle \mathsf T'_{i,j}(Z), \mathsf t'_{i,j}\:|\: Z\in \mathfrak{gl}_K, (i,j)\in \mathbb N^2\rangle,\\
    \end{split}
\end{equation}
such that following equations hold in $\mathsf A^{(K)}$
\begin{align}\label{eqn_schematic commutators 1'}
    [\mathsf T_{n,m}(X),\mathsf T'_{p,q}(Y)]=\mathsf T'_{n+p,m+q}([X,Y])+\bar f_{n,m,p,q}^{'X,Y},
\end{align}
\begin{align}\label{eqn_schematic commutators 2'}
    [\mathsf t_{n,m},\mathsf T'_{p,q}(X)]=(nq-mp)\mathsf T'_{n+p-1,m+q-1}(X)+\bar g_{n,m,p,q}^{'X},
\end{align}
\begin{align}\label{eqn_schematic commutators 3'}
    [\mathsf t_{n,m},\mathsf t'_{p,q}]=(nq-mp)\mathsf t'_{n+p-1,m+q-1}+\bar h'_{n,m,p,q},
\end{align}
where $\bar f_{n,m,p,q}^{'X,Y}$ (resp. $\bar g_{n,m,p,q}^{'X}$, resp. $\bar h'_{n,m,p,q}$) is the image of $f_{n,m,p,q}^{'X,Y}$ (resp. $g_{n,m,p,q}^{'X}$, resp. $h'_{n,m,p,q}$) in $\mathsf A^{(K)}$.
\end{proposition}

\noindent The proof of Proposition \ref{prop: filtration'} is analogous to that of Proposition \ref{prop: filtration} and we omit it.\\

Let us choose a basis $\mathfrak B:=\{X_1,\cdots,X_{K^2-1}\}$ of $\mathfrak{sl}_K$, so that $\mathfrak{B}_+:=\{1\}\cup \mathfrak B$ is a basis of $\mathfrak{gl}_K$. We fix a total order $1\preceq X_1\preceq\cdots\preceq X_{K^2-1}$ on $\mathfrak B_+$. Then we put the dictionary order on the set $\mathfrak{G}(\mathsf A^{(K)}):=\{\mathsf T_{n,m}(X),\mathsf t_{n,m}\:|\: X\in \mathfrak B, (n,m)\in \mathbb N^2\}$, in other words $\mathsf T_{n,m}(X)\preceq\mathsf T_{n',m'}(X')$ if and only only if $n<n'$ or $n=n'$ and $m<m'$ or $(n,m)=(n',m')$ and $X\preceq X'$ \footnote{We choose dictionary order for convenience, we can also choose another order and the argument works verbatim.}. Similarly we put the dictionary order on the set $\mathfrak{G}(\mathsf D^{(K)}):=\{\mathsf T_{n,m}(X)\:|\: X\in \mathfrak B_+, (n,m)\in \mathbb N^2\}$ and the set $\mathfrak{G}(\mathsf B^{(K)}):=\{\mathsf T'_{n,m}(X),\mathsf t'_{n,m}\:|\: X\in \mathfrak B, (n,m)\in \mathbb N^2\}$.

\begin{definition}\label{def: basis of A, B, D}
Define the set of ordered monomials in $\mathfrak{G}(\mathsf A^{(K)})$ (resp. $\mathfrak{G}(\mathsf B^{(K)})$, resp. $\mathfrak{G}(\mathsf D^{(K)})$) as 
\begin{align}
    \mathfrak B(\star):=\{1\}\cup\{\mathcal O_1\cdots\mathcal O_n\:|\: n\in \mathbb N_{>0}, \mathcal O_1\preceq \cdots\preceq\mathcal O_n\in \mathfrak{G}(\star)\},
\end{align}
where $\star$ is $\mathsf A^{(K)}$ or $\mathsf B^{(K)}$ or $\mathsf D^{(K)}$.
\end{definition}

\begin{lemma}\label{lem: ordered monoimials}
$\mathsf A^{(K)}$ (resp. $\mathsf B^{(K)}$, resp. $\mathsf D^{(K)}$) is generated by $\mathfrak B(\mathsf A^{(K)})$ (resp. $\mathfrak{B}(\mathsf B^{(K)})$, resp. $\mathfrak{B}(\mathsf D^{(K)})$) as $\mathbb C[\epsilon_1,\epsilon_2]$-module.
\end{lemma}

\begin{proof}
Let us first prove the claim for $\mathsf A^{(K)}$. Obviously $F_0\mathsf A^{(K)}$ is generated by $1$ as $\mathbb C[\epsilon_1,\epsilon_2]$-module. Assume that $F_s\mathsf A^{(K)}$ is generated by elements in $\mathfrak B(\mathsf A^{(K)})$, then Proposition \ref{prop: filtration} implies that we can reorder any monomials in $\mathfrak{G}(\mathsf A^{(K)})$ with total degree $s+1$ into the non-decreasing order modulo terms in $F_{s}\mathsf A^{(K)}$, therefore $F_{s+1}\mathsf A^{(K)}$ is generated by elements in $\mathfrak B(\mathsf A^{(K)})$. $F_{\bullet}\mathsf A^{(K)}$ is obviously exhaustive, thus $\mathsf A^{(K)}$ is generated by $\mathfrak B(\mathsf A^{(K)})$. For $\mathsf D^{(K)}$, it is enough to use \eqref{eqn_schematic commutators 1} for the induction step. And for $\mathsf B^{(K)}$, use Proposition \ref{prop: filtration'} instead.
\end{proof}

\begin{lemma}\label{lem: t[m,n] preserves D^K}
For all $(n,m)\in \mathbb N^2$, the adjoint action of $\mathsf t_{n,m}$ preserves the subalgebra $\mathsf D^{(K)}$.
\end{lemma}
\begin{proof}
This follows from \eqref{eqn_schematic commutators 2}.
\end{proof}

\subsection{PBW theorems for \texorpdfstring{$\mathsf A^{(K)}$}{Ak, Bk, and Dk}, \texorpdfstring{$\mathsf B^{(K)}$}{TEXT}, and \texorpdfstring{$\mathsf D^{(K)}$}{TEXT}}\label{subsec: PBW}
Recall that the $\gl_K$-extended rational Cherednik algebra $\mathcal H^{(K)}_N$ \cite{guay2005cherednik,etingof2002symplectic,opdam2000lecture,etingof2010lecture} is defined as the quotient of the semi-direct product
\begin{align*}
    \mathbb C[\mathfrak{S}_N]\ltimes \left(\mathbb C\langle x_1,\cdots,x_N,y_1,\cdots,y_N\rangle\otimes \mathfrak{gl}_K^{\otimes N}\right)
\end{align*}
by the relations
\begin{align*}
    [x_i,x_j]=0&,\quad [y_i,y_j]=0,\\
    [y_i,x_j]=\delta_{ij}(\epsilon_2-\epsilon_1&\sum_{l\neq i}s_{il}\Omega_{il})+(1-\delta_{ij})\epsilon_1 s_{ij}\Omega_{ij},
\end{align*}
where $s_{ij}\in \mathfrak{S}_N$ is the elementary permutation of $ij$ positions, and $\Omega_{ij}\in \mathfrak{gl}_K^{\otimes N}$ is the tensor $1\otimes\cdots \otimes E^a_b\otimes\cdots \otimes E^b_a\otimes\cdots\otimes 1$ (put $1$ on the $k$-th site for $k\neq i$ or $j$). It is well-known that $\mathcal H^{(K)}_N$ admits an embedding $\mathcal H^{(K)}_N\hookrightarrow \mathbb C[\mathfrak{S}_N]\ltimes \left(D(\mathbb C^N_{\mathrm{disj}})\otimes \mathfrak{gl}_K^{\otimes N}\right)[\epsilon_1,\epsilon_2]$, where $D(\mathbb C^N_{\mathrm{disj}})$ is the algebra of differential operators on the configuration space of $N$-points on $\mathbb C$. We identify $x_1,\cdots,x_N$ with the coordinates of $\mathbb C^N$, then the explicit formula of Dunkl embedding is given by
\begin{align}
    y_i\mapsto \epsilon_2 \partial_i+\epsilon_1\sum_{j\neq i}\frac{1}{x_i-x_j}s_{ij}\Omega_{ij}.
\end{align}
Recall that the spherical subalgebra $\mathrm S\mathcal H^{(K)}_N$ is defined as $\mathbf e\mathcal H^{(K)}_N \mathbf e$, where $\mathbf e:=\frac{1}{N!}\sum_{g\in \mathfrak S_N}g$ is the projector to the $\mathfrak S_N$-invariants. Then Dunkl embedding restricts to an embedding of $\mathbb C[\epsilon_1,\epsilon_2]$-algebras $\mathrm S\mathcal H^{(K)}_N\hookrightarrow \left(D(\mathbb C^N_{\mathrm{disj}})\otimes \mathfrak{gl}_K^{\otimes N}\right)^{\mathfrak{S}_N}[\epsilon_1,\epsilon_2]$. Examples of images of elements are as follows:
\begin{align}
    \sum_{i=1}^Ny_i \mathbf e\mapsto \epsilon_2 \sum_{i=1}^N\partial_i,\quad\sum_{i=1}^Ny_i^2  \mathbf e \mapsto \epsilon_2^2 \sum_{i=1}^N\partial_i^2-\sum_{i\neq j}^N\frac{\epsilon_1}{(x_i-x_j)^2}(\epsilon_2 \Omega_{ij}+\epsilon_1).
\end{align}

\begin{lemma}\label{lem: map rho_N}
The map on generators
\begin{align}\label{eqn: map rho_N}
    \rho_N(\mathsf T_{n,m}(X))=\sum_{i=1}^N\sym(x_i^my_i^n)X_i\mathbf e,\quad\rho_N(\mathsf t_{n,m})=\frac{1}{\epsilon_2}\sum_{i=1}^N\sym(x_i^my_i^n)\mathbf e,
\end{align}
uniquely determines a surjective algebra homomorphism $\rho_N:\mathsf A^{(K)}[\epsilon_2^{-1}]\twoheadrightarrow \mathrm S\mathcal H^{(K)}_N[\epsilon_2^{-1}]$.
\end{lemma}

\begin{proof}
The computation is essentially the same as that of generators and relations of quantized Gieseker varieties in the Calogero representations, we refer to Appendix \ref{sec: quantum ADHM quiver variety} for details.
\end{proof}

\begin{remark}
It is obvious from the formula that the map $\rho_N$ intertwines between the $\SL_2$ action \eqref{eqn: sl2 action} and the following $\SL_2$ action on $\mathcal H^{(K)}_N$:
\begin{align}
    \begin{pmatrix}
a&b\\c&d
    \end{pmatrix}
    \curvearrowright \mathcal H^{(K)}_N: \;x_i\mapsto dx_i+cy_i,\; y_i\mapsto bx_i+ay_i.
\end{align}
\end{remark}

\begin{remark}
The grading \eqref{eqn: grading on A} is compatible with the natural grading on $\mathrm{S}\mathcal K^{(K)}_N$. Namely the grading on $\mathrm{S}\mathcal H^{(K)}_N$ is by setting $\deg y_i=-1,\deg x_i=1$. Then the map $\rho_N:\mathsf D^{(K)}\to \mathrm{S}\mathcal H^{(K)}_N$ is $\mathbb Z$-graded.
\end{remark}

\begin{proposition}\label{prop: compare with Kalinov's DDCA}
    Let $(t,k)\in \mathbb C^2$ and denote by $\mathcal D_{t,k}(K)$ the DDCA defined in \cite[5.3.3]{kalinov2021deformed}. Assume that $t\neq 0$, then there is surjective homomorphism of algebras
\begin{align*}
    \mathsf A^{(K)}/(\epsilon_1=k,\epsilon_2=t)\twoheadrightarrow \mathcal D_{t,k}(K).
\end{align*}
\end{proposition}
\begin{proof}
It is shown in \cite{kalinov2021deformed} that $\mathcal D_{t,k}(K)$ has a set of generators $\{T_{m,n}(X)\:|\: X\in \mathfrak{gl}_K,(m,n)\in \mathbb N^2\}$. Consider the map
\begin{align}\label{eqn: map from A to D}
    \mathsf T_{n,m}(X)\mapsto T_{m,n}(X),\quad \mathsf t_{n,m}\mapsto \frac{1}{t}T_{m,n}(1).
\end{align}
We claim that for any transcendental number $\nu\in \mathbb C$ the map \eqref{eqn: map from A to D} generates a surjective algebra homomorphism 
\begin{align*}
   \mathsf A^{(K)}/(\epsilon_1=k,\epsilon_2=t)\twoheadrightarrow \widetilde{\mathcal D}_{t,k,\nu}(K)= \mathcal D_{t,k}(K)/(T_{0,0}(1)=\nu).
\end{align*}
In fact, if we fix a non-principle ultrafilter $\mathcal F$ on $\mathbb N$, and choose an isomorphism $\prod_{\mathcal F}\overline{\mathbb Q}\cong \mathbb C$ such that $\prod_{\mathcal F} n=\nu$, then $\widetilde{\mathcal D}_{t,k,\nu}(K)$ is by definition the restricted ultraproduct $\prod_{\mathcal F}^r \mathrm S\mathcal H^{(K)}_n/(\epsilon_1=k_n,\epsilon_2=t_n)$, and the parameters $t,k$ are identified to ultraproducts $\prod_{\mathcal F}t_n=t,\prod_{\mathcal F}k_n=k$, where $t_n,k_n\in \overline{\mathbb Q}$. From the Lemma \ref{lem: map rho_N} we know that if $t_n\neq 0$ then \eqref{eqn: map from A to D} gives rise to a surjective algebra homomorphism 
\begin{align*}
    \mathsf A^{(K)}/(\epsilon_1=k,\epsilon_2=t)/(\epsilon_1-k_n,\epsilon_2-t_n)\twoheadrightarrow \mathrm S\mathcal H^{(K)}_n/(\epsilon_1=k_n,\epsilon_2=t_n).
\end{align*}
By the definition of ultraproduct, $t_n$ are generically nonzero since $t\neq 0$, so the map \eqref{eqn: map from A to D} generically generates surjective algebra homomorphisms $\mathsf A^{(K)}/(\epsilon_1=k,\epsilon_2=t)/(\epsilon_1-k_n,\epsilon_2-t_n)\twoheadrightarrow \mathrm S\mathcal H^{(K)}_n/(\epsilon_1=k_n,\epsilon_2=t_n)$. By the property of ultraproduct, we conclude that \eqref{eqn: map from A to D} generates a surjective algebra morphism $\mathsf A^{(K)}/(\epsilon_1=k,\epsilon_2=t)\twoheadrightarrow \widetilde{\mathcal D}_{t,k,\nu}(K)$.

Now $\mathcal D_{t,k}(K)$ is a free $\mathbb C[T_{0,0}(1)]$-module, and the map \eqref{eqn: map from A to D} generates algebra homomorphisms for all transcendental $T_{0,0}(1)$, therefore \eqref{eqn: map from A to D} generates a surjective algebra homomorphism $\mathsf A^{(K)}/(\epsilon_1=k,\epsilon_2=t)\twoheadrightarrow \mathcal D_{t,k}(K)$.
\end{proof}

\begin{theorem}\label{thm: PBW}
$\mathsf A^{(K)}$ (resp. $\mathsf B^{(K)}$, resp. $\mathsf D^{(K)}$) is a free $\mathbb C[\epsilon_1,\epsilon_2]$-module with basis $\mathfrak B(\mathsf A^{(K)})$ (resp. $\mathfrak{B}(\mathsf B^{(K)})$, resp. $\mathfrak{B}(\mathsf D^{(K)})$).
\end{theorem}

\begin{proof}
By Lemma \ref{lem: ordered monoimials}, it suffices to show that there is no nontrivial relations among elements in $\mathfrak B(\mathsf A^{(K)})$ or $\mathfrak{B}(\mathsf B^{(K)})$ or $\mathfrak{B}(\mathsf D^{(K)})$. Localize to $\mathbb C[\epsilon_1^{\pm},\epsilon_2^{\pm}]$, the basis elements in $\mathfrak B(\mathsf A^{(K)})$ and $\mathfrak{B}(\mathsf B^{(K)})$ and $\mathfrak{B}(\mathsf D^{(K)})$ are the same up to scaling, so it is enough to prove the claim for $\mathfrak B(\mathsf A^{(K)})$. For all $(t,k)\in \mathbb C^{\times}\times \mathbb C$, it is known that the image of $\mathfrak B(\mathsf A^{(K)})$ in $\mathcal D_{t,k}(K)$ forms a $\mathbb C$-basis by \cite[5.3.4]{kalinov2021deformed}, this implies that there is no nontrivial relations among elements in $\mathfrak B(\mathsf A^{(K)})$.
\end{proof}

\begin{corollary}\label{cor: truncation of D}
The homomorphism $\rho_N$ in Lemma \ref{lem: map rho_N} restricts to a surjective algebra homomorphism $\rho_N:\mathsf D^{(K)}\twoheadrightarrow \mathrm S\mathcal H^{(K)}_N$, and $\ker(\prod_N \rho_N)=0$. Moreover, for all $(t,k)\in \mathbb C^2$, the map $\mathsf T_{n,m}(X)\mapsto T_{m,n}(X)$ generates an algebra isomorphism $\mathsf D^{(K)}/(\epsilon_1=k,\epsilon_2=t)\cong\mathcal D_{t,k}(K)$.
\end{corollary}

\begin{proof}
The first claim follows from the flatness of $\mathsf D^{(K)}$ as a $\mathbb C[\epsilon_2]$-module which is a consequence of PBW theorem. The proof of the last claim is analogous to that of Proposition \ref{prop: compare with Kalinov's DDCA} and we omit it. Finally, let $f\in\ker(\prod_N \rho_N)$, suppose that $f\neq 0$ then there exists $(t,k)\in \mathbb C^2$ such that $f(\epsilon_1=k,\epsilon_2=t)\neq 0$, so the map $\mathsf D^{(K)}/(\epsilon_1=k,\epsilon_2=t)\to\mathcal D_{t,k}(K)$ has a nontrivial kernel, which is a contradiction, hence $\ker(\prod_N \rho_N)=0$.
\end{proof}

\begin{corollary}\label{cor: A is DDCA}
The $\mathbb C[\epsilon_2]$-algebra homomorphism $\mathsf A^{(K)}/(\epsilon_1)\to U(D_{\epsilon_2}(\mathbb C)\otimes \mathfrak{gl}_K^{\sim})$ sending $\mathsf T_{n,m}(X),\mathsf t_{n,m}\in \mathsf A^{(K)}/(\epsilon_1)$ to the same symbols in $D_{\epsilon_2}(\mathbb C)\otimes \mathfrak{gl}_K^{\sim}$ is an isomorphism. Moreover this isomorphism restricts to a $\mathbb C[\epsilon_2]$-algebra isomorphism $\mathsf D^{(K)}/(\epsilon_1)\cong U(D_{\epsilon_2}(\mathbb C)\otimes \mathfrak{gl}_K)$.
\end{corollary}

\begin{proof}
Since $\mathfrak B(\mathsf A^{(K)})$ is simultaneously a $\mathbb C[\epsilon_2]$-basis of $\mathsf A^{(K)}/(\epsilon_1)$ and a $\mathbb C[\epsilon_2]$-basis of $U(D_{\epsilon_2}(\mathbb C)\otimes \mathfrak{gl}_K^{\sim})$, and the homomorphism $\mathsf A^{(K)}/(\epsilon_1)\to U(D_{\epsilon_2}(\mathbb C)\otimes \mathfrak{gl}_K^{\sim})$ induces identity on $\mathfrak B(\mathsf A^{(K)})$, we conclude that $\mathsf A^{(K)}/(\epsilon_1)\cong U(D_{\epsilon_2}(\mathbb C)\otimes \mathfrak{gl}_K^{\sim})$. The second claim is proven similarly.
\end{proof}

\begin{corollary}\label{cor: compare two DDCAs}
Assume that $K>1$, then $\mathbb D^{(K)}[\epsilon_3^{-1}]=\mathsf D^{(K)}[\epsilon_3^{-1}]$, in other words $\mathsf D^{(K)}$ is generated by $\{\mathsf T_{n,m}(X)\:|\: X\in \mathfrak{sl}_K,(n,m)\in \{0,1\}^2\}$ if $\epsilon_3$ is invertible.
\end{corollary}

\begin{proof}
The idea of proof is essentially explained in \cite[6.3.1]{kalinov2021deformed}. In effect, \cite[6.3.1]{kalinov2021deformed} can be restated as the claim that if $K>3$ then for all $(t,k)\in \mathbb C^2$ such that $t+Kk\neq 0$, the DDC algebra $\mathcal D_{t,k}(K)$ is generated by $\{T_{m,n}(X)\:|\: X\in \mathfrak{sl}_K,(m,n)\in \{0,1\}^2\}$. We shall simplify the argument in the \emph{loc. cit.} and relax the technical assumption to include all $K>1$ cases, and our argument works for $\mathbb C[\epsilon_1,\epsilon_2,\epsilon_3^{-1}]$-families, not just for complex numbers $(t,k)$.

Notice that $\{\mathsf T_{n,0}(X)\:|\: X\in \mathfrak{sl}_K,n\in \mathbb N \}$ forms the current Lie algebra $\mathfrak{sl}_K[y]$ and similarly $\{\mathsf T_{0,n}(X)\:|\: X\in \mathfrak{sl}_K,n\in \mathbb N \}$ also forms the current Lie algebra $\mathfrak{sl}_K[x]$, thus $\{\mathsf T_{n,0}(X), \mathsf T_{0,n}(X)\:|\: X\in \mathfrak{sl}_K,n\in \mathbb N \}\subset \mathbb D^{(K)}$. Then it follows that $[\mathsf t_{2,0},\mathbb D^{(K)}]\subset \mathbb D^{(K)}$ since the adjoint action of $\mathsf t_{2,0}$ maps the generators of $\mathbb D^{(K)}$ to the subset $\{\mathsf T_{n,0}(X)\:|\: X\in \mathfrak{sl}_K,n\in \{1,2\} \}$ which is contain in $\mathbb D^{(K)}$. Since $\mathsf D^{(K)}$ is generated by the image of the adjoint action of $\mathsf t_{2,0}$ on the subset $\{\mathsf T_{0,n}(X)\:|\: X\in \mathfrak{gl}_K,n\in \mathbb N \}$, it suffices to show that if $K>1$ then for all $n\in \mathbb N$ there exists $X\in \mathfrak{gl}_K$ such that $\mathrm{tr}(X)\neq 0$ and that $\mathsf T_{0,n}(X)\in \mathbb D^{(K)}[\epsilon_3^{-1}]$. 

To this end, we set $X=Y=H_1:=E^1_1-E^2_2$ in the first equation of \eqref{eqn: A3}, and get
\begin{align*}
    [\mathsf T_{1,0}(H_1),\mathsf T_{0,n}(H_1)]&=-\epsilon_3 n\mathsf T_{0,n-1}((H_1)^2)+\epsilon_1 \sum_{m=0}^{n-1}\mathsf T_{0,m,0,n-1-m}((H_1\otimes H_1-(H_1)^2\otimes 1)\cdot \Omega).
\end{align*}
Simple computation shows that
\begin{align*}
    (H_1\otimes H_1-(H_1)^2\otimes 1)\cdot \Omega=-(E^1_2\otimes E^2_1+E^2_1\otimes E^1_2)-\sum_{i=1}^2\sum_{j\neq i}^KE^i_j\otimes E^j_i,
\end{align*}
therefore $\mathsf T_{m,n-1-m}((H_1\otimes H_1-(H_1)^2\otimes 1)\cdot \Omega)\in \mathbb D^{(K)}$. This implies that $\mathsf T_{0,n-1}(E^1_1+E^2_2)\in \mathbb D^{(K)}[\epsilon_3^{-1}]$. This concludes the proof.
\end{proof}

\subsection{Relation to Costello's DDCA}
Define $\mathscr O_{\epsilon_1}(\mathcal M_{\epsilon_2}(N,K))$ to be the quantum Higgs branch algebra for the ADHM quiver gauge theory of gauge rank $N$ and flavor rank $K$. $\mathscr O_{\epsilon_1}(\mathcal M_{\epsilon_2}(N,K))$ is generated by $\GL_N$-invariant monomials in $\{X^i_j,Y^i_j,I^a_i,J^j_a\:|\: 1\le i,j\le N,1\le a\le K\}$ with relations
\begin{equation}
\begin{aligned}
[X^i_j,Y^k_l]=&\epsilon_1\delta^i_l\delta^k_j,\;[J^j_a,I^b_i]=\epsilon_1\delta^j_i\delta_a^b,\\
g(X,Y,I,J)&\left(:{[X,Y]^i_j}:+I^a_jJ^i_a-\epsilon_2\delta^i_j\right)=0,\\
\end{aligned}
\end{equation}
and other commutations between symbols $X,Y,I,J$ are zero. Here $g(X,Y,I,J)$ means arbitrary polynomials in $X,Y,I,J$, and normal ordering convention is such that $Y$ is to the left of $X$ and that $I$ to the left of $J$.

\begin{lemma}\label{lem: map to Higgs branch}
The map on generators
\begin{align}
    p_N(\mathsf T'_{n,m}(E^a_b))=I^a\sym (X^nY^m)J_b,\quad  p_N(\mathsf t'_{n,m})=\mathrm{Tr}(\sym (X^nY^m)),
\end{align}
uniquely determines a surjective algebra homomorphism $p_N:\mathsf B^{(K)}\twoheadrightarrow \mathscr O_{\epsilon_1}(\mathcal M_{\epsilon_2}(N,K))$.
\end{lemma}

\begin{proof}
By the flatness of $\mathsf B^{(K)}$ and $ \mathscr O_{\epsilon_1}(\mathcal M_{\epsilon_2}(N,K))$, it suffices to localize $\epsilon_1$ and show that 
\begin{align}
    p_N(\mathsf T_{n,m}(E^a_b))=\frac{1}{\epsilon_1}I^a\sym (X^nY^m)J_b,\quad  p_N(\mathsf t_{n,m})=\frac{1}{\epsilon_1}\mathrm{Tr}(\sym (X^nY^m)),
\end{align}
extends to a surjective algebra homomorphism $p_N:\mathsf A^{(K)}[\epsilon_1^{-1}]\twoheadrightarrow \mathscr O_{\epsilon_1}(\mathcal M_{\epsilon_2}(N,K))[\epsilon_1^{-1}]$. In the Appendix, we show that the above formula gives rise to an algebra homomorphism, see Lemma \ref{Lemma_Basic Commutation Relation} and Proposition \ref{Proposition_The Key Commutation Relation}, and it is surjective because $\mathscr O_{\epsilon_1}(\mathcal M_{\epsilon_2}(N,K))[\epsilon_1^{-1}]$ is generated by $\frac{1}{\epsilon_1}I^a\sym (X^nY^m)J_b$ and $\frac{1}{\epsilon_1}\mathrm{Tr}(\sym (X^nY^m))$.
\end{proof}

\begin{remark}
It is obvious from the formula that the map $p_N$ intertwines between the $\SL_2$ action \eqref{eqn: sl2 action} and the following $\SL_2$ action on $\mathscr O_{\epsilon_1}(\mathcal M_{\epsilon_2}(N,K))$:
\begin{align}
    \begin{pmatrix}
a&b\\c&d
    \end{pmatrix}
    \curvearrowright \mathscr O_{\epsilon_1}(\mathcal M_{\epsilon_2}(N,K)): \;X\mapsto aX+bY,\; Y\mapsto cX+dY.
\end{align}
\end{remark}

\begin{remark}
The grading \eqref{eqn: grading on A} is compatible with the natural grading on $\mathscr O_{\epsilon_1}(\mathcal M_{\epsilon_2}(N,K))$. Namely the grading on $\mathscr O_{\epsilon_1}(\mathcal M_{\epsilon_2}(N,K))$ is by setting $\deg X^i_j=\deg I^a_i=-1,\deg Y^i_j=\deg J^i_a=1$. Then the map $p_N:\mathsf B^{(K)}\to \mathscr O_{\epsilon_1}(\mathcal M_{\epsilon_2}(N,K))$ is $\mathbb Z$-graded.
\end{remark}

Recall that Costello defined a version of deformed double current algebra \cite{costello2017holography}. In short, his DDC algebra $\mathscr O_{\epsilon_1}(\mathcal M_{\epsilon_2}(\bullet,K))$ is the subalgebra of $\prod_N \mathscr O_{\epsilon_1}(\mathcal M_{\epsilon_2}(N,K))$ generated by elements $\{I^a\sym (X^nY^m)J_b\}_{N}$ and $\{\mathrm{Tr}(\sym (X^nY^m))\}_{N}$. In other words, it is the image of the map $p_{\bullet}:=\prod_N p_N: \mathsf B^{(K)}\to \prod_N \mathscr O_{\epsilon_1}(\mathcal M_{\epsilon_2}(N,K))$.

\begin{theorem}\label{thm: compare with Costello's DDCA}
$\ker(p_{\bullet})=0$. In particular the map $p_{\bullet}:\mathsf B^{(K)}\to \mathscr O_{\epsilon_1}(\mathcal M_{\epsilon_2}(\bullet,K))$ is a $\mathbb C[\epsilon_1,\epsilon_2]$-algebra isomorphism.
\end{theorem}

\begin{proof}
It is shown in \cite[Proposition 15.0.2]{costello2017holography} that the images of the generators $\mathfrak G(\mathsf B^{(K)})$ in $\mathbb C_{\epsilon_1=0}[\mathcal M_{\epsilon_2}(\bullet,K)]$ are algebraically independent for generic $\epsilon_2$. By the flatness of $\mathsf B^{(K)}/(\epsilon_1)$ and $\mathbb C[\mathcal M_{\epsilon_2}(\bullet,K)]$, the modulo $\epsilon_1$ map $\mathsf B^{(K)}/(\epsilon_1)\to \mathbb C[\mathcal M_{\epsilon_2}(\bullet,K)]$ is injective therefore it is an isomorphism. In other words the kernel of $\mathsf B^{(K)}\to \mathscr O_{\epsilon_1}(\mathcal M_{\epsilon_2}(\bullet,K))$ is contained in the ideal $(\epsilon_1)$. By the flatness of $\mathscr O_{\epsilon_1}(\mathcal M_{\epsilon_2}(\bullet,K))$, if $\epsilon_1f$ is in the kernel, then $f$ is in the kernel too. This implies that the kernel of $\mathsf B^{(K)}\to \mathscr O_{\epsilon_1}(\mathcal M_{\epsilon_2}(\bullet,K))$ is contained in the ideal $\cap_n (\epsilon_1^n)$, which is zero because $\mathsf B^{(K)}$ is a free $\mathbb C[\epsilon_1]$-module, thus $\mathsf B^{(K)}\to \mathscr O_{\epsilon_1}(\mathcal M_{\epsilon_2}(\bullet,K))$ is an isomorphism.
\end{proof}

\subsection{Duality automorphism}
\begin{proposition}\label{prop: duality for A}
    The map on generators 
    \begin{align}\label{eqn: duality for A}
        \sigma(\mathsf t_{n,m})=\mathsf t_{n,m},\quad \sigma(\mathsf T_{n,m}(X))= -\mathsf T_{n,m}(X^{\mathrm t})-\epsilon_1\mathrm{tr}(X)\mathsf t_{n,m},\quad \sigma(\epsilon_2)= \epsilon_3,
    \end{align}
    extends to a $\mathbb C[\epsilon_1]$-algebra automorphism $\sigma:\mathsf A^{(K)}\cong \mathsf A^{(K)}$. Moreover $\sigma$ preserves the subalgebras $\mathsf B^{(K)}$ and $\mathbb D^{(K)}$.
\end{proposition}

\begin{proof}
By flatness, it suffices to show that the map on generators 
\begin{align}
        \sigma(\mathsf t'_{n,m})=\mathsf t'_{n,m},\quad \sigma(\mathsf T'_{n,m}(X))= -\mathsf T'_{n,m}(X^{\mathrm t})-\epsilon_1\mathrm{tr}(X)\mathsf t'_{n,m},\quad \sigma(\epsilon_2)= \epsilon_3,
\end{align}
extends to a $\mathbb C[\epsilon_1]$-algebra automorphism $\sigma:\mathsf B^{(K)}\cong \mathsf B^{(K)}$. Consider the automorphism of $\mathscr O_{\epsilon_1}(\mathcal M_{\epsilon_2}(N,K))$ defined by 
\begin{equation}
    \begin{split}
        X^i_j\mapsto X^j_i,\; Y^j_i\mapsto Y^i_j&,\; I^a_i\mapsto J^i_a,\; J^j_b\mapsto -I^b_j\\
    \epsilon_1\mapsto\epsilon_1&,\;\epsilon_2\mapsto\epsilon_3.
    \end{split}
\end{equation}
Under this map, the commutation relations and quantum moment map equation are preserved, thus it induces a $\mathbb C[\epsilon_1]$-algebra automorphism on $\mathscr O_{\epsilon_1}(\mathcal M_{\epsilon_2}(N,K))$. By direct computation, $p_N$ intertwines between $\sigma$ and the above automorphism when restricted to generators. Since $\ker(p_{\bullet})=0$, we conclude that $\sigma$ is an algebra homomorphism, therefore it is a $\mathbb C[\epsilon_1]$-algebra automorphism. Moreover, $\sigma(\mathsf T_{n,m}(X))= -\mathsf T_{n,m}(X^{\mathrm t})$ for $X\in \mathfrak{sl}_K$, in particular $\sigma$ preserves $\mathbb D^{(K)}$.
\end{proof}
It is easy to see that $\sigma$ does not preserves $\mathsf D^{(K)}$ unless $K=1$.
\begin{definition}\label{def: D tilde}
We define the subalgebra $\widetilde{\mathsf D}^{(K)}$ to be the image $\sigma(\mathsf D^{(K)})$. Equivalently, $\widetilde{\mathsf D}^{(K)}$ is the $\mathbb C[\epsilon_1,\epsilon_2]$-subalgebra of $\mathsf A^{(K)}$ generated by $\mathfrak{G}(\widetilde{\mathsf D}^{(K)}):=\{\mathsf T_{n,m}(X),\epsilon_3\mathsf t_{n,m}\:|\: X\in \mathfrak{sl}_K,(n,m)\in \mathbb N^2\}$.
\end{definition}

Composing the map $\rho_N:\mathsf A^{(K)}\to \left(D(\mathbb C^N_{\mathrm{disj}})\otimes \mathfrak{gl}_K^{\otimes N}\right)^{\mathfrak{S}_N}[\epsilon_1,\epsilon_2^{-1}]$ with the duality automorphism $\sigma$, we obtain another homomorphism $\widetilde\rho_N:\mathsf A^{(K)}\to \left(D(\mathbb C^N_{\mathrm{disj}})\otimes \mathfrak{gl}_K^{\otimes N}\right)^{\mathfrak{S}_N}[\epsilon_1,\epsilon_3^{-1}]$ which is uniquely determined by 
\begin{align}\label{eqn: third truncation}
    \widetilde\rho_N(\mathsf T_{0,n}(E^a_b))=\sum_{i=1}^N F^a_{b,i}x_i^n,\quad \widetilde\rho_N(\mathsf t_{2,0})=\epsilon_3\sum_{i=1}^N \partial_i^2-2\sum_{i<j}^N \frac{\epsilon_1\Omega_{ij}+\epsilon_1^2\epsilon_2/\epsilon_3^2}{(x_i-x_j)^2}.
\end{align}
Here $F^a_{b,i}$ is related to $\mathfrak{gl}_K$ fundamental representation matrix $E^a_{b,i}$ by
\begin{align}
    F^a_{b,i}=-E^b_{a,i}-\frac{\epsilon_1}{\epsilon_3}\delta^b_a,
\end{align}
and $\Omega_{ij}=F^a_{b,i}F^b_{a,j}$. By definition of $\widetilde\rho_N$ and Corollary \ref{cor: truncation of D}, the image $\widetilde\rho_N(\widetilde{\mathsf D}^{(K)})$ lies inside the subalgebra $\left(D(\mathbb C^N_{\mathrm{disj}})\otimes \mathfrak{gl}_K^{\otimes N}\right)^{\mathfrak{S}_N}[\epsilon_1,\epsilon_2]$.

\begin{definition}
Define $\mathsf B^{(K)}_N:=\mathscr O_{\epsilon_1}(\mathcal M_{\epsilon_2}(N,K))$, $\mathsf D^{(K)}_N:=\mathrm S\mathcal H^{(K)}_N$, and $\widetilde{\mathsf D}^{(K)}_N:=\widetilde\rho_N(\widetilde{\mathsf D}^{(K)})$.
\end{definition}
By duality automorphism $\sigma$, $\widetilde{\mathsf D}^{(K)}_N$ is isomorphic to $\mathsf D^{(K)}_N$ up to reparametrization $\epsilon_2\leftrightarrow\epsilon_3$.

\subsection{Relations between various DDCAs}
The relations between different versions of DDCAs are summarized in the Figure \eqref{relations between DDCAs}.
\begin{figure}[ht]
\centering
\includegraphics[scale=0.5]{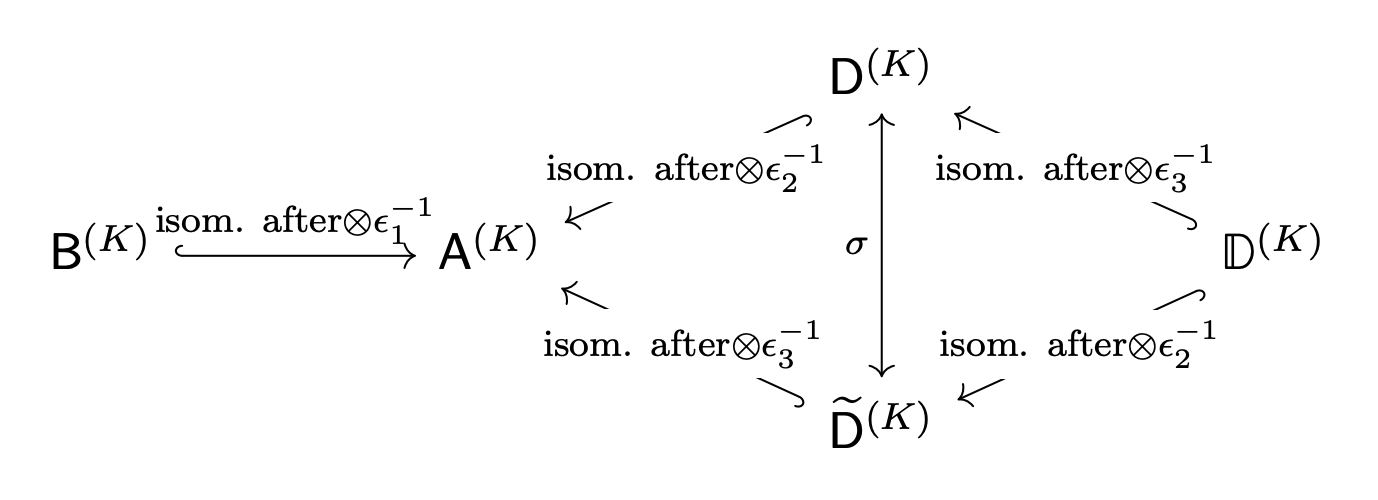}
\caption{Relations between the algebras $\mathsf A^{(K)}$, $\mathsf B^{(K)}$, $\mathsf D^{(K)}$, $\widetilde{\mathsf D}^{(K)}$, and $\mathbb D^{(K)}$. Here all hook arrows are algebra embeddings and they become isomorphism after localizing the parameter $\epsilon_i$. }
\label{relations between DDCAs}
\end{figure}
Different versions of DDCAs have been studied in the literature:
\begin{itemize}
    \item $\mathsf B^{(K)}$ is isomorphic to Costello's DDCA $\mathscr O_{\epsilon_1}(\mathcal M_{\epsilon_2}(\bullet,K))$ defined as large-$N$-limit of quantum ADHM quiver variety in \cite{costello2017holography}, see Theorem \ref{thm: compare with Costello's DDCA}. $\mathscr O_{\epsilon_1}(\mathcal M_{\epsilon_2}(\bullet,K))$ also shows up in the recent study on the large $N$ limit of the matrix models arising from $3d$ non-commutative Chern-Simons theory \cite{hu2023quantum}. 
    \item $\mathsf D^{(K)}$ is isomorphic to the DDCA $\mathcal D_{\epsilon_2,\epsilon_1}(K)$ studied by Etingof, Kalinov and Rains in \cite{etingof2023new,kalinov2021deformed}, see Corollary \ref{cor: truncation of D}.
    \item $\mathbb D^{(K)}$ is isomorphic to the imaginary $1$-shifted affine Yangian of $\mathfrak{sl}_K$ defined by Guay in \cite{guay2007affine} if $K>2$ (Proposition \ref{prop: DDCA=shifted Yangian}), therefore $\mathbb D^{(K)}$ is isomorphic to Guay's DDCA if $K>3$ \cite[Theorem 11.1]{guay2007affine}, see also \cite[Section 6]{kalinov2021deformed}.
    \item If $K=1$ then $\mathsf A^{(1)}$ is the $1$-shifted affine Yangian studied in the literature \cite{kodera2018quantized,rapcak2020cohomological}. $\mathbb D^{(1)}$ is not well-defined, nevertheless $\mathsf B^{(1)},\mathsf D^{(1)}$ and $\widetilde{\mathsf D}^{(1)}$ are defined, moreover the duality is enhanced to a triality automorphism $\mathfrak{S}_3\curvearrowright\{\mathsf B^{(1)},\mathsf D^{(1)},\widetilde{\mathsf D}^{(1)}\}$ where $\mathfrak{S}_3$ permutes the parameters $\epsilon_1,\epsilon_2,\epsilon_3$, see also \cite{creutzig2022trialities}.
\end{itemize}

\subsection{Vertical and horizontal filtrations on \texorpdfstring{$\mathsf A^{(K)}$}{Ak}}

In this subsection we give two more filtration on $\mathsf A^{(K)}$, a vertical one and a horizontal one, whose meaning is self-evident in the following definition.

\begin{definition}\label{def: filtration_vertical and horizontal}
The vertical filtration $0=V_{-1}\mathsf A^{(K)}\subset V_0\mathsf A^{(K)}\subset V_1\mathsf A^{(K)}\subset\cdots$ is an increasing filtration induced by setting the degrees on generators as
\begin{align}
    \deg_v\epsilon_1=\deg_v\epsilon_2=0,\quad \deg_v\mathsf T_{n,m}(X)=\deg_v\mathsf t_{n,m}=n.
\end{align}
The horizontal filtration $0=H_{-1}\mathsf A^{(K)}\subset H_0\mathsf A^{(K)}\subset H_1\mathsf A^{(K)}\subset\cdots$ is an increasing filtration induced by setting the degrees on generators as
\begin{align}
    \deg_h\epsilon_1=\deg_h\epsilon_2=0,\quad \deg_h\mathsf T_{n,m}(X)=\deg_h\mathsf t_{n,m}=m.
\end{align}
Note that the automorphism $\tau:\mathsf A^{(K)}\cong \mathsf A^{(K)}$ defined by \eqref{eqn: automorphism tau} interchanges the vertical and horizontal filtration: $\tau(H_n\mathsf A^{(K)})=V_n\mathsf A^{(K)}$, $\tau(V_n\mathsf A^{(K)})=H_n\mathsf A^{(K)}$.

We also define shifted vertical and horizontal filtrations $0=\tilde V_{-1}\mathsf A^{(K)}\subset \tilde V_0\mathsf A^{(K)}\subset \tilde V_1\mathsf A^{(K)}\subset\cdots$ and $0=\tilde H_{-1}\mathsf A^{(K)}\subset \tilde H_0\mathsf A^{(K)}\subset \tilde H_1\mathsf A^{(K)}\subset\cdots$ by setting the degrees on generators as
\begin{equation}
\begin{split}
 \deg_{\tilde v}\epsilon_1=\deg_{\tilde v}\epsilon_2=0&,\quad \deg_{\tilde v}\mathsf T_{n,m}(X)=\deg_{\tilde v}\mathsf t_{n,m}=n+1,\\
 \deg_{\tilde h}\epsilon_1=\deg_{\tilde h}\epsilon_2=0&,\quad \deg_{\tilde h}\mathsf T_{n,m}(X)=\deg_{\tilde h}\mathsf t_{n,m}=m+1.
\end{split}
\end{equation}
\end{definition}

Note that on the Cherednik algebra $\mathcal H^{(K)}_N$ there is an order filtration $\mathrm{Ord}^{\bullet}\mathcal H^{(K)}_N$ by letting the degree of $y_i,1\le i\le N$ to be $1$ and all other generators are of degree $0$. It is easy to see that the truncation $\rho_N:\mathsf A^{(K)}\to \mathrm{S}\mathcal H^{(K)}_N[\epsilon_2^{-1}]$ is filtered with respect to the (shifted) vertical filtration and order filtration, i.e. $\rho_N(V_{n}\mathsf A^{(K)})\subset \mathrm{Ord}^{n}\mathrm{S}\mathcal H^{(K)}_N[\epsilon_2^{-1}]$ for all $n\in \mathbb N$.

\begin{proposition}\label{prop: filtration_vertical and horizontal}
The commutators between generators of $\mathsf A^{(K)}$ can be schematically written as
\begin{equation}\label{eqn_schematic commutators_vertical}
\begin{split}
&[\mathsf T_{n,m}(X),\mathsf T_{p,q}(Y)]=\mathsf T_{n+p,m+q}([X,Y])\pmod{V_{n+p-1}\tilde V_{n+p}H_{m+q-1} \tilde H_{m+q}\mathsf D^{(K)}},\\
&[\mathsf t_{n,m},\mathsf T_{p,q}(X)]=(nq-mp)\mathsf T_{n+p-1,m+q-1}(X)\pmod{V_{n+p-2}\tilde V_{n+p-1}H_{m+q-2}\tilde H_{m+q-1}\mathsf D^{(K)}},\\
&[\mathsf t_{n,m},\mathsf t_{p,q}]=(nq-mp)\mathsf t_{n+p-1,m+q-1}\pmod{V_{n+p-2}\tilde V_{n+p-1}H_{m+q-2}\tilde H_{m+q-1}\mathsf A^{(K)}},
\end{split}
\end{equation}
where $V_i\tilde V_jH_k\tilde H_l \mathsf A^{(K)}$ is the short hand notation for $V_i\mathsf A^{(K)}\cap\tilde V_j\mathsf A^{(K)}\cap H_k\mathsf A^{(K)}\cap\tilde H_l \mathsf A^{(K)}$.
\end{proposition}

\noindent The proof of Proposition \ref{prop: filtration_vertical and horizontal} is similar to that of \ref{prop: filtration} and we omit it.\\

It follows from Proposition \ref{prop: filtration_vertical and horizontal} that the associated graded algebras with respect to vertical and horizontal filtrations are are isomorphic to the double current algebra $$\mathrm{gr}_V\mathsf A^{(K)}\cong U(\mathscr O(\mathbb C^2)\otimes \gl_K)\cong \mathrm{gr}_H\mathsf A^{(K)},$$ and the associated graded algebras with respect to shifted vertical and horizontal filtrations are are isomorphic to the free commutative algebra $$\mathrm{gr}_{\tilde V}\mathsf A^{(K)}\cong \mathrm{Sym}(\mathscr O(\mathbb C^2)\otimes \gl_K)\cong \mathrm{gr}_{\tilde H}\mathsf A^{(K)}.$$ Then it is easy to see that 
$\bar\rho_N:\mathrm{gr}_V\mathsf A^{(K)}\to \mathrm{gr}_{\mathrm{Ord}}\mathrm{S}\mathcal H^{(K)}_N[\epsilon_2^{-1}]\cong \left(\mathscr O(\mathbb C^2)\otimes\gl_K^{\otimes N}\right)^{\mathfrak S_N}[\epsilon_2^{-1}]$ is simply
\begin{align*}
\bar\rho_N(y^nx^m\otimes X)=\sum_{i=1}^N y_i^nx_i^m\otimes X_i,\; (X\in\mathfrak{sl}_K),\quad
\bar\rho_N(y^nx^m\otimes 1)=\frac{1}{\epsilon_2}\sum_{i=1}^N y_i^nx_i^m\otimes 1.
\end{align*}
In particular $\bigcap_N\ker(\bar\rho_N)=0$. Thus it follows that 
\begin{align}\label{eqn: verticcal vs order filtration}
    V_n\mathsf A^{(K)}=\bigcap _N\rho_N^{-1}(\mathrm{Ord}^n\mathrm{S}\mathcal H^{(K)}_N[\epsilon_2^{-1}]).
\end{align}
\begin{proposition}\label{prop: characterize vertical filtration}
$V_0\mathsf A^{(K)}$ is generated by $\mathsf T_{0,n}(X),\mathsf t_{0,n},(n=0,1,\cdots)$. Moreover, an element $f\in \mathsf A^{(K)}$ is in $V_k\mathsf A^{(K)}$ if and only if 
$$\mathrm{ad}_{\mathsf t_{0,n_1}}\circ \cdots\circ\mathrm{ad}_{\mathsf t_{0,n_{k+1}}}(f)=0, \; \forall (n_1,\cdots,n_{k+1})\in \mathbb N^{k+1}.$$
\end{proposition}

\begin{proof}
The corresponding statement for the Cherednik algebras is proven in \cite[Proposition 1.2]{schiffmann2013cherednik} (which works for matrix extended Cherednik algebras as well), namely $\mathrm{Ord}^0\mathrm{S}\mathcal H^{(K)}_N$ is generated by $\mathbf e\:\mathbb C[x_i\:|\:1\le i\le N]\mathbf e$. Moreover, an element $f\in \mathrm{S}\mathcal H^{(K)}_N$ is of order $\le k$ if and only if 
$$\mathrm{ad}_{z_1}\circ \cdots\circ\mathrm{ad}_{z_{k+1}}(f)=0, \; \forall z_1,\cdots,z_{k+1}\in \mathbf e\:\mathbb C[x_i\:|\:1\le i\le N]\mathbf e.$$
Then the proposition follows from \eqref{eqn: verticcal vs order filtration}.
\end{proof}

\begin{remark}
Since the automorphism $\tau:\mathsf A^{(K)}\cong \mathsf A^{(K)}$ defined by \eqref{eqn: automorphism tau} interchanges the vertical and horizontal filtration, we also have similar result for the horizontal filtration: an element $f\in \mathsf A^{(K)}$ is in $H_k\mathsf A^{(K)}$ if and only if 
$$\mathrm{ad}_{\mathsf t_{n_1,0}}\circ \cdots\circ\mathrm{ad}_{\mathsf t_{n_{k+1},0}}(f)=0, \; \forall (n_1,\cdots,n_{k+1})\in \mathbb N^{k+1}.$$
\end{remark}

\subsection{Unsymmetrized generators}
In the definition of $\mathsf B^{(K)}$, the symmetrization is used to define the uniform-in-$N$ generators $I^a\sym(X^nY^m)J_b$ and $\mathrm{Tr}\sym(X^nY^m)$. It also makes sense to talk about unsymmetrized version of generators. In fact, by the definition of uniform-in-$N$ algebra, for every word $\mathbf r$ with two letters, there are an elements $\{I^a\mathbf r(X,Y)J_b\}_N$ and $\{\mathrm{Tr}\:\mathbf r(X,Y)\}_N$ in the uniform-in-$N$ algebra $\mathscr O_{\epsilon_1}(\mathcal M_{\epsilon_2}(\bullet,K))\cong \mathsf B^{(K)}$.

\begin{proposition}\label{prop: other generators}
For every word $\mathbf r$ with two letters, there exists $\mathsf T_{\mathbf r}(E^a_b)\in \mathsf D^{(K)}$ and $\mathsf t_{\mathbf r}\in \mathsf A^{(K)}$ such that
\begin{equation}
\begin{split}
    \epsilon_1 p_N(\mathsf T_{\mathbf r}(E^a_b))=I^a\mathbf r(X,Y)J_b &,\quad \rho_N(\mathsf T_{\mathbf r}(E^a_b))=\sum_{i=1}^N \mathbf r(y_i,x_i)E^a_{b,i},\\
    \epsilon_1 p_N(\mathsf t_{\mathbf r})&=\mathrm{Tr}\:\mathbf r(X,Y).
\end{split}
\end{equation}
\end{proposition}

\begin{proof}
Step 1. We set $\mathsf T_{\mathbf r}(E^a_b):=\frac{1}{\epsilon_1}\{I^a\mathbf r(X,Y)J_b\}_N\in \mathsf B^{(K)}[\epsilon_1^{-1}]$, and we claim that $\rho_N(\mathsf T_{\mathbf r}(E^a_b))=\sum_{i=1}^N \mathbf r(y_i,x_i)E^a_{b,i}$. The claim implies that $\mathsf T_{\mathbf r}(E^a_b)\in \mathsf D^{(K)}$, by the PBW theorem for $\mathsf D^{(K)}$ and the triviality of $\ker(\prod_N\rho_N)$. To prove the claim, we use the induction on the number of letters in $\mathbf r$, the initial cases of zero and one letter are obviously true. For the induction step, suppose that $\mathbf r$ can be written as the augmentation of three binary sequences $\mathbf r_1 XY\mathbf r_2$, then let $\mathbf r'=\mathbf r_1 YX\mathbf r_2$, and we have 
\begin{align*}
    &p_N(\mathsf T_{\mathbf r}(E^a_b)-\mathsf T_{\mathbf r'}(E^a_b))=\frac{1}{\epsilon_1}(I^a\mathbf r_1(X,Y)[X,Y]\mathbf r_2(X,Y))J_b)\\
    &=\frac{1}{\epsilon_1}(-\epsilon_3I^a\mathbf r_1(X,Y)\mathbf r_2(X,Y)J_b-I^a\mathbf r_1(X,Y)J_cI^c\mathbf r_2(X,Y)J_b)\\
    &=p_N(-\epsilon_3 \mathsf T_{\mathbf r_1\mathbf r_2}(E^a_b)-\epsilon_1 \mathsf T_{\mathbf r_1}(E^a_c) \mathsf T_{\mathbf r_2}(E^c_b)),
\end{align*}
which implies that
\begin{align}\label{eqn: transformation to unsymmetrized basis}
    \mathsf T_{\mathbf r}(E^a_b)-\mathsf T_{\mathbf r'}(E^a_b)=-\epsilon_3 \mathsf T_{\mathbf r_1\mathbf r_2}(E^a_b)-\epsilon_1 \mathsf T_{\mathbf r_1}(E^a_c) \mathsf T_{\mathbf r_2}(E^c_b).
\end{align}
On the other hand, we have 
\begin{align*}
    &\sum_{i=1}^N (\mathbf r(y_i,x_i)-\mathbf r'(y_i,x_i))E^a_{b,i}=\sum_{i=1}^N E^a_{b,i}\mathbf r_1(y_i,x_i)[y_i,x_i]\mathbf r_2(y_i,x_i)\\
    &=\epsilon_2\sum_{i=1}^N E^a_{b,i}\mathbf r_1(y_i,x_i)\mathbf r_2(y_i,x_i)-\epsilon_1\sum_{i\neq j}^N E^a_{c,i}\mathbf r_1(y_i,x_i)E^c_{b,j}\mathbf r_2(y_j,x_j)\\
    &=-\epsilon_3\sum_{i=1}^N E^a_{b,i}\mathbf r_1(y_i,x_i)\mathbf r_2(y_i,x_i)-\epsilon_1\left(\sum_{i=1}^N E^a_{c,i}\mathbf r_1(y_i,x_i)\right)\left(\sum_{j=1}^N E^c_{b,j}\mathbf r_2(y_j,x_j)\right).
\end{align*}
By our induction assumption, we have $\rho_N(\mathsf T_{\mathbf r}(E^a_b)-\mathsf T_{\mathbf r'}(E^a_b))=\sum_{i=1}^N (\mathbf r(y_i,x_i)-\mathbf r'(y_i,x_i))E^a_{b,i}$. This implies that $\rho_N(\mathsf T_{\mathbf r}(E^a_b))-\sum_{i=1}^N \mathbf r(y_i,x_i)E^a_{b,i}$ is independent of the ordering of letters in $\mathbf r$, thus $$\rho_N(\mathsf T_{\mathbf r}(E^a_b))-\sum_{i=1}^N \mathbf r(y_i,x_i)E^a_{b,i}=\rho_N(\mathsf T_{n,m}(E^a_b))-\sum_{i=1}^N \sym(y_i^nx_i^m)E^a_{b,i}=0.$$

\noindent Step 2. We set $\mathsf t_{\mathbf r}:=\frac{1}{\epsilon_1}\{\mathrm{Tr}\:\mathbf r(X,Y)\}_N\in \mathsf B^{(K)}[\epsilon_1^{-1}]$, and we claim that $\mathsf t_{\mathbf r}\in \mathsf A^{(K)}$. To prove the claim, we use the induction on the number of letters in $\mathbf r$, the initial cases of zero and one letter are obviously true. For the induction step, suppose that $\mathbf r$ can be written as the augmentation of three binary sequences $\mathbf r_1 XY\mathbf r_2$, then let $\mathbf r'=\mathbf r_1 YX\mathbf r_2$, and we have 
\begin{align*}
    &p_N(\mathsf t_{\mathbf r}-\mathsf t_{\mathbf r'})=\frac{1}{\epsilon_1}\mathrm{Tr}(\mathbf r_1(X,Y)[X,Y]\mathbf r_2(X,Y))\\
    &=\frac{1}{\epsilon_1}(-\epsilon_3\mathrm{Tr}(\mathbf r_1(X,Y)\mathbf r_2(X,Y))-\mathrm{Tr}(\mathbf r_1(X,Y)J_cI^c\mathbf r_2(X,Y))-\epsilon_1\mathrm{Tr}(\mathbf r_1(X,Y))\mathrm{Tr}(\mathbf r_2(X,Y)))\\
    &=p_N(-\epsilon_3 \mathsf t_{\mathbf r_1\mathbf r_2}- \mathsf T_{\mathbf r_2\mathbf r_1}(1)) \pmod{\mathsf B^{(K)}_N}.
\end{align*}
By our induction assumption and the result of Step 1, we have $p_N(\mathsf t_{\mathbf r}-\mathsf t_{\mathbf r'})\in p_N(\mathsf A^{(K)})$ for all $N$. This implies that $\mathsf t_{\mathbf r}-\mathsf t_{\mathbf r'}\in \mathsf A^{(K)}$ by the PBW theorem for $\mathsf A^{(K)}$ and the triviality of $\ker(\prod_N p_N)$. Therefore for any reordering $\mathbf r''$ of $\mathbf r$, we have $\mathsf t_{\mathbf r}-\mathsf t_{\mathbf r''}\in \mathsf A^{(K)}$. Summing over all possible ordering, we get $\mathsf t_{\mathbf r}-\mathsf t_{n,m}\in \mathsf A^{(K)}$, thus $\mathsf t_{\mathbf r}\in \mathsf A^{(K)}$.
\end{proof}

\begin{remark}\label{rmk: unsymmetrized basis}
By \eqref{eqn: transformation to unsymmetrized basis} and an induction argument, we see that $$\mathsf T_{\mathbf r}(E^a_b)\equiv \mathsf T_{n,m}(E^a_b)\pmod{V_{n-1}\tilde V_n H_{m-1}\tilde H_{m}\mathsf D^{(K)}},$$
where $n,m$ are the numbers of $X$s and $Y$s in $\mathbf r(X,Y)$.
\end{remark}

\begin{remark}
If $K=1$, then the same computation in the proof of Proposition \ref{prop: other generators} shows that
\begin{align}
    \rho_N(\mathsf t_{\mathbf r})=\frac{1}{\epsilon_2}\sum_{i=1}^N\mathbf r(y_i,x_i).
\end{align}
In fact we have $\mathsf t_{\mathbf r}-\mathsf t_{\mathbf r'}=-\epsilon_3 \mathsf t_{\mathbf r_1\mathbf r_2}-\epsilon_1\epsilon_2 \mathsf t_{\mathbf r_1} \mathsf t_{\mathbf r_2}$, and 
\begin{align*}
    \sum_{i=1}^N (\mathbf r(y_i,x_i)-\mathbf r'(y_i,x_i))=-\epsilon_3\sum_{i=1}^N \mathbf r_1(y_i,x_i)\mathbf r_2(y_i,x_i)-\epsilon_1\left(\sum_{i=1}^N \mathbf r_1(y_i,x_i)\right)\left(\sum_{j=1}^N \mathbf r_2(y_j,x_j)\right).
\end{align*}
By induction, $\rho_N(\mathsf t_{\mathbf r})-\frac{1}{\epsilon_2}\sum_{i=1}^N\mathbf r(y_i,x_i)$ does not depend on the ordering of $\mathbf r$, thus 
\begin{align*}
    \rho_N(\mathsf t_{\mathbf r})-\frac{1}{\epsilon_2}\sum_{i=1}^N\mathbf r(y_i,x_i)=\rho_N(\mathsf t_{n,m})-\frac{1}{\epsilon_2}\sum_{i=1}^N\sym (y_i^nx_i^m)=0.
\end{align*}
\end{remark}

\subsection{\texorpdfstring{$\mathcal{B}$}{B}-algebra and the Yangian algebra of \texorpdfstring{$\mathfrak{gl}_K$}{gl(K)}}\label{subsec: B-algebra}
Recall that if $A=\oplus_{i\in\mathbb Z}A^i$ is a $\mathbb Z$-graded algebra with homogeneous components $A^i$, then one can define a new algebra $\mathcal B(A)$, called the $\mathcal{B}$-algebra 
\begin{align}
    \mathcal B(A)=A^0/\left(\sum_{i>0}A^i\cdot A^{-i}\right).
\end{align}
Note that if $A$ is commutative, then $\Spec \mathcal{B}(A)=(\Spec A)^{\mathbb C^{\times}}$, where $\mathbb C^{\times}$-action on $\Spec A$ is induced from grading.\\

Recall that $\mathsf A^{(K)}$ is graded by \eqref{eqn: grading on A}. The following theorem was conjectured by Costello \cite[Section 2.3]{costello2017holography}.

\begin{theorem}\label{thm:B-algebra and Yangian}
Under the grading \eqref{eqn: grading on A}, there is an algebra isomorphism
\begin{align}
    \mathcal B(\mathsf D^{(K)})\cong Y_{\epsilon_1}(\mathfrak{gl}_K)[\epsilon_2]
\end{align}
between $\mathcal B$-algebra of $\mathsf D^{(K)}$ and the Yangian algebra of $\mathfrak{gl}_K$.
\end{theorem}

Recall that $Y_{\epsilon_1}(\mathfrak{gl}_K)$ is the $\mathbb C[\epsilon_1]$-algebra generated by $\{T^a_{b;n}\:|\: 1\le a,b\le K,n\in \mathbb N\}$ with relations:
\begin{align}\label{eqn: RTT relations}
    [T^a_b(u),T^c_d(v)]=\frac{\epsilon_1}{u-v}\left(T^c_b(u)T^a_d(v)-T^c_b(v)T^a_d(u)\right),
\end{align}
where $T^a_b(u)=\delta^a_b+\epsilon_1\sum_{n\ge 0}T^a_{b;n}u^{-n-1}$. 
\begin{lemma}\label{lem: embedding of Yangian}
Let $\mathbf r_n(x,y):=(yx)^n$, then the map $T^a_{b;n}\mapsto \mathsf T_{\mathbf r_n}(E^a_b)$ generates a $\mathbb C[\epsilon_1]$-algebra embedding $Y_{\epsilon_1}(\gl_K)\hookrightarrow \mathsf D^{(K)}$. The image of this embedding is in the degree zero subalgebra under the grading \eqref{eqn: grading on A}.
\end{lemma}

\begin{proof}
The computation in \cite{moosavian2021towards} shows that the map
\begin{align}
    T^a_{b;n}\mapsto \frac{1}{\epsilon_1}I^a(YX)^nJ_b=p_N(\mathsf T_{\mathbf r_n}(E^a_b))
\end{align}
extends to a $\mathbb C[\epsilon_1]$-algebra homomorphism from $Y_{\epsilon_1}(\mathfrak{gl}_K)$ to $\mathsf B^{(K)}_N[\epsilon_1^{-1}]$. Then the lemma follows from Theorem \ref{thm: compare with Costello's DDCA} and the PBW theorem for $\mathsf D^{(K)}$ (Theorem \ref{thm: PBW}). Finally, $\mathsf T_{\mathbf r_n}(E^a_b)$ has degree zero because its image in $\mathsf B^{(K)}_N[\epsilon_1^{-1}]$ has degree zero for all $N$.
\end{proof}

\begin{lemma}\label{lem: B-algbera of polynomial-like algebra}
Let $R$ be a base ring and let $A$ be a $\mathbb Z$-graded $R$-algebra, assume that $A$ possesses a set of homogeneous elements $\mathfrak{G}(A)$ together with a total order which refines the partial order given by $\mathbb Z$-grading, such that $A$ is a free $R$-module with a PBW basis $\mathfrak{B}(A):=$non-increasing ordered monomials of elements in $\mathfrak{G}(A)$, assume moreover that the $R$-span of the subset of $\mathfrak{B}(A)$ consisting of monomials in degree zero elements $\mathfrak{G}(A)^0\subset\mathfrak{G}(A)$ is a subalgebra $B$, then the $\mathcal B$-algebra $\mathcal B(A)$ is isomorphic to $B$. 
\end{lemma}

\begin{proof}
By the assumption, the degree zero subalgebra $A^0$ has $R$-basis $\mathfrak B(A)^0:=\mathfrak{B}(A)\cap A^0$. Let $\mathfrak B(A)^0_0$ be the subset of $\mathfrak B(A)^0$ consisting of non-increasing ordered monomials in $\mathfrak{G}(A)^0$, and let $\mathfrak B(A)^0_1$ be the complement of $\mathfrak B(A)^0_0$ in $\mathfrak B(A)^0$, then elements in $\mathfrak B(A)^0_1$ are of the form: $a_1\cdots a_n$, where $\deg a_1>0$, this is because if $\deg a_1\le 0$ then all other elements $a_2,\cdots,a_n$ have non-positive degrees, so all $a_i$ to be of degree zero i.e. $a_1\cdots a_n\in \mathfrak B(A)^0_0$, a contradiction. Thus $\mathfrak B(A)^0_1$ belongs to the ideal $\sum_{i>0} A^i\cdot A^{-i}$, and the projection $B\to \mathcal B(A)$ is an $R$-module isomorphism. Since $B$ is a subalgebra of $A^0$, the projection $B\to \mathcal B(A)$ is algebra homomorphism, thus $B$ is isomorphic to $ \mathcal B(A)$.
\end{proof}

\begin{proof}[Proof of Theorem \ref{thm:B-algebra and Yangian}]
According to our previous computation \eqref{eqn: transformation to unsymmetrized basis}, the transformation between unsymmetrized generators and the symmetrized ones are triangular with respect to the filtration in Section \ref{subsec: flitration on A}, therefore $\mathsf D^{(K)}$ has a set of generators $$\tilde{\mathfrak G}(\mathsf D^{(K)}):=\{\mathsf T_{n,m}(E^a_b),\mathsf T_{\mathbf r_n}(E^a_b)\:|\: 1\le a,b\le K, (n,m)\in \mathbb N^2, n\neq m\},$$ where $\mathbf r_n(x,y)=(xy)^n$. Note that elements in $\tilde{\mathfrak G}(\mathsf D^{(K)})$ are homogeneous under the grading \eqref{eqn: grading on A} and the degree zero subset $\tilde{\mathfrak G}(\mathsf D^{(K)})^0=\{\mathsf T_{\mathbf r_n}(E^a_b)\:|\: 1\le a,b\le K, n\in \mathbb N\}$. Choose a refinement of the partial order on $\tilde{\mathfrak G}(\mathsf D^{(K)})$ given by $\mathbb Z$-grading and we get a total order $\preceq$ on $\tilde{\mathfrak G}(\mathsf D^{(K)})$. Since the proof of Theorem \ref{thm: PBW} does not depend on the choice of total order, we conclude that $\mathsf D^{(K)}$ possesses a $\mathbb C[\epsilon_1,\epsilon_2]$-basis $\tilde{\mathfrak B}(\mathsf D^{(K)}):=$non-increasing ordered monomials of elements in $\tilde{\mathfrak G}(\mathsf D^{(K)})$. By Lemma \ref{lem: embedding of Yangian}, the $\mathbb C[\epsilon_1,\epsilon_2]$-span of the subset of $\tilde{\mathfrak B}(\mathsf D^{(K)})$ consisting of monomials in degree zero elements $\tilde{\mathfrak G}(\mathsf D^{(K)})^0$ is a subalgebra which is isomorphic to $Y_{\epsilon_1}(\gl_K)[\epsilon_2]$. The assumptions in Lemma \ref{lem: B-algbera of polynomial-like algebra} are satisfied if we set $R=\mathbb C[\epsilon_1,\epsilon_2], A=\mathsf D^{(K)}, \mathfrak{G}(A)=\tilde{\mathfrak G}(\mathsf D^{(K)})$, thus $Y_{\epsilon_1}(\gl_K)[\epsilon_2]$ is isomorphic to $\mathcal B(\mathsf D^{(K)})$.
\end{proof}

\begin{remark}
The image of the Yangian $Y_{\epsilon_1}(\mathfrak{gl}_K)$ generators $T^a_{b;n}$ in the spherical Cherednik algebra $\mathrm{S}\mathcal H^{(K)}_N$ is 
\begin{align}
    \rho_N(T^a_{b;n})=\sum_{i=1}^N E^a_{b,i}(x_iy_i)^n.
\end{align}
This is compatible with the observation in \cite{bernard1993yang}.
\end{remark}

\subsection{A simpler definition of \texorpdfstring{$\mathsf A^{(1)}$}{A1}}
\begin{proposition}\label{prop: simpler definition, K=1}
$\mathsf A^{(1)}$ is generated over $\mathbb C[\epsilon_1,\epsilon_2]$ by $\{\mathsf t_{3,0},\mathsf t_{2,0},\mathsf t_{1,0},\mathsf t_{1,1},\mathsf t_{0,n}\:|\: n\in \mathbb N\}$ with relations
\begin{equation}\label{eqn: A1', K=1}
     [\mathsf t_{2,0},\mathsf t_{0,2}]=4\mathsf t_{1,1},\;[\mathsf t_{1,1},\mathsf t_{2,0}]=-2\mathsf t_{2,0},\;[\mathsf t_{1,1},\mathsf t_{0,2}]=2\mathsf t_{0,2},\tag{A$1_1$}
\end{equation}
\begin{equation}\label{eqn: A2', K=1}
\begin{split}
    [\mathsf t_{2,0},\mathsf t_{n,0}]=0,\; [\mathsf t_{0,2},\mathsf t_{0,m}]=0,&\; [\mathsf t_{1,1},\mathsf t_{0,m}]=m\mathsf t_{0,m}, \;(0\le n\le 3 \text{ and }m\ge 0)\\
    \mathrm{ad}^3_{\mathsf t_{2,0}}(\mathsf t_{0,3})=48\mathsf t_{3,0},&\; [\mathsf t_{2,0},\mathsf t_{0,1}]=2\mathsf t_{1,0},
\end{split}\tag{A$2_1$}
\end{equation}
\begin{equation}\label{eqn: A3', K=1}
\begin{split}
        [\mathsf t_{1,0},\mathsf t_{0,0}]=0,\;[\mathsf t_{1,0},\mathsf t_{0,1}]=\mathsf t_{0,0},\; [\mathsf t_{1,0},\mathsf t_{0,3}]=3\mathsf t_{0,2},
    \end{split}\tag{A$3_1$}
\end{equation}
\begin{equation}\label{eqn: A4', K=1}
\begin{split}
    \begin{split}
    [\mathsf t_{3,0},\mathsf t_{0,n}]=&\frac{3}{4n+4}\mathrm{ad}^2_{\mathsf t_{2,0}} (\mathsf t_{0,n+1})-\frac{n(n-1)(n-2)}{4}\sigma_2\mathsf t_{0,n-3}\\
&+\frac{3\sigma_3}{2}\sum_{m=0}^{n-3}(m+1)(n-2-m)\mathsf t_{0,m}\mathsf t_{0,n-3-m}, \;(n\ge 3)
    \end{split}
    \end{split}\tag{A$4_1$}
\end{equation}
where we set $\sigma_2=\epsilon_1\epsilon_2+\epsilon_2\epsilon_3+\epsilon_3\epsilon_1$ and $\sigma_3=\epsilon_1\epsilon_2\epsilon_3$.
\end{proposition}

Define $\mathsf A^{(1)}_{\mathrm{new}}$ to be the $\mathbb C[\epsilon_1,\epsilon_2]$-algebra generated by $\{\mathsf t_{3,0},\mathsf t_{2,0},\mathsf t_{1,0},\mathsf t_{1,1},\mathsf t_{0,n}\:|\: n\in \mathbb N\}$ with relations \eqref{eqn: A1', K=1}-\eqref{eqn: A4', K=1}. Then we have an obvious surjective algebra homomorphism $\mathsf A^{(1)}_{\mathrm{new}}\to \mathsf A^{(1)}$. To show that this is an isomorphism, we need to recover the relations \eqref{eqn: A1, K=1} and \eqref{eqn: A2, K=1} in the Lemma \ref{lem: relations A^1}.

As a preliminary step, we show that the $\mathsf t_{2,0}$ acts on $\mathsf A^{(1)}_{\mathrm{new}}$ locally nilpotently, namely we have the following.
\begin{lemma}\label{lem: local nilpotency, K=1}
For all $n\in \mathbb N$, the equations $\mathrm{ad}_{\mathsf t_{2,0}}^{n+1}(\mathsf t_{0,n})=0$ holds in $\mathsf A^{(1)}_{\mathrm{new}}$.
\end{lemma}
\begin{proof}
We prove it by induction on $n$. For $n\le 3$, this equation is implied by \eqref{eqn: A1', K=1} and \eqref{eqn: A2', K=1}, so we assume that $n>3$ and that the equation $\mathrm{ad}_{\mathsf t_{2,0}}^{m+1}(\mathsf t_{0,m})=0$ holds for all $m<n$.
Then we have 
\begin{align*}
    &[\mathsf t_{3,0},\mathrm{ad}^{n}_{\mathsf t_{0,2}}(\mathsf t_{0,n-1})]=\mathrm{ad}^{n}_{\mathsf t_{0,2}}([\mathsf t_{3,0},\mathsf t_{0,n-1}])=\frac{3}{4n}\mathrm{ad}^{n+2}_{\mathsf t_{2,0}} (\mathsf t_{0,n})\\
&-\frac{(n-1)(n-2)(n-3)}{4}\sigma_2\mathrm{ad}^{n}_{\mathsf t_{2,0}} (\mathsf t_{0,n-4})+\frac{3\sigma_3}{2}\sum_{m=0}^{n-4}(m+1)(n-3-m)\mathrm{ad}^{n}_{\mathsf t_{2,0}} (\mathsf t_{0,m}\mathsf t_{0,n-4-m}),
\end{align*}
and by induction assumption, the left-hand-side of the above equation is zero and the right-hand-side of the above equation equals to $\frac{3}{4n}\mathrm{ad}^{n+2}_{\mathsf t_{2,0}} (\mathsf t_{0,n})$, this shows that $\mathrm{ad}^{n+2}_{\mathsf t_{2,0}} (\mathsf t_{0,n})=0$. On the other hand, \eqref{eqn: A1', K=1} implies that $\{\mathsf t_{2,0},\mathsf t_{1,1},\mathsf t_{0,2}\}$ is an $\mathfrak{sl}_2$-triple, and \eqref{eqn: A2', K=1} implies that $\mathsf t_{0,n}$ is a highest weight vector of $\mathfrak{sl}_2$ with highest weight $n$. The nilpotency $\mathrm{ad}^{n+2}_{\mathsf t_{2,0}} (\mathsf t_{0,n})=0$ implies that the $\mathfrak{sl}_2$-action on $\mathsf t_{0,n}$ generates an irreducible representation with highest weight $n$, thus $\mathrm{ad}_{\mathsf t_{2,0}}^{n+1}(\mathsf t_{0,n})=0$.
\end{proof}

Define $\mathsf t_{n,m}:=\frac{m!}{2^n(n+m)!}\mathrm{ad}_{\mathsf t_{2,0}}^{n}(\mathsf t_{0,n+m})$, then Lemma \ref{lem: local nilpotency, K=1} implies that for $p+q=2$, the equations
\begin{equation}\label{eqn: A1, K=1, repeated}
\begin{split}
   [\mathsf t_{p,q},\mathsf t_{n,m}]&=(mp-nq)\mathsf t_{p+n-1,q+m-1}
\end{split}
\end{equation}
hold in $\mathsf A^{(1)}_{\mathrm{new}}$.

\begin{proof}[Proof of Proposition \ref{prop: simpler definition, K=1}]
It suffices to show that \eqref{eqn: A1, K=1, repeated} holds for $p+q\le 1$ as well. Let us first verify the equations $[\mathsf t_{1,0},\mathsf t_{0,n}]=n\mathsf t_{0,n-1}$, which is automatic for $n\le 3$ by \eqref{eqn: A2', K=1} and \eqref{eqn: A3', K=1}. For $n>3$, we proceed by induction on $n$, assume that $[\mathsf t_{1,0},\mathsf t_{0,m}]=m\mathsf t_{0,m-1}$ holds for all $m<n$. Applying $\mathrm{ad}_{\mathsf t_{0,2}}^{2}$ to both sides of \eqref{eqn: A4', K=1}, and we get 
\begin{align*}
[\mathsf t_{1,2},\mathsf t_{0,n-1}]=(n-1)\mathsf t_{0,n}.
\end{align*}
On the other hand, applying $\mathrm{ad}_{\mathsf t_{2,0}}$ to the equation $[\mathsf t_{1,0},\mathsf t_{0,3}]=3\mathsf t_{0,2}$ and we get
\begin{align*}
[\mathsf t_{1,0},\mathsf t_{1,2}]=2\mathsf t_{1,1}.
\end{align*}
Thus we have
\begin{align*}
[\mathsf t_{1,0},\mathsf t_{0,n}]=\frac{1}{n-1}[\mathsf t_{1,0},[\mathsf t_{1,2},\mathsf t_{0,n-1}]]=\frac{2}{n-1}[\mathsf t_{1,1},\mathsf t_{0,n-1}]+[\mathsf t_{1,2},\mathsf t_{0,n-2}]=n\mathsf t_{0,n-1}.
\end{align*}
This proves the induction step. 

Next, applying adjoint actions of $\mathsf t_{2,0}$ on both sides of the equation $[\mathsf t_{1,0},\mathsf t_{0,n}]=n\mathsf t_{0,n-1}$, and we see that \eqref{eqn: A1, K=1, repeated} holds for $(p,q)=(1,0)$. Then applying adjoint actions of $\mathsf t_{0,2}$ on both sides of the equation \eqref{eqn: A1, K=1, repeated} with $(p,q)=(1,0)$, and we see that \eqref{eqn: A1, K=1, repeated} holds for $(p,q)=(0,1)$. Finally
\begin{align*}
[\mathsf t_{0,0},\mathsf t_{n,m}]=[[\mathsf t_{1,0},\mathsf t_{0,1}],\mathsf t_{n,m}]=[[\mathsf t_{1,0},[\mathsf t_{0,1},\mathsf t_{n,m}]]-[[\mathsf t_{0,1},[\mathsf t_{1,0},\mathsf t_{n,m}]]=0.
\end{align*}
We have verified all relations in \eqref{eqn: A1, K=1}, therefore the surjective algebra homomorphism $\mathsf A^{(1)}_{\mathrm{new}}\to \mathsf A^{(1)}$ admits a section $\mathsf A^{(1)}\to \mathsf A^{(1)}_{\mathrm{new}}$ and it is obvious from the construction that these two maps are inverse to each other.
\end{proof}

\subsection{A simpler definition of \texorpdfstring{$\mathsf A^{(K)}$}{Ak}, \texorpdfstring{$K>1$}{K>1}}
In this subsection, we continue using the convention $\epsilon_3:=-K\epsilon_1-\epsilon_2$.
\begin{theorem}\label{thm: simpler definition, K>1}
If $K>1$, then $\mathsf A^{(K)}[\epsilon_2^{-1},\epsilon_3^{-1}]$ is generated over $\mathbb C[\epsilon_1,\epsilon_2^{\pm},\epsilon_3^{\pm}]$ by $\{\mathsf T_{1,0}(X),\mathsf T_{0,n}(X)\:|\: X\in \mathfrak{gl}_K, n\in \mathbb N\}$ and $\{\mathsf t_{2,0},\mathsf t_{1,1},\mathsf t_{0,2}\}$ with relations
\begin{equation}\label{eqn: A0'}
\begin{split}
     \mathsf T_{0,2}(1)=&\epsilon_2 \mathsf t_{0,2},\\
    \mathsf T_{n,m}(aX+bY)=a\mathsf T_{n,m}(X)+b\mathsf T_{n,m}(Y),&\forall (a,b)\in \mathbb C^2, \forall (n,m)\in (1,0)\sqcup (0,\mathbb N),
\end{split}\tag{A0'}
\end{equation}
\begin{equation}\label{eqn: A1'}
    [\mathsf T_{0,m}(X),\mathsf T_{0,n}(Y)]=\mathsf T_{0,m+n}([X,Y]),\tag{A1'}
\end{equation}
and for $(p,q)=(2,0)$ or $(0,2)$, and for all $(m,n)\in \mathbb N^2$ such that $m+n=2$, and for all $(r,s)\in \mathbb N^2$ such that $r+s\le 1$,
\begin{equation}\label{eqn: A2'}
\begin{split}
   [\mathsf t_{p,q},\mathsf t_{n,m}]&=(mp-nq)\mathsf t_{p+n-1,q+m-1},\\
    [\mathsf t_{p,q},\mathsf T_{r,s}(X)]&=(ps-rq)T_{p+r-1,q+s-1}(X).
\end{split}\tag{A2'}
\end{equation}
Use the notation $\mathsf T_{u,r,t,s}(X\otimes Y):=\mathsf T_{u,r}(X)\mathsf T_{t,s}(Y)$ for $X,Y\in \mathfrak{gl}_K$, and $\Omega:=E^a_b\otimes E^b_a\in \mathfrak{gl}_K^{\otimes 2}$, then
\begin{equation}\label{eqn: A3'}
\begin{split}
        [\mathsf T_{1,0}(X),\mathsf T_{0,n}(Y)]&=\frac{1}{2n+2}[\mathsf t_{2,0},\mathsf T_{0,n+1}([X,Y])]-\frac{\epsilon_3 n}{2}\mathsf T_{0,n-1}(\{X,Y\})-n\epsilon_1\mathrm{tr}(Y) \mathsf T_{0,n-1}(X)\\
&+\epsilon_1\sum_{m=0}^{n-1}\frac{m+1}{n+1}\mathsf T_{0,m,0,n-1-m}(([X,Y]\otimes 1)\cdot \Omega)\\
&+\epsilon_1 \sum_{m=0}^{n-1}\mathsf T_{0,m,0,n-1-m}((X\otimes Y-XY\otimes 1)\cdot \Omega)
    \end{split}\tag{A3'}
\end{equation}
\end{theorem}

Define $\mathsf A^{(K)}_{\mathrm{new}}$ to be the $\mathbb C[\epsilon_1,\epsilon_2^{\pm},\epsilon_3^{\pm}]$-algebra generated by $\{\mathsf T_{1,0}(X),\mathsf T_{0,n}(X)\:|\: X\in \mathfrak{gl}_K, n\in \mathbb N\}$ and $\{\mathsf t_{2,0},\mathsf t_{1,1},\mathsf t_{0,2}\}$ with relations \eqref{eqn: A0'}-\eqref{eqn: A3'}. Then we have an obvious surjective algebra homomorphism $\mathsf A^{(K)}_{\mathrm{new}}\to \mathsf A^{(K)}[\epsilon_2^{-1},\epsilon_3^{-1}]$ and we will show that this is an isomorphism. As a preliminary step, we show that the $\mathsf t_{2,0}$ and $\mathsf t_{0,2}$ act on $\mathsf A^{(K)}_{\mathrm{new}}$ locally nilpotently, namely we have the following.
\begin{lemma}\label{lem: local nilpotency}
For all $n\in \mathbb N$ and for all $X\in \mathfrak{gl}_K$, 
\begin{align}
    [\mathsf t_{0,2},\mathsf T_{0,n}(X)]=0,\quad [\mathsf t_{1,1},\mathsf T_{0,n}(X)]=n\mathsf T_{0,n}(X),\quad\mathrm{ad}_{\mathsf t_{2,0}}^{n+1}(\mathsf T_{0,n}(X))=0.
\end{align}
\end{lemma}

\begin{proof}
By \eqref{eqn: A1'} and \eqref{eqn: A2'}, it is straightforward to see that for all $n\in \mathbb N$ and for all $X\in \mathfrak{sl}_K$, $[\mathsf t_{0,2},\mathsf T_{0,n}(X)]=0$ and $[\mathsf t_{1,1},\mathsf T_{0,n}(X)]=n\mathsf T_{0,n}(X)$ and $\mathrm{ad}_{\mathsf t_{2,0}}^{n+1}(\mathsf T_{0,n}(X))=0$. It remains to show that $[\mathsf t_{0,2},\mathsf T_{0,n}(1)]=0$ and $[\mathsf t_{1,1},\mathsf T_{0,n}(1)]=n\mathsf T_{0,n}(1)$ and $\mathrm{ad}_{\mathsf t_{2,0}}^{n+1}(\mathsf T_{0,n}(1))=0$ for all $n\in \mathbb N_{>1}$ (since these are known for $n=0,1$).

Set $X=Y=H_1:=E^1_1-E^2_2$ in \eqref{eqn: A3'}, and we get
\begin{align*}
    [\mathsf T_{1,0}(H_1),\mathsf T_{0,n+1}(H_1)]&=-\epsilon_3 n\mathsf T_{0,n}((H_1)^2)+\epsilon_1 \sum_{m=0}^{n}\mathsf T_{0,m,0,n-m}((H_1\otimes H_1-(H_1)^2\otimes 1)\cdot \Omega).
\end{align*}
Simple computation shows that
\begin{align*}
    (H_1\otimes H_1-(H_1)^2\otimes 1)\cdot \Omega=-(E^1_2\otimes E^2_1+E^2_1\otimes E^1_2)-\sum_{i=1}^2\sum_{j\neq i}^KE^i_j\otimes E^j_i,
\end{align*}
therefore $\mathsf T_{m,n-m}((H_1\otimes H_1-(H_1)^2\otimes 1)\cdot \Omega)$ satisfies
\begin{align}\label{eqn: local nilpotency}
    [\mathsf t_{0,2},\mathcal O]=0,\quad [\mathsf t_{1,1},\mathcal O]=n\mathcal O,\quad\mathrm{ad}_{\mathsf t_{2,0}}^{n+2}(\mathcal O)=0,
\end{align}
where $\mathcal O=\mathsf T_{m,n-m}((H_1\otimes H_1-(H_1)^2\otimes 1)\cdot \Omega)$. On the other hand, \eqref{eqn: local nilpotency} also holds for $\mathcal O=[\mathsf T_{1,0}(H_1),\mathsf T_{0,n+1}(H_1)]$, thus \eqref{eqn: local nilpotency} holds for $\mathcal O=\mathsf T_{0,n}((H_1)^2)$, since $\epsilon_3$ is invertible. This implies that \eqref{eqn: local nilpotency} holds for $\mathcal O=\mathsf T_{0,n}(1)$. Finally, it follows from \eqref{eqn: local nilpotency} for $\mathcal O=\mathsf T_{0,n}(1)$ that the action of the $\mathfrak{sl}_2$-triple $\{\mathsf t_{2,0},\mathsf t_{1,1},\mathsf t_{0,2}\}$ on $\mathsf T_{0,n}(1)$ generates a finite dimensional cyclic representation (therefore irreducible) with $\mathsf T_{0,n}(1)$ being the highest weight vector of weight $n$, hence $\mathrm{ad}_{\mathsf t_{2,0}}^{n+1}(\mathsf T_{0,n}(1))=0$.
\end{proof}

Define $\mathsf T_{n,m}(X):=\frac{m!}{2^n(n+m)!}\mathrm{ad}_{\mathsf t_{2,0}}^{n}(\mathsf T_{0,n+m}(X))$ and define $\mathsf t_{n,m}:=\frac{1}{\epsilon_2}\mathsf T_{n,m}(1)$, then Lemma \ref{lem: local nilpotency} implies that 
for $(p,q)\in \mathbb N^2$ such that $p+q=2$, and for all $(m,n)\in \mathbb N^2$,
\begin{equation}\label{eqn: implies relation A2}
\begin{split}
   [\mathsf t_{p,q},\mathsf t_{n,m}]&=(mp-nq)\mathsf t_{p+n-1,q+m-1},\\
    [\mathsf t_{p,q},\mathsf T_{n,m}(X)]&=(mp-nq)T_{p+r-1,q+s-1}(X).
\end{split}
\end{equation}
\begin{lemma}\label{lemma: implies relation A2}
    \eqref{eqn: implies relation A2} also holds for $(p,q)\in \mathbb N^2$ such that $p+q\le 1$.
\end{lemma}
\begin{proof}
By \eqref{eqn: A1'}, $\mathsf t_{0,0}=\frac{1}{\epsilon_2}\mathsf T_{0,0}(1)$ commutes with all $\mathsf T_{0,m}(X)$. Since $[\mathsf t_{0,0},\mathsf t_{2,0}]=0$ by \eqref{eqn: A2'}, we see that $\mathsf t_{0,0}$ is central. Next, we set $X=1$ in \eqref{eqn: A3'} and get 
\begin{align*}
    [\mathsf T_{1,0}(1),\mathsf T_{0,n}(Y)]&=-\epsilon_3 n\mathsf T_{0,n-1}(Y)-n\epsilon_1\mathrm{tr}(Y) \mathsf T_{0,n-1}(1)\\
    &+\epsilon_1 \sum_{m=0}^{n-1}\mathsf T_{0,m,0,n-1-m}((1\otimes Y-Y\otimes 1)\cdot \Omega).
\end{align*}
Using \eqref{eqn: A1'} we expand 
\begin{align*}
    \sum_{m=0}^{n-1}\mathsf T_{0,m,0,n-1-m}((1\otimes Y-Y\otimes 1)\cdot \Omega)=n\epsilon_1 \mathrm{tr}(Y)\mathsf T_{0,n-1}(1)-Kn\epsilon_1 \mathsf T_{0,n-1}(Y),
\end{align*}
thus $[\mathsf t_{1,0},\mathsf T_{0,n}(Y)]=n \mathsf T_{0,n-1}(Y)$. Using \eqref{eqn: A2'} again, we conclude that \eqref{eqn: implies relation A2} holds for $(p,q)=(1,0)$ and $(0,1)$. 
\end{proof}

We define increasing filtration $0=G_{-1}\mathsf A^{(K)}_{\mathrm{new}}\subset G_0\mathsf A^{(K)}_{\mathrm{new}}\subset G_1\mathsf A^{(K)}_{\mathrm{new}}\cdots$ which is \textit{different} from the one in subsection \ref{subsec: flitration on A}. Define the degree on generators as $\deg \epsilon_1=\deg \epsilon_2=0$ and
\begin{equation}
    \begin{split}   
    \text{for }X\in \mathfrak{sl}_K,\;\deg \mathsf T_{n,m}(X)=n+m,\text{ and }
    \deg \mathsf T_{n,m}(1)=n+m+1,
    \end{split}
\end{equation}
this gives rise to a grading on the tensor algebra $\mathbb C[\epsilon_1,\epsilon_2^{\pm},\epsilon_3^{\pm}]\langle \mathsf T_{n,m}(X), \mathsf T_{n,m}(1)\:|\: X\in \mathfrak{sl}_K,(n,m)\in \mathbb N^2\rangle$. We define $G_{i}\mathsf A^{(K)}_{\mathrm{new}}$ to be the image of the span of homogeneous elements in the tensor algebra of degrees $\le i$.

Recall that we choose a basis $\mathfrak B:=\{X_1,\cdots,X_{K^2-1}\}$ of $\mathfrak{sl}_K$, so that $\mathfrak{B}_+:=\{1\}\cup \mathfrak B$ is a basis of $\mathfrak{gl}_K$.

\begin{proposition}\label{prop: filtration_new}
For all $(n,m,p,q)\in \mathbb N^4$ and $X,Y\in \mathfrak B$, there exists 
\begin{equation}
    \begin{split}
        f_{n,m,p,q}^{X,Y}&\in G_{n+m+p+q-1}\mathbb C[\epsilon_1,\epsilon_2^{\pm},\epsilon_3^{\pm}]\langle \mathsf T_{i,j}(Z)\:|\: Z\in \mathfrak{gl}_K, (i,j)\in \mathbb N^2\rangle,\\
        g_{n,m,p,q}^{X}&\in G_{n+m+p+q-3}\mathbb C[\epsilon_1,\epsilon_2^{\pm},\epsilon_3^{\pm}]\langle \mathsf T_{i,j}(Z)\:|\: Z\in \mathfrak{gl}_K, (i,j)\in \mathbb N^2\rangle,\\
        h_{n,m,p,q}&\in G_{n+m+p+q-1}\mathbb C[\epsilon_1,\epsilon_2^{\pm},\epsilon_3^{\pm}]\langle \mathsf T_{i,j}(Z)\:|\: Z\in \mathfrak{gl}_K, (i,j)\in \mathbb N^2\rangle,
    \end{split}
\end{equation}
such that following equations hold in $\mathsf A^{(K)}_{\mathrm{new}}$
\begin{align}\label{eqn_schematic commutators 1_new}
    [\mathsf T_{n,m}(X),\mathsf T_{p,q}(Y)]=\mathsf T_{n+p,m+q}([X,Y])+\bar f_{n,m,p,q}^{X,Y},
\end{align}
\begin{align}\label{eqn_schematic commutators 2_new}
    [\mathsf T_{n,m}(X),\mathsf t_{p,q}]=(nq-mp)\mathsf T_{n+p-1,m+q-1}(X)+\bar g_{n,m,p,q}^{X},
\end{align}
\begin{align}\label{eqn_schematic commutators 3_new}
    [\mathsf t_{n,m},\mathsf t_{p,q}]=\bar h_{n,m,p,q},
\end{align}
where $\bar f_{n,m,p,q}^{X,Y}$ (resp. $\bar g_{n,m,p,q}^{X}$, resp. $\bar h_{n,m,p,q}$) is the image of $f_{n,m,p,q}^{X,Y}$ (resp. $g_{n,m,p,q}^{X}$, resp. $h_{n,m,p,q}$) in $\mathsf A^{(K)}_{\mathrm{new}}$.
\end{proposition}

\begin{proof}
We construct $f_{n,m,p,q}^{X,Y}$, $g_{n,m,p,q}^{X}$ and $h_{n,m,p,q}$ inductively. First of all, we set $f^{X,Y}_{0,0,0,n}=0$ and set $g^{X}_{n,m,p,q}=0$ for $n+m\le 1$, using \eqref{eqn: implies relation A2} and Lemma \ref{lemma: implies relation A2}. For $f^{X,Y}_{1,0,0,n}$, we notice that the term $$\sum_{m=0}^{n-1}\frac{m+1}{n+1}\mathsf T_{0,m,0,n-1-m}(([X,Y]\otimes 1)\cdot \Omega)+\sum_{m=0}^{n-1}\mathsf T_{0,m,0,n-1-m}((X\otimes Y-XY\otimes 1)\cdot \Omega)$$ can be written as sum over quadratic monomials in $\{\mathsf T_{0,i}(Z)\:|\: Z\in \mathfrak{sl}_K,i\in \mathbb N\}$ with total degree $n-1$, thus we set $\mathsf T_{1,n}([X,Y])+f^{X,Y}_{1,0,0,n}$ to be the RHS of \eqref{eqn: A3'}. 

Next we set $f^{X,Y}_{1,0,i+1,n-i-1}$ to be the lift of $\frac{1}{2n-2i}[\mathsf t_{2,0},\bar f^{X,Y}_{1,0,i,n-i}]$, inductively for all $i<n$. Then we set $f^{X,Y}_{0,1,p,q}$ to be the lift of $\frac{1}{2}[\bar f^{X,Y}_{1,0,p,q},\mathsf t_{0,2}]-p\bar f^{X,Y}_{1,0,p-1,q+1}$.

Assume that $f^{X,Y}_{n,m,p,q}$ and $g^{X}_{n,m,p,q}$ have been constructed for all $(n+m)\le s$ and all $X,Y\in \mathfrak B$. Then for every $X\in \mathfrak B$ we fix $Z_1,Z_2\in \mathfrak{sl}_K$ such that $[Z_1,Z_2]=X$, so 
\begin{align*}
    [\mathsf T_{0,1}(Z_1),\mathsf T_{0,s}(Z_2)]=\mathsf T_{0,s+1}(X).
\end{align*}
Using the identity
\begin{align*}
    &[\mathsf T_{0,s+1}(X),\mathsf T_{p,q}(Y)]=[[\mathsf T_{0,1}(Z_1),\mathsf T_{0,s}(Z_2)],\mathsf T_{p,q}(Y)]\\
    =&-[\mathsf T_{0,s}(Z_2),[\mathsf T_{0,1}(Z_1),\mathsf T_{p,q}(Y)]]+[\mathsf T_{0,1}(Z_1),[\mathsf T_{0,s}(Z_2),\mathsf T_{p,q}(Y)]]\\
    =&\mathsf T_{p,q+s+1}([X,Y])-\bar f^{Z_2,[Z_1,Y]}_{0,s,p,q+1}+\bar f^{Z_1,[Z_2,Y]}_{0,1,p,q+s}-[\mathsf T_{0,s}(Z_2),\bar f^{Z_1,Y}_{0,1,p,q}]+[\mathsf T_{0,1}(Z_1),\bar f^{Z_2,Y}_{0,s,p,q}]
\end{align*}
we set $f^{X,Y}_{0,s+1,p,q}$ to be the lift of $-\bar f^{Z_2,[Z_1,Y]}_{0,s,p,q+1}+\bar f^{Z_1,[Z_2,Y]}_{0,1,p,q+s}-[\mathsf T_{0,s}(Z_2),\bar f^{Z_1,Y}_{0,1,p,q}]+[\mathsf T_{0,1}(Z_1),\bar f^{Z_2,Y}_{0,s,p,q}]$. Note that we expand $[\mathsf T_{0,s}(Z_2),\bar f^{Z_1,Y}_{0,1,p,q}]$ using \eqref{eqn_schematic commutators 1_new} and $f^{Z_2,W}_{0,s,r,t}$ and $g^{Z_2}_{0,s,r,t}$ for all $W\in \mathfrak{B}_+$ and all $(r,t)\in \mathbb N^2$, and expand $[\mathsf T_{0,1}(Z_1),\bar f^{Z_2,Y}_{0,s,p,q}]$ similarly. By induction hypothesis, we have
\begin{align*}
    \deg f^{X,Y}_{0,s+1,p,q}&\le \max\{\deg f^{Z_2,[Z_1,Y]}_{0,s,p,q+1},\deg f^{Z_1,[Z_2,Y]}_{0,1,p,q+s}, s+\deg f^{Z_1,Y}_{0,1,p,q},1+\deg  f^{Z_2,Y}_{0,s,p,q}\}\\
    &\le p+q+s-1.
\end{align*}
Next, we set $f^{X,Y}_{i+1,s-i,p,q}$ to be the lift of $\frac{1}{2(s+1-i)}[\mathsf t_{2,0},\bar f^{X,Y}_{i,s+1-i,p,q}]-\frac{q}{s+1-i}\bar f^{X,Y}_{i,s+1-i,p+1,q-1}$, inductively for all $0\le i\le s$.

Similarly, using the identity
\begin{align*}
    &[\mathsf T_{0,s+1}(X),\mathsf t_{p,q}]=[[\mathsf T_{0,1}(Z_1),\mathsf T_{0,s}(Z_2)],\mathsf t_{p,q}]\\
    =&-[\mathsf T_{0,s}(Z_2),[\mathsf T_{0,1}(Z_1),\mathsf t_{p,q}]]+[\mathsf T_{0,1}(Z_1),[\mathsf T_{0,s}(Z_2),\mathsf t_{p,q}]]\\
    =&-(ps+p)\mathsf T_{p-1,q+s}(X)+p\bar f^{Z_2,Z_1}_{0,s,p-1,q}-ps\bar f^{Z_1,Z_2}_{0,1,p-1,q+s-1}-[\mathsf T_{0,s}(Z_2),\bar g^{Z_1}_{0,1,p,q}]+[\mathsf T_{0,1}(Z_1),\bar g^{Z_2}_{0,s,p,q}]
\end{align*}
we set $g^{X}_{0,s+1,p,q}$ to be the lift of $p\bar f^{Z_2,Z_1}_{0,s,p-1,q}-ps\bar f^{Z_1,Z_2}_{0,1,p-1,q+s-1}-[\mathsf T_{0,s}(Z_2),\bar g^{Z_1}_{0,1,p,q}]+[\mathsf T_{0,1}(Z_1),\bar g^{Z_2}_{0,s,p,q}]$. By induction hypothesis, we have
\begin{align*}
    \deg g^{X}_{0,s+1,p,q}&\le \max\{\deg f^{Z_2,Z_1}_{0,s,p-1,q},\deg f^{Z_1,Z_2}_{0,1,p-1,q+s-1}, s+\deg g^{Z_1}_{0,1,p,q},1+\deg  g^{Z_2}_{0,s,p,q}\}\\
    &\le p+q+s-2.
\end{align*}
Next, we set $g^{X}_{i+1,s-i,p,q}$ to be the lift of $\frac{1}{2(s+1-i)}[\mathsf t_{2,0},\bar g^{X}_{i,s+1-i,p,q}]-\frac{q}{s+1-i}\bar g^{X}_{i,s+1-i,p+1,q-1}$, inductively for all $0\le i\le s$. 

The above steps conclude the construction of $f^{X,Y}_{n,m,p,q}$ and $g^{X}_{n,m,p,q}$ for all $X,Y\in \mathfrak B$ and all $(n,m,p,q)\in \mathbb N^4$. It remains to construct $h_{n,m,p,q}$.

Set $H_{ij}:=E^i_i-E^j_j$, then we get
\begin{align*}
    \sum_{i<j}^K[\mathsf T_{1,0}(H_{ij}),\mathsf T_{0,n+1}(H_{ij})]&=-(K-1)\epsilon_2\epsilon_3 n\mathsf t_{0,n}-\epsilon_1 K\sum_{i\neq j}^K\sum_{m=0}^{n}\mathsf T_{0,m}(E^i_j)\mathsf T_{0,n-m}(E^j_i).
\end{align*}
So we set $-(K-1)\epsilon_2\epsilon_3 nh_{0,n,p,q}$ to be the lift of 
\begin{align*}
    &\sum_{i<j}^K\left([\bar g^{H_{ij}}_{1,0,p,q},\mathsf T_{0,n+1}(H_{ij})]+q\bar f^{H_{ij},H_{ij}}_{p,q-1,0,n+1}+[\mathsf T_{1,0}(H_{ij}),\bar g^{H_{ij}}_{0,n+1,p,q}]-p(n+1)\bar f^{H_{ij},H_{ij}}_{1,0,p-1,n+q}\right)+\\
    +&\epsilon_1 K \sum_{i\neq j}^K\sum_{m=0}^{n}\{\bar g^{E^i_j}_{0,m,p,q}-mp\mathsf T_{p-1,m+q-1}(E^i_j),\mathsf T_{0,n-m}(E^j_i)\}.
\end{align*}
The degree of $h_{0,n,p,q}$ is bounded above by the maximum of
\begin{align*}
    \deg g^{H_{ij}}_{1,0,p,q}+n+1,\;\deg f^{H_{ij},H_{ij}}_{p,q-1,0,n+1}&,\; 1+\deg g^{H_{ij}}_{0,n+1,p,q},\;\deg  f^{H_{ij},H_{ij}}_{1,0,p-1,n+q},\\
    n-m+\deg g^{E^i_j}_{0,m,p,q}&,\; p+q+n-2,
\end{align*}
which is bounded above by $p+q+n-1$. Finally, we set $h_{i+1,n-1-i,p,q}$ to be the lift of $\frac{1}{2(n-i)}[\mathsf t_{2,0},\bar h_{i,n-i,p,q}]-\frac{q}{n-i}\bar h_{i,n-i,p+1,q-1}$, inductively for all $0\le i< n$. 
\end{proof}

\begin{proof}[Proof of Theorem \ref{thm: simpler definition, K>1}]
We fix a total order $1<X_1<\cdots<X_{K^2-1}$ on $\mathfrak B_+$. Then we put the dictionary order on the set $\mathfrak{G}(\mathsf A^{(K)}):=\{\mathsf T_{n,m}(X)\:|\: X\in \mathfrak B_+, (n,m)\in \mathbb N^2\}$, in other words $\mathsf T_{n,m}(X)<\mathsf T_{n',m'}(X')$ if and only only if $n<n'$ or $n=n'$ and $m<m'$ or $(n,m)=(n',m')$ and $X<X'$. Define the set of ordered monomials in $\mathfrak{G}(\mathsf A^{(K)}$ as 
\begin{align}
    \mathfrak B(\mathsf A^{(K)}):=\{1\}\cup\{\mathcal O_1\cdots\mathcal O_n\:|\: n\in \mathbb N_{>0}, \mathcal O_1\le \cdots\le\mathcal O_n\in \mathfrak{G}(\mathsf A^{(K)})\}.
\end{align}
We claim that natural map $\mathbb C[\epsilon_1,\epsilon_2^{\pm},\epsilon_3^{\pm}]\cdot \mathfrak B(\mathsf A^{(K)})\to \mathsf A^{(K)}_{\mathrm{new}}$ is surjective. Obviously $G_0\mathsf A^{(K)}_{\mathrm{new}}$ is generated by polynomials in $\{\mathsf T_{0,0}(X)\:|\: X\in \mathfrak{sl}_K\}$, which is contained in the image of $\mathbb C[\epsilon_1,\epsilon_2^{\pm},\epsilon_3^{\pm}]\cdot \mathfrak B(\mathsf A^{(K)})$ by the PBW theorem for $\mathfrak{sl}_K$. Suppose that $G_s\mathsf A^{(K)}_{\mathrm{new}}$ is generated by elements in $\mathfrak B(\mathsf A^{(K)})$. Consider another natural filtration $W_{\bullet}\mathsf A^{(K)}_{\mathrm{new}}$ such that $W_{d}\mathsf A^{(K)}_{\mathrm{new}}$ is spanned by monomials consisting of $\le d$ elements in $\mathfrak{G}(\mathsf A^{(K)})$. Then $W_{1}\mathsf A^{(K)}_{\mathrm{new}}$ is contained in the image of $\mathbb C[\epsilon_1,\epsilon_2^{\pm},\epsilon_3^{\pm}]\cdot \mathfrak B(\mathsf A^{(K)})$ by definition. Suppose that $G_{s+1}\mathsf A^{(K)}_{\mathrm{new}}\cap W_{d}\mathsf A^{(K)}_{\mathrm{new}}$ is generated by elements in $\mathfrak B(\mathsf A^{(K)})$, then Proposition \ref{prop: filtration_new} implies that we can reorder any monomials in $G_{s+1}\mathsf A^{(K)}_{\mathrm{new}}\cap W_{d+1}\mathsf A^{(K)}_{\mathrm{new}}$ into the non-decreasing order modulo terms in $G_{s}\mathsf A^{(K)}+G_{s+1}\mathsf A^{(K)}_{\mathrm{new}}\cap W_{d}\mathsf A^{(K)}_{\mathrm{new}}$, therefore $G_{s+1}\mathsf A^{(K)}_{\mathrm{new}}\cap W_{d+1}\mathsf A^{(K)}_{\mathrm{new}}$ is generated by elements in $\mathfrak B(\mathsf A^{(K)})$. $G_{\bullet}\mathsf A^{(K)}_{\mathrm{new}}$ and $W_{\bullet}\mathsf A^{(K)}_{\mathrm{new}}$ are obviously exhaustive, thus $\mathsf A^{(K)}_{\mathrm{new}}$ is generated by $\mathfrak B(\mathsf A^{(K)})$.

Finally, the composition $\mathbb C[\epsilon_1,\epsilon_2^{\pm},\epsilon_3^{\pm}]\cdot \mathfrak B(\mathsf A^{(K)})\to \mathsf A^{(K)}_{\mathrm{new}}\to \mathsf A^{(K)}[\epsilon_2^{-1},\epsilon_3^{-1}]$ is isomorphism beause of the PBW theorem for $\mathsf A^{(K)}$ (Theorem \ref{thm: PBW}), thus the natural map $\mathsf A^{(K)}_{\mathrm{new}}\to \mathsf A^{(K)}[\epsilon_2^{-1},\epsilon_3^{-1}]$ must be an isomorphism.
\end{proof}

\section{A map from \texorpdfstring{$\mathsf A^{(K)}$}{Ak} to the Mode Algebras of Rectangular W-Algebras}\label{sec: map A^K to W}

In this section, we construct a map from $\mathsf A^{(K)}$ to the mode algebra of rectangular W-algebra $\mathcal W^{\kappa}(\gl_{KL},(L^K))$, later in section \ref{subsec: map from affine Yangian to Y} we will compare our map with the map from the affine Yangian of $\mathfrak{sl}_K$ to the mode algebra of $\mathcal W^{\kappa}(\gl_{KL},(L^K))$ obtained by Kodera-Ueda in \cite{kodera2022coproduct}, see also \cite{ueda2022affine}.

To begin with, let us briefly recall the definition of rectangular W-algebra \cite{prochazka2015exploring,eberhardt2019matrix,kodera2022coproduct,gaiotto2019vertex,arakawa2017explicit}.

Let $\alpha$ be a variable, and define a symmetric bilinear form $\kappa_{\alpha}$ on $\mathfrak{gl}_K=\mathfrak{sl}_K\oplus \mathfrak{z}_K$ as follows
\begin{equation}\label{eqn: inner form kappa(alpha)}
    \kappa_{\alpha}(X,Y)=\alpha\mathrm{Tr}(XY)+\mathrm{Tr}(X)\mathrm{Tr}(Y)=\begin{cases}
    \alpha\mathrm{Tr}(XY),& \text{if }X\in \mathfrak{sl}_K,\\
    (\alpha+K)\mathrm{Tr}(XY),& \text{if }X\in \mathfrak{z}_K.
    \end{cases}
\end{equation}
Explicitly 
\begin{align}
    \kappa_{\alpha}(J^a_b,J^c_d)=\alpha \delta^a_d\delta^c_b+\delta^a_b\delta^c_d.
\end{align}
We define the affine Lie algebra $\widehat{\mathfrak{gl}}_K^{\alpha}$ as the Lie algebra $\mathfrak{gl}_K[t,t^{-1}]\oplus \mathbb C\cdot \mathbf 1$ with commutation relation
\begin{align*}
    [X\otimes t^n,Y\otimes t^m]=[X,Y]\otimes t^{m+n}+n\delta_{n,-m} \kappa_{\alpha}(X,Y)\mathbf 1.
\end{align*}
We use $X_n$ to denote the element $X\otimes t^n$. It is a direct sum of the affine Lie algebra $\widehat{\mathfrak{sl}}_K$ of level $\alpha$ and the Heisenberg Lie algebra $\widehat{\mathfrak{z}}_K$ of level $\alpha+K$. Note that $\widehat{\mathfrak{gl}}_K^{\alpha}$ is usually denoted by $\mathfrak{gl}(K)_{\alpha,1}$ in the literature. 

We define $\mathfrak{U}(\widehat{\mathfrak{gl}}_K^{\alpha})$ as the current (mode) algebra of the vertex algebra $V^{\kappa_\alpha}({\mathfrak{gl}}_K)$. It is the completion of the universal enveloping algebra of $\widehat{\mathfrak{gl}}_K^{\alpha}$ in a certain topology.

$\mathfrak{U}(\widehat{\mathfrak{gl}}_K^{\alpha})$ is the basic building block of the rectangular W-algebra
\begin{align*}
    \mathcal W^{(K)}_L:=\mathcal W^{\kappa}(\mathfrak{gl}_{LK},f_L),
\end{align*}
defined as the quantum Drinfeld-Sokolov reduction of $V^{\kappa}(\mathfrak{gl}_{LK})$ with respect to the unipotent element $f_L$ associated to the partition $(L^K)$. Here the inner product $\kappa$ on $\mathfrak{gl}_{LK}$ is defined as 
\begin{equation}
    \kappa(X,Y)=\begin{cases}
    (\alpha+K-KL)\mathrm{Tr}(XY),& \text{if }X\in \mathfrak{sl}_{KL},\\
    (\alpha+K)\mathrm{Tr}(XY),& \text{if }X\in \mathfrak{z}_{KL}.
    \end{cases}
\end{equation}
Note that $\mathcal W^{(K)}_1$ is $V^{\kappa_{\alpha}}(\mathfrak{gl}_K)$ by definition. We define $\mathfrak{U}(\mathcal W^{(K)}_L)$ as the current (mode) algebra of $\mathcal W^{(K)}_L$.

$\mathcal W^{(K)}_L$ can be characterized using the Miura map \cite{arakawa2017explicit}
\begin{align*}
    \mathcal W^{(K)}_L\hookrightarrow V^{\kappa_{\alpha}}(\mathfrak{gl}_K)^{\otimes L}.
\end{align*}
We recall the construction here. Consider $L$ copies of affine Kac-Moody vertex algebra $V^{\kappa_{\alpha}}(\mathfrak{gl}_K)$, and let $J^{[i]}(z)=\left(J^{a[i]}_b(z)\right)_{1\le a,b\le K}$ be the $K\times K$ matrix whose $(a,b)$ entry is the field $J^{a[i]}_b(z)$ of the $i$-th copy of $V^{\kappa_{\alpha}}(\mathfrak{gl}_K)$. Consider the matrix valued differential operator (Miura operator)
\begin{align}\label{eqn: Miura operator}
    (\alpha\partial-J^{[1]}(z))(\alpha\partial-J^{[2]}(z))\cdots(\alpha\partial-J^{[L]}(z))=(\alpha\partial)^L+\sum_{r=1}^L (-1)^r(\alpha\partial)^{L-r}W^{(r)}(z).
\end{align}
Then $\mathcal W^{(K)}_L$ is the vertex subalgebra of $V^{\kappa_{\alpha}}(\mathfrak{gl}_K)^{\otimes L}$ generated by fields $W^{(1)},W^{(2)},\cdots,W^{(L)}$.

We write down the explicit form of the spin $1,2,3$ fields here:
\begin{equation}
    \begin{split}
        W^{a(1)}_b(z)&=\sum_{i=1}^L J^{a[i]}_b(z),\\
        W^{a(2)}_b(z)&=\sum_{i<j}^L J^{a[i]}_c(z)J^{c[j]}_b(z)+\alpha \sum_{i=1}^L(L-i)\partial J^{a[i]}_b(z),\\
        W^{a(3)}_b(z)&=\sum_{i<j<k}^L J^{a[i]}_c(z)J^{c[j]}_d(z)J^{d[k]}_b(z)+\frac{\alpha^2}{2}\sum_{i=1}^{L}(L-i)(L-i-1)\partial^2 J^{a[i]}_b(z)\\
        &~+\alpha \sum_{i<j}^L \left((L-i-1)\partial J^{a[i]}_c(z)J^{c[j]}_b(z)+(L-j)J^{a[i]}_c(z)\partial J^{c[j]}_b(z)\right).
    \end{split}
\end{equation}
We present the $W^{(1)}W^{(n)}$ OPE here:
\begin{equation}\label{W1Wn OPE_finite L}
    \begin{split}
        W^{a(1)}_b(z)W^{c(n)}_{d}(w)\sim \:& \sum_{i=0}^{n-1}\frac{(L-i)!}{(L-n)!}\frac{\alpha^{n-1-i}}{(z-w)^{n+1-i}}(\alpha\delta^c_bW^{a(i)}_{d}(w)+\delta^a_b W^{c(i)}_{d}(w))\\
&+\frac{\delta^c_bW^{a(n+1)}_{d}(w)-\delta^a_dW^{c(n+1)}_{b}(w)}{z-w}.
    \end{split}
\end{equation}
For the proof, see Appendix \ref{app: W1Wn OPE}.

\bigskip It is shown in \cite{ueda2022affine} that if $\alpha\neq 0$ then $\mathfrak{U}(\mathcal W^{(K)}_L)$ is topologically generated by the modes $W^{a(1)}_{b,n}$ and $W^{a(2)}_{b,n}$, where the modes of the field $W^{(r)}$ are defined as
\begin{align}
    W^{a(r)}_b(z)=\sum_{n\in \mathbb Z}W^{a(r)}_{b,n}z^{-n-r}.
\end{align}
It is known that $\mathcal W^{(K)}_L$ is a conformal vertex algebra, and it possesses a unique stress-energy operator $T(z)$ such that
\begin{itemize}
    \item[(1)] $W^{a(r)}_{b}(z)$ has conformal weight $r$ w.r.t. $T(z)$,
    \item[(2)] $W^{a(1)}_{b}(z)$ are primary of spin $1$ w.r.t $T(z)$.
\end{itemize}
$T(z)$ is given by the equation
\begin{align}\label{eqn: stress-energy operator_finite L}
    T(z)=\frac{1}{2(\alpha+K)}:W^{a(1)}_{b}W^{b(1)}_{a}:(z)+\frac{\alpha(L-1)}{2(\alpha+K)}\partial W^{a(1)}_{a}(z)-\frac{1}{\alpha+K}W^{a(2)}_{a}(z),
\end{align}
with central charge 
\begin{align}\label{eqn: central charge_finite L}
    c=\frac{KL}{\alpha+K}(1+\alpha K-(L^2-1)\alpha^2).
\end{align}

\begin{remark}\label{rmk: shift automorphism of W}
$\mathcal W^{(K)}_L$ is preserved under the automorphism $\eta^{\otimes L}_{\beta}:E^{a[i]}_b(z)\mapsto E^{a[i]}_b(z)+\delta^a_b\frac{\beta}{z}$, in fact $\eta^{\otimes L}_{\beta}$ transforms the Miura operator to
\begin{align*}
    \left(\alpha\partial-\frac{\beta}{z}\right)^L+\sum_{r=1}^L (-1)^r\left(\alpha\partial-\frac{\beta}{z}\right)^{L-r}W^{(r)}(z),
\end{align*}
therefore $\eta^{\otimes L}_{\beta}$ transforms $W^{(r)}(z)$ by
\begin{align}
    W^{(r)}(z)\mapsto W^{(r)}(z)+\sum_{s=1}^r \binom{L+s-r}{s}\binom{{\beta}/{\alpha}}{s}\frac{\alpha^s \cdot s!}{z^s} W^{(r-s)}(z),
\end{align}
where $W^{(0)}(z)$ is set to be the constant identity matrix. For example
\begin{equation*}
    \begin{split}
        W^{(1)}(z)\mapsto & W^{(1)}(z)+\frac{\beta L}{z}W^{(0)}(z),\\
    W^{(2)}(z)\mapsto & W^{(2)}(z)+\frac{\beta (L-1)}{z}W^{(1)}(z)+\frac{\beta (\beta-\alpha) L(L-1)}{2z^2}W^{(0)}(z),\\
    W^{(3)}(z)\mapsto & W^{(3)}(z)+\frac{\beta (L-2)}{z}W^{(2)}(z)+\frac{\beta (\beta-\alpha) (L-1)(L-2)}{2z^2}W^{(1)}(z)\\
    &+ \binom{L}{3}\frac{\beta(\beta-\alpha)(\beta-2\alpha)}{z^3} W^{(0)}(z).
    \end{split}
\end{equation*}
\end{remark}

It is shown in \cite{gaiotto2022miura} that when $K=1$ and $\alpha=\frac{\epsilon_3}{\epsilon_1}$, there exists an algebra homomorphism $\Psi_L:\mathsf A^{(1)}\to \mathfrak{U}(\mathcal W^{(1)}_L)$ which is uniquely determined by
\begin{equation}
    \begin{split}
        \Psi_{L}(\mathsf t_{2,0})&=\frac{\epsilon_1^2}{\epsilon_2}\left(V_{-2}+\frac{\alpha}{2}\sum_{n=-\infty}^{\infty}|n|:W^{(1)}_{-n-1}W^{(1)}_{n-1}:\right)\\
    \Psi_{L}(\mathsf t_{0,m})&=\frac{1}{\epsilon_2}W^{(1)}_m,
    \end{split}
\end{equation}
where $V_{-2}$ is a mode of quasi-primary field $V(z)=\sum_{n\in \mathbb Z}V_{n}z^{-n-3}$ defined as
\begin{align}
    V(z):=\frac{1}{3}\sum_{i=1}^L:J^{[i]}(z)J^{[i]}(z)J^{[i]}(z):+\alpha\sum_{i<j}^L:J^{[i]}(z)\partial J^{[j]}(z):
\end{align}
Here our convention for the normally-ordered product of three operators is that
\begin{align*}
    :ABC:=\frac{1}{2}\big(:A(:BC:):+:(:AB:)C:\big).
\end{align*}
It is also shown in \cite{gaiotto2022miura} that when $K=1$ and $\alpha=\frac{\epsilon_3}{\epsilon_1}$, there exists an algebra homomorphism $\Delta_L:\mathsf A^{(1)}\to \mathsf A^{(1)}\widetilde{\otimes}\mathfrak{U}(\mathcal W^{(1)}_L)$ which is uniquely determined by
\begin{equation}
    \begin{split}
    \Delta_L(\mathsf t_{2,0})&=\square (\mathsf t_{2,0})+2\epsilon_1\epsilon_3\sum_{n=1}^{\infty}n\mathsf t_{0,n-1}\otimes W^{(1)}_{-n-1},\\
    \Delta_L(\mathsf t_{0,m})&=\square (\mathsf t_{0,m}),
    \end{split}
\end{equation}
where $\square(x):=x\otimes 1+1\otimes \Psi_L(x)$ for $x\in \mathsf A^{(1)}$. Here $\widetilde\otimes$ is the completed tensor product defined in the Appendix \ref{sec: Completion of Tensor Product}. It is straightforward to see that 
\begin{align}
    (\mathfrak C_{\mathsf A}\otimes 1)\circ \Delta_L=\Psi_L,
\end{align}
where $\mathfrak C_{\mathsf A}: \mathsf A^{(1)}\to \mathbb C[\epsilon_1,\epsilon_2]$ is the natural augmentation morphism sending all generators $\mathsf t_{m,n}$ to zero. Moreover, it is also straightforward to show that
\begin{align}
    (\Delta_{L_1}\otimes 1)\circ \Delta_{L_2}=(1\otimes \Delta_{L_1,L_2})\circ  \Delta_{L_1+L_2},
\end{align}
where $\Delta_{L_1,L_2}: \mathcal W^{(1)}_{L_1+L_2}\to \mathcal W^{(1)}_{L_1}\otimes \mathcal W^{(1)}_{L_2}$ is the (injective) vertex algebra map induced from the splitting the Miura operator into two parts.

More generally, the splitting of the Miura operator \eqref{eqn: Miura operator} induces an injective vertex algebra map $\Delta_{L_1,L_2}: \mathcal W^{(K)}_{L_1+L_2}\to \mathcal W^{(K)}_{L_1}\otimes \mathcal W^{(K)}_{L_2}$, and its explicit form is given by
\begin{equation}
    \begin{split}
    \Delta_{L_1,L_2}(W^{a(r)}_b(z))=\sum_{\substack{(s,t,u)\in\mathbb N^3\\s+t+u=r}} \binom{L_2-t}{u}\alpha^u \partial^u W^{a(s)}_c(z)\otimes W^{c(t)}_b(z),
\end{split}
\end{equation}
where we have set $W^{a(0)}_b(z)=\delta^a_b$.

The main goal of this section is to generalize the construction of $\Delta_L$ and $\Psi_L$ to the cases when $K>1$.

\bigskip In the rest of this section, we freely extend $\mathcal W^{(K)}_L$ by $\mathbb C[\epsilon_1]$, and set
\begin{align}
    \boxed{\epsilon_1\alpha=\epsilon_3,\quad \epsilon_1\bar\alpha=\epsilon_2.}
\end{align}

\subsection{A map from \texorpdfstring{$\mathsf A^{(K)}$}{Ak} to the mode algebra of affine Kac-Moody vertex algebra}

\begin{proposition}\label{prop: map A to affine Kac-Moody}
For all $K\in\mathbb N_{\ge 1}$, there is an algebra homomorphism $$\Psi_1: \mathsf A^{(K)}\to \mathfrak{U}(\widehat{\mathfrak{gl}}_K^{\alpha})[\bar\alpha^{-1}]$$ which is uniquely determined by the map on generators
\begin{equation}
\boxed{
\begin{aligned}
&\Psi_1(\mathsf T_{0,n}(E^a_b))= J^a_{b,n},\\
&\Psi_1(\mathsf T_{1,n}(E^a_b))= \epsilon_1\sum_{m=n}^{\infty}J^a_{c,n-1-m}J^c_{b,m}+\frac{\epsilon_3 n}{2}J^a_{b,n-1}+\epsilon_1\sum_{k=0}^{n-1}\frac{k+1}{n+1}J^a_{c,n-1-k}J^c_{b,k},\\
&\Psi_1(\mathsf t_{2,0})=\frac{\epsilon_1}{6\bar\alpha}\sum_{k,l\in \mathbb Z}\left(:J^a_{b,-k-l-2}J^b_{c,k} J^c_{a,l}:+:J^b_{a,-k-l-2}J^c_{b,k} J^a_{c,l}:\right)\\
&-\epsilon_1\sum_{n=1}^{\infty}n J^{a}_{b,-n-1}J^{b}_{a,n-1}-\frac{\epsilon_1}{\bar\alpha}\sum_{n=1}^{\infty}nJ^{a}_{a,-n-1}J^{b}_{b,n-1},
\end{aligned}}
\end{equation}
in particular
\begin{equation}
\Psi_1(\mathsf t_{1,n})= \frac{1}{2\bar\alpha}\sum_{m\in \mathbb Z}:J^a_{b,-m+n-1}J^b_{a,m}:+\frac{\alpha n}{2\bar\alpha}J^a_{a,n-1}=-\mathsf L_{n-1}+\frac{\alpha n}{2\bar\alpha}J^a_{a,n-1}.
\end{equation}
Here $T(z)=\sum_{n\in \mathbb Z}\mathsf L_nz^{-n-2}$ is the mode expansion of the stress-energy operator.
\end{proposition}

\begin{proof}
The case $K=1$ is treated in \cite{gaiotto2022miura}, so we assume $K\ge 2$. By Theorem \ref{thm: simpler definition, K>1}, it is enough to check relations \eqref{eqn: A0'}-\eqref{eqn: A3'}. \eqref{eqn: A0'} and \eqref{eqn: A1'} are straightforward to check. Let us examine \eqref{eqn: A2'} first. Since $-\Psi_1(\mathsf t_{1,1})$ is the stress-energy operator $\mathsf L_0$ (modulo a central term), which satisfies $[\mathsf L_0,J^a_{b,n}]=-nJ^a_{b,n}$, so we have
\begin{align*}
    [\Psi(\mathsf t_{1,1}),\Psi(\mathsf t_{2,0})]=-2\Psi(\mathsf t_{2,0}),\quad [\Psi(\mathsf t_{1,1}),\Psi(\mathsf t_{0,2})]=2\Psi(\mathsf t_{0,2}).
\end{align*}
To prove the other commutation relations, we need the following equation
\begin{equation}\label{eqn: [t[2,0],J(x)]}
\begin{split}
&\frac{1}{\epsilon_1}[\Psi_1(\mathsf t_{2,0}),J^a_b(x)]=\alpha \partial^2(J^a_b(x)_- - J^a_b(x)_+)-\frac{1}{2}\partial(:J^a_c(x)J^c_b(x):+:J^c_b(x)J^a_c(x):)\\
&+\oint_{|x|>|w|}\frac{J^a_c(x)J^c_b(w)-J^c_b(x)J^a_c(w)}{(x-w)^2}\:\frac{dw}{2\pi i} +\oint_{|w|>|x|}\frac{J^c_b(w)J^a_c(x)-J^a_c(w)J^c_b(x)}{(x-w)^2}\:\frac{dw}{2\pi i},
\end{split}
\end{equation}
where $A(x)=A(x)_++A(x)_-=\sum_{n<0}A_n x^{-n-1}+\sum_{n\ge 0}A_n x^{-n-1}$ is the decomposition of a local field $A(x)$ into non-negative powers and negative powers in coordinate $x$. Taking the negative Fourier modes of \eqref{eqn: [t[2,0],J(x)]}, we get
\begin{equation}\label{eqn: [t[2,0],J_n]}
\begin{split}
&\frac{1}{2n}[\Psi_1(\mathsf t_{2,0}),J^a_{b,n}]=\frac{\epsilon_1}{2n}\oint \frac{1}{\epsilon_1}[\Psi_1(\mathsf t_{2,0}),J^a_b(x)]x^n\:\frac{dx}{2\pi i}\\
&=\epsilon_1\sum_{m=n-1}^{\infty}J^a_{c,n-2-m}J^c_{b,m}+\frac{\epsilon_3 (n-1)}{2}J^a_{b,n-2}+\epsilon_1\sum_{k=0}^{n-2}\frac{k+1}{n}J^a_{c,n-2-k}J^c_{b,k},\;(n>0)
\end{split}
\end{equation}
Using \eqref{eqn: [t[2,0],J_n]} we get
\begin{align*}
    [\Psi(\mathsf t_{2,0}),\Psi(\mathsf t_{0,2})]=\frac{2}{\bar\alpha}\sum_{m\in \mathbb Z}:J^a_{b,-m}J^b_{a,m}:+\frac{2\alpha}{\bar\alpha}J^a_{a,0}=4\Psi(\mathsf t_{1,1}),
\end{align*}
so the first line of \eqref{eqn: A2'} is verified. For the second line of \eqref{eqn: A2'}, showing that $[\Psi_1(\mathsf t_{2,0}),\Psi_1(\mathsf T_{1,0}(E^a_b))]$ vanishes is the only term that requires efforts and we present the computation here. First we rewrite 
\begin{align}
    \Psi_1(\mathsf T_{1,0}(E^a_b))=\epsilon_1\oint_{|x|>|y|}\frac{J^a_c(x)J^c_b(y)}{x-y}\:\frac{dx}{2\pi i}\:\frac{dy}{2\pi i}.
\end{align}
Then we expand the commutator
\begin{equation}\label{eqn: commutator to be proven}
\begin{split}
\frac{1}{\epsilon_1^2}[\Psi_1(\mathsf t_{2,0}),\Psi_1(\mathsf T_{1,0}(E^a_b))]=\frac{1}{\epsilon_1}\oint_{|x|>|y|}\frac{[\Psi_1(\mathsf t_{2,0}),J^a_c(x)]J^c_b(y)+J^a_c(x)[\Psi_1(\mathsf t_{2,0}),J^c_b(y)]}{x-y}\:\frac{dx}{2\pi i}\:\frac{dy}{2\pi i}.
\end{split}
\end{equation}
The right-hand-side of \eqref{eqn: commutator to be proven} is the sum of the following three terms
\begin{equation}\label{eqn: 1st term}
\begin{split}
2\alpha\oint_{|x|>|y|}\frac{(J^a_c(x)_--J^a_c(x)_+)J^c_b(y)+J^a_c(x)(J^c_b(y)_--J^c_b(y)_+)}{(x-y)^3}\:\frac{dx}{2\pi i}\:\frac{dy}{2\pi i},
\end{split}
\end{equation}
and
\begin{equation}\label{eqn: 2nd term}
\begin{split}
\frac{1}{2}\oint_{|x|>|y|}\frac{J^a_c(x)(:J^d_b(y)J^c_d(y):+:J^c_d(y)J^d_b(y):)-(:J^a_d(x)J^d_c(x):+:J^d_c(x)J^a_d(x):)J^c_b(y)}{(x-y)^2}\:\frac{dx}{2\pi i}\:\frac{dy}{2\pi i},
\end{split}
\end{equation}
and
\begin{equation}\label{eqn: 3rd term}
\begin{split}
&\oint_{|x|>|y|>|w|}\frac{J^a_c(x)(J^c_d(y)J^d_b(w)-J^d_b(y)J^c_d(w))}{(x-y)(y-w)^2}\:\frac{dx}{2\pi i}\:\frac{dy}{2\pi i}\:\frac{dw}{2\pi i}\\
&+\oint_{|x|>|w|>|y|}\frac{J^a_c(x)(J^d_b(w)J^c_d(y)-J^c_d(w)J^d_b(y))}{(x-y)(y-w)^2}\:\frac{dx}{2\pi i}\:\frac{dw}{2\pi i}\:\frac{dy}{2\pi i}\\
&+\oint_{|x|>|w|>|y|}\frac{(J^a_d(x)J^d_c(w)-J^d_c(x)J^a_d(w))J^c_b(y)}{(x-y)(x-w)^2}\:\frac{dx}{2\pi i}\:\frac{dw}{2\pi i}\:\frac{dy}{2\pi i}\\
&+\oint_{|w|>|x|>|y|}\frac{(J^d_c(w)J^a_d(x)-J^a_d(w)J^d_c(x))J^c_b(y)}{(x-y)(x-w)^2}\:\frac{dw}{2\pi i}\:\frac{dx}{2\pi i}\:\frac{dy}{2\pi i}.
\end{split}
\end{equation}
We will show that \eqref{eqn: 1st term} vanishes and \eqref{eqn: 2nd term} cancels with \eqref{eqn: 3rd term}. We simplify these terms as follows.\\

\noindent $\bullet$ \eqref{eqn: 1st term}. Notice that 
\begin{align}
    \oint_{|x|>|y|}\frac{A(x)_-B(y)}{(x-y)^3}\:\frac{dx}{2\pi i}\:\frac{dy}{2\pi i}=0=\oint_{|x|>|y|}\frac{A(x)B(y)_+}{(x-y)^3}\:\frac{dx}{2\pi i}\:\frac{dy}{2\pi i}
\end{align}
for any pair of local fields $A(x),B(y)$, so \eqref{eqn: 1st term} equals to
\begin{equation}
\begin{split}
2\alpha\oint_{|x|>|y|}\frac{-J^a_c(x)J^c_b(y)+J^a_c(x)J^c_b(y)}{(x-y)^3}\:\frac{dx}{2\pi i}\:\frac{dy}{2\pi i}=0.
\end{split}
\end{equation}
\noindent $\bullet$ \eqref{eqn: 2nd term}. Using the identity
\begin{align}
    :A(z)B(z):=\oint_{|w|>|z|}\frac{A(w)B(z)}{w-z}\:\frac{dw}{2\pi i}\:\frac{dz}{2\pi i}-\oint_{|z|>|w|}\frac{B(z)A(w)}{w-z}\:\frac{dz}{2\pi i}\:\frac{dw}{2\pi i}
\end{align}
for any pair of local fields $A(z),B(z)$, we expand \eqref{eqn: 2nd term} into
\begin{equation*}
\begin{split}
&\frac{1}{2}\oint_{|x|>|w|>|y|}\left(\frac{1}{(x-y)^2(w-y)}+\frac{1}{(x-w)^2(w-y)}\right)\times\\
&\qquad\times (J^a_c(x)J^d_b(w)J^c_d(y)+J^a_c(x)J^c_d(w)J^d_b(y))\:\frac{dx}{2\pi i}\:\frac{dw}{2\pi i}\:\frac{dy}{2\pi i}\\
&-\frac{1}{2}\oint_{|x|>|w|>|y|}\left(\frac{1}{(x-y)^2(x-w)}+\frac{1}{(w-y)^2(x-w)}\right)\times\\
&\qquad\times (J^d_c(x)J^a_d(w)J^c_b(y)+J^a_d(x)J^d_c(w)J^c_b(y))\:\frac{dx}{2\pi i}\:\frac{dw}{2\pi i}\:\frac{dy}{2\pi i}\\
=&\oint_{|x|>|w|>|y|}\frac{J^a_c(x)J^d_b(w)J^c_d(y)-J^c_d(x)J^a_c(w)J^d_b(y)}{(x-w)(x-y)(w-y)}\:\frac{dx}{2\pi i}\:\frac{dw}{2\pi i}\:\frac{dy}{2\pi i}\\
&~+\frac{1}{2}\oint_{|x|>|w|>|y|}\frac{w-y}{(x-y)^2(x-w)^2}(J^a_c(x)J^d_b(w)J^c_d(y)+J^a_c(x)J^c_d(w)J^d_b(y))\:\frac{dx}{2\pi i}\:\frac{dw}{2\pi i}\:\frac{dy}{2\pi i}\\
&~-\frac{1}{2}\oint_{|x|>|w|>|y|}\frac{x-w}{(x-y)^2(w-y)^2}(J^d_c(x)J^a_d(w)J^c_b(y)+J^a_d(x)J^d_c(w)J^c_b(y))\:\frac{dx}{2\pi i}\:\frac{dw}{2\pi i}\:\frac{dy}{2\pi i}.\\
\end{split}
\end{equation*}
The last two lines vanish identically, in fact
\begin{align*}
&\oint_{|x|>|w|>|y|}\frac{w-y}{(x-y)^2(x-w)^2}(J^a_c(x)J^d_b(w)J^c_d(y)+J^a_c(x)J^c_d(w)J^d_b(y))\:\frac{dx}{2\pi i}\:\frac{dw}{2\pi i}\:\frac{dy}{2\pi i}\\
=&\frac{1}{2}\oint_{|x|>|w|>|y|}\frac{w-y}{(x-y)^2(x-w)^2}(J^a_c(x)J^d_b(w)J^c_d(y)+J^a_c(x)J^c_d(w)J^d_b(y))\:\frac{dx}{2\pi i}\:\frac{dw}{2\pi i}\:\frac{dy}{2\pi i}\\
&~+\frac{1}{2}\oint_{|x|>|y|>|w|}\frac{y-w}{(x-y)^2(x-w)^2}(J^a_c(x)J^d_b(w)J^c_d(y)+J^a_c(x)J^c_d(w)J^d_b(y))\:\frac{dx}{2\pi i}\:\frac{dy}{2\pi i}\:\frac{dw}{2\pi i},
\end{align*}
by deforming the integration contour we pick up the OPE and the above simplifies to
\begin{align*}
\frac{1}{2}\oint_{|x|>|y|}\frac{1}{(x-y)^4}2J^a_c(x)(\alpha\delta^d_d\delta^c_b+\delta^d_b\delta^c_d)\:\frac{dx}{2\pi i}\:\frac{dy}{2\pi i}=\oint_{|x|>|y|}\frac{\alpha K+1}{(x-y)^4}J^a_b(x)\:\frac{dx}{2\pi i}\:\frac{dy}{2\pi i}=0.
\end{align*}
The last line vanishes for similar reason. Therefore we conclude that
\begin{align}
    {\eqref{eqn: 2nd term}}=\oint_{|x|>|w|>|y|}\frac{J^a_c(x)J^d_b(w)J^c_d(y)-J^c_d(x)J^a_c(w)J^d_b(y)}{(x-w)(x-y)(w-y)}\:\frac{dx}{2\pi i}\:\frac{dw}{2\pi i}\:\frac{dy}{2\pi i}.
\end{align}
\noindent $\bullet$ \eqref{eqn: 3rd term}. We substitute variables in each line of \eqref{eqn: 3rd term} into the order $|x|>|w|>|y|$, then summing all terms and arriving at
\begin{align}
    {\eqref{eqn: 3rd term}}=\oint_{|x|>|w|>|y|}\frac{J^c_d(x)J^a_c(w)J^d_b(y)-J^a_c(x)J^d_b(w)J^c_d(y)}{(x-w)(x-y)(w-y)}\:\frac{dx}{2\pi i}\:\frac{dw}{2\pi i}\:\frac{dy}{2\pi i}.
\end{align}
Plugging \eqref{eqn: 1st term}, \eqref{eqn: 2nd term} and \eqref{eqn: 3rd term} back into \eqref{eqn: commutator to be proven}, we conclude that $[\Psi_1(\mathsf t_{2,0}),\Psi_1(\mathsf T_{1,0}(E^a_b))]=0$, thus we have verified all commutation relations in \eqref{eqn: A2'}.\\

The relations \eqref{eqn: A3'} are much easier to show. In fact \eqref{eqn: [t[2,0],J_n]} implies that 
\begin{align}
    \frac{1}{2n+2}[\Psi_1(\mathsf t_{2,0}),\Psi_1(\mathsf T_{0,n+1}(E^a_b))]=\Psi_1(\mathsf T_{1,n}(E^a_b)),
\end{align}
and it is straightforward to compute that
\begin{equation}
\begin{split}
&[\Psi_1(\mathsf T_{1,0}(E^a_b)),\Psi_1(\mathsf T_{0,n}(E^c_d))]=\\
&\delta^c_b \Psi_1(\mathsf T_{1,n}(E^a_d))-\delta^a_d \Psi_1(\mathsf T_{1,n}(E^c_b))-\frac{\epsilon_3 n}{2}\left(\delta^c_b J^a_{d,n-1}+\delta^a_d J^c_{b,n-1}\right)-n\epsilon_1\delta^c_d J^a_{b,n-1}\\
&+\epsilon_1\sum_{m=0}^{n-1}J^a_{d,m}J^c_{b,n-1-m}
-\epsilon_1\sum_{m=0}^{n-1}\frac{m+1}{n+1}\delta^a_d J^c_{f,m}J^f_{b,n-1-m}-\epsilon_1\sum_{m=0}^{n-1}\frac{n-m}{n+1}\delta^c_b J^a_{f,m}J^f_{d,n-1-m},
\end{split}
\end{equation}
which implies \eqref{eqn: A3'}. This finishes the proof of the Proposition \ref{prop: map A to affine Kac-Moody}.
\end{proof}

\subsection{Mixed coproduct: the basic case}
In the Appendix \ref{sec: Completion of Tensor Product} we have defined a notion of completed tensor product of $\mathbb Z$-graded $\mathbb C[\epsilon_1,\epsilon_2]$-algebras, we apply this construction to the algebras $\mathsf A^{(K)}$ and $\mathfrak{U}(\widehat{\mathfrak{gl}}_K^{\alpha})$ and get the $\mathbb Z$-graded $\mathbb C[\epsilon_1,\epsilon_2]$-algebra $\mathsf A^{(K)}\widetilde\otimes\mathfrak{U}(\widehat{\mathfrak{gl}}_K^{\alpha})$, where the $\mathbb Z$-grading is determined by
\begin{align}\label{eqn: grading on tensor product}
    \deg \mathsf T_{n,m}(E^a_b)=\deg \mathsf t_{n,m}=m-n,\quad \deg J^{a}_{b,n}=n.
\end{align}
\begin{proposition}\label{prop: mixed coproduct L=1}
For all $K\in\mathbb N_{\ge 1}$, there is an algebra homomorphism $$\Delta_1: \mathsf A^{(K)}\to \mathsf A^{(K)}\widetilde\otimes\mathfrak{U}(\widehat{\mathfrak{gl}}_K^{\alpha})[\bar\alpha^{-1}]$$ which is uniquely determined by the map on generators
\begin{equation}
\boxed{
\begin{aligned}
&\Delta_1(\mathsf T_{0,n}(E^a_b))= \square(\mathsf T_{0,n}(E^a_b)),\\
&\Delta_1(\mathsf T_{1,n}(E^a_b))=\square(\mathsf T_{1,n}(E^a_b))+ \epsilon_1\sum_{m=n}^{\infty}\left(\mathsf T_{0,m}(E^c_b)\otimes J^{a}_{c,n-m-1}-\mathsf T_{0,m}(E^a_c)\otimes J^{c}_{b,n-m-1}\right)\\
&~+\epsilon_3n\mathsf T_{0,n-1}(E^a_b)\otimes 1+\epsilon_1\sum_{m =0}^{n-1}\frac{m+1}{n+1}\left(\mathsf T_{0,m}(E^c_b)\otimes J^{a}_{c,n-m-1}-\mathsf T_{0,m}(E^a_c)\otimes J^{c}_{b,n-m-1}\right),\\
&\Delta_1(\mathsf t_{2,0})=\square(\mathsf t_{2,0})-2\epsilon_1\sum_{n=1}^{\infty}n \left(\mathsf T_{0,n-1}(E^a_b)\otimes J^{b}_{a,-n-1}+\epsilon_1\mathsf t_{0,n-1}\otimes J^{a}_{a,-n-1}\right),
\end{aligned}}
\end{equation}
where $\square(x):=x\otimes 1+1\otimes \Psi_1(x)$. In particular $\Delta_1(\mathsf t_{1,n})=\square(\mathsf t_{1,1})+\epsilon_3 n \mathsf t_{0,0}\otimes 1$.
\end{proposition}

\begin{proof}
The case $K=1$ is treated in \cite{gaiotto2022miura}, so we assume $K\ge 2$. By Theorem \ref{thm: simpler definition, K>1}, it is enough to check equations \eqref{eqn: A0'}-\eqref{eqn: A3'}. The only term that requires efforts is showing that $[\Delta_1(\mathsf t_{2,0}),\Delta_1(\mathsf T_{1,0}(E^a_b))]=0$, and we present the computation here.

First we introduce notation
\begin{align}
    E^a_b(z):=\sum_{m=0}^{\infty}\mathsf T_{0,m}(E^a_b)z^{-m-1},\quad \mathsf t(z):=\sum_{m=0}^{\infty}\mathsf t_{0,m}z^{-m-1},
\end{align}
and rewrite
\begin{equation}
\begin{split}
&\Delta_1(\mathsf T_{1,0}(E^a_b))=\square(\mathsf T_{1,0}(E^a_b))+\epsilon_1\oint_{|x|<|y|}\frac{E^a_c(x)\otimes J^c_b(y)-E^c_b(x)\otimes J^a_c(y)}{x-y}\:\frac{dx}{2\pi i}\:\frac{dy}{2\pi i}\\
&\Delta_1(\mathsf t_{2,0})=\square(\mathsf t_{2,0})-2\epsilon_1\oint_{|x|<|y|}\frac{E^a_b(x)\otimes J^b_a(y)+\epsilon_1\mathsf t(x)\otimes J^a_a(y)}{(x-y)^2}\:\frac{dx}{2\pi i}\:\frac{dy}{2\pi i}.
\end{split}
\end{equation}
Note that only the non-negative part $J(y)_+:=\sum_{n<0}J_n y^{-n-1}$ of the local field $J(y)$ in the integrand contributes to the integral. Since $[\square(\mathsf t_{2,0}),\square(\mathsf T_{1,0}(E^a_b))]=0$, the commutator $\frac{1}{2\epsilon_1^2}[\Delta_1(\mathsf T_{1,0}(E^a_b)),\Delta_1(\mathsf t_{2,0})]$ can be written as the sum of the following three terms:
\begin{equation}\label{eqn: 1st term_coproduct}
\begin{split}
    &\frac{1}{\epsilon_1}\oint_{|x|<|y|}\frac{E^d_c(x)}{x-y}\otimes \left([\Psi_1(\mathsf T_{1,0}(E^a_b)),\partial J^c_d(y)]+\frac{1}{2}[\Psi_1(\mathsf t_{2,0}),\delta^c_bJ^a_d(y)-\delta^a_dJ^c_b(y)]\right)\:\frac{dx}{2\pi i}\:\frac{dy}{2\pi i}\\
    &-\epsilon_2\oint_{|x|<|y|}\frac{\mathsf t(x)\otimes \partial^2 J^a_b(y)}{x-y}\:\frac{dx}{2\pi i}\:\frac{dy}{2\pi i},
\end{split}
\end{equation}
and
\begin{equation}\label{eqn: 2nd term_coproduct}
\begin{split}
    &-\frac{1}{\epsilon_1}\oint_{|x|<|y|}\left([\mathsf T_{1,0}(E^a_b),\partial E^c_d(x)]+\frac{1}{2}[\mathsf t_{2,0},\delta^c_bE^a_d(y)-\delta^a_dE^c_b(y)]\right)\otimes \frac{J^d_c(y)}{x-y}\:\frac{dx}{2\pi i}\:\frac{dy}{2\pi i}\\
    &+\oint_{|x|<|y|}\frac{\partial^2 E^a_b(x)\otimes J^c_c(y)}{x-y}\:\frac{dx}{2\pi i}\:\frac{dy}{2\pi i},
\end{split}
\end{equation}
and
\begin{equation}\label{eqn: 3rd term_coproduct}
\begin{split}
    &\oint_{|x|<|y|}\frac{E^d_c(x)\otimes (J^a_d(y)_+\partial J^c_b(y)_+ +J^c_b(y)_+\partial J^a_d(y)_+)}{x-y}\:\frac{dx}{2\pi i}\:\frac{dy}{2\pi i}\\
    &-\oint_{|x|<|y|}\frac{E^a_c(x)\otimes J^e_b(y)_+\partial J^c_e(y)_+ + E^d_b(x)\otimes J^a_e(y)_+\partial J^e_d(y)_+}{x-y}\:\frac{dx}{2\pi i}\:\frac{dy}{2\pi i}\\
    &+\oint_{|x|<|y|}\frac{(\partial E^c_b(x) E^a_d(x)+\partial E^a_d(x)E^c_b(x))\otimes J^d_c(y)}{x-y}\:\frac{dx}{2\pi i}\:\frac{dy}{2\pi i}\\
    &-\oint_{|x|<|y|}\frac{\partial E^c_e(x)E^e_b(x) \otimes J^a_c(y)+\partial E^e_d(x)E^a_e(x)\otimes J^d_b(y)}{x-y}\:\frac{dx}{2\pi i}\:\frac{dy}{2\pi i}.
\end{split}
\end{equation}
It is straightforward to compute:
\begin{equation}\label{eqn: positive part}
\begin{split}
&\frac{1}{\epsilon_1}[\Psi_1(\mathsf T_{1,0}(E^a_b)),\partial J^c_d(z)_+]+\frac{1}{2\epsilon_1}[\Psi_1(\mathsf t_{2,0}),\delta^c_bJ^a_d(z)_+-\delta^a_dJ^c_b(z)_+]=\delta^c_d\partial^2J^a_b(z)_+\\
&+\frac{\alpha}{2}\partial^2(\delta^c_bJ^a_d(z)_++\delta^a_dJ^c_b(z)_+)-\partial(J^a_d(z)_+J^c_b(z)_+)\\
&+\delta^c_b J^a_e(z)_+\partial J^e_d(z)_+ + \delta^a_d \partial J^c_e(z)_+J^e_b(z)_+,
\end{split}
\end{equation}
and
\begin{equation}\label{eqn: negative part}
\begin{split}
&\frac{1}{\epsilon_1}[\mathsf T_{1,0}(E^a_b),\partial E^c_d(z)]+\frac{1}{2\epsilon_1}[\mathsf t_{2,0},\delta^c_bE^a_d(z)-\delta^a_dE^c_b(z)]=\delta^c_d\partial^2E^a_b(z)\\
&+\frac{\alpha}{2}\partial^2(\delta^c_bE^a_d(z)+\delta^a_dE^c_b(z))+\partial(E^a_d(z)E^c_b(z))\\
&-\delta^c_b 
E^a_e(z)\partial E^e_d(z) - \delta^a_d \partial E^c_e(z)E^e_b(z).
\end{split}
\end{equation}
Plugging \eqref{eqn: positive part} into \eqref{eqn: 1st term_coproduct} and \eqref{eqn: negative part} into \eqref{eqn: 2nd term_coproduct}, and taking the sum of \eqref{eqn: 1st term_coproduct} and \eqref{eqn: 2nd term_coproduct} and \eqref{eqn: 3rd term_coproduct}, we arrive at
\begin{equation}\label{eqn: check coproduct_final}
\begin{split}
&\oint_{|x|<|y|}\frac{E^d_c(x)\otimes [J^c_b(y)_+,\partial J^a_d(y)_+]}{x-y}\:\frac{dx}{2\pi i}\:\frac{dy}{2\pi i}-\oint_{|x|<|y|}\frac{E^a_c(x)\otimes [J^e_b(y)_+,\partial J^c_e(y)_+]}{x-y}\:\frac{dx}{2\pi i}\:\frac{dy}{2\pi i}\\
&+\oint_{|x|<|y|}\frac{[\partial E^c_b(x),E^a_d(x)]\otimes J^d_c(y)}{x-y}\:\frac{dx}{2\pi i}\:\frac{dy}{2\pi i}-\oint_{|x|<|y|}\frac{[\partial E^e_d(x),E^a_e(x)]\otimes J^d_b(y)}{x-y}\:\frac{dx}{2\pi i}\:\frac{dy}{2\pi i}.
\end{split}
\end{equation}
Using the following identifies
\begin{align*}
[J^a_b(x)_+,\partial J^c_d(x)_+]=\frac{1}{2}\partial^2(\delta^c_bJ^a_d(x)_+-\delta^a_dJ^c_b(x)_+),\quad [\partial E^a_b(x),E^c_d(x)]=\frac{1}{2}\partial^2(\delta^a_dE^c_b(x)-\delta^c_bE^a_d(x))
\end{align*}
\eqref{eqn: check coproduct_final} reduces to zero. This proves that $[\Delta_1(\mathsf T_{1,0}(E^a_b)),\Delta_1(\mathsf t_{2,0})]=0$. The other relations are much easier and we omit the details.
\end{proof}

\subsection{Bootstrap the general mixed coproducts from the basic one}

The $\mathbb Z$-grading on $\mathfrak{U}(\widehat{\mathfrak{gl}}_K^{\alpha})$ naturally induces a $\mathbb Z$-grading on $\mathfrak{U}(\mathcal W^{(K)}_L)\subset\mathfrak{U}(\widehat{\mathfrak{gl}}_K^{\alpha})^{\widehat{\otimes}L}$, thus we have a completed tensor product algebra $ \mathsf A^{(K)}\widetilde\otimes\mathfrak{U}(\mathcal W^{(K)}_L)$.

Using the algebra homomorphism $\Delta_1: \mathsf A^{(K)}\to \mathsf A^{(K)}\widetilde\otimes\mathfrak{U}(\widehat{\mathfrak{gl}}_K^{\alpha})[\bar\alpha^{-1}]$, we can bootstrap a series of algebra homomorphisms $\Delta_L:\mathsf A^{(K)}\to \mathsf A^{(K)}\widetilde\otimes\mathfrak{U}(\mathcal W^{(K)}_L)[\bar\alpha^{-1}]$ for all positive integers $L$.

\begin{theorem}\label{thm: mixed coproduct}
For all $K\in\mathbb N_{\ge 1}$, there is an algebra homomorphism $$\Psi_L: \mathsf A^{(K)}\to \mathfrak{U}(\mathcal W^{(K)}_L)[\bar\alpha^{-1}]$$ which is uniquely determined by the map on generators
\begin{equation}\label{eqn: Psi_L}
\boxed{
\begin{aligned}
&\Psi_L(\mathsf T_{0,n}(E^a_b))= W^{a(1)}_{b,n},\\
&\Psi_L(\mathsf t_{1,n})=\frac{1}{\bar\alpha}\left(\frac{1}{2}\sum_{k\in \mathbb Z}:W^{a(1)}_{b,n-1-k}W^{b(1)}_{a,k}:+\frac{\alpha n}{2} W^{a(1)}_{a,n-1}-W^{a(2)}_{a,n-1}\right)=-\mathsf L_{n-1}+\frac{\alpha n L}{2\bar\alpha}W^{a(1)}_{a,n-1},\\
&\Psi_L(\mathsf T_{1,0}(E^a_b))=\epsilon_1\sum_{m\ge 0}W^{a(1)}_{c,-m-1}W^{c(1)}_{b,m}-\epsilon_1 W^{a(2)}_{b,-1},\\
&\Psi_L(\mathsf t_{2,0})=\frac{\epsilon_1}{\bar\alpha}\left(V_{-2}-\bar\alpha\sum_{n=1}^{\infty}n\:W^{a(1)}_{b,-n-1}W^{b(1)}_{a,n-1}-\sum_{n=1}^{\infty}n\: W^{a(1)}_{a,-n-1}W^{b(1)}_{b,n-1}\right).
\end{aligned}}
\end{equation}
where $V_{-2}$ is a mode of quasi-primary field $V(z)=\sum_{n\in \mathbb Z}V_{n}z^{-n-3}$ defined as
\begin{equation}\label{eqn: Delta_L}
\begin{split}
V(z):=&\frac{1}{6}\sum_{i=1}^L\left(:J^{a[i]}_b(z)J^{b[i]}_c(z)J^{c[i]}_a(z):+:J^{b[i]}_a(z)J^{c[i]}_b(z)J^{a[i]}_c(z):\right)\\
&-\bar\alpha\sum_{i<j}^L:J^{a[i]}_b(z)\partial J^{b[j]}_a(z):-\sum_{i<j}^L:J^{a[i]}_a(z)\partial J^{b[j]}_b(z):\\
=&\frac{1}{6}\left(:W^{a(1)}_b(z)W^{b(1)}_c(z)W^{c(1)}_a(z):+:W^{b(1)}_a(z)W^{c(1)}_b(z)W^{a(1)}_c(z):\right)\\
&+W^{a(3)}_a(z)-:W^{a(1)}_b(z)W^{b(2)}_a(z):+\text{ total derivatives},
\end{split}
\end{equation}
from the last equation we see that $V_{-2}\in \mathfrak{U}(\mathcal W^{(K)}_L)$.

Moreover there is an algebra homomorphism
$$\Delta_L: \mathsf A^{(K)}\to \mathsf A^{(K)}\widetilde\otimes\mathfrak{U}(\mathcal W^{(K)}_L)[\bar\alpha^{-1}]$$ which is uniquely determined by the map on generators
\begin{equation}
\boxed{
\begin{aligned}
&\Delta_L(\mathsf T_{0,n}(E^a_b))=\square( \mathsf T_{0,n}(E^a_b),\\
&\Delta_L(\mathsf t_{1,n})=\square(\mathsf t_{1,n})+\epsilon_3 nL\mathsf t_{0,n-1}\otimes 1,\\
&\Delta_L(\mathsf T_{1,0}(E^a_b))=\square(\mathsf T_{1,0}(E^a_b))+\epsilon_1\sum_{m=0}^{\infty}\left(\mathsf T_{0,m}(E^c_b)\otimes W^{a(1)}_{c,-m-1}-\mathsf T_{0,m}(E^a_c)\otimes W^{c(1)}_{b,-m-1}\right),\\
&\Delta_L(\mathsf t_{2,0})=\square(\mathsf t_{2,0})-2\epsilon_1\sum_{n=1}^{\infty}n \left(\mathsf T_{0,n-1}(E^a_b)\otimes W^{b(1)}_{a,-n-1}+\epsilon_1\mathsf t_{0,n-1}\otimes W^{a(1)}_{a,-n-1}\right),
\end{aligned}}
\end{equation}
where $\square(x):=x\otimes 1+1\otimes \Psi_L(x)$.
\end{theorem}

\begin{proof}
The case $K=1$ is treated in \cite{gaiotto2022miura}, so we assume $K\ge 2$. We prove by induction on $L$. The initial case $L=1$ is proven in Proposition \ref{prop: map A to affine Kac-Moody} and \ref{prop: mixed coproduct L=1}. For general $L$, let us first show that $\Delta_L$ generates an algebra homomorphism. In fact, consider the injective vertex algebra map $\Delta_{1,L-1}: \mathcal W^{(K)}_{L}\to V^{\kappa_\alpha}(\gl_K)\otimes \mathcal W^{(K)}_{L-1}$ induced from the splitting the Miura operator into two parts, then it is easy to check that $(1\otimes \Delta_{1,L-1})\circ \Delta_L$ agrees with $(\Delta_1\otimes 1)\circ\Delta_{L-1}$ on the generators $\mathsf T_{0,n}(E^a_b)$, $\mathsf T_{1,0}(E^a_b)$ and $\mathsf t_{2,0}$, thus $(1\otimes \Delta_{1,L-1})\circ \Delta_L$ generates an algebra homomorphism. Since $\Delta_{1,L-1}$ is injective, $\Delta_L$ must generates an algebra homomorphism. For the first statement of the theorem, we notice that 
\begin{align}
    (\mathfrak C_{\mathsf A}\otimes 1)\circ \Delta_L=\Psi_L,
\end{align}
where $\mathfrak C_{\mathsf A}: \mathsf A^{(K)}\to \mathbb C[\epsilon_1,\epsilon_2]$ is the natural augmentation morphism sending all generators $\mathsf T_{m,n}(X)$ and $\mathsf t_{m,n}$ to zero. Therefore $\Psi_L$ is also an algebra homomorphism.
\end{proof}

\begin{corollary}
The maps $\Delta_L$ and $\Delta_{L_1,L_2}$ are compatible in the following sense:
\begin{equation}
(\Delta_{L_1}\otimes 1)\circ \Delta_{L_2}=(1\otimes\Delta_{L_1,L_2})\circ \Delta_{L_1+L_2}
\end{equation}
\end{corollary}

Recall that $\mathsf D^{(K)}$ is the subalgebra of $\mathsf A^{(K)}$ generated by $\{\mathsf T_{n,m}(X)\:|\: X\in \gl_K,(n,m)\in \mathbb N^2\}$.
\begin{proposition}\label{prop: image of D^K}
The image of $\mathsf D^{(K)}$ under the map $\Delta_L$ is contained in $\mathsf D^{(K)}\widetilde\otimes\mathfrak{U}(\mathcal W^{(K)}_L)$, and the image of $\mathsf D^{(K)}$ under the map $\Psi_L$ is contained in $\mathfrak{U}(\mathcal W^{(K)}_L)$.
\end{proposition}

\begin{proof}
Since $\mathsf D^{(K)}$ is generated by the adjoint actions of $\mathsf t_{2,0}$ on $\{\mathsf T_{0,n}(X)\:|\: X\in \gl_K,n\in \mathbb N\}$, and obviously $\Delta_L(\mathsf T_{0,n}(X))\in \mathsf D^{(K)}\widetilde\otimes\mathfrak{U}(\mathcal W^{(K)}_L)$, it suffices to show that $\mathrm{ad}_{\Delta_L(\mathsf t_{2,0})}$ preserves the subalgebra $\mathsf D^{(K)}\widetilde\otimes\mathfrak{U}(\mathcal W^{(K)}_L)$. Write $\Delta_L(\mathsf t_{2,0})=\mathsf t_{2,0}\otimes 1+1\otimes \Psi_L(\mathsf t_{2,0})+$cross terms. Then the adjoint action of $\mathsf t_{2,0}\otimes 1$ obviously preserves $\mathsf D^{(K)}\widetilde\otimes\mathfrak{U}(\mathcal W^{(K)}_L)$, and equation \eqref{eqn: [t[2,0],J(x)]} implies that $1\otimes \Psi_L(\mathsf t_{2,0})$ also preserves $\mathsf D^{(K)}\widetilde\otimes\mathfrak{U}(\mathcal W^{(K)}_L)$. For the cross terms, we only need to check $\epsilon_1\mathsf t_{0,n-1}\otimes W^{a(1)}_{a,-n-1}$, which also equals to $\mathsf T_{0,n-1}(1)\otimes \frac{1}{\bar\alpha}W^{a(1)}_{a,-n-1}$, and its commutator with an element $A\otimes B\in \mathsf D^{(K)}\widetilde\otimes\mathfrak{U}(\mathcal W^{(K)}_L)$ can be written as 
\begin{align*}
    \epsilon_1[\mathsf t_{0,n-1},A]\otimes W^{a(1)}_{a,-n-1}B+A\mathsf T_{0,n-1}(1)\otimes \frac{1}{\bar\alpha}[W^{a(1)}_{a,-n-1},B].
\end{align*}
$[\mathsf t_{0,n-1},A]\in \mathsf D^{(K)}$ by Lemma \ref{lem: t[m,n] preserves D^K}, and $\frac{1}{\bar\alpha}[W^{a(1)}_{a,-n-1},B]\in \mathfrak{U}(\mathcal W^{(K)}_L)$ because $\frac{1}{\bar\alpha}[W^{a(1)}_{a,-n-1},J^{b[i]}_{c,m}]=(n+1)\delta^b_c\delta_{n+1,m}$. Thus $\Delta_L(\mathsf D^{(K)})\subset\mathsf D^{(K)}\widetilde\otimes\mathfrak{U}(\mathcal W^{(K)}_L)$. The second statement follows from the first by applying augmentation map.
\end{proof}

\subsection{Vertical filtration on the rectangular W-algebra}
For every W-algebra, there is a natural increasing filtration \cite[4.9-4.11]{arakawa2007representation} attached to it. Let us denote it by $0=V_{-1}\mathcal W^{(K)}_L\subset V_{0}\mathcal W^{(K)}_L\subset\cdots \subset \mathcal W^{(K)}_L$. In our notation, $V_{\bullet}\mathcal W^{(K)}_L$ is induced from the filtration degree assignment
\begin{align*}
    \deg_V \alpha=0,\quad  \deg_V W^{a(r)}_{b}(z)=r-1.
\end{align*}
It is known that the associated graded vertex algebra $\mathrm{gr}_F\: \mathcal W^{(K)}_L$ is isomorphic to $V^{\kappa^{\sharp}}(\mathfrak{gl}_{KL}^{f_L})$ for a specific level $\kappa^{\sharp}$. The level $\kappa^{\sharp}$ in the rectangular case can be read out from \cite{arakawa2017explicit}, and we summarize it in the following lemma.
\begin{lemma}\label{lem: associated graded of W-algebra}
$\mathrm{gr}_V\mathcal W^{(K)}_L\cong V^{\alpha L,L}(\gl_{K}\otimes\mathbb C[z]/(z^L))$. Equivalently, the OPEs in $\mathcal W^{(K)}_L$ have the following form:
\begin{align*}
    W^{a(r)}_{b}(z) W^{c(s)}_{d}(w)\sim \frac{\delta^c_b W^{a(r+s-1)}_{d}(w)-\delta^a_d W^{c(r+s-1)}_{b}(w)}{z-w}+\delta_{r,1}\delta_{s,1}\frac{\alpha L\delta^a_d\delta^c_b+L\delta ^a_b\delta^c_d}{(z-w)^2}\pmod{V_{r+s-3}\mathcal W^{(K)}_L}.
\end{align*}
\end{lemma}
Consider the stress-energy operator $T(z)=\sum_{n\in \mathbb Z}\mathsf L_n z^{-n-2}$. Since $\mathsf L_{n}W^{a(r)}_{b,-r}|0\rangle$ has conformal weight $n-r$, we see that $\mathsf L_{n}W^{a(r)}_{b,-r}|0\rangle\in V_{r-2}\mathcal W^{(K)}_L$ whenever $n>0$, thus we get the following.
\begin{proposition}\label{prop: TW OPE}
The OPE between stress-energy operator and $W^{a(r)}_{b}$ has the form
\begin{align}
    T(z)W^{a(r)}_{b}(w)\sim \frac{r W^{a(r)}_{b}(w)}{(z-w)^2}+\frac{\partial W^{a(r)}_{b}(w)}{z-w}\pmod{V_{r-2}\mathcal W^{(K)}_L}.
\end{align}
\end{proposition}
It follows from Lemma \ref{lem: associated graded of W-algebra} and the formula \eqref{eqn: Psi_L} of the map $\Psi_L$ that $\Psi_L$ respects the vertical filtrations on $\mathsf A^{(K)}$ and $\mathcal W^{(K)}_L$.
\begin{proposition}
For all $n\in \mathbb Z$, $\Psi_L(V_n\mathsf A^{(K)})\subset V_n\mathfrak U(\mathcal W^{(K)}_L)[\bar\alpha^{-1}]$.
\end{proposition}

\begin{proof}
It suffices to show that $\mathrm{ad}_{\Psi_L(\mathsf t_{2,0})}(V_n\mathfrak U(\mathcal W^{(K)}_L))\subset V_{n+1}\mathfrak U(\mathcal W^{(K)}_L)$. We note that $\Psi_L(\mathsf t_{2,0})\equiv W^{a(3)}_{a,-2}\pmod{V_1\mathfrak U(\mathcal W^{(K)}_L)}$, thus it is enough to show that $\mathrm{ad}_{W^{a(3)}_{a,-2}}(V_n\mathfrak U(\mathcal W^{(K)}_L))\subset V_{n+1}\mathfrak U(\mathcal W^{(K)}_L)$, which obviously follows from Lemma \ref{lem: associated graded of W-algebra}.
\end{proof}

\begin{proposition}\label{prop: Psi_L filtered}
$\Psi_L(\mathsf T_{n,m}(E^a_b))\equiv (-\epsilon_1)^n W^{a(n+1)}_{b,m-n}\pmod{V_{n-1}\mathfrak U(\mathcal W^{(K)}_L)}$.
\end{proposition}

The proof of Proposition \ref{prop: Psi_L filtered} will be given in the next section. Some preliminary steps are given here. We discuss the $K=1$ case and $K>1$ case separately.

\bigskip \noindent $\bullet\; K=1$. Since $W^{(r)}_{-n}|0\rangle$ has degree $-n$ under the grading \eqref{eqn: grading on tensor product}, we have $\deg W^{(3)}_{-2}W^{(r)}_{-r}|0\rangle=-r-2$. Modulo $V_{r-1}\mathcal W^{(1)}_L$, $W^{(3)}_{-2}W^{(r)}_{-r}|0\rangle$ contains only linear terms in $W^{(r+1)}_{-r-2}|0\rangle$ or quadratic terms in $W^{(s)}_{-s}W^{(t)}_{-t}|0\rangle$ with $s+t=r+2$, i.e.
\begin{equation}\label{eqn: W^3W^r preliminary, K=1}
\begin{split}
    W^{(3)}_{-2}W^{(r)}_{-r}|0\rangle\equiv  \mu W^{(r+1)}_{-r-2}|0\rangle+\sum_{s=1}^{\lfloor \frac{r+2}{2}\rfloor}\nu_{s}W^{(s)}_{-s}W^{(r+2-s)}_{s-r-2}|0\rangle \pmod{V_{r-1}\mathcal W^{(1)}_L},
\end{split}
\end{equation}
for some $\mu,\nu_s\in \mathbb C[\alpha]$.

\begin{lemma}\label{lem: vanishing of nu_s}
Assume that $L>r>1$, then $\nu_s=0$ for all $1\le s\le \lfloor \frac{r+2}{2}\rfloor$. 
\end{lemma}

\begin{proof}
Consider the image of two sides of \eqref{eqn: W^3W^r preliminary, K=1} in the Zhu's $C_2$-algebra of $V^{-\bar\alpha}(\gl_1)^{\otimes L}$ (which is the polynomial ring over $\mathbb C[\alpha]$ freely generated by $J^{[1]},\cdots,J^{[L]}$). The left-hand-side is a polynomial of at most $r+1$ $J$'s, then so is the right-hand-side. On the other hand, the right-hand-side equals to $\sum_{s=1}^{\lfloor \frac{r+2}{2}\rfloor}\nu_{s}\overline{W}^{(s)}\overline{W}^{(r+2-s)}$ modulo $V_{r-2}(\mathbb C[\alpha,\overline{W}^{(1)},\cdots,\overline{W}^{(L)}])$, where 
\begin{align*}
    \overline{W}^{(s)}=\sum_{i_1<\cdots<i_s}J^{[i_1]}\cdots J^{[i_s]},
\end{align*}
and the vertical filtration $V_{\bullet}(\mathbb C[\alpha,\overline{W}^{(1)},\cdots,\overline{W}^{(L)}])$ is induced by setting $\deg_V(\alpha)=0,\deg_V \overline{W}^{(s)}=s-1$. Let us grade $\mathbb C[\alpha,J^{[1]},\cdots,J^{[L]}]$ by $\deg\alpha=0,\deg J^{[i]}=1$, then the degree $r+2$ subspace of $\mathfrak S_L$-invariant subspace of $\mathbb C[\alpha,J^{[1]},\cdots,J^{[L]}]$ is a free $\mathbb C[\alpha]$-module with a basis given by monomials $\overline W^{(s_1)}\cdots\overline W^{(s_n)}$ such that $\sum_{i=1}^n s_i=r+2$. Since the left-hand-side has trivial degree $r+2$ component, we conclude that $\nu_s=0$ for all $1\le s\le \lfloor \frac{r+2}{2}\rfloor$. 
\end{proof}


\begin{lemma}\label{lem: mu independent of L}
Assume that $L>r>1$, then $\mu$ does not depend on $L$.
\end{lemma}

\begin{proof}
We write $\mu[L]$ to indicate the $L$-dependence. Consider the coproduct $\Delta_{L-1,1}:\mathcal W^{(1)}_L\to \mathcal W^{(1)}_{L-1}\otimes \mathcal W^{(K)}_1$, then we have
\begin{align*}
    \Delta_{L-1,1}(W^{(r)}_{-r}|0\rangle)\equiv W^{(r)}_{-r}|0\rangle\otimes  |0\rangle\pmod{V_{r-2}(\mathcal W^{(1)}_{L-1}\otimes \mathcal W^{(1)}_1)}.
\end{align*}
Then it follows that
\begin{align*}
    \Delta_{L-1,1}(W^{(3)}_{-2}W^{(r)}_{-r}|0\rangle)&\equiv \Delta_{L-1,1}(W^{(3)}_{-2})W^{(r)}_{-r}|0\rangle\otimes  |0\rangle \equiv W^{(3)}_{-2}W^{(r)}_{-r}|0\rangle\otimes  |0\rangle \\
    &\equiv \mu[L-1] W^{(r+1)}_{-r-1}|0\rangle|0\rangle\otimes |0\rangle \quad \pmod{V_{r-2}(\mathcal W^{(1)}_{L-1}\otimes \mathcal W^{(1)}_1)}.
\end{align*}
From the above we conclude that $\mu[L]=\mu[L-1]$, thus $\mu$ does not depend on $L$.
\end{proof}
Now we read from \eqref{eqn: W^3W^r preliminary, K=1} that
\begin{align*}
    [W^{(3)}_{-2},W^{(r)}_{n}]\equiv -(n+r-1)\mu W^{(r+1)}_{n-2}\pmod{V_{r-1}\mathfrak{U}(\mathcal W^{(1)}_L)}.
\end{align*}
The above formula together with $\Psi_L(\mathsf T_{0,n}(1))=W^{(1)}_{n}$ and $\Psi_L(\mathsf t_{2,0})\equiv \frac{\epsilon_1}{\bar\alpha} W^{(3)}_{-2}\pmod{V_1\mathfrak{U}(\cal W^{(K)}_L)}$ imply that $\exists \mu_{n}\in \mathbb C[\epsilon_1,\alpha]$ such that
\begin{align*}
    \Psi_L(\mathsf T_{n,m}(1))\equiv \mu_{n}  W^{(n+1)}_{m-n}\pmod{V_{n-1}\mathfrak{U}(\mathcal W^{(1)}_L)}.
\end{align*}
Moreover, according to Lemma \ref{lem: mu independent of L}, $\mu_{n}$ is independent of $L$ as long as $L > n$.

\bigskip \noindent $\bullet\; K>1$. We notice that
\begin{align*}
    [W^{a(r)}_{b,n},W^{c(s)}_{d,m}]\equiv \delta^c_b W^{a(r+s-1)}_{d,n+m}-\delta^a_d W^{c(r+s-1)}_{b,n+m}\pmod{V_{r+s-3}\mathfrak{U}(\cal W^{(K)}_L)}.
\end{align*}
The above formula together with $\Psi_L(\mathsf T_{0,n}(E^a_b))=W^{a(1)}_{b,n}$ and $\Psi_L(\mathsf T_{1,0}(E^a_b))\equiv -\epsilon_1 W^{a(2)}_{b,-1}\pmod{V_0\mathfrak{U}(\cal W^{(K)}_L)}$ imply that 
\begin{align*}
    \Psi_L(\mathsf T_{n,m}(E^a_b))\equiv (-\epsilon_1)^n  W^{a(n+1)}_{b,m-n}\pmod{V_{n-1}\mathfrak{U}(\mathcal W^{(K)}_L)}, \quad \forall {E}^a_b\in \mathfrak{sl}_K.
\end{align*}
For the $\gl_1$ part, one might proceed as follows. First, we have generalization of \eqref{eqn: W^3W^r preliminary, K=1} to $K>1$:
\begin{equation}\label{eqn: W^3W^r preliminary, K>1}
\begin{split}
    W^{a(3)}_{a,-2}W^{b(r)}_{b,-r}|0\rangle &\equiv  \mu W^{b(r+1)}_{b,-r-2}|0\rangle+\sum_{s=1}^{\lfloor \frac{r+2}{2}\rfloor}\left(\nu_{s,1}W^{c(s)}_{c,-s}W^{d(r+2-s)}_{d,s-r-2}+\nu_{s,2}W^{c(s)}_{d,-s}W^{d(r+2-s)}_{c,s-r-2}\right)|0\rangle\\
    &\pmod{V_{r-1}\mathcal W^{(K)}_L},
\end{split}
\end{equation}
for some $\mu,\nu_{s,1},\nu_{s,2}\in \mathbb C[\alpha]$. However, the argument in the proof of Lemma \ref{lem: vanishing of nu_s} does not work when $K>1$, because the image of left-hand-side of \eqref{eqn: W^3W^r preliminary, K>1} in the Zhu's $C_2$ algebra of $V^{\alpha,1}(\gl_K)^{\otimes L}$ has leading order $r+2$, in contrast to $r+1$ when $K=1$. So we proceed in an inductive way instead.

\begin{lemma}\label{lem: induction argument}
Assume that for a fixed $L>n>1$ and arbitrary $m\in \mathbb N$, $\Psi_L(\mathsf T_{n,m}(E^a_b))\equiv (-\epsilon_1)^n W^{a(n+1)}_{b,m-n}\pmod{V_{n-1}\mathfrak U(\mathcal W^{(K)}_L)}$, then $\exists f_{n+1}\in \mathbb C[\epsilon_1,\alpha]$ such that 
\begin{align*}
    \Psi_L(\mathsf T_{n+1,m}(E^a_b))\equiv (-\epsilon_1)^{n+1} W^{a(n+2)}_{b,m-n-1}+f_{n+1}\delta^a_bW^{c(n+2)}_{c,m-n-1}\pmod{V_{n}\mathfrak U(\mathcal W^{(K)}_L)}.
\end{align*}
Moreover, $f_{n+1}$ does not depend on $L$.
\end{lemma}

\begin{proof}
Consider the identity $\Psi_L([\mathsf t_{2,0},\mathsf T_{n,0}(1)])=0$, and let us replace its left-hand-side by the image in $\mathfrak{U}(\mathcal W^{(K)}_L)$, namely
\begin{align*}
    0=[\Psi_L(\mathsf t_{2,0}),\Psi_L(\mathsf T_{n,0}(1))]\equiv [\frac{\epsilon_1}{\bar\alpha}W^{a(3)}_{a,-2},(-\epsilon_1)^n W^{b(n+1)}_{b,-n} ]\pmod{V_{n}\mathfrak{U}(\mathcal W^{(K)}_L)}.
\end{align*}
Plug \eqref{eqn: W^3W^r preliminary, K>1} into the above equation, and we find
\begin{align*}
    0\equiv \sum_{s=1}^{\lfloor \frac{n+3}{2}\rfloor}\left(\nu_{s,1}\sum_{k\in \mathbb Z}:W^{c(s)}_{c,-k}W^{d(n+3-s)}_{d,k-n-2}:+\nu_{s,2}\sum_{l\in \mathbb Z}:W^{c(s)}_{d,-l}W^{d(n+3-s)}_{c,l-n-2}:\right)\pmod{V_{n}\mathfrak{U}(\mathcal W^{(K)}_L)}.
\end{align*}
This implies that $\nu_{s,1}=\nu_{s,2}=0$ for all $1\le s\le \lfloor \frac{n+3}{2}\rfloor$.  Then one deduces from \eqref{eqn: W^3W^r preliminary, K>1} that 
\begin{align*}
    [W^{a(3)}_{a,-2},W^{b(n+1)}_{b,m}]\equiv -(n+m)\mu W^{b(n+2)}_{b,m-2}\pmod{V_{n-1}\mathfrak{U}(\mathcal W^{(K)}_L)}.
\end{align*}
Thus we have
\begin{align*}
    \Psi_L(\mathsf T_{n+1,m}(1))&=\frac{1}{2(m+1)}[\Psi_L(\mathsf t_{2,0}),\Psi_L(\mathsf T_{n,m+1}(1))]\equiv \frac{-(-\epsilon_1)^{n+1}}{2(m+1)\bar\alpha}[W^{a(3)}_{a,-2},W^{b(n)}_{b,m+1-n}]\\
    &\equiv \frac{(-\epsilon_1)^{n+1}\mu}{2\bar\alpha}W^{b(n+2)}_{b,m-1-n}.
\end{align*}
This proves the first statement, with $f_{n+1}=\frac{(-\epsilon_1)^{n+1}}{K}(\frac{\mu}{2\bar\alpha}-1)$. Note that $f_{n+1}\in \mathbb C[\epsilon_1,\alpha]$ by Proposition \ref{prop: image of D^K}. Finally, the statement that $f_{n+1}$ does not depend on $L$ is proven in the same way as Lemma \ref{lem: mu independent of L} and we omit the detail.
\end{proof}



\subsection{Duality isomorphism of the rectangular W-algebra}\label{subsec: duality of W}
The rectangular $\mathcal W$-algebra that we have discussed so far is denoted by $\mathcal W^{(K)}_{0,0,L}$ in the literature \cite{gaiotto2022miura}. In this subsection we introduce another realization of the rectangular $\mathcal W$-algebra, which is the $\mathcal W^{(K)}_{0,L,0}$ in \cite{gaiotto2022miura}.

We define the $\mathbb C[\bar\alpha]$-vertex algebra $\widetilde{\mathcal W}^{(K)}_{L}$ to be the vertex subalgebra of $V^{\kappa_{\bar\alpha}}(\mathfrak{gl}_K)^{\otimes L}$ which is strongly generated by the fields $U^{a(r)}_{b}(z)=\sum_{m\in \mathbb Z}U^{a(r)}_{b,m}z^{-m-r}$ in the Miura operator:
\begin{equation}
\begin{split}
 (\bar\alpha\partial+J^{a[1]}_b(z)E^b_a)&(\bar\alpha\partial+J^{a[2]}_b(z)E^b_a)\cdots(\bar\alpha\partial+J^{a[L]}_b(z)E^b_a)\\
    &=(\bar\alpha\partial)^L+\sum_{r=1}^L (\bar\alpha\partial)^{L-r}U^{a(r)}_b(z)E^b_a,
\end{split}
\end{equation}
where $E^b_a$ is the elementary matrix, and $J^{a[i]}_b(z),i=1,\cdots,L$ are $L$ copies of affine Kac-Moody currents in $V^{\kappa_{\bar\alpha}}(\mathfrak{gl}_K)$. $\widetilde V_{-2}$ is a mode of quasi-primary field $\widetilde V(z)=\sum_{n\in \mathbb Z}\widetilde V_{n}z^{-n-3}$ defined as
\begin{equation}
\begin{split}
\widetilde V(z):=&\frac{1}{6}\sum_{i=1}^L\left(:J^{a[i]}_b(z)J^{b[i]}_c(z)J^{c[i]}_a(z):+:J^{b[i]}_a(z)J^{c[i]}_b(z)J^{a[i]}_c(z):\right)\\
&+\alpha\sum_{i<j}^L:J^{a[i]}_b(z)\partial J^{b[j]}_a(z):+\sum_{i<j}^L:J^{a[i]}_a(z)\partial J^{b[j]}_b(z):\\
=&\frac{1}{6}\left(:U^{a(1)}_b(z)U^{b(1)}_c(z)U^{c(1)}_a(z):+:U^{b(1)}_a(z)U^{c(1)}_b(z)U^{a(1)}_c(z):\right)\\
&+U^{a(3)}_a(z)-:U^{a(1)}_b(z)U^{b(2)}_a(z):+\text{ total derivatives},
\end{split}
\end{equation}
from the last equation we see that $\widetilde V_{-2}\in \mathfrak{U}(\widetilde{\mathcal W}^{(K)}_L)$. 

There exists an isomorphism of vertex algebras $\sigma_L:\mathcal W^{(K)}_L\cong \widetilde{\mathcal W}^{(K)}_L$ which is induced by 
\begin{align*}
    \sigma_L(\alpha)=\bar\alpha,\quad \sigma_L(J^{a[i]}_{b}(z))=-J^{b[i]}_{a}(z).
\end{align*}
Composing the duality transform \eqref{eqn: duality for A} $\sigma:\mathsf A^{(K)}\cong \mathsf A^{(K)}$ with the representation $\Psi_L: \mathsf A^{(K)}\to \mathfrak{U}(\mathcal W^{(K)}_L)[\bar\alpha^{-1}]$, and then applying the isomorphism $\sigma_L:\mathcal W^{(K)}_L\cong \widetilde{\mathcal W}^{(K)}_L$, we get a new map $$\widetilde\Psi_L=\sigma_L\circ\Psi_L\circ\sigma:\mathsf A^{(K)}\to \mathfrak{U}(\widetilde{\mathcal W}^{(K)}_L)[\alpha^{-1}]$$ which is uniquely determined by
\begin{equation}\label{eqn: another map to W-algebra}
\boxed{
\begin{aligned}
    &\widetilde\Psi_L(\mathsf T_{0,n}(E^a_b))= U^{a(1)}_{b,n}+\frac{1}{\alpha}\delta^a_bU^{c(1)}_{c,n},\\
    &\widetilde\Psi_L(\mathsf t_{1,n})=-\mathsf L_{n-1}-\frac{\bar\alpha n L}{2\alpha}U^{a(1)}_{a,n-1}\\
    &\widetilde\Psi_L(\mathsf T_{1,0}(E^a_b))= -\epsilon_1\sum_{m\ge 0}\left(U^{c(1)}_{b,-m-1}U^{a(1)}_{c,m}+\frac{1}{\alpha}\delta^a_bU^{c(1)}_{d,-m-1}U^{d(1)}_{c,m}\right)+\epsilon_1 U^{a(2)}_{b,-1}+\frac{\epsilon_1}{\alpha}\delta^a_bU^{c(2)}_{c,-1}\\
    &\widetilde\Psi_L(\mathsf t_{2,0})=-\frac{\epsilon_1}{\alpha}\left(\widetilde V_{-2}+\alpha\sum_{n=1}^{\infty}n\:U^{a(1)}_{b,-n-1}U^{b(1)}_{a,n-1}+\sum_{n=1}^{\infty}n\: U^{a(1)}_{a,-n-1}U^{b(1)}_{b,n-1}\right).
\end{aligned}}
\end{equation}
Similarly we can define algebra homomorphism
$$\widetilde\Delta_L=(\sigma\otimes\sigma_L)\circ\Delta_L\circ\sigma: \mathsf A^{(K)}\to \mathsf A^{(K)}\widetilde\otimes\mathfrak{U}(\widetilde{\mathcal W}^{(K)}_L)[\alpha^{-1}]$$ which is uniquely determined by
\begin{equation}
\boxed{
\begin{aligned}
&\widetilde\Delta_L(\mathsf T_{0,n}(E^a_b))=\square( \mathsf T_{0,n}(E^a_b)),\\
&\widetilde\Delta_L(\mathsf t_{1,n})=\square(\mathsf t_{1,n})+\epsilon_3 nL\mathsf t_{0,n-1}\otimes 1,\\
&\widetilde\Delta_L(\mathsf T_{1,0}(E^a_b))=\square(\mathsf T_{1,0}(E^a_b))+\epsilon_1\sum_{m=0}^{\infty}\left(\mathsf T_{0,m}(E^c_b)\otimes U^{a(1)}_{c,-m-1}-\mathsf T_{0,m}(E^a_c)\otimes U^{c(1)}_{b,-m-1}\right),\\
&\widetilde\Delta_L(\mathsf t_{2,0})=\square(\mathsf t_{2,0})-2\epsilon_1\sum_{n=1}^{\infty}n \:\mathsf T_{0,n-1}(E^a_b)\otimes U^{b(1)}_{a,-n-1},
\end{aligned}}
\end{equation}
where $\square(x):=x\otimes 1+1\otimes \widetilde\Psi_L(x)$.

\bigskip The maps $\widetilde\Delta_L$ and $\Delta_{L_1,L_2}$ are compatible in the following sense:
\begin{equation}\label{compatible between coproducts_cd tilde}
(\widetilde\Delta_{L_1}\otimes 1)\circ \widetilde\Delta_{L_2}=(1\otimes\Delta_{L_1,L_2})\circ \widetilde\Delta_{L_1+L_2}
\end{equation}
Recall that $\widetilde{\mathsf D}^{(K)}$ is the subalgebra of $\mathsf A^{(K)}$ defined as $\sigma(\mathsf D^{(K)})$ (Definition \ref{def: D tilde}). By the Proposition \ref{prop: image of D^K} and our construction of $\widetilde\Delta_L$ and $\widetilde\Psi_L$, the image of $\widetilde{\mathsf D}^{(K)}$ under the map $\widetilde\Delta_L$ is contained in the subalgebra $\widetilde{\mathsf D}^{(K)}\widetilde\otimes\mathfrak{U}(\widetilde{\mathcal W}^{(K)}_L)$, and the image of $\widetilde{\mathsf D}^{(K)}$ under the map $\widetilde\Psi_L$ is contained in the subalgebra $\mathfrak{U}(\widetilde{\mathcal W}^{(K)}_L)$.

\subsection{An anti-involution of the mode algebra of the rectangular W-algebra}

Consider the following anti-involution of $\mathfrak{U}(\widehat{\gl}^\alpha_K)^{\widehat{\otimes} L}$:
\begin{align}\label{eqn: anti-involution s_L}
    \mathfrak s_L: J^{a[i]}_{b,n}\mapsto J^{b[L+1-i]}_{a,-n}+(2i-L-1)\alpha\delta^a_b\delta_{n,0}, \quad 1\le i\le L,1\le a,b\le K,n\in \mathbb Z.
\end{align}
\begin{lemma}
The mode algebra of rectangular W-algebra is invariant under $\mathfrak s_L$, i.e. $\mathfrak s_L(\mathfrak{U}(\mathcal{\mathcal W}^{(K)}_L))=\mathfrak{U}(\mathcal{\mathcal W}^{(K)}_L)$.
\end{lemma}

\begin{proof}
Let us extend $\mathfrak s_L$ to an anti-involution on the mode-valued differential operators by defining
\begin{align*}
    \mathfrak s_L(z)=z^{-1},\quad \mathfrak s_L(\partial_z)=z^{2}\partial_z.
\end{align*}
Then $\mathfrak s_L$ acts on the Miura operator by
\begin{align*}
&\mathfrak s_L\left((\alpha\partial-J^{[1]}(z))^{a_1}_{a_2}(\alpha\partial-J^{[2]}(z))^{a_2}_{a_3}\cdots(\alpha\partial-J^{[L]}(z))^{a_L}_{a_{L+1}}\right)=\\
=&(\alpha z^2\partial-(L-1)\alpha z-z^2J^{[1]}(z))^{a_{L+1}}_{a_L}\cdots(\alpha z^2\partial+(L-3)\alpha z-z^2J^{[L-1]}(z))^{a_3}_{a_2}\times\\
&\times(\alpha z^2\partial+(L-1)\alpha z-z^2J^{[L]}(z))^{a_2}_{a_1}\\
=& z^{L+1} (\alpha \partial-J^{[1]}(z))^{a_{L+1}}_{a_L}\cdots(\alpha \partial-J^{[L]}(z))^{a_2}_{a_1} z^{L-1}\\
=&z^{L+1} \left((\alpha\partial)^L+\sum_{r=1}^L (-1)^r(\alpha\partial)^{L-r}W^{a_{L+1}(r)}_{a_{1}}(z)\right)z^{L-1}.
\end{align*}
On the other hand,
\begin{align*}
&\mathfrak s_L\left((\alpha\partial)^L+\sum_{r=1}^L (-1)^r(\alpha\partial)^{L-r}W^{a_1(r)}_{a_{L+1}}(z)\right)=\\
=&(\alpha z^2\partial)^L+\sum_{r=1}^L (-1)^r\sum_{n\in \mathbb Z}\mathfrak s_L\left(W^{a_1(r)}_{a_{L+1},n}\right)z^{n+r} (\alpha z^2\partial)^{L-r},
\end{align*}
therefore we see that $\mathfrak s_L\left(W^{a(r)}_{b,n}\right)$ is a linear combination of $W^{b(r)}_{a,-n}, W^{b(r-1)}_{a,-n},\cdots ,W^{b(1)}_{a,-n}$ and $\delta^{a}_{b}\delta_{n,0}$. This proves the lemma.
\end{proof}

Consider the anti-involution $\mathfrak{s}_{\mathsf A}:\mathsf A^{(K)}\cong \mathsf A^{(K)\mathrm{op}}$ such that
\begin{align}\label{eqn: anti-involution s_A}
    \mathfrak{s}_{\mathsf A}(\mathsf t_{n,m})=(-1)^n \mathsf t_{n,m},\quad \mathfrak{s}_{\mathsf A}(\mathsf T_{n,m}(X))=(-1)^n \mathsf T_{n,m}(X^{\mathrm{t}}).
\end{align}

\begin{definition}\label{def: Psi minus}
We define the algebra homomorphism $\Psi^-_L:\mathsf A^{(K)}\to \mathfrak{U}(\mathcal W^{(K)}_L)[\bar\alpha^{-1}]$ to be the composition
\begin{align}
    \Psi^-_L:=\mathfrak{s}_L\circ \Psi_L\circ \mathfrak{s}_{\mathsf A}.
\end{align}
\end{definition}
Direct computation shows that
\begin{equation}\label{eqn: Psi^- obvious equations}
\begin{split}
&\Psi^-_L(\mathsf T_{0,n}(E^a_b))=W^{a(1)}_{b,-n},\quad \Psi^-_L(\mathsf t_{0,n})=\frac{1}{\epsilon_2}W^{a(1)}_{a,-n},\\
&\Psi^-_L(\mathsf t_{1,n})=\mathsf L_{1-n}-\frac{\alpha nL}{2\bar\alpha}W^{a(1)}_{a,1-n},\\
&\Psi^-_L(\mathsf T_{1,1}(E^a_b))=-\Psi_L(\mathsf T_{1,1}(E^a_b)).
\end{split}
\end{equation}
A less obvious equation is the following
\begin{equation}\label{eqn: Psi^-(t[2,2])}
    \Psi^-_L(\mathsf t_{2,2})=\Psi_L(\mathsf t_{2,2}).
\end{equation}
To see this equation, let us write $\Psi_L(\mathsf t_{2,2})$:
\begin{equation}
\begin{split}
&\Psi_L(\mathsf t_{2,2})=\frac{1}{6}[\Psi_L(\mathsf t_{2,0}),\Psi_L(\mathsf t_{1,3})]\\
=&\frac{\epsilon_1}{\bar\alpha}\sum_{i=1}^L\Bigg(\frac{1}{6}\sum_{k,l\in \mathbb Z}\left(:J^{a[i]}_{b,-k-l}J^{b[i]}_{c,k} J^{c[i]}_{a,l}:+:J^{b[i]}_{a,-k-l}J^{c[i]}_{b,k} J^{a[i]}_{c,l}:\right)\\
&~+\left(L-i+\frac{1}{2}\right)\alpha\sum_{m\in \mathbb Z}:J^{a[i]}_{b,-m}J^{b[i]}_{a,m}:+\left((L-i)(L+1-i)+\frac{1}{2}\right)\alpha^2J^{a[i]}_{a,0}\\
&~-\sum_{n=1}^{\infty}n(\bar\alpha J^{a[i]}_{b,-n}J^{b[i]}_{a,n}+ J^{a[i]}_{a,-n}J^{b[i]}_{b,n})-\frac{1}{3}(\bar\alpha J^{a[i]}_{b,0}J^{b[i]}_{a,0}+J^{a[i]}_{a,0}J^{b[i]}_{b,0})\Bigg)\\
&~-\frac{2\epsilon_1}{\bar\alpha}\sum_{i<j}^L\bigg(\sum_{m=1}^\infty m(\bar\alpha J^{a[i]}_{b,m}J^{b[j]}_{a,-m}+ J^{a[i]}_{a,m}J^{b[j]}_{b,-m})+\frac{1}{3}(\bar\alpha J^{a[i]}_{b,0}J^{b[j]}_{a,0}+J^{a[i]}_{a,0}J^{b[j]}_{b,0})\bigg),
\end{split}
\end{equation}
then it is straightforward to check that $\mathfrak{s}_L(\Psi_L(\mathsf t_{2,2}))=\Psi_L(\mathsf t_{2,2})$, which implies \eqref{eqn: Psi^-(t[2,2])}.

\subsection{Restricted mode algebra of the rectangular W-algebra}
For any graded vertex algebra $\mathcal V$, there is a notion of restricted mode algebra $U(\mathcal V)$ (Definition \ref{def: restricted mode alg}), together with an algebra homomorphism $U(\mathcal V)\to \mathfrak{U}(\mathcal V)$ which is injective when $\mathcal V$ satisfies certain reasonable technical assumptions (Proposition \ref{prop: U(V) is a sub of mode algebra}). It is known that W-algebras satisfy the assumptions in Proposition \ref{prop: U(V) is a sub of mode algebra}, so we can regard $U(\mathcal W^{(K)}_L)$ as a subalgebra of $\mathfrak{U}(\mathcal W^{(K)}_L)$. Using the notation of Definition \ref{def: restricted mode alg} we can rewrite the image of $\mathsf t_{2,0}$ under the maps $\Psi_L$ and $\widetilde\Psi_L$ as
\begin{equation}\label{eqn: image of t[2,0] in restricted modes}
\begin{split}
\Psi_L(\mathsf t_{2,0})&=\frac{\epsilon_1}{\bar\alpha} \left(V_{-2}- \mathcal O\left(W^{a(1)}_{b,-1}|0\rangle,W^{b(1)}_{a,-1}|0\rangle;\frac{\bar\alpha}{(z_1-z_2)^2}\right)-\mathcal O\left(W^{a(1)}_{a,-1}|0\rangle,W^{b(1)}_{b,-1}|0\rangle;\frac{1}{(z_1-z_2)^2}\right)\right),\\
\widetilde\Psi_L(\mathsf t_{2,0})&=-\frac{\epsilon_1}{\alpha} \left(\widetilde V_{-2}+\mathcal O\left(U^{a(1)}_{b,-1}|0\rangle,U^{b(1)}_{a,-1}|0\rangle;\frac{\alpha}{(z_1-z_2)^2}\right)+\mathcal O\left(U^{a(1)}_{a,-1}|0\rangle,U^{b(1)}_{b,-1}|0\rangle;\frac{1}{(z_1-z_2)^2}\right)\right),
\end{split}
\end{equation}
in particular, the images $\Psi_L(\mathsf A^{(K)})$ (respectively $\widetilde\Psi_L(\mathsf A^{(K)})$) is in the positive modes subalgebra $U_+(\mathcal W^{(K)}_L)$ (respectively $U_+(\widetilde{\mathcal W}^{(K)}_L)$).

\section{Matrix Extended \texorpdfstring{$\mathcal W_{\infty}$}{W inf} Vertex Algebra}\label{sec: W(infinity)}
In this section, we construct the large-$L$-limit of the rectangular W-algebra $\mathcal W^{(K)}_L$, denoted by $\mathcal W^{(K)}_{\infty}$, then all results in the previous section about $\mathcal W^{(K)}_L$ can be packaged into statements about $\mathcal W^{(K)}_{\infty}$. This construction generalizes the $K=1$ case in \cite{linshaw2021universal}. The q-deformed version of matrix extended $\mathcal W_\infty$ is studied in \cite{neguct2022deformed}.

$\mathcal W^{(K)}_{\infty}$ will be a vertex algebra over the base ring $\mathbb C[\alpha,\lambda]$. We define the underlying $\mathbb C[\alpha,\lambda]$ module using the large-$L$ limit of the PBW basis of $\mathcal W^{(K)}_L$, where $L$ will be promoted to the formal variable $\lambda$. The state-operators map is roughly speaking the analytic continuation of state-operators maps of $\mathcal W^{(K)}_L$, after treating $L$ as a formal variable. As we will show later in this section, all structure constants in the $W^{a(r)}_{b}$ basis are polynomials in $\alpha$ and $L$, and this enables us to define the structure constants of $\mathcal W^{(K)}_{\infty}$ by these polynomials. Note that we do not provide the explicit formula of these polynomials, instead we only show their existence, i.e. our approach is implicit. An explicit formula of the structure constants was conjectured in \cite{eberhardt2019matrix}.

We state the main result of this section as follows.

\begin{theorem}\label{thm: matrix extended W(infty)}
For every $K\in \mathbb N_{\ge 1}$, there exists a $\mathbb Z$-graded vertex algebra $\mathcal W^{(K)}_{\infty}$ over the base ring $\mathbb C[\alpha,\lambda]$ with strong generators $W^{a(r)}_b(z),1\le a,b\le K,r=1,2,\cdots$ and
\begin{itemize}
    \item a $\mathbb C[\alpha]$-vertex algebra map $\Delta_{\mathcal W}:\mathcal W^{(K)}_{\infty}\to \mathcal W^{(K)}_{\infty}\otimes \mathcal W^{(K)}_{\infty}$,
    \item a $\mathbb C[\alpha]$-vertex algebra map $\pi_L:\mathcal W^{(K)}_{\infty}\to \mathcal W^{(K)}_{L}$ for every positive integer $L$,
\end{itemize}
such that 
\begin{itemize}
    \item[(1)] $W^{a(1)}_b(z)$ generates a $\mathbb C[\alpha,\lambda]$-vertex subalgebra $V^{\kappa_{\lambda\alpha,\lambda}}(\gl_K)\hookrightarrow \mathcal W^{(K)}_{\infty}$, where $\kappa_{\lambda\alpha,\lambda}$ is the inner form $\kappa_{\lambda\alpha,\lambda}(E^a_b,E^c_d)=\lambda\alpha\delta^c_b\delta^a_d+\lambda\delta^a_b\delta^c_d$.
    \item[(2)] $ \Delta_{\mathcal W}(\lambda)=\lambda\otimes 1+1\otimes \lambda$, and $\Delta_{\mathcal W}$ acts on strong generators by 
    \begin{equation}\label{eqn: W(infinity) coproduct}
    \begin{split}
    \Delta_{\mathcal W}(W^{a(r)}_b(z))=\sum_{\substack{(s,t,u)\in\mathbb N^3\\s+t+u=r}} \binom{1\otimes \lambda-t}{u}(\alpha\partial)^u W^{a(s)}_c(z)\otimes W^{c(t)}_b(z),
    \end{split}
    \end{equation}
    where we set $W^{a(0)}_b(z)=\delta^a_b$.
    \item[(3)] $\pi_L(\lambda)=L$ and $\pi_L(W^{a(r)}_b(z))=W^{a(r)}_b(z)$ for $r=1,\cdots,L$ and $\pi_L(W^{a(r)}_b(z))=0$ for $r>L$.
\end{itemize}
We call $\mathcal W^{(K)}_{\infty}$ the $\mathfrak{gl}_K$-extended $\mathcal W_{\infty}$ algebra\footnote{In the usual convention, this will be called the matrix extended $\mathcal W_{1+\infty}$ algebra, and the notation $\mathcal W_{\infty}$ is reserved for the decomposition $\mathcal W_{1+\infty}=V(\gl_1)\otimes \mathcal W_{\infty}$. Here we drop ``$1+$'' to simplify notation because we do not consider quotienting out the $\gl_1$ component in this paper.}. 
\end{theorem}

It is obvious from the theorem that $\Delta_{\mathcal W}$ is compatible with the finite coproduct $\Delta_{L_1,L_2}$:
\begin{align}
    (\pi_{L_1}\otimes\pi_{L_2})\circ \Delta_{\mathcal W}=\Delta_{L_1,L_2}\circ \pi_{L_1+L_2}.
\end{align}
And it is coassociative: $(\Delta_{\mathcal W}\otimes 1)\circ \Delta_{\mathcal W}=(1\otimes \Delta_{\mathcal W})\circ \Delta_{\mathcal W}$.\\

We also note that the automorphism $\eta_{\beta}^{\otimes L}$ defined in the Remark \ref{rmk: shift automorphism of W} lifts to an automorphism $\boldsymbol{\eta}_{\beta}:\mathfrak{U}(\mathcal W^{(K)}_{\infty})\cong \mathfrak{U}(\mathcal W^{(K)}_{\infty})$ such that 
\begin{align}
    \pi_L\circ\boldsymbol{\eta}_{\beta}=\eta_{\beta}^{\otimes L}\circ \pi_L,
\end{align}
and $\tau_{\beta}$ maps the currents as follows:
\begin{align}\label{eqn: shift automorphism of W}
    \boldsymbol{\eta}_\beta(W^{(r)}(z))= W^{(r)}(z)+\sum_{s=1}^r \binom{\lambda+s-r}{s}
    \binom{{\beta}/{\alpha}}{s}\frac{\alpha^s \cdot s!}{z^s}  W^{(r-s)}(z),
\end{align}
where $W^{(0)}(z)$ is set to be the constant identity matrix.\\

The rest of this section is devoted to the proof of Theorem \ref{thm: matrix extended W(infty)}. The difficult part is to show that structure constants in $W^{a(r)}_{b}$ basis are polynomials in $\alpha$ and $L$. The main tool that we use is the (matrix extended) pseododifferential symbols, namely we consider the correlation functions of the Miura operators of $\mathcal W^{(K)}_L$ and show that their dependence on $L$ are polynomials, and this allows us to promote these correlation functions to pseododifferential operators from which we extract the set of correlation functions between $W^{a(r)}_b,r=1,2\cdots$ such that $L$ becomes a parameter $\lambda$ and these correlation functions depend on $\lambda$ in a polynomial way, and these correlators are essentially equivalent to the state-operator map. 

\subsection{Matrix-valued pseudodifferential symbols}

\begin{definition}
Let $R$ be a unital associative (possibly non-commutative) $\mathbb C$-algebra, an $R$-valued pseudodifferential symbol $D$ of $n$ variables is the following formal expression
\begin{align}
    D=\partial_{x_1}^{\mu_1}\cdots \partial_{x_n}^{\mu_n}+\sum_{r\in \mathbb N_{\ge 0}^{n}\backslash 0}\partial_{x_1}^{\mu_1-r_1}\cdots \partial_{x_n}^{\mu_n-r_n}\cdot D_r,
\end{align}
where $(\mu_1,\cdots,\mu_n)$ is a set of $n$ formal variables which will be called the leading order, $D_r\in \mathbb C(x_1,\cdots,x_n)\otimes R$, i.e. $R$-valued rational functions in variables $x_1,\cdots,x_n$.
The space of such symbols is denoted by $\Psi\mathrm{DS}_n(R)$. We also define a multiplication on $\Psi\mathrm{DS}_n(R)$ by extending the commutation relation
\begin{align}\label{eq: PDS commutator}
    [f(x),\partial_{x_i}^{\mu}]=\sum_{s\ge 1}(-1)^s \binom{\mu}{s}\partial_{x_i}^{\mu-s}\cdot\left(\partial_{x_i}^sf(x)\right).
\end{align}
\end{definition}

\begin{proposition}\label{prop: formal powers of PsiDS}
Let $D\in \Psi\mathrm{DS}_n(R)$ be an $R$-valued pseudodifferential symbol with leading order $(\mu_1,\cdots,\mu_n)$, then there exists a unique $R$-valued pseudo-differential symbol $D^{\lambda}\in \Psi\mathrm{DS}_n(R[\lambda])$ with leading order $(\mu_1\lambda,\cdots,\mu_n\lambda)$, written as
\begin{align}
    D^{\lambda}=\partial_{x_1}^{\mu_1\lambda}\cdots \partial_{x_n}^{\mu_n\lambda}+\sum_{r\in \mathbb N_{\ge 0}^{n}\backslash 0}\partial_{x_1}^{\mu_1\lambda-r_1}\cdots \partial_{x_n}^{\mu_n\lambda-r_n}\cdot D_r(\lambda),
\end{align}
such that
\begin{itemize}
    \item[(1)] $D_r(\lambda)\in \mathbb C(x_1,\cdots,x_n)\otimes R[\lambda]$, i.e. $\mathbb C(x_1,\cdots,x_n)\otimes R$-valued polynomials in $\lambda$.
    \item[(2)] For every positive integer $l$, the $l$-th power of $D$ is obtained from $D^{\lambda}$ by specializing $\lambda=l$.
\end{itemize}
\end{proposition}

\begin{proof}
The uniqueness follows from the polynomial dependence of $D_r$ on $\lambda$ together with the condition that $D^{l}=D^{\lambda}|_{\lambda=l}$ for all $l\in \mathbb N_{>0}$. It remains to show the existence. Let us write
\begin{align*}
    D^{l}=\partial_{x_1}^{\mu_1 l}\cdots \partial_{x_n}^{\mu_n l}+\sum_{r\in \mathbb N_{\ge 0}^{n}\backslash 0}\partial_{x_1}^{\mu_1 l-r_1}\cdots \partial_{x_n}^{\mu_n l-r_n}\cdot D_r(l), \;(l\in \mathbb N_{>0}),
\end{align*}
then it suffices to show that $D_r(l)$ depends on $l$ in a polynomial way. We proceed by induction on $|r|:=\sum_{i=1}^nr_i$. For simplicity we define $D_0(l)=\mathrm{id}_{R}$ which is constant in $l$.

Consider $D_r(l+1)-D_r(l)$, it is a polynomial in $l$ if and only if $D_r(l)$ is a polynomial in $l$. By direct computation of $D^{l+1}=D^l\cdot D$, we obtain:
\begin{align}\label{eqn: D(l+1)-D(l)}
D_r(l+1)-D_r(l)=\sum_{\substack{u,v,s\in \mathbb N^{n}_{\ge 0}\\u+v+s=r\\u\neq r}}(-1)^{|s|}\binom{\mu_1-v_1}{s_1}\cdots \binom{\mu_n-v_n}{s_n}\left(\partial_{x_1}^{s_1}\cdots \partial_{x_n}^{s_n}D_u(l)\right)\cdot D_{v}(1).
\end{align}
By the induction hypothesis, $D_u(l)$ are polynomials in $l$ for all $u\in \mathbb N^{n}_{\ge 0}$ such that $|u|<|r|$, so every summand in the right-hand-side of \eqref{eqn: D(l+1)-D(l)} is a polynomial in $l$, thus $D_r(l)$ is a polynomial in $l$. This completes the proof.
\end{proof}

To apply the above general results to our construction of matrix extended $\mathcal W_{\infty}$ algebra, let us consider the following $\gl_K^{\otimes n}$-valued differential symbol in $n$ variables: 
\begin{align}
    \mathcal D:=\langle(\partial_{x_1}\delta^{a_1}_{b_1}-\alpha^{-1}J^{a_1}_{b_1}(x_1))(\partial_{x_2}\delta^{a_2}_{b_2}-\alpha^{-1}J^{a_2}_{b_2}(x_2))\cdots (\partial_{x_n}\delta^{a_n}_{b_n}-\alpha^{-1}J^{a_n}_{b_n}(x_n))\rangle,
\end{align}
where $a_i$ and $b_i$ for $i=1,\cdots,n$ are $\gl_K$ indice. $\mathcal D$ is the correlator between $n$ copies of elementary Miura operators $\mathcal L^1_1(x_i)=\alpha\partial_{x_i}-J(x_i)$ (up to scaling), and we will write down an explicit formula of $\mathcal D$ in terms of Cherednik algebra elements in the next section, but its explicit form will not play a role in the construction of $\mathcal W^{(K)}_{\infty}$ algebra. 

According to Proposition \ref{prop: formal powers of PsiDS}, there exists a $\gl_K^{\otimes n}$-valued pseudodifferential symbol $\mathcal D^{\lambda}\in \Psi\mathrm{DS}_n(\gl_K^{\otimes n}[\lambda])$ written as
\begin{align}
    \mathcal D^{\lambda}=\partial_{x_1}^{\lambda}\cdots \partial_{x_n}^{\lambda}+\sum_{r\in \mathbb N_{\ge 0}^{n}\backslash 0}\partial_{x_1}^{\lambda-r_1}\cdots \partial_{x_n}^{\lambda-r_n}\cdot \mathcal D_r(\lambda),
\end{align}
such that
\begin{itemize}
    \item[(1)] $\mathcal D_r(\lambda)\in \mathbb C(x_1,\cdots,x_n)\otimes \gl_K^{\otimes n}[\lambda]$.
    \item[(2)] For every positive integer $L$, the $L$-th power of $\mathcal D$ is obtained from $\mathcal D^{\lambda}$ by specializing $\lambda=l$.
\end{itemize}
On the other hand, we can consider the composition between $L$ copies of elementary Miura operators $\mathcal L^L_1(x):=\mathcal L^1_1(x)^{[1]}\mathcal L^1_1(x)^{[2]}\cdots \mathcal L^1_1(x)^{[L]}$ where the superscript $[i]$ means $J^a_b(x)$ current takes value in the $i$-th copy of affine Kac-Moody algebra, and by the definition of the $\mathcal W^{(K)}_L$ algebra we have
\begin{align}\label{eqn: finite L Miura operator}
    \mathcal L^L_1(x)=(\alpha\partial_x)^L+\sum_{r=1}^L (-1)^r(\alpha\partial_x)^{L-r}W^{(r)}(x).
\end{align}
Since different copies of affine Kac-Moody algebras do not interact with each other, the correlators between these composite Miura operators completely decouples:
\begin{align}\label{eqn: correlators decouple}
    \langle \mathcal L^L_1(x_1)\mathcal L^L_1(x_2)\cdots \mathcal L^L_1(x_n)\rangle=\alpha^L\mathcal D^L.
\end{align}

\begin{proposition}
The correlator $\langle W^{a_1(r_1)}_{b_1}(x_1)\cdots W^{a_n(r_n)}_{b_n}(x_n)\rangle$ is defined for all $L\in \mathbb Z_{\ge 0}$ (including the cases when $L<r_i$) and its value depends on $L$ in a polynomial way. Moreover it vanishes for $L< \max(r_1,\cdots,r_n)$.
\end{proposition}\label{prop: U-correlators are polynomials}
\begin{proof}
If $L\ge \max(r_1,\cdots,r_n)$, then expand the left-hand-side of \eqref{eqn: correlators decouple} and we find
\begin{align*}
\sum_{r\in \mathbb N^n_{[0,L]}}(-1)^{\sum_{i=1}^nr_i}(\alpha\partial_{x_1})^{L-r_1}\cdots (\alpha\partial_{x_n})^{L-r_n}\cdot \langle W^{a_1(r_1)}_{b_1}(x_1)\cdots W^{a_n(r_n)}_{b_n}(x_n)\rangle ,
\end{align*}
compare with the right-hand-side and we find
\begin{align}
\langle W^{a_1(r_1)}_{b_1}(x_1)\cdots W^{a_n(r_n)}_{b_n}(x_n)\rangle =\alpha^{|r|}\mathcal D_r(L)^{a_1\cdots a_n}_{b_1\cdots b_n},
\end{align}
thus $\langle W^{a_1(r_1)}_{b_1}(x_1)\cdots W^{a_n(r_n)}_{b_n}(x_n)\rangle$ depend on $L$ in a polynomial way. Now the polynomial $\mathcal D_r(L)$ vanishes if $L\in \mathbb Z_{\ge 1}$ and $L<\max(r_1,\cdots,r_n)$ so we can define $$\langle W^{a_1(r_1)}_{b_1}(x_1)\cdots W^{a_n(r_n)}_{b_n}(x_n)\rangle:=0$$ in this case.
\end{proof}

The vanishing of the correlator $\langle W^{a_1(r_1)}_{b_1}(x_1)\cdots W^{a_n(r_n)}_{b_n}(x_n)\rangle$ when $L\in \mathbb Z_{\ge 1}$ and $L<\max(r_1,\cdots,r_n)$ is compatible with the fact that there is no operator $W^{a(r)}_{b}(z)$ in $\mathcal W^{(K)}_L$ if $r>L$. However, the Proposition \ref{prop: U-correlators are polynomials} should be better understood as a result of analytic continuation at this point. 

Similar technique can be used to prove the following generalization.
\begin{proposition}\label{prop: correlators are polynomials in L}
The correlator $\langle J^{c_1[s_1]}_{d_1}(y_1)\cdots J^{c_m[s_m]}_{d_m}(y_m) W^{a_1(r_1)}_{b_1}(x_1)\cdots W^{a_n(r_n)}_{b_n}(x_n)\rangle$ is defined for all $L\ge \max(s_1,\cdots,s_m)$ and its value depends on $L$ in a polynomial way. Moreover it vanishes for $L< \max(r_1,\cdots,r_n)$.

Equivalently, $\langle 0| J^{c_1[s_1]}_{d_1,k_1}\cdots J^{c_m[s_m]}_{d_m,k_m} W^{a_1(r_1)}_{b_1,i_1}\cdots W^{a_n(r_n)}_{b_n,i_n}|0\rangle$ is defined for all $L\ge \max(s_1,\cdots,s_m)$ and its value depends on $L$ in a polynomial way. Moreover it vanishes for $L<\max(r_1,\cdots,r_n)$.
\end{proposition}

\subsection{The underlying vector space of \texorpdfstring{$\mathcal W^{(K)}_{\infty}$}{Wk inf}}
In this subsection, we define the underlying vector space of $\mathcal W^{(K)}_{\infty}$.

\begin{definition}\label{def: vector space of W(infinity)}
Let $\mathbf W$ be the $\mathbb C$-vector space generated by the basis
$\Omega$ and $A^{a_m,s_m}_{b_m,n_m}\cdots A^{a_1,s_1}_{b_1,n_1}$ for those integer indices $ 1\le a_i,b_i\le K$ and $s_m\ge \cdots \ge s_1\ge 1$ and $n_i\le -s_i$ such that for every $1\le i<m$ either $s_{i+1}>s_i$ or $n_{i+1}\le n_i$ holds.

We give a $\mathbb Z$-grading on $\mathbf W$ by setting $\deg \Omega=0$ and $\deg A^{a_m,s_m}_{b_m,n_m}\cdots A^{a_1,s_1}_{b_1,n_1}=\sum_{i=1}^m n_i$. Write the homogeneous decomposition $\mathbf W=\bigoplus_{d\in \mathbb Z_{\le 0}}\mathbf W_d$, note that each $\mathbf W_d$ is finite dimensional. Write $\mathbf W_{\ge n}:=\bigoplus_{d=n}^0\mathbf W_d$ for $n\in \mathbb Z_{\le 0}$. 

We also give an increasing $\mathbb N$-filtration $F_{\bullet}\mathbf W$ by letting $F_l \mathbf W$ be the span of the basis elements $\Omega$ and $A^{a_m,s_m}_{b_m,n_m}\cdots A^{a_1,s_1}_{b_1,n_1}$ such that $s_m\le l$, in particular $F_0\mathbf W$ is spanned by the element $\Omega$. Write $F_l\mathbf W_d:=F_l\mathbf W\cap \mathbf W_d$ and $F_l\mathbf W_{\ge n}:=F_l\mathbf W\cap \mathbf W_{\ge n}$.
\end{definition}

It follows immediately from the definition that $\mathbf W_{\ge -n}\subset F_n\mathbf W$. Moreover, $F_{l}\mathbf W$ has a linear complement $K_l \mathbf W$ spanned by the basis elements $A^{a_m,s_m}_{b_m,n_m}\cdots A^{a_1,s_1}_{b_1,n_1}$ such that $s_m> l$, let $\pi_l:\mathbf W\to F_l\mathbf W$ be the projection with kernel $K_l\mathbf W$.

\begin{proposition}[PBW theorem for W-algebras {\cite[Theorem 3.1]{arakawa2017explicit}}]
The $\mathbb C[\alpha]$-module map $F_L\mathbf W\otimes\mathbb C[\alpha]\to \mathcal W^{(K)}_L:$ $$\Omega\mapsto|0\rangle,\quad A^{a_m,s_m}_{b_m,n_m}\cdots A^{a_1,s_1}_{b_1,n_1}\mapsto W^{a_m(s_m)}_{b_m,n_m}\cdots W^{a_1(s_1)}_{b_1,n_1}|0\rangle$$ is an isomorphism.
\end{proposition}

Let $V^{\kappa_{\alpha}}(\gl_K^{\oplus \mathbb N})$ be the affine Kac-Moody vertex algebra associated to the countable infinite sum $\gl_K^{\oplus \mathbb N}$, with the inner product $\kappa_{\alpha}$ \eqref{eqn: inner form kappa(alpha)} on each direct summand. Denote its dual vacuum module by $V^{\kappa_{\alpha}}(\gl_K^{\oplus \mathbb N})^{\vee}$. Explicitly, $V^{\kappa_{\alpha}}(\gl_K^{\oplus \mathbb N})^{\vee}$ has a basis $$\langle 0|J^{a_1[s_1]}_{b_1,n_1}\cdots J^{a_m[s_m]}_{b_m,n_m}$$ for those integer indices $ 1\le a_i,b_i\le K$ and $1\le s_1\le\cdots \le s_m$ and $n_i\ge 1$ such that for every $1\le i<m$ either $s_{i}<s_{i+1}$ or $n_{i}\le n_{i+1}$ holds.

\begin{definition}
Define a bilinear map $G:V^{\kappa_{\alpha}}(\gl_K^{\oplus \mathbb N})^{\vee}\otimes \mathbf W\to \mathbb C[\alpha,\lambda]$ as follows: 
\begin{equation}
\begin{split}
&G(\langle 0|J^{c_1[s_1]}_{d_1,k_1}\cdots J^{c_m[s_m]}_{d_m,k_m}, A^{a_n,r_n}_{b_n,i_n}\cdots A^{a_1,r_1}_{b_1,i_1}):=\\
&\text{the polynomial in $\alpha$ and $L$ of }\langle 0|J^{c_1[s_1]}_{d_1,k_1}\cdots J^{c_m[s_m]}_{d_m,k_m}W^{a_n(r_n)}_{b_n,i_n}\cdots W^{a_1(r_1)}_{b_1,i_1}|0\rangle,
\end{split}
\end{equation}
and replace the variable $L$ by $\lambda$, using the Proposition \ref{prop: correlators are polynomials in L}.
\end{definition}

\begin{lemma}\label{lem: nondegenerate pairing_finite L}
Let $\langle 0|J^{c_1[s_1]}_{d_1,k_1}\cdots J^{c_m[s_m]}_{d_m,k_m}$ be an elements in $V^{\kappa_{\alpha}}(\gl_K^{\oplus \mathbb N})^{\vee}$ and take $L\ge s_m$, then the map
$G(\langle 0|J^{c_1[s_1]}_{d_1,k_1}\cdots J^{c_m[s_m]}_{d_m,k_m},-)|_{\lambda=L}: \mathbf W\to \mathbb C[\alpha]$ factors through the projection $\pi_L:\mathbf W\twoheadrightarrow F_L\mathbf W$.
\end{lemma}

\begin{proof}
This follows from the Proposition \ref{prop: correlators are polynomials in L} that $\langle 0|J^{c_1[s_1]}_{d_1,k_1}\cdots J^{c_m[s_m]}_{d_m,k_m}W^{a_n(r_n)}_{b_n,i_n}\cdots W^{a_1(r_1)}_{b_1,i_1}|0\rangle$ vanishes for $r_n>L$, i.e. $G(\langle 0|J^{c_1[s_1]}_{d_1,k_1}\cdots J^{c_m[s_m]}_{d_m,k_m},-)|_{\lambda=L}$ is zero on $K_L\mathbf W$, thus it factors through $\pi_L:\mathbf W\twoheadrightarrow F_L\mathbf W$.
\end{proof}

\begin{lemma}\label{lem: nondegenerate pairing_global}
The pairing $G:V^{\kappa_{\alpha}}(\gl_K^{\oplus \mathbb N})^{\vee}\otimes \mathbf W\to \mathbb C[\alpha,\lambda]$ is non-degenerate on the second component.
\end{lemma}

\begin{proof}
Suppose that $A\in \mathbf W$ such that $\forall v\in V^{\kappa_{\alpha}}(\gl_K^{\oplus \mathbb N})^{\vee}$ we have $G(v,A)=0$. Take $L$ such that $A\in F_L\mathbf W$, and consider the subspace $V^{\kappa_{\alpha}}(\oplus_{i=1}^L\gl_K^{[i]})^{\vee}\subset V^{\kappa_{\alpha}}(\gl_K^{\oplus \mathbb N})^{\vee}$, then $G|_{\lambda=L}:V^{\kappa_{\alpha}}(\oplus_{i=1}^L\gl_K^{[i]})^{\vee}\otimes F_L\mathbf W\to \mathbb C[\alpha]$ agrees with the canonical pairing $V^{\kappa_{\alpha}}(\oplus_{i=1}^L\gl_K^{[i]})^{\vee}\otimes \mathcal W^{(K)}_L\to \mathbb C[\alpha]$ after identifying $F_{L}\mathbf W\otimes\mathbb C[\alpha]\cong \mathcal W^{(K)}_L$ and treating $\mathcal W^{(K)}_L$ as a subspace of $V^{\kappa_{\alpha}}(\oplus_{i=1}^L\gl_K^{[i]})$. Since the pairing $V^{\kappa_{\alpha}}(\oplus_{i=1}^L\gl_K^{[i]})^{\vee}\otimes V^{\kappa_{\alpha}}(\oplus_{i=1}^L\gl_K^{[i]})\to \mathbb C[\alpha]$ is non-degenerate, we must have $A=0$.
\end{proof}

\begin{lemma}\label{lem: eventually rational}
Suppose that $A[L]$ is an $L$-dependent element of $\mathbf W_{-d}\otimes\mathbb C[\alpha]$ for a fixed $d$, where $L$ takes values in $\mathbb Z_{\ge n}$ for a fixed positive integer $n$. Assume moreover that for all $v\in V^{\kappa_{\alpha}}(\bigoplus_{i=1}^d\gl_K^{[i]})^{\vee}$, $G(v,A[L])|_{\lambda=L}$ is a polynomial in $L$ when $L\ge \max(n,d)$, then the $L$-dependence of $A[L]$ is eventually rational. More specifically, there exists a unique $A'[\lambda]\in \mathbf W_{-d}\otimes\mathbb C(\lambda)[\alpha]$ such that $A'[L]=A[L]$ when $L\ge \max(n,d)$.
\end{lemma}

\begin{proof}
Notice that $F_L\mathbf W_{-d}=\mathbf W_{-d}$ when $L\ge d$, then Lemma \ref{lem: nondegenerate pairing_finite L} together with Lemma \ref{lem: nondegenerate pairing_global} implies that the pairing $G: V^{\kappa_{\alpha}}(\bigoplus_{i=1}^d\gl_K^{[i]})^{\vee}_d\otimes \mathbf W_{-d}\to \mathbb C[\alpha,\lambda]$ is nondegenerate on the second component, and for all $L\ge d$ the pairing $G|_{\lambda=L}: V^{\kappa_{\alpha}}(\bigoplus_{i=1}^d\gl_K^{[i]})^{\vee}_d\otimes \mathbf W_{-d}\to \mathbb C[\alpha]$ is also nondegenerate on the second component. Now the nondegeneracy of the pairings $G$ and $G|_{\lambda=L}$ for all $L\in \mathbb Z_{\ge d}$ implies that there exists a divisor $f\in \mathbb C[\alpha,\lambda]$ such that $f|_{\lambda=L}\neq 0$ for all $L\in \mathbb Z_{\ge d}$, and that the induced map $V^{\kappa_{\alpha}}(\bigoplus_{i=1}^d\gl_K^{[i]})^{\vee}_d \otimes \mathbb C[\alpha,\lambda]_{(f)}\to \mathbf (W_{-d})^{\vee}\otimes \mathbb C[\alpha,\lambda]_{(f)}$ is surjective. In particular, if we allow base change with $f$ being inverted, then the dual basis of $W_{-d}$ can be obtained by the pairing $G$ using elements in $V^{\kappa_{\alpha}}(\bigoplus_{i=1}^d\gl_K^{[i]})^{\vee}_d$. This implies that there exists a unique $A'[\lambda]\in \mathbf W_{-d}\otimes \mathbb C[\alpha,\lambda]_{(f)}$ such that $A'[L]=A[L]$ for all $L\ge \max(n,d)$. This in turn implies that there exists a divisor $g\in \mathbb C[\alpha,\lambda]$ such that $g|_{\lambda=L}\in \mathbb C^{\times}$ for all $L\in \mathbb Z_{\ge \max(n,d)}$ and that $A'[\lambda]\in \mathbf W_{-d}\otimes \mathbb C[\alpha,\lambda]_{(g)}$. Write $g=\sum_{i=0}^m g_i(\lambda)\alpha^i$, then for all $i>0$ and all $L\in \mathbb Z_{\ge \max(n,d)}$, $g_i(L)=0$, thus $g=g_0(\lambda)$, in particular $A'[\lambda]\in \mathbf W_{-d}\otimes\mathbb C(\lambda)[\alpha]$. This finishes the proof.
\end{proof}

\subsection{The state-operator map}
\begin{lemma}\label{lem: state-operator map on W(infinity)}
Let $A\in F_r\mathbf W_{-p}, B\in F_s\mathbf W_{-q}$, then there exists a unique element $A_{(n)}B[\lambda]\in\mathbf W_{-p-q+n+1}\otimes \mathbb C(\lambda)[\alpha]$ such that $A_{(n)}B[L]=Y_L(A,z)B[z^n]$ for all $L\ge \max(r,s,p+q-n-1)$, where $Y_L(-,z)$ is the state-operator map for $\mathcal W^{(K)}_L$ (identified with $F_L\mathbf W\otimes\mathbb C[\alpha]$) and $Y_L(A,z)B[z^n]$ is the Fourier coefficient of $z^n$ in $Y_L(A,z)B$. 
\end{lemma}

\begin{proof}
Let $d=p+q-n-1$. If $L\ge \max(r,s)$, then both $A$ and $B$ are in the subspace $F_L\mathbf W$, so we have the vertex algebra action $Y_L(A,z)B\in F_L\mathbf W\otimes\mathbb C[\alpha](\!(z)\!)$, in particular $Y_L(A,z)B[z^n]\in F_L\mathbf W_{-d}$. Now suppose that $L\ge \max(r,s,d)$, then for all $v\in V^{\kappa_{\alpha}}(\bigoplus_{i=1}^d\gl_K^{[i]})^{\vee}$, the pairing $G(v,Y_L(A,z)B[z^n])|_{\lambda=L}$ equals to the correlator
$$\oint_{|z|=1}\langle vY_L(A,z)B\rangle z^n\frac{dz}{2\pi i}$$
in the vertex algebra $V^{\kappa_{\alpha}}(\bigoplus_{i=1}^L\gl_K^{[i]})$, in particular it is a polynomial in $L$ by the Proposition \ref{prop: correlators are polynomials in L}. Applying the Lemma \ref{lem: eventually rational} we obtain the unique $A_{(n)}B[\lambda]\in\mathbf W_{-d}\otimes \mathbb C(\lambda)[\alpha]$ such that $A_{(n)}B[L]=Y_L(A,z)B[z^n]$ for all $L\ge \max(r,s,d)$.
\end{proof}

Finally we can define the state-operator map on $\mathbf W$. For homogeneous elements $A,B\in \mathbf W$, we define 
\begin{align}\label{eqn: state-operator map on W(infinity)}
Y_{\lambda}(A,z)B:=\sum_{n\in \mathbb Z} \frac{A_{(n)}B[\lambda]}{z^{n+1}}\in \mathbf W\otimes\mathbb C(\lambda)[\alpha](\!(z)\!),
\end{align}
where $A_{(n)}B[\lambda]$ is obtained by Lemma \ref{lem: state-operator map on W(infinity)}.

\begin{proposition}\label{prop: state-operator map on W(infinity)}
The data $(\mathbf W\otimes\mathbb C(\lambda)[\alpha], Y_{\lambda}(-,z),\Omega )$ is a $\mathbb Z$-graded vertex algebra over the base ring $\mathbb C(\lambda)[\alpha]$.
\end{proposition}

\begin{proof}
We need to check the vacuum axiom and the commutation identity
\begin{align}\label{eqn: Borcherds identity}
    [A_{(n)},B_{(m)}]=\sum_{k\ge 0}\binom{n}{k}(A_{(k)}B)_{(n+m-k)},
\end{align}
note that this identity implies the locality axiom. The strategy is to check these axioms for $\lambda$ taking values in all sufficiently large $L$, then use the fact that two rational functions $f(\lambda),g(\lambda)$ are identical if and only if $f(L)=g(L)$ for all sufficiently large integer $L$.

For the vacuum axiom, notice that $Y_L(\Omega,z)A=A$ for all $L$ such that $A\in F_L\mathbf W$ (so that the left-hand-side is defined), thus $Y_{\lambda}(\Omega,z)A=A$; on the other hand, $Y_L(A,z)\Omega[z^0]=A$ and $Y_L(A,z)\Omega[z^n]=0$ for all $n>0$ and $L$ such that $A\in F_L\mathbf W$, thus $Y_{\lambda}(A,z)\Omega\in \mathbf W\otimes\mathbb C(\lambda)[\alpha][z]$ and $Y_{\lambda}(A,z)\Omega|_{z=0}=A$.

For the commutation identity \eqref{eqn: Borcherds identity}, it suffices to check the equation holds when it acts on elements, i.e. to check $[A_{(n)},B_{(m)}]C=\sum_{k\ge 0}\binom{n}{k}(A_{(k)}B)_{(n+m-k)}C$ holds. This holds when $\lambda=L$ and $L$ being sufficiently large (e.g. $A,B,C\in F_L\mathbf W$), thus the equation holds as identity between rational functions in $\lambda$.

Finally, the state-operator map $Y_{\lambda}(-,z)$ is compatible with the natural $\mathbb Z$-grading on $\mathbf W\otimes\mathbb C(\lambda)[\alpha]$, thus $(\mathbf W\otimes\mathbb C(\lambda)[\alpha], Y_{\lambda}(-,z),\Omega )$ is a $\mathbb Z$-graded vertex algebra over the base ring $\mathbb C(\lambda)[\alpha]$.
\end{proof}

From this point, we shall use the standard notation $W^{a(r)}_{b}(z)=\sum_{n\in \mathbb Z}W^{a(r)}_{b,n}z^{-n-r}$ to denote the vertex operator $Y_{\lambda}(A^{a,r}_{b,-r},z)$, and also write $|0\rangle:=\Omega$. Note that
\begin{align*}
    A^{a_m,s_m}_{b_m,n_m}\cdots A^{a_1,s_1}_{b_1,n_1}=W^{a_m(s_m)}_{b_m,n_m}\cdots W^{a_1(s_1)}_{b_1,n_1}|0\rangle,
\end{align*}
since this equation holds for all sufficiently large $L$. In particular, $W^{a(r)}_{b}(z),r=1,2,\cdots$ is a set of strong generators of $\mathbf W\otimes\mathbb C(\lambda)[\alpha]$. 


By construction, $W^{a(1)}_b(z)$ generates a $\mathbb C[\alpha,\lambda]$-vertex subalgebra $V^{\kappa_{\lambda\alpha,\lambda}}(\gl_K)\hookrightarrow \mathbf W\otimes\mathbb C(\lambda)[\alpha]$, where $\kappa_{\lambda\alpha,\lambda}$ is the inner form $\kappa_{\lambda\alpha,\lambda}(E^a_b,E^c_d)=\lambda\alpha\delta^c_b\delta^a_d+\lambda\delta^a_b\delta^c_d$.

Using the reduction to sufficiently large $L$ technique again, we see that $\mathbf W\otimes\mathbb C(\lambda)[\bar\alpha^{\pm}]$ is conformal and it has stress-energy operator 
\begin{align}\label{eqn: stress-energy operator}
    T(z)=-\frac{1}{2\bar\alpha}:W^{a(1)}_{b}W^{b(1)}_{a}:(z)-\frac{\alpha(\lambda-1)}{2\bar\alpha}\partial W^{a(1)}_{a}(z)+\frac{1}{\bar\alpha}W^{a(2)}_{a}(z)
\end{align}
with central charge 
\begin{align}\label{eqn: central charge}
    c=\frac{K\lambda}{\bar\alpha}((\lambda^2-1)\alpha^2-\alpha K-1).
\end{align}
Moreover, $W^{a(r)}_{b}(z)$ has conformal weight $r$ w.r.t. $T(z)$,
and $W^{a(1)}_{b}(z)$ are primary of spin $1$ w.r.t $T(z)$.

\subsection{Polynomiality of the structure constants}
So far we have constructed a candidate for $\mathcal W^{(K)}_{\infty}$ over the base ring $\mathbb C(\lambda)[\alpha]$, we would like to show that all the structure constants in the $U^{a(r)}_b$ basis actually takes value in polynomial ring $\mathbb C[\lambda,\alpha]$.

Define a $\mathbb C$-linear map $\Delta_{\lambda_1,\lambda_2}:\mathbf W\to \mathbf W\otimes \mathbb C(\lambda_1)\otimes \mathbf W\otimes\mathbb C(\lambda_2)[\alpha]$ as follows: $\Omega\mapsto \Omega\otimes\Omega$ and
\begin{align*}
A^{a_m,s_m}_{b_m,n_m}\cdots A^{a_1,s_1}_{b_1,n_1}\mapsto \oint_{|z_m|>\cdots>|z_1|}\Delta(W^{a_m(s_m)}_{b_m}(z_m))\cdots \Delta(W^{a_1(s_1)}_{b_1}(z_1))|0\rangle\otimes|0\rangle \prod_{j=1}^m\frac{z_j^{n_j+s_j-1}dz_j}{2\pi i},
\end{align*}
where $\Delta(W^{a(r)}_{b}(z))$ is defined by the formula
\begin{align*}
\Delta(W^{a(r)}_b(z))=\sum_{\substack{(s,t,u)\in\mathbb N^3\\s+t+u=r}} \binom{\lambda_2-t}{u}(\alpha\partial)^u W^{c(s)}_b(z)\otimes W^{a(t)}_c(z).
\end{align*}

\begin{lemma}\label{lem: injectivity of coproduct}
$\Delta_{\lambda_1,\lambda_2}$ is injective.
\end{lemma}
\begin{proof}
Obviously $\Delta_{\lambda_1,\lambda_2}$ preserves the $\mathbb Z$-grading, so we only need to show the injectivity for homogeneous component. Let $L_1,L_2\gg d$ be sufficiently large integers, then for a nonzero element $A\in \mathbf W_{-d}$ we have $$\Delta_{\lambda_1,\lambda_2}(A)\bigg\rvert_{\substack{\lambda_1=L_1\\ \lambda_2=L_2}}=\Delta_{L_1,L_2}(A),$$ where on the right-hand-side we treat $A$ as element in $\mathcal W^{(K)}_{L_1+L_2}=F_{L_1+L_2}\mathbf W\otimes\mathbb C[\alpha]$. Since $\Delta_{L_1,L_2}$ is injective, $\Delta_{\lambda_1,\lambda_2}(A)\neq 0$. 
\end{proof}

\begin{lemma}
For all $A,B\in \mathbf W$ and all $n\in \mathbb Z$ the equation 
\begin{align}\label{eqn: coproduct preserves state-operator map}
\Delta_{\lambda_1,\lambda_2}(A_{(n)}B[\lambda_1+\lambda_2])=\Delta_{\lambda_1,\lambda_2}(A)_{(n)}\Delta_{\lambda_1,\lambda_2}(B)[\lambda_1,\lambda_2]
\end{align}
holds in $\mathbf W\otimes  \mathbf W\otimes\mathbb C(\lambda_1,\lambda_2)[\alpha]$.
\end{lemma}

\begin{proof}
Suppose that $A\in \mathbf W_{-p}, B\in \mathbf W_{-q}$, then for $L_1,L_2\gg p+q+|n|$, $$\Delta_{\lambda_1,\lambda_2}(A_{(n)}B[\lambda_1+\lambda_2])\bigg\rvert_{\substack{\lambda_1=L_1\\ \lambda_2=L_2}}=\Delta_{L_1,L_2}(Y_{L_1+L_2}(A,z)B[z^n]),$$ where on the right-hand-side we treat $A, B$ and $Y_{L_1+L_2}(A,z)B[z^n]$ as element in $\mathcal W^{(K)}_{L_1+L_2}=F_{L_1+L_2}\mathbf W\otimes\mathbb C[\alpha]$. Since $\Delta_{L_1,L_2}$ is a vertex algebra map, i.e. preserves state-operator map, thus \eqref{eqn: coproduct preserves state-operator map} holds when $\lambda_1=L_1,\lambda_2=L_2$. Since it holds for infinitely many pairs $(L_1,L_2)$, \eqref{eqn: coproduct preserves state-operator map} must hold as rational functions in $\lambda_1,\lambda_2$. 
\end{proof}

\begin{proposition}\label{prop: polynomiality}
For all $A,B\in \mathbf W$ and all $n\in \mathbb Z$, $A_{(n)}B[\lambda]\in \mathbf W\otimes\mathbb C[\lambda,\alpha]$.
\end{proposition}

\begin{proof}
Suppose that $A\in \mathbf W_{-p}, B\in \mathbf W_{-q}$ and $A_{(n)}B[\lambda]$ has a pole at $\lambda=\mu\in \mathbb C$, then we can tune $\lambda_1,\lambda_2\in \mathbb C$ such that $\lambda_1+\lambda_2=\mu$, and that $\Delta_{\lambda_1,\lambda_2}(A)_{(n)}\Delta_{\lambda_1,\lambda_2}(B)[\lambda_1,\lambda_2]$ is regular, and that $\Delta_{\lambda_1,\lambda_2}|_{\mathbf W_{-p-q+n+1}}$ is also regular (it is possible because coefficients are rational functions of the form $f(\lambda_1)g(\lambda_2)$). Using the equation \eqref{eqn: coproduct preserves state-operator map}, we see that $\Delta_{\lambda_1,\lambda_2}(A_{(n)}B[\mu])$ is regular. However, $A_{(n)}B[\mu]$ is singular and the map $\Delta_{\lambda_1,\lambda_2}|_{\mathbf W_{-p-q+n+1}}$ is regular by our assumption, and $\Delta_{\lambda_1,\lambda_2}|_{\mathbf W_{-p-q+n+1}}$ is injective according to the Lemma \ref{lem: injectivity of coproduct}, this is a contradiction. Therefore $A_{(n)}B[\lambda]$ has no pole, i.e. it is an element of $\mathbf W\otimes\mathbb C[\lambda,\alpha]$.
\end{proof}

Finally, we can define the $\mathcal W^{(K)}_{\infty}$ algebra to be the $\mathbb Z$-graded vertex algebra $(\mathbf W\otimes\mathbb C[\lambda,\alpha], Y_{\lambda}(-,z),\Omega)$, which is defined over the base ring $\mathbb C[\alpha,\lambda]$ with strong generators $W^{a(r)}_b(z),r=1,2,\cdots$. Note that $W^{a(1)}_b(z)$ generates a $\mathbb C[\alpha,\lambda]$-vertex subalgebra $V^{\kappa_{\lambda\alpha,\lambda}}(\gl_K)\hookrightarrow \mathcal W^{(K)}_{\infty}$. Moreover, the $\Delta_{\lambda_1,\lambda_2}$ map that we define in the beginning of this subsection becomes a $\mathbb C[\alpha]$-vertex algebra map $\Delta_{\mathcal W}:\mathcal W^{(K)}_{\infty}\to \mathcal W^{(K)}_{\infty}\otimes \mathcal W^{(K)}_{\infty}$, such that $ \Delta_{\mathcal W}(\lambda)=\lambda\otimes 1+1\otimes \lambda$, and it acts on strong generators by \eqref{eqn: W(infinity) coproduct}.

\bigskip To conclude this subsection, we note that the correlator $\langle 0| W^{a_1(r_1)}_{b_1,i_1}\cdots W^{a_n(r_n)}_{b_n,i_n}|0\rangle$ in the $\mathcal W^{(K)}_{\infty}$ algebra is a polynomial in $\alpha$ and $\lambda$, and it agrees with the one determined by finite $L$ correlators in the Proposition \ref{prop: U-correlators are polynomials}.




\subsection{Proof of Theorem \ref{thm: matrix extended W(infty)}}
It remains to construct the truncation map, we define $\pi_{L}:\mathcal W^{(K)}_{\infty}\to \mathcal W^{(K)}_{L}$ to be the $\mathbb C[\alpha]$-linear map induced by setting $\lambda=L$ and composing with the projection $\mathbf W\to F_L\mathbf W$ and the identification $\mathcal W^{(K)}_{L}=F_L\mathbf W\otimes\mathbb C[\alpha]$.

\begin{lemma}
$\pi_1: \mathcal W^{(K)}_{\infty}\to \mathcal W^{(K)}_{1}=V^{\kappa_{\alpha}}(\gl_K)$ is a vertex algebra map.
\end{lemma}

\begin{proof}
Consider the pairing $P_{\lambda}: (\mathcal W^{(K)}_{1})^{\vee}\otimes \mathbf W\to \mathbb C[\alpha,\lambda]$ given by the following
\begin{equation}
\begin{split}
P_{\lambda}(\langle 0|W^{c_1(1)}_{d_1,k_1}\cdots W^{c_m(1)}_{d_m,k_m}, A^{a_n,r_n}_{b_n,i_n}\cdots A^{a_1,r_1}_{b_1,i_1}):=\langle 0|W^{c_1(1)}_{d_1,k_1}\cdots W^{c_m(1)}_{d_m,k_m}W^{a_n(r_n)}_{b_n,i_n}\cdots W^{a_1(r_1)}_{b_1,i_1}|0\rangle,
\end{split}
\end{equation}
where the right-hand-side is the correlator in $\mathcal W^{(K)}_{\infty}$. By the Proposition \ref{prop: U-correlators are polynomials}, $P_1:(\mathcal W^{(K)}_{1})^{\vee}\otimes \mathbf W\to \mathbb C[\alpha]$ vanishes identically on the subspace $(\mathcal W^{(K)}_{1})^{\vee}\otimes K_1\mathbf W$, and the restriction of $P_1$ on the subspace $(\mathcal W^{(K)}_{1})^{\vee}\otimes F_1\mathbf W$ is the same as the canonical pairing $V^{\kappa_{\alpha}}(\gl_K)^{\vee}\otimes V^{\kappa_{\alpha}}(\gl_K)\to \mathbb C[\alpha]$ which is nondegenerate on both components. So $K_1\mathbf W$ is the kernel of $P_1$.

Therefore it remains to show that vertex operators in $\mathcal W^{(K)}_{\infty}$ evaluated at $\lambda=1$ maps $K_1\mathbf W$ to the kernel of $P_1$. In fact, let $A=A^{c,t}_{d,-t}$ and $B=A^{a_n,r_n}_{b_n,i_n}\cdots A^{a_1,r_1}_{b_1,i_1}$ be elements in $\mathbf W$ such that $r_k>1$ for some $1\le k\le n$, i.e. $B\in K_1\mathbf W$, then for all $v\in (\mathcal W^{(K)}_{1})^{\vee}$, we have
\begin{align*}
    P_{\lambda}(v, Y_{\lambda}(A,z)B)|_{\lambda=1}=\langle 0|v \:W^{c(t)}_{d}(z)W^{a_n(r_n)}_{b_n,i_n}\cdots W^{a_1(r_1)}_{b_1,i_1}|0\rangle|_{\lambda=1},
\end{align*}
which is zero by Proposition \ref{prop: U-correlators are polynomials}. Thus $Y_{\lambda}(A,z)B|_{\lambda=1}\in K_1\mathbf W\otimes\mathbb C[\alpha](\!(z)\!)$. This finishes the proof.
\end{proof}

Now let us bootstrap the general truncation from the $\pi_1$, using the coproduct $\Delta_{\mathcal W}$. Suppose that $\pi_1,\cdots,\pi_{n-1}$ are known to be vertex algebra maps for some $n>1$, then consider the composition $(\pi_{n-1}\otimes\pi_1)\circ \Delta_{\mathcal W}:\mathcal W^{(K)}_{\infty}\to \mathcal W^{(K)}_{n-1}\otimes \mathcal W^{(K)}_{1}$, this maps $\lambda$ to $n$, and it is a vertex algebra map by our assumption. Using coproduct formula \eqref{eqn: W(infinity) coproduct} we find that $(\pi_{n-1}\otimes\pi_1)\circ \Delta_{\mathcal W}(W^{a(r)}_{b}(z))=0$ for all $r>n$, so this map factors through the projection $\pi_n:\mathbf W\to F_n\mathbf W$. Moreover, the coproduct formula \eqref{eqn: W(infinity) coproduct} also implies that $(\pi_{n-1}\otimes\pi_1)\circ \Delta_{\mathcal W}(W^{a(r)}_{b}(z))=\Delta_{n-1,1}(W^{a(r)}_{b}(z))$ for all $r=1,\cdots,n$, where on the right-hand-side $\Delta_{n-1,1}$ is the coproduct $\mathcal W^{(K)}_n\to \mathcal W^{(K)}_{n-1}\otimes\mathcal W^{(K)}_1$, in particular it is a vertex algebra embedding. This implies that $\pi_n$ is also a vertex algebra map, and finishes the proof of Theorem \ref{thm: matrix extended W(infty)}. 



\subsection{An integral form of \texorpdfstring{$\mathcal W^{(K)}_{\infty}$}{Wk inf}}
Define $\mathds W^{a(r)}_b(z):=\epsilon_1^{r-1}W^{a(r)}_b(z)$, $\mathsf c:=\lambda/\epsilon_1$, then this basis provide an ``integral'' form of the $\mathcal W^{(K)}_{\infty}$ algebra over the base ring $\mathbb C[\epsilon_1,\epsilon_2,\mathsf c]$. According to \eqref{W^1W^n OPE}, the $\mathds W^{(1)}\mathds W^{(n)}$ OPE is
\begin{equation}\label{eqn:  dsW1 dsWn OPE}
    \begin{split}
        \mathds W^{a(1)}_b(z)\mathds W^{c(1)}_d(w)\sim &\frac{\mathsf c(\epsilon_3\delta^a_d\delta^c_b+\epsilon_1\delta^a_b\delta^c_d)}{(z-w)^2}+\frac{\delta^c_b\mathds W^{a(1)}_d(w)-\delta^a_d\mathds W^{c(1)}_b(w)}{z-w},\\       
        \mathds W^{a(1)}_b(z)\mathds W^{c(n)}_d(w)\sim & \frac{\mathsf c(\epsilon_1\mathsf c-1)\cdots(\epsilon_1\mathsf c-n+1)\epsilon_3^{n-1}(\epsilon_3\delta^a_d\delta^c_b+\epsilon_1\delta^a_b\delta^c_d)}{(z-w)^{n+1}}\\
        &+\sum_{i=1}^{n-1}\frac{(\epsilon_1\mathsf c-i)\cdots(\epsilon_1\mathsf c-n+1)\epsilon_3^{n-1-i}(\epsilon_3\delta^c_b\mathds W^{a(i)}_d(w)+\epsilon_1\delta^a_b\mathds W^{c(i)}_d(w))}{(z-w)^{n+1-i}}\\
        &+\frac{\delta^c_b\mathds W^{a(n)}_d(w)-\delta^a_d\mathds W^{c(n)}_b(w)}{z-w}.\qquad (n>1)
    \end{split}
\end{equation}
We notice that OPE coefficients in the above formulae are polynomials in $\epsilon_1,\epsilon_2,\mathsf c$. This is true in general, proven in the next lemma.

\begin{lemma}\label{lem: integral basis}
The structure constants in the $\mathds W^{a(r)}_b, r=1,2,\cdots$ basis are polynomials in $\epsilon_1,\epsilon_2$, and $\mathsf c$.
\end{lemma}

\begin{proof}
Consider the scaled $\gl_K$ affine Kac-Moody generator $\tilde{J}^{a}_b(z):=\epsilon_1 {J}^{a}_b(z)$, and scale the Miura operator accordingly $\tilde{\mathcal L}(z):=\epsilon_1\mathcal L(z)=\epsilon_3\partial_z-\tilde{J}(z)$, then the scaled generators $\tilde W^{a(r)}_b(z):=\epsilon_1^r W^{a(r)}_b(z)$ of $\mathcal W^{(K)}_L$ is given by
\begin{align*}
    (\epsilon_3\partial_z)^L+\sum_{r=1}^L (-1)^r(\epsilon_3\partial_z)^{L-r}\tilde W^{(r)}(z):=\tilde{\mathcal L}^L(z):=\tilde{\mathcal L}(z)^{[1]}\cdots \tilde{\mathcal L}(z)^{[L]}.
\end{align*}
Then it is easy to see that structure constants of $\mathcal W^{(K)}_L$ in the $\tilde W^{a(r)}_b$ basis are polynomials in $\epsilon_1,\epsilon_2$. Moreover, if we set $\epsilon_1\to 0$ then all structure constants become zero and we a commutative vertex algebra, which implies that structure constants of $\mathcal W^{(K)}_L$ in the $\tilde W^{a(r)}_b$ basis are divisible by $\epsilon_1$. Passing to $\mathcal W^{(K)}_{\infty}$, we find that structure constants in the $\tilde W^{a(r)}_b$ basis are polynomials in $\epsilon_1,\epsilon_2,\lambda$ and are divisible by $\epsilon_1$.

Now $\mathds W^{a(r)}_b$ is related to $\tilde W^{a(r)}_b$ by $\mathds W^{a(r)}_b=\tilde W^{a(r)}_b/\epsilon_1$, so it remains to show that the leading order pole in the OPE between $\tilde W^{a(r)}_b$, schematically written as 
\begin{align*}
\tilde W^{a(r)}_b(z)\tilde W^{c(s)}_d(w)\sim \frac{C^{ac(r,s)}_{bd}(\epsilon_1,\epsilon_2,\lambda)}{(z-w)^{r+s}}+\text{linear in $\tilde W$ + higher order in $\tilde W$ terms},
\end{align*}
has the property that $C^{ac(r,s)}_{bd}(\epsilon_1,\epsilon_2,\lambda)$ is divisible by $\lambda$. This can be proven as follows. Consider the correlator $\langle \tilde W^{a(r)}_b(z)\tilde W^{c(s)}_d(w)\rangle$, which equals to $\frac{C^{ac(r,s)}_{bd}(\epsilon_1,\epsilon_2,\lambda)}{(z-w)^{r+s}}$. On the other hand, the correlator vanishes when $\lambda=0$ as explained in Proposition \ref{prop: U-correlators are polynomials}. Thus $C^{ac(r,s)}_{bd}(\epsilon_1,\epsilon_2,\lambda)$ must be divisible by $\lambda$. Since $C^{ac(r,s)}_{bd}(\epsilon_1,\epsilon_2,\lambda)$ is also divisible by $\epsilon_1$, $C^{ac(r,s)}_{bd}(\epsilon_1,\epsilon_2,\epsilon_1 \mathsf c)/\epsilon_1^2$ is therefore a polynomial in $\epsilon_1,\epsilon_2$ and $\mathsf c$. This finishes the proof.
\end{proof}

\begin{remark}\label{rmk: E1 to 0 limit}
In the limit $\epsilon_1\to 0$, the OPE between $\mathds W^{(r)}$ and $\mathds W^{(s)}$ can be written as
\begin{align*}
\mathds W^{a(r)}_b(z)\mathds W^{c(s)}_d(w)\sim \frac{\mathds C^{ac(r,s)}_{bd}(\epsilon_2,\mathsf c)}{(z-w)^{r+s}}+\text{linear in $\mathds W$ },
\end{align*}
where $\mathds C^{ac(r,s)}_{bd}(\epsilon_2,\mathsf c)=\underset{\epsilon_1\shortrightarrow 0}{\lim} \epsilon_1^{-2}C^{ac(r,s)}_{bd}(\epsilon_1,\epsilon_2,\epsilon_1\mathsf c)$.
\end{remark}

\begin{definition}
Denote by $\mathsf W^{(K)}_{\infty}$ the $\mathbb C[\epsilon_1,\epsilon_2,\mathsf c]$-vertex algebra strongly generated by $\mathds W^{a(r)}_b(z)$. We call it the integral form of $\mathcal W^{(K)}_{\infty}$. 
\end{definition}
As a $\mathbb C[\epsilon_1,\epsilon_2,\mathsf c]$-module $\mathsf W^{(K)}_{\infty}$ is isomorphic to the free module $\mathbf W\otimes \mathbb C[\epsilon_1,\epsilon_2,\mathsf c]$. The spin $1$ fields $\mathds W^{a(1)}_{b}(z)$ generate a affine Kac-Moody vertex subalgebra $V^{\kappa_{\epsilon_3\mathsf c,\epsilon_1 \mathsf c}}(\gl_K)$, where $\kappa_{\epsilon_3\mathsf c,\epsilon_1 \mathsf c}$ is the inner product $\kappa_{\epsilon_3\mathsf c,\epsilon_1 \mathsf c}(E^a_b,E^c_d)=\epsilon_3\mathsf c\delta^c_b\delta^a_d+\epsilon_1 \mathsf c\delta^a_b\delta^c_d$. Moreover it has $\mathbb C[\epsilon_1,\epsilon_2]$-vertex algebra map $\Delta_{\mathsf W}: \mathsf W^{(K)}_{\infty}\to \mathsf W^{(K)}_{\infty}\otimes \mathsf W^{(K)}_{\infty}$ such that $ \Delta_{\mathsf W}(\mathsf c)=\mathsf c\otimes 1+1\otimes \mathsf c$, and $\Delta_{\mathsf W}$ acts on strong generators by 
\begin{equation}\label{eqn: W(infinity) coproduct_integral form}
    \begin{split}
    \Delta_{\mathsf W}(\mathds W^{a(r)}_b(z))=\square(\mathds W^{a(r)}_b(z))+\epsilon_1\sum_{\substack{u\in \mathbb N,(s,t)\in\mathbb N^2_{>0}\\s+t+u=r}} \binom{\epsilon_1\otimes \mathsf c-t}{u}(\epsilon_3\partial)^u \mathds W^{c(s)}_b(z)\otimes \mathds W^{a(t)}_c(z),
    \end{split}
\end{equation}
where $\square(\mathds W^{a(r)}_b(z))=\mathds W^{a(r)}_b(z)\otimes 1+1\otimes \mathds W^{a(r)}_b(z)$.\\

We also note that the subalgebra $\mathfrak U(\mathsf  W^{(K)}_\infty)\subset \mathfrak{U}(\mathcal W^{(K)}_\infty)$ is invariant under the automorphism $\boldsymbol{\eta}_{\beta}: \mathfrak{U}(\mathcal W^{(K)}_\infty)\cong \mathfrak{U}(\mathcal W^{(K)}_\infty)$ in \eqref{eqn: shift automorphism of W}, and we scale the shifting parameter as follows:
\begin{align}
\boldsymbol{\tau}_\beta:=\boldsymbol{\eta}_{\beta/\epsilon_1},
\end{align}
then $\boldsymbol{\tau}_\beta$ shifts the currents by
\begin{equation}\label{eqn: shift automorphism of W_integral basis}
\begin{split}
\boldsymbol{\tau}_\beta(\mathds W^{(r)}(z))= &\mathds W^{(r)}(z)+\sum_{s=1}^{r} \binom{\epsilon_1\mathsf c+s-r}{s} \binom{{\beta}/{\epsilon_3}}{s}\frac{\epsilon_3^s \cdot s!}{z^s}\mathds W^{(r-s)}(z),
\end{split}
\end{equation}
where $\mathds W^{(0)}(z)$ is set to be $\frac{1}{\epsilon_1}$ times the constant identity matrix.\\

In the $\mathds W^{(s)}$ basis, the stress-energy operator $T(z)$ in \eqref{eqn: stress-energy operator} is written as
\begin{align}\label{eqn: stress-energy operator_integral basis}
    T(z)=-\frac{\epsilon_1}{2\epsilon_2}:\mathds W^{a(1)}_{b}\mathds W^{b(1)}_{a}:(z)-\frac{\epsilon_3(\epsilon_1\mathsf c-1)}{2\epsilon_2}\partial \mathds W^{a(1)}_{a}(z)+\frac{1}{\epsilon_2}\mathds W^{a(2)}_{a}(z)
\end{align}
with central charge 
\begin{align}\label{eqn: central charge_integral basis}
    c=\frac{K\mathsf c}{\epsilon_2}((\epsilon_1^2\mathsf c^2-1)\epsilon_3^2-\epsilon_1\epsilon_3 K-\epsilon_1^2).
\end{align}
\begin{lemma}\label{lem: restricted mode alg of W(infinity)}
$\mathsf W^{(K)}_{\infty}$ satisfies the technical assumptions in the Proposition \ref{prop: U(V) is a sub of mode algebra}, i.e.
$\mathsf W^{(K)}_{\infty}$ has a Hamiltonian $H$, and an increasing filtration $F$ such that $\mathrm{gr}_F\mathsf W^{(K)}_{\infty}$ is commutative, and an $H$-invariant subspace $U$ of $\mathcal W^{(K)}_{\infty}$ such that its image $\overline U$ in $\mathrm{gr}_F\mathsf W^{(K)}_{\infty}$ generate a PBW basis of $\mathrm{gr}_F\mathsf W^{(K)}_{\infty}$. In particular the canonical map $U(\mathsf W^{(K)}_{\infty})\to \mathfrak U(\mathsf W^{(K)}_{\infty})$ is injective.
\end{lemma}

\begin{proof}
We take $H$ to be the negative of grading, then $H$ is a Hamiltonian since $H=L_0$ when $\epsilon_3$ is invertible (so that stress-energy operator $T(z)$ is defined) and $\mathsf W^{(K)}_{\infty}$ is flat over the $\mathbb C[\epsilon_3]$. Next we take the filtration $F$ to be the one in the Definition \ref{def: vector space of W(infinity)}, and take $U$ to be the $\mathbb C[\epsilon_1,\epsilon_2,\mathsf c]$-span of $\mathds W^{a(r)}_{b,-s}|0\rangle$ where $1\le a,b\le K$ and $s\in \mathbb Z_{\ge 1}$. Then $\mathrm{gr}_F\mathsf W^{(K)}_{\infty}$ is obviously commutative, and the image $\overline U$ in $\mathrm{gr}_F\mathsf W^{(K)}_{\infty}$ generate a PBW basis of $\mathrm{gr}_F\mathsf W^{(K)}_{\infty}$ by the construction.
\end{proof}

Another feature of the integral form $\mathsf W^{(K)}_{\infty}$ is that its OPEs modulo $\mathsf c$ only involve nontrivial operators, i.e. $\mathds W^{a(r)}_{s,n}\mathds W^{c(s)}_{d,m}|0\rangle \in \mathsf W_{<0}\otimes\mathbb C[\epsilon_1,\epsilon_2]$ modulo $\mathsf c$, where $\mathsf W_{<0}=\bigoplus_{d=1}^{\infty}\mathsf W_{-d}$. This allows us to define an augmentation
\begin{align}\label{eqn: augmentation of W(infinity)}
    \mathfrak C_{\mathsf W}: \mathsf W^{(K)}_{\infty}\to \mathbb C[\epsilon_1,\epsilon_2]
\end{align}
where the right hand side is treated as the trivial vertex algebra over the base ring $\mathbb C[\epsilon_1,\epsilon_2]$. $\mathfrak C_{\mathsf W}$ simply maps all nontrivial generators $\mathds W^{a(s)}_{b}(z)$ to zero, and also maps $\mathsf c$ to zero. By the above discussions $\mathfrak C_{\mathsf W}$ is a vertex algebra map. Moreover, $\mathfrak C_{\mathsf W}$ is a counit for the coproduct $\Delta_{\mathsf W}$, in fact it is straightforward to see from \eqref{eqn: W(infinity) coproduct_integral form} that
\begin{equation}
(\mathfrak C_{\mathsf W}\otimes 1)\circ \Delta_{\mathsf W}=\mathrm{id}=( 1\otimes \mathfrak C_{\mathsf W})\circ \Delta_{\mathsf W}.
\end{equation}

The maps $\Psi_L:\mathsf A^{(K)}\to \mathfrak{U}(\mathcal W^{(K)}_L)[\bar\alpha^{-1}]$ can be promoted in a unique way to a map $$\Psi_{\infty}:\mathsf A^{(K)}\to \mathfrak{U}(\mathsf W^{(K)}_\infty)[\epsilon_2^{-1}]$$ such that $\Psi_L=\pi_L\circ \Psi_{\infty}$:
\begin{equation}\label{eqn: map to W(infinity)}
\boxed{
\begin{aligned}
    &\Psi_{\infty}(\mathsf T_{0,n}(E^a_b))=\mathds W^{a(1)}_{b,n},\quad \Psi_{\infty}(\mathsf t_{0,n})=\frac{1}{\epsilon_2}\mathds W^{a(1)}_{a,n},\\
    &\Psi_\infty(\mathsf T_{1,0}(E^a_b))= \epsilon_1\sum_{m\ge 0}\mathds W^{a(1)}_{c,-m-1}\mathds W^{c(1)}_{b,m}- \mathds W^{a(2)}_{b,-1}\\
    &\Psi_{\infty}(\mathsf t_{2,0})=\frac{1}{\epsilon_2}\left(\mathds V_{-2}-\epsilon_1\epsilon_2\sum_{n=1}^{\infty}n\:\mathds W^{a(1)}_{b,-n-1}\mathds W^{b(1)}_{a,n-1}-\epsilon_1^2\sum_{n=1}^{\infty}n\: \mathds W^{a(1)}_{a,-n-1}\mathds W^{b(1)}_{b,n-1}\right).
\end{aligned}}
\end{equation}
Here $\mathds V_{-2}$ is a mode of quasi-primary field $\mathds V(z)=\sum_{n\in \mathbb Z}\mathds V_{n}z^{-n-3}$ defined as
\begin{equation}
\begin{split}
\mathds V(z):=&\frac{\epsilon_1^2}{6}\left(:\mathds W^{a(1)}_b(z)\mathds W^{b(1)}_c(z)\mathds W^{c(1)}_a(z):+:\mathds W^{b(1)}_a(z)\mathds W^{c(1)}_b(z)\mathds W^{a(1)}_c(z):\right)\\
&+\mathds W^{a(3)}_a(z)-\epsilon_1:\mathds W^{a(1)}_b(z)\mathds W^{b(2)}_a(z):.
\end{split}
\end{equation}
Note that the image $\Psi_\infty(\mathsf A^{(K)})$ is contained in the positive restricted mode algebra $U_+(\mathsf W^{(K)}_\infty)[\epsilon_2^{-1}]$. It also follows from \eqref{eqn: map to W(infinity)} that the action of $\Psi_\infty(\mathsf A^{(K)})$ on the vacuum $|0\rangle$ factors through the augmentation $\mathfrak C_{\mathsf A}:\mathsf A^{(K)}\to \mathbb C[\epsilon_1,\epsilon_2]$ such that all generators $\mathsf T_{m,n}(X)$ and $\mathsf t_{m,n}$ are mapped to zero.

\smallskip Similarly the maps $\Delta_{L}:\mathsf A^{(K)}\to\mathsf A^{(K)}\widetilde\otimes \mathfrak U(\mathcal W^{(K)}_L)[\epsilon_2^{-1}]$ can be promoted in a unique way to a map 
\begin{align*}
    \Delta_{\infty}:\mathsf A^{(K)}\to\mathsf A^{(K)}\widetilde\otimes \mathfrak U(\mathsf W^{(K)}_\infty)[\epsilon_2^{-1}]
\end{align*}
such that $\Delta_{L}=(1\otimes \pi_L)\circ \Delta_{\infty}$:
\begin{equation}\label{eqn: AW(infinity) coproduct}
\boxed{
\begin{aligned}
&\Delta_\infty(\mathsf T_{0,n}(E^a_b))=\square( \mathsf T_{0,n}(E^a_b),\\
&\Delta_\infty(\mathsf t_{1,n})=\square(\mathsf t_{1,n})+\epsilon_1\epsilon_3 n\:\mathsf t_{0,n-1}\otimes \mathsf c,\\
&\Delta_\infty(\mathsf T_{1,0}(E^a_b))=\square(\mathsf T_{1,0}(E^a_b))+\epsilon_1\sum_{m=0}^{\infty}\left(\mathsf T_{0,m}(E^c_b)\otimes \mathds W^{a(1)}_{c,-m-1}-\mathsf T_{0,m}(E^a_c)\otimes \mathds W^{c(1)}_{b,-m-1}\right),\\
&\Delta_\infty(\mathsf t_{2,0})=\square(\mathsf t_{2,0})-2\epsilon_1\sum_{n=1}^{\infty}n \left(\mathsf T_{0,n-1}(E^a_b)\otimes\mathds W^{b(1)}_{a,-n-1}+\epsilon_1\mathsf t_{0,n-1}\otimes\mathds W^{a(1)}_{a,-n-1}\right),
\end{aligned}}
\end{equation}
where $\square(x):=x\otimes 1+1\otimes \Psi_\infty(x)$. Comparing \eqref{eqn: AW(infinity) coproduct} with \eqref{eqn: W(infinity) coproduct_integral form}, we observe the following.

\begin{proposition}\label{prop: AW coproduct compatible with WW coproduct}
$\Delta_{\infty}$ is compatible with $\Delta_{\mathsf W}:\mathfrak U(\mathsf W^{(K)}_\infty)\to \mathfrak U(\mathsf W^{(K)}_\infty\otimes \mathsf W^{(K)}_\infty)$, i.e.
\begin{align*}
    (\Psi_{\infty}\otimes 1)\circ\Delta_{\infty}=\Delta_{\mathsf W}\circ\Psi_{\infty}.
\end{align*}
\end{proposition}
By the Lemma \ref{lem: restricted mode alg of W(infinity)} and our previous results \eqref{eqn: image of t[2,0] in restricted modes}, the image of $\Psi_{\infty}$ lies in the restricted mode algebra $U(\mathsf W^{(K)}_\infty)[\epsilon_2^{-1}]$, and the image of $\Delta_{\infty}$ lies in the subalgebra $\mathsf A^{(K)}\widetilde\otimes U(\mathsf W^{(K)}_\infty)[\epsilon_2^{-1}]$. The uniform-in-$L$ version of Proposition \ref{prop: image of D^K} is the following:
\begin{proposition}
$\Psi_\infty({\mathsf D}^{(K)})\subset {U}(\mathsf W^{(K)}_\infty)$ and $\Delta_\infty( {\mathsf D}^{(K)})\subset \mathsf D^{(K)}\widetilde\otimes {U}(\mathsf W^{(K)}_\infty)$.
\end{proposition}

\begin{proof}
It suffices to show that $\mathrm{ad}_{\Psi_{\infty}(\mathsf t_{2,0})}({U}(\mathsf W^{(K)}_\infty))\subset {U}(\mathsf W^{(K)}_\infty)$, equivalently $\mathrm{ad}_{\epsilon_2\Psi_{\infty}(\mathsf t_{2,0})}({U}(\mathsf W^{(K)}_\infty))\subset \epsilon_2\cdot{U}(\mathsf W^{(K)}_\infty)$. So we need to show that $\mathrm{ad}_{\epsilon_2\Psi_{\infty}(\mathsf t_{2,0})}\equiv 0\pmod{\epsilon_2}$. Since $ {U}(\mathsf W^{(K)}_\infty/(\epsilon_2))$ is torsion-free over $\mathbb C[\epsilon_1]$, so it is enough to localize $\epsilon_1$ and show that $\mathrm{ad}_{\epsilon_2\Psi_{\infty}(\mathsf t_{2,0})}\equiv 0\pmod{\epsilon_2}$ on $\mathbb C[\epsilon_1^\pm]$. Notice that $\mathsf W^{(K)}_\infty[\epsilon_1^{-1}]/(\epsilon_2)$ is isomorphic to $\mathcal W^{(K)}_{\infty}/(\bar\alpha)$, thus we only need to show that $\mathrm{ad}_{\bar\alpha\Psi_{\infty}(\mathsf t_{2,0})}\equiv 0\pmod{\bar\alpha}$ when acting on ${U}(\mathcal W^{(K)}_\infty)$, which follows from Proposition \ref{prop: image of D^K}.
\end{proof}

Another remark is that the constants $C^{ac(r,s)}_{bd}(\epsilon_1,\epsilon_2,\lambda)$ in the proof of Lemma \ref{lem: integral basis} can be determined completely by the following formula:
\begin{equation}
\begin{split}
(\epsilon_3\partial_z)^{\lambda}(\epsilon_3\partial_w)^{\lambda}+&\sum_{(r,s)\in \mathbb N^2\backslash 0}(-1)^{r+s}(\epsilon_3\partial_z)^{\lambda-r}(\epsilon_3\partial_w)^{\lambda-s}\cdot \frac{C^{ac(r,s)}_{bd}(\epsilon_1,\epsilon_2,\lambda)}{(z-w)^{r+s}}E^a_b\otimes E^c_d\\
&=\left(\epsilon_3^2\partial_z\partial_w+\frac{\epsilon_1\epsilon_3E^e_f\otimes E^f_e+\epsilon_1^2}{(z-w)^2}\right)^{\lambda},
\end{split}
\end{equation}
where we treat both sides as $\gl_K^{\otimes 2}$-valued pseudodifferential symbols. This formula is derived from Corollary \ref{cor: correlator of Miura}. Let us assume this result for now and defer the proof until Section \ref{sec: Miura Operators as Intertwiners}, then direct computation shows that $C^{ac(r,s)}_{bd}(\epsilon_1,\epsilon_2,\lambda)$ can be written as 
$$C^{ac(r,s)}_{bd}(\epsilon_1,\epsilon_2,\lambda)=(-\epsilon_1)^{r+s}C^{0,0}_{r,s}(\epsilon_3/\epsilon_1,\lambda)\delta^a_b\delta^c_d+(-\epsilon_1)^{r+s}\tilde C^{0,0}_{r,s}(\epsilon_3/\epsilon_1,\lambda)\delta^a_d\delta^b_c,$$
where $C^{0,0}_{r,s}$ and $\tilde C^{0,0}_{r,s}$ are given by \cite[(3.9), (3.10)]{eberhardt2019matrix}. Moreover the $\mathcal O(\epsilon_1\lambda)$-term in $C^{ac(r,s)}_{bd}(\epsilon_1,\epsilon_2,\lambda)$ is
\begin{equation}
\begin{split}
(-1)^{s+1}\epsilon_1\epsilon_3^{s+r-1} \lambda(r+s-2)!\setlength\arraycolsep{1pt}
{}_2 F_1\left[\begin{matrix}&1-r,1-s\\&2-r-s\end{matrix};1\right]\delta^a_d\delta^c_b.
\end{split}
\end{equation}
On the other hand, the $\mathcal O(\epsilon_1\lambda^0)$-term in the $\tilde W$-linear terms of the $\tilde W^{a(r)}_b(z)\tilde W^{c(s)}_d(w)$ OPE is $$\epsilon_1\frac{\delta^c_b\tilde W^{a(r+s-1)}_d(w)-\delta^a_d\tilde W^{c(r+s-1)}_b(w)}{z-w}+\mathcal{O}(\epsilon_1\epsilon_3)\text{-terms}.$$ Note that the $\mathcal{O}(\epsilon_1\epsilon_3)$-terms in the above equation only involve $\tilde W^{(k)}$ for $k<r+s-1$. Combine the above OPE computation with the Remark \ref{rmk: E1 to 0 limit}, we see that
\begin{lemma}\label{lem: E1 to 0 limit}
$\mathfrak{U}(\mathsf W^{(K)}_\infty)/(\epsilon_1=0)$ is a completion of the universal enveloping algebra of a Lie algebra $\mathscr O_{\epsilon_2,\mathsf c}(\mathbb C\times \mathbb C^{\times})\otimes\gl_K$, such that $\mathscr O_{0,0}(\mathbb C\times \mathbb C^{\times})\otimes\gl_K$ is the tensor product of the function ring on $\mathbb C\times \mathbb C^{\times}$ and the matrix algebra $\gl_K$. 
\end{lemma}
The identification between $\mathfrak{U}(\mathsf W^{(K)}_\infty)/(\epsilon_1=\epsilon_2=\mathsf c=0)$ and $U(\mathscr O(\mathbb C\times \mathbb C^{\times})\otimes\gl_K)$ is that 
\begin{align}\label{eqn: mode algebra of W(infinity) at degenerate limit}
    \mathds W^{a(s)}_{b,n}\mapsto  x^{n+s-1}(-y)^{s-1}E^a_b,
\end{align}
where $x$ is the coordinate on $\mathbb C^{\times}$ and $y$ is the coordinate of $\mathbb C$ and $E^a_b\in \gl_K$ is the elementary matrix. We will figure out the structure of the Lie algebra $\mathscr O_{\epsilon_2,\mathsf c}(\mathbb C\times \mathbb C^{\times})\otimes\gl_K$ shortly in the next subsection.

It follows from \eqref{eqn: map to W(infinity)} that $\Psi_{\infty}(\mathsf t_{2,0})=\frac{1}{\epsilon_2}\mathds W^{a(3)}_{a,-2}$, and $\Psi_{\infty}(\mathsf T_{0,n}(E^a_b))=\mathds W^{a(1)}_{b,n}$. The first order pole in the OPE between $\mathds W^{a(3)}_a$ and $\mathds W^{c(s)}_d$ in $\mathsf W^{(K)}_\infty/(\epsilon_1=\mathsf c=0)$ can be computed:
\begin{equation}\label{eqn: W^3W^s OPE modulo epsilon_1}
\begin{split}
\mathds W^{a(3)}_{a,-2}\mathds W^{c(s)}_{d,-s}|0\rangle =\epsilon_2\left(\mu_{s,1}\mathds W^{c(s+1)}_{d,-s-2}+\mu_{s,2}\delta^c_d\mathds W^{e(s+1)}_{e,-s-2}\right)|0\rangle+\mathcal O(\epsilon_2^2)\text{-terms},
\end{split}
\end{equation}
where $\mu_{s,1},\mu_{s,2}\in \mathbb C$ are two complex numbers which are going to be determined. From the OPE we compute the commutator: 
\begin{align}\label{eqn: [W^3,W^s] modulo epsilon_1}
    \frac{1}{\epsilon_2}[\mathds W^{a(3)}_{a,-2},\mathds W^{c(s)}_{d,n}]=-(n+s-1)(\mu_{s,1}\mathds W^{c(s+1)}_{d,n-2}+\mu_{s,2}\delta^c_d\mathds W^{e(s+1)}_{e,n-2})+\mathcal O(\epsilon_2)\text{-terms}
\end{align}
such that the $\mathcal O(\epsilon_2)$-terms are linear combinations of $\mathds W^{(i)}_{n-2}$ for $i\le s$. Therefore we find 
\begin{equation}\label{eqn: map to W(infinity) modulo epsilon_1}
\begin{split}
\Psi_{\infty}(\mathsf T_{n,m}(E^a_b))&=\frac{m!}{2^n(m+n)!}\mathrm{ad}_{\frac{1}{\epsilon_2}\mathds W^{c(3)}_{c,-2}}^{n}(\mathds W^{a(1)}_{b,m+n})\\
&=h_{n,1}\mathds W^{a(n+1)}_{b,m-n}+h_{n,2}\delta^a_b\mathds W^{d(n+1)}_{d,m-n}+\epsilon_2\cdot(\text{linear combination of $\mathds W^{(i)}_{m-n}$ for $i\le n$})
\end{split}
\end{equation}
in $\mathfrak{U}(\mathsf W^{(K)}_\infty)/(\epsilon_1=\mathsf c=0)$, where $h_{n,1},h_{n,2}\in \mathbb C$ are certain complex numbers which are going to be determined.

\begin{lemma}\label{lem: map to W(infinity) modulo epsilon_1}
In \eqref{eqn: W^3W^s OPE modulo epsilon_1}, $\mu_{s,1}=2$ and $\mu_{s,2}=0$. In \eqref{eqn: map to W(infinity) modulo epsilon_1}, $h_{n,1}=(-1)^n$ and $h_{n,2}=0$.
\end{lemma}

\begin{proof}
The following OPE in $\mathcal W^{(K)}_L$ is straightforward to compute:
\begin{align*}
    W^{a(1)}_{a,1}W^{b(r+1)}_{c,-r-1}|0\rangle=-\bar\alpha (L-r) W^{b(r)}_{c,-r}|0\rangle.
\end{align*}
Then we have $\mathds W^{a(1)}_{a,1}\mathds W^{b(r+1)}_{c,-r-1}|0\rangle=\epsilon_2 (r-\epsilon_1\mathsf c) W^{b(r)}_{c,-r}|0\rangle$ in $\mathsf W^{(K)}_\infty$, therefore
\begin{align*}
    \frac{1}{\epsilon_2}[\mathds W^{a(1)}_{a,1},\mathds W^{b(r+1)}_{c,n}]= r \mathds W^{b(r)}_{c,n+1}
\end{align*}
in $\mathfrak{U}(\mathsf W^{(K)}_\infty)/(\epsilon_1=\mathsf c=0)$. Let us apply the adjoint action of $\Psi_\infty(\mathsf t_{0,1})=\frac{1}{\epsilon_2}\mathds W^{a(1)}_{a,1}$ $n$ times to the two sides of \eqref{eqn: map to W(infinity) modulo epsilon_1} and we find
\begin{align*}
    \Psi_{\infty}(\mathsf T_{0,m}(E^a_b))=(-1)^n h_{n,1}\mathds W^{a(1)}_{b,m}+(-1)^nh_{n,2}\delta^a_b\mathds W^{d(1)}_{d,m}.
\end{align*}
Comparing with \eqref{eqn: map to W(infinity)}, we find $h_{n,1}=(-1)^n$ and $h_{n,2}=0$. Next we plug \eqref{eqn: map to W(infinity) modulo epsilon_1} with $h_{n,1}=(-1)^n,h_{n,2}=0$ to the left-hand-side of \eqref{eqn: [W^3,W^s] modulo epsilon_1} and get
\begin{align*}
    \frac{1}{\epsilon_2}[\mathds W^{a(3)}_{a,-2},\mathds W^{c(s)}_{d,n}]&=(-1)^{s-1}[\Psi_\infty(\mathsf t_{2,0}),\Psi_\infty(\mathsf T_{s-1,n+s-1}(E^c_d))]+\mathcal O(\epsilon_2)\text{-terms}\\
    &=(-1)^{s-1}2(n+s-1)\Psi_\infty(\mathsf T_{s,n+s-2}(E^c_d))+\mathcal O(\epsilon_2)\text{-terms}\\
    &=-2(n+s-1)\mathds W^{c(s+1)}_{d,n-2}+\mathcal O(\epsilon_2)\text{-terms}.
\end{align*}
Comparing with right-hand-side of \eqref{eqn: [W^3,W^s] modulo epsilon_1}, we find $\mu_{s,1}=2$ and $\mu_{s,2}=0$.
\end{proof}

\begin{proposition}\label{prop: embed A^(K) into mode algebra of W(infinity)}
The composition of ${\mathsf D}^{(K)}\overset{ \Psi_{\infty}}{\longrightarrow}\mathfrak{U}(\mathsf W^{(K)}_\infty)\to \mathfrak{U}(\mathsf W^{(K)}_\infty)/(\mathsf c=0)$ is injective. In particular, $\Psi_{\infty}:{\mathsf A}^{(K)}\to \mathfrak{U}(\mathsf W^{(K)}_\infty)[\epsilon_2^{-1}]$ is an embedding.
\end{proposition}

\begin{proof}
By the flatness of ${\mathsf D}^{(K)}$ and $\mathfrak{U}(\mathsf W^{(K)}_\infty)/(\mathsf c=0)$ over the base ring $\mathbb C[\epsilon_1,\epsilon_2]$, to show the injectivity of the map ${\mathsf D}^{(K)}\to \mathfrak{U}(\mathsf W^{(K)}_\infty)/(\mathsf c=0)$, it suffices to show the injectivity after modulo $\epsilon_1,\epsilon_2$.

Identifying $\mathfrak{U}(\mathsf W^{(K)}_\infty)/(\epsilon_1=\epsilon_2=\mathsf c=0)$ with (a completion of) $U(\mathscr O(\mathbb C\times \mathbb C^{\times})\otimes\gl_K)$ via \eqref{eqn: mode algebra of W(infinity) at degenerate limit}, and identifying ${\mathsf D}^{(K)}/(\epsilon_1=\epsilon_2=0)$ with $U(\mathscr O(\mathbb C\times \mathbb C)\otimes\gl_K)$ using Corollary \ref{cor: A is DDCA}, then the equation \eqref{eqn: map to W(infinity) modulo epsilon_1} implies that $\Psi_{\infty}: {\mathsf D}^{(K)}/(\epsilon_1=\epsilon_2=0)\to \mathfrak{U}(\mathsf W^{(K)}_\infty)/(\epsilon_1=\epsilon_2=\mathsf c=0)$ is nothing but the one induced by the natural embedding $\mathscr O(\mathbb C\times \mathbb C)\otimes\gl_K\hookrightarrow \mathscr O(\mathbb C\times \mathbb C^{\times})\otimes\gl_K$. This proves the injectivity of $\Psi_\infty$ modulo $\epsilon_1,\epsilon_2$, which in turn implies the injectivity of $\Psi_\infty$ by the flatness of ${\mathsf D}^{(K)}$ and $\mathfrak{U}(\mathsf W^{(K)}_\infty)/(\mathsf c=0)$ over $\mathbb C[\epsilon_1,\epsilon_2]$.
\end{proof}

\subsection{Vertical filtration on \texorpdfstring{$\mathsf W^{(K)}_{\infty}$}{Wk inf}}\label{subsec: vertical filtration on W_inf}

Let us equip $\mathbf W$ with a filtration $0=V_{-1}\mathbf W\subset V_0\mathbf W\subset V_1\mathbf W\subset\cdots \subset \mathbf W$, where $V_s\mathbf W$ is spanned by the vectors $A^{a_n,r_n}_{b_n,i_n}\cdots A^{a_1,r_1}_{b_1,i_1}$ such that $\sum_{j=1}^n(r_j-1)\le s$. Then $V_{\bullet}\mathbf W$ induces a filtration $V_{\bullet}\mathsf W^{(K)}_{\infty}$.
\begin{lemma}
$V_{\bullet}\mathsf W^{(K)}_{\infty}$ is a vertex algebra filtration, i.e. $\partial V_{s}\mathsf W^{(K)}_{\infty}\subset V_{s}\mathsf W^{(K)}_{\infty}$ and $Y(V_{s}\mathsf W^{(K)}_{\infty},z)V_{t}\mathsf W^{(K)}_{\infty}\subset V_{s+t}\mathsf W^{(K)}_{\infty}[\![z^{\pm}]\!]$. Moreover, there is vertex algebra isomorphism $\mathrm{gr}_V\mathsf W^{(K)}_{\infty}\cong V^{\epsilon_3\mathsf c,\epsilon_1\mathsf c}(\gl_K[z])$, equivalently the OPEs in $\mathsf W^{(K)}_{\infty}$ have the following form:
\begin{align*}
    \mathds W^{a(r)}_{b}(z) \mathds W^{c(s)}_{d}(w)\sim \frac{\delta^c_b \mathds W^{a(r+s-1)}_{d}(w)-\delta^a_d \mathds W^{c(r+s-1)}_{b}(w)}{z-w}+\delta_{r,1}\delta_{s,1}\frac{\epsilon_3\mathsf c\delta^a_d\delta^c_b+\epsilon_1\mathsf c\delta ^a_b\delta^c_d}{(z-w)^2}\pmod{V_{r+s-3}\mathsf W^{(K)}_{\infty}}.
\end{align*}.
\end{lemma}

\begin{proof}
The two statements follow directly from their finite-$L$ counterparts, which are known, see Lemma \ref{lem: associated graded of W-algebra}.
\end{proof}

Now we are ready to prove Proposition \ref{prop: Psi_L filtered}, in fact it is deduced from its uniform-in-$L$ version.

\begin{proposition}\label{prop: Psi_inf filtered}
$\Psi_\infty(\mathsf T_{n,m}(E^a_b))\equiv (-1)^n \mathds W^{a(n+1)}_{b,m-n}\pmod{V_{n-1}\mathfrak U(\mathsf W^{(K)}_\infty)}$.
\end{proposition}

\begin{proof}
We prove the statement by induction on $n$. The cases when $n=0$ or $n=1$ are automatically true by \eqref{eqn: map to W(infinity)}. Now assume that the statement is true for all $n$ such that $n\le r$, and let us deduce that it also holds for $n=r+1$. By Lemma \ref{lem: induction argument}, $\exists f_{r+1}\in \mathbb C[\epsilon_1,\epsilon_2]$ such that 
\begin{align}\label{eqn: induction step}
    \Psi_\infty(\mathsf T_{r+1,m}(E^a_b))\equiv (-1)^{r+1} \mathds W^{a(r+2)}_{b,m-r-1}+f_{r+1}\delta^a_b \mathds W^{c(r+2)}_{c,m-r-1}\pmod{V_{r}\mathfrak U(\mathcal W^{(K)}_\infty)}.
\end{align}
Notice that $\Psi_\infty$ is a graded homomorphism, where we set $\deg\epsilon_1=\deg\epsilon_2=1,\deg\mathsf c=-1$, and the grading on $\mathsf D^{(K)}$ is given by $\deg\mathsf T_{n,m}(X)=n$, and the grading on $\mathfrak U(\mathcal W^{(K)}_\infty)$ is given by $\deg\mathds W^{a(n)}_{b,m}=n-1$. Comparing the degrees of two sides of \eqref{eqn: induction step}, we see that $f_{r+1}\in \mathbb C$. By \eqref{eqn: map to W(infinity) modulo epsilon_1} and Lemma \ref{lem: map to W(infinity) modulo epsilon_1}, we see that $f_{r+1}\equiv 0\pmod{\epsilon_1}$, which implies that $f_{r+1}=0$. This finishes the induction step.
\end{proof}

\begin{remark}
As a corollary, we see the following relation holds:
\begin{equation}\label{eqn: W^3W^r}
\begin{split}
    \mathds W^{a(3)}_{a,-2}\mathds W^{b(r)}_{c,-r}|0\rangle &\equiv  2\epsilon_2 \mathds W^{b(r+1)}_{c,-r-2}|0\rangle\pmod{V_{r-1}\mathsf W^{(K)}_\infty}, \quad\forall r>1.
\end{split}
\end{equation}
For $r=1$, it can be computed directly that
\begin{align}
    \mathds W^{a(3)}_{a,-2}\mathds W^{b(1)}_{c,-r}|0\rangle &\equiv  (\epsilon_1\mathsf c-2)\left(\epsilon_3 \mathds W^{b(2)}_{c,-r-2}+\epsilon_1 \delta^b_c\mathds W^{d(2)}_{d,-r-2}\right)|0\rangle\pmod{V_{0}\mathsf W^{(K)}_\infty}
\end{align}
\end{remark}

\subsection{Zhu algebra of \texorpdfstring{$\mathsf W^{(K)}_{\infty}$}{Wk inf}}

The Zhu algebra $\mathrm{Zhu}(\mathcal V)$ of a graded vertex algebra $\mathcal V$ is defined to be the $\mathcal B$-algebra of the mode algebra, i.e. $\mathrm{Zhu}(\mathcal V)=\mathcal B(\mathfrak{U}(\mathcal V))=\mathfrak{U}(\mathcal V)_0/\sum_{i>0}\mathfrak{U}(\mathcal V)_{i}\mathfrak{U}(\mathcal V)_{-i}$, where $\mathfrak{U}(\mathcal V)_{d}$ is the homogeneous degree $d$ component of $\mathfrak{U}(\mathcal V)$.

\begin{theorem}
There is an algebra isomorphism 
\begin{align}
    \mathrm{Zhu}(\mathsf W^{(K)}_{\infty})\cong Y_{\epsilon_1}(\gl_K)[\epsilon_2,\mathsf c]
\end{align}
between the Zhu algebra of $\mathsf W^{(K)}_{\infty}$ and the Yangian of $\gl_K$.
\end{theorem}

\begin{proof}
Recall the RTT generators $T^a_{b;n}$ of $Y_{\epsilon_1}(\gl_K)$ in \eqref{eqn: RTT relations}, and consider the map $T^a_{b;n}\mapsto \Psi_{\infty}(\mathsf T_{\mathbf r_n}(E^a_b))$ where $\mathsf T_{\mathbf r_n}(E^a_b)$ is defined in \ref{lem: embedding of Yangian}, then it uniquely determines a $\mathbb C[\epsilon_1]$-algebra homomorphism $Y_{\epsilon_1}(\gl_K)\to \mathfrak{U}(\mathsf W^{(K)}_{\infty})$. Since $\Psi_{\infty}(\mathsf T_{\mathbf r_n}(E^a_b))$ has degree zero, this map induces a map $Y_{\epsilon_1}(\gl_K)[\epsilon_2,\mathsf c]\to \mathrm{Zhu}(\mathsf W^{(K)}_{\infty})$. We claim that this map is an isomorphism.

To prove this claim, we consider the filtration $V_{\bullet}\mathrm{Zhu}(\mathsf W^{(K)}_{\infty})$ induced from the vertical filtration $V_{\bullet}\mathfrak{U}(\mathsf W^{(K)}_{\infty})$. By Remark \ref{rmk: unsymmetrized basis} $\mathsf T_{\mathbf r_n}(E^a_b)\equiv \mathsf T_{n,n}(E^a_b)\pmod{V_{n-1}\mathsf A^{(K)}}$, then Proposition \ref{prop: Psi_inf filtered} implies that $\Psi_\infty(\mathsf T_{\mathbf r_n}(E^a_b))\equiv (-1)^n\mathds W^{a(n+1)}_{b,0}\pmod{V_{n-1}\mathfrak{U}(\mathsf W^{(K)}_{\infty})}$.
By the construction of $\mathsf W^{(K)}_\infty$, $\mathds W^{a(s)}_{b}$ strongly generates $\mathsf W^{(K)}_{\infty}$ and form a PBW basis of $\mathsf W^{(K)}_{\infty}$, thus $\mathrm{Zhu}(\mathsf W^{(K)}_{\infty})$ is generated by $\{\mathds W^{a(n)}_{b,0}\:|\: 1\le a,b\le K, n\in\mathbb Z_{\ge 1}\}$. Moreover $\mathrm{gr}_V\mathrm{Zhu}(\mathsf W^{(K)}_{\infty})$ is isomorphic to the universal enveloping algebra $U(\gl_K[z])\otimes\mathbb C[\epsilon_1,\epsilon_2,\mathsf c]$. The pullback of the vertical filtration $V_\bullet\mathsf A^{(K)}$ induces a filtration $V_{\bullet}Y_{\epsilon_1}(\gl_K)$ on the Yangian, and $\mathrm{gr}_VY_{\epsilon_1}(\gl_K)\cong U(\gl_K[z])\otimes\mathbb C[\epsilon_1]$. It follows that the map $Y_{\epsilon_1}(\gl_K)[\epsilon_2,\mathsf c]\to \mathrm{Zhu}(\mathsf W^{(K)}_{\infty})$ becomes an isomorphism after passing to associated graded algebra with respect to $V_\bullet$, whence itself is an isomorphism.
\end{proof}

\subsection{Linear degeneration limit}

\begin{proposition}\label{prop: linear degeneration}
$\mathfrak{U}(\mathsf W^{(K)}_\infty)/(\epsilon_1=0)$ is a completion of the universal enveloping algebra of the Lie algebra $D_{\epsilon_2}(\mathbb C^{\times})\otimes \gl_K$ centrally extended by $\epsilon_2\mathsf c$ times the standard 2-cocycle.
\end{proposition}

\begin{proof}
By the Lemma \ref{lem: E1 to 0 limit}, $\mathfrak{U}(\mathsf W^{(K)}_\infty)/(\epsilon_1=0)$ is a completion of the universal enveloping algebra of the Lie algebra $\mathscr O_{\epsilon_2,\mathsf c}(\mathbb C\times \mathbb C^{\times})\otimes\gl_K$ which is a flat deformation of $\mathscr O(\mathbb C\times \mathbb C^{\times})\otimes\gl_K$ over the base ring $\mathbb C[\epsilon_2,\mathsf c]$. 

We claim that $\mathscr O_{\epsilon_2,0}(\mathbb C\times \mathbb C^{\times})\otimes\gl_K$ is isomorphic to $D_{\epsilon_2}(\mathbb C^{\times})\otimes \gl_K$. To prove this claim, we begin with the positive mode $\mathscr O_{\epsilon_2,0}(\mathbb C\times \mathbb C)\otimes\gl_K$. It follows from \eqref{eqn: map to W(infinity) modulo epsilon_1} that $\Psi_{\infty}$ maps ${\mathsf D}^{(K)}/(\epsilon_1)$ to $ U(\mathscr O_{\epsilon_2,0}(\mathbb C\times \mathbb C)\otimes\gl_K)$ and according to Proposition \ref{prop: embed A^(K) into mode algebra of W(infinity)} this map is injective. It is also surjective because every generator $\mathds W^{a(s)}_{b,n},\;(n\ge 1-s)$ is in the image of $\Psi_{\infty}$ by \eqref{eqn: map to W(infinity) modulo epsilon_1}. Thus $\Psi_{\infty}$ induces an isomorphism ${\mathsf D}^{(K)}/(\epsilon_1)\cong U(\mathscr O_{\epsilon_2,0}(\mathbb C\times \mathbb C)\otimes\gl_K)$, and using the Corollary \ref{cor: A is DDCA} we arrive at an isomorphism
\begin{align}\label{eqn: mode alg mod c and epsilon_1}
    D_{\epsilon_2}(\mathbb C)\otimes \gl_K\cong \mathscr O_{\epsilon_2,0}(\mathbb C\times \mathbb C)\otimes\gl_K
\end{align} 
which maps $E^a_bx^m(\epsilon_2\partial_x)^n$ to $(-1)^n\mathds W^{a(n+1)}_{b,m-n}+\epsilon_2\cdot(\text{linear combination of $\mathds W^{(i)}_{m-n}$ for $i\le n$})$.

Next, consider the meromorphic coproduct for the restricted mode algebra
$$\Delta_{\mathsf W^{(K)}_{\infty}}:U(\mathsf W^{(K)}_{\infty})\to U(\mathsf W^{(K)}_{\infty})\otimes U_+(\mathsf W^{(K)}_{\infty})(\!(w^{-1})\!)$$
which is a special case of the general construction in Lemma \ref{lem: meromorphic coproduct for restricted modes}, and compose it with the map $\mathfrak C_{\mathsf W}\otimes 1$ where $\mathfrak C_{\mathsf W}$ is the truncation map defined in \eqref{eqn: augmentation of W(infinity)}, then modulo $\epsilon_1,\mathsf c$ on both the domain and codomain, we get a $\mathbb C[\epsilon_2]$-algebra map $S_{\mathsf W}(w):U(\mathsf W^{(K)}_{\infty}/(\epsilon_1,\mathsf c))\to U_+(\mathsf W^{(K)}_{\infty}/(\epsilon_1,\mathsf c))(\!(w^{-1})\!)$. Restrict the domain of $\Delta_{\mathsf W^{(K)}_{\infty}}$ to the subalgebra $U(\mathscr O_{\epsilon_2,0}(\mathbb C\times \mathbb C^{\times})\otimes\gl_K)$, then we find the image is contained in $U(\mathscr O_{\epsilon_2,0}(\mathbb C\times \mathbb C)\otimes\gl_K)(\!(w^{-1})\!)$. In fact it is easy to compute that
\begin{align}\label{eqn: shift map of W(infinity) modulo epsilon_1}
   S_{\mathsf W}(w)(\mathds W^{a(s)}_{b,n})=\sum_{m=1-s}^{\infty} \binom{n+s-1}{m+s-1}w^{n-m}\mathds W^{a(s)}_{b,m},\quad (n\in \mathbb Z).
\end{align}
In particular, $S_{\mathsf W}(w)$ is induced by a Lie algebra map $\mathscr O_{\epsilon_2,0}(\mathbb C\times \mathbb C^{\times})\otimes\gl_K\to \mathscr O_{\epsilon_2,0}(\mathbb C\times \mathbb C)\otimes\gl_K (\!(w^{-1})\!)$. Modulo $\epsilon_2$, this map is exactly a special case of the map $\Delta(w)_{0,1}$ that will be constructed in \eqref{eqn: meromorphic coproduct finite N}, in particular $S_{\mathsf W}(w)$ is injective modulo $\epsilon_2$. Since both $\mathscr O_{\epsilon_2,0}(\mathbb C\times \mathbb C^{\times})\otimes\gl_K$ and $\mathscr O_{\epsilon_2,0}(\mathbb C\times \mathbb C)\otimes\gl_K (\!(w^{-1})\!)$ are flat over $\mathbb C[\epsilon_2]$, we conclude that $S_{\mathsf W}(w)$ is injective. Direct computation using the formula \eqref{eqn: shift map of W(infinity) modulo epsilon_1} and the identification \eqref{eqn: mode alg mod c and epsilon_1} shows that 
$$S_{\mathsf W}(w)(\mathds W^{a(1)}_{b,n})=\Delta(w)_{0,1}(E^a_bx^{n}),\;\forall n\in \mathbb Z,$$
where $\Delta(w)_{0,1}: D_{\epsilon_2}(\mathbb C^{\times})\otimes\gl_K\hookrightarrow D_{\epsilon_2}(\mathbb C)\otimes\gl_K(\!(w^{-1})\!)$ is the map constructed in \eqref{eqn: meromorphic coproduct finite N}. Since $D_{\epsilon_2}(\mathbb C^{\times})\otimes\gl_K$ is generated by the subalgebras $D_{\epsilon_2}(\mathbb C)\otimes\gl_K$ and $\mathscr O(\mathbb C^{\times})\otimes \gl_K$ when $\epsilon_2$ is invertible, we deduce that $\Delta(w)_{0,1}(D_{\epsilon_2}(\mathbb C^{\times})\otimes\gl_K[\epsilon_2^{-1}])\subset S_{\mathsf W}(w)(\mathscr O_{\epsilon_2,0}(\mathbb C\times \mathbb C^{\times})\otimes\gl_K[\epsilon_2^{-1}])$. This inclusion further implies that $\Delta(w)_{0,1}(D_{\epsilon_2}(\mathbb C^{\times})\otimes\gl_K)\subset S_{\mathsf W}(w)(\mathscr O_{\epsilon_2,0}(\mathbb C\times \mathbb C^{\times})\otimes\gl_K)$ because the cokernel of $S_{\mathsf W}(w)$ has no $\epsilon_2$-torsion (since $S_{\mathsf W}(w)$ is injective modulo $\epsilon_2$). Taking the composition $S_{\mathsf W}(w)^{-1}\circ \Delta(w)_{0,1}$ we get an embedding $D_{\epsilon_2}(\mathbb C^{\times})\otimes\gl_K\subset \mathscr O_{\epsilon_2,0}(\mathbb C\times \mathbb C^{\times})\otimes\gl_K$. Moreover, by inductively applying adjoint actions of $\frac{1}{\epsilon_2}x(\epsilon_2\partial_x)^2=\frac{1}{\epsilon_2}\mathds W^{c(3)}_{c,-1}+h\mathds W^{c(2)}_{c,-1}$ for some $h\in \mathbb C[\epsilon_2]$, starting from $E^a_bx^{-m}=\mathds W^{a(1)}_{b,-m}$, we can show that for all $m>0$, $$E^a_bx^{-m}(\epsilon_2\partial_x)^n\mapsto (-1)^n\mathds W^{a(n+1)}_{b,-m-n}+\epsilon_2\cdot(\text{linear combination of $\mathds W^{(i)}_{-m-n}$ for $i\le n$}).$$ In particular the embedding $D_{\epsilon_2}(\mathbb C^{\times})\otimes\gl_K\subset \mathscr O_{\epsilon_2,0}(\mathbb C\times \mathbb C^{\times})\otimes\gl_K$ is surjective. This proves our claim.

Finally, there is a unique cocycle $\omega$ up to coboundary on the Lie algebra $D_{\epsilon_2}(\mathbb C^{\times})\otimes\gl_K$ such that it is $\GL_K$ invariant, involves a single trace in $\gl_K$, and modulo $\epsilon_2$ it equals to 
\begin{align*}
    \omega (f,g)=\frac{1}{2\pi i}\oint_{|x|=1,y=0}\mathrm{Tr}f\partial g.
\end{align*}
The cocycle that governs the central extension $\mathscr O_{\epsilon_2,\mathsf c}(\mathbb C\times \mathbb C^{\times})\otimes\gl_K$ has the same property as $\epsilon_2\mathsf c\cdot \omega$, so it must be equal to $\epsilon_2\mathsf c\cdot \omega$. This finishes the proof.
\end{proof}

From the above discussions we conclude that
\begin{align}\label{eq: linear degeneration}
    \mathsf W^{(K)}_\infty/(\epsilon_1=0)\cong U_{\epsilon_2\mathsf c}(D_{\epsilon_2}(\mathbb C^{\times})\otimes \gl_K)\otimes_{U(D_{\epsilon_2}(\mathbb C)\otimes \gl_K)}\mathbb C[\epsilon_2],
\end{align}
this vertex algebra is also known as matrix extended linear $\mathcal W_{1+\infty}$ algebra \cite{Costello:2016nkh,ben2010symmetry}.

\bigskip Using the free $\beta\gamma$-$bc$ system description of linear $\mathcal W_{1+\infty}$ algebra in \cite[Section 15]{Costello:2016nkh}, we can write down the isomorphism \eqref{eq: linear degeneration} explicitly. Namely, let $D^{\mathrm{ch}}(T^*\Hom(\mathbb C^K,\mathbb C^{N|M}))$ be $(N+M)K$ pairs of $\beta\gamma$-$bc$ systems, i.e. the vertex algebra freely generated by $\{\beta^a_i,\gamma^i_a\:|\:1\le a\le K,1\le i\le N+M\}$ where $\beta^a_i,\gamma^i_a$ is bosonic for $1\le i\le N$ and is fermionic for $N<i\le N+M$, with the OPEs:
\begin{align*}
    \gamma^i_a(z)\beta^b_j(w)\sim \frac{\delta^b_a\delta^i_j}{z-w}.
\end{align*}
Denote $$\mathscr O^{a(m)}_{b}(z)=:\partial^{m-1}\beta^a_i(z)\gamma^i_b(z):,\quad 1\le a,b,\le K,\quad m\in \mathbb N_{>0}.$$ 
Then $\mathscr O^{a(m)}_{b}(z)$ generates the $\GL_{N|M}$-invariants $\mathcal F^{(K)}_{N|M}=D^{\mathrm{ch}}(T^*\Hom(\mathbb C^K,\mathbb C^{N|M}))^{\GL_{N|M}}$.
It is straightforward to compute the OPEs:
\begin{equation}\label{OmOn OPE}
\begin{split}
\mathscr O^{a(m)}_{b}(z)\mathscr O^{c(n)}_{d}(w)\sim\: &(-1)^m(N-M)\frac{m!n!\delta^a_d\delta^c_b}{(z-w)^{m+n}}+\sum_{\ell=0}^{n-1}\frac{n!\delta^c_b}{\ell!(z-w)^{n-\ell}}\mathscr O^{a(m+\ell)}_{d}(w)\\
&+\sum_{\substack{i,j\in \mathbb Z_{\ge 0}\\i+j<m}}\frac{(-1)^{m-j}m!\delta^a_d}{i!j!(z-w)^{m-i-j}}\partial^i\mathscr O^{c(n+j)}_{b}(w).
\end{split}
\end{equation}
We note that the structure constants in $\mathcal F^{(K)}_{N|M}$ only depends on the difference $N-M$. In fact, there is a natural vertex algebra projection $\mathcal F^{(K)}_{N+1|M+1}\twoheadrightarrow \mathcal F^{(K)}_{N|M}$ that maps $\mathscr O^{a(m)}_{b}(z)$ to $\mathscr O^{a(m)}_{b}(z)$. Let us fix $L=N-M$, and define $\mathcal F^{(K)}_{L+\infty|\infty}$ to be the vertex algebra freely generated by $\{\mathscr O^{a(m)}_{b}\:|\:  1\le a,b,\le K,\: m\in \mathbb N_{>0}\}$ with OPE \eqref{OmOn OPE}. Combine \cite[Proposition 15.3.7]{Costello:2016nkh} with Proposition \ref{prop: linear degeneration}, we arrive at the following identification.
\begin{lemma}
There exists a vertex algebra isomorphism
\begin{align}
    \mathsf W^{(K)}_\infty/(\epsilon_1=0,\epsilon_2=1,\mathsf c=L)\cong \mathcal F^{(K)}_{L+\infty|\infty}.
\end{align}
Moreover, the map is such that
\begin{align}\label{W to O}
    \mathds W^{a(n)}_b(z)\mapsto \mathscr O^{a(n)}_b(z)+\sum_{\ell=1}^{n-1}\lambda_{n,\ell} \partial^{n-\ell}\mathscr O^{a(\ell)}_b(z)+\delta^a_b\sum_{\ell=0}^{n-1}\mu_{n,\ell} \partial^{n-\ell}\mathscr O^{c(\ell)}_c(z).
\end{align}
where $\lambda_{n,\ell},\mu_{n,\ell}\in \mathbb C$.
\end{lemma}


The coefficients $\lambda_{n,\ell},\mu_{n,\ell}$ in \eqref{W to O} can be constrained using the OPE. In fact, we have the following explicit form of an isomorphism.
\begin{theorem}
The map $\mathds W^{a(n)}_b(z)\mapsto \mathscr O^{a(n)}_b(z)$ generates a vertex algebra isomorphism between $\mathsf W^{(K)}_\infty/(\epsilon_1=0,\epsilon_2=1,\mathsf c=L)$ and $\mathcal F^{(K)}_{L+\infty|\infty}$.
\end{theorem}

\begin{proof}
Our strategy is to compare the $\mathds W^{(1)}\mathds W^{(n)}$ OPE with the $\mathscr O^{(1)}\mathscr O^{(n)}$ OPE, it turns out that if $L\neq 0$ then this will completed fix the coefficients $\lambda_{n,\ell},\mu_{n,\ell}$ in \eqref{W to O} to be zero. The case $L=0$ will be deduced from the $L\neq 0$ cases using the polynomiality of the OPE coefficients in $\mathsf W^{(K)}_\infty$.

Let us first assume that $L\neq 0$. Then we claim that $\lambda_{n,\ell},\mu_{n,\ell}$ in \eqref{W to O} must be zero for all $n,\ell$. We prove the claim by induction on $n$. The claim automatically holds for $n=0$. Suppose that $n>0$ and the claim holds for all $n'$ such that $n'<n$. We shall abuse the notation by denoting the image of $\mathds W$ by the same symbol. Using \eqref{W^1W^n OPE} and \eqref{OmOn OPE}, we get
\begin{align*}
   \mathds W^{c(1)}_c(z)\mathds W^{a(n)}_b(w)-\mathscr O^{c(1)}_{c}(z)\mathscr O^{a(n)}_{b}(w)\sim & \left[\frac{-\mathsf c(n-1)!\delta^a_b}{(z-w)^{n+1}}+\sum_{i=1}^{n-1}\frac{(n-1)!}{(i-1)!}\frac{\mathds W^{a(i)}_b(w)}{(z-w)^{n+1-i}}\right]\\
   &-\left[\frac{-\mathsf c(n-1)!\delta^a_b}{(z-w)^{n+1}}+\sum_{i=1}^{n-1}\frac{(n-1)!}{(i-1)!}\frac{\mathscr O^{a(i)}_b(w)}{(z-w)^{n+1-i}}\right].
\end{align*}
By the induction hypothesis, the right-hand-side of the above OPE vanishes. On the other hand, $\mathds W^{a(n)}_b(z)$ is mapped to $\mathscr O^{a(n)}_b(z)+\sum_{\ell=1}^{n-1}\lambda_{n,\ell} \partial^{n-\ell}\mathscr O^{a(\ell)}_b(z)+\delta^a_b\sum_{\ell=1}^{n-1}\mu_{n,\ell} \partial^{n-\ell}\mathscr O^{c(\ell)}_c(z)$. Thus we have
\begin{align*}
     \mathds W^{c(1)}_c(z)\mathds W^{a(n)}_b(w)-\mathscr O^{c(1)}_{c}(z)\mathscr O^{a(n)}_{b}(w)=\sum_{\ell=1}^{n-1}\lambda_{n,\ell} \mathscr O^{c(1)}_{c}(z) \partial^{n-\ell}\mathscr O^{a(\ell)}_b(w)+\delta^a_b\sum_{\ell=1}^{n-1}\mu_{n,\ell} \mathscr O^{c(1)}_{c}(z)\partial^{n-\ell}\mathscr O^{d(\ell)}_d(w),
\end{align*}
and the $\frac{1}{(z-w)^2}$ term on the right-hand-side is
\begin{align*}
    \frac{1}{(z-w)^2}\left[\sum_{\ell=2}^{n-1}\lambda_{n,\ell} (\ell-1) \partial^{n-\ell}\mathscr O^{a(\ell-1)}_b(w)+\delta^a_b\sum_{\ell=2}^{n-1}\mu_{n,\ell}(\ell-1) \partial^{n-\ell}\mathscr O^{d(\ell-1)}_d(w)\right],
\end{align*}
which must vanish. Thus $\lambda_{n,\ell}=\mu_{n,\ell}=0$ for all $1<\ell<n$. Using the aforementioned vanishing result, the following OPE can be simplified as
\begin{align}\label{W1Wn-O1On}
   \mathds W^{c(1)}_d(z)\mathds W^{a(n)}_b(w)-\mathscr O^{c(1)}_{d}(z)\mathscr O^{a(n)}_{b}(w)\sim \frac{\lambda_{n,1}}{z-w}(\delta^a_d\partial^{n-1}\mathscr O^{c(1)}_b(w)-\delta^c_b\partial^{n-1}\mathscr O^{a(1)}_d(w)).
\end{align}
On the hand hand, we can replace $\mathds W^{a(n)}_b(z)$ by $\mathscr O^{a(n)}_b(z)+\lambda_{n,1} \partial^{n-1}\mathscr O^{a(\ell)}_b(z)+\delta^a_b\mu_{n,1} \partial^{n-1}\mathscr O^{c(\ell)}_c(z)$, and get
\begin{align*}
    \mathds W^{c(1)}_d(z)\mathds W^{a(n)}_b(w)-\mathscr O^{c(1)}_{d}(z)\mathscr O^{a(n)}_{b}(w)=\lambda_{n,1} \mathscr O^{c(1)}_{d}(z) \partial^{n-1}\mathscr O^{a(1)}_b(w)+\delta^a_b\mu_{n,1} \mathscr O^{c(1)}_{d}(z)\partial^{n-1}\mathscr O^{e(1)}_e(w),
\end{align*}
and the $\frac{1}{(z-w)^{n+1}}$ term on the right-hand-side is
\begin{align*}
    -n!L\frac{\delta^c_b\delta^a_b\lambda_{n,1}+\delta^a_b\delta^c_d\mu_{n,1}}{(z-w)^{n+1}},
\end{align*}
which must vanish according to \eqref{W1Wn-O1On}. Note that $L\neq 0$ by the assumption, so $\delta^c_b\delta^a_b\lambda_{n,1}+\delta^a_b\delta^c_d\mu_{n,1}=0$. Since $\delta^c_b\delta^a_b$ and $\delta^a_b\delta^c_d$ are linearly independent, we conclude that $\lambda_{n,1}=\mu_{n,1}=0$, thus $\mathds W^{a(n)}_b(z)$ is mapped to $\mathscr O^{a(n)}_b(z)$. This finishes the induction step, thus proving the theorem in the case $L\neq 0$.

In particular, the $\mathds W^{(m)}\mathds W^{(n)}$ OPE in $\mathsf W^{(K)}_{\infty}/(\epsilon_1=0,\epsilon_2=1)$ is
\begin{equation}\label{WmWn OPE, epsilon_1=0}
\begin{split}
\mathds W^{a(m)}_{b}(z)\mathds W^{c(n)}_{d}(w)\sim\: &(-1)^m\mathsf c\frac{m!n!\delta^a_d\delta^c_b}{(z-w)^{m+n}}+\sum_{\ell=0}^{n-1}\frac{n!\delta^c_b}{\ell!(z-w)^{n-\ell}}\mathds W^{a(m+\ell)}_{d}(w)\\
&+\sum_{\substack{i,j\in \mathbb Z_{\ge 0}\\i+j<m}}\frac{(-1)^{m-j}m!\delta^a_d}{i!j!(z-w)^{m-i-j}}\partial^i\mathds W^{c(n+j)}_{b}(w),
\end{split}
\end{equation}
whenever $\mathsf c\in \mathbb Z\backslash\{0\}$. By the polynomiality of the OPE coefficients of $\mathsf W^{(K)}_\infty$ (Proposition \ref{prop: polynomiality}), the OPE \eqref{WmWn OPE, epsilon_1=0} holds for all $\mathsf c$. This implies the theorem in the case $L=0$.
\end{proof}




\subsection{Compatibility between coproduct and meromorphic coproduct}

Applying the functoriality of the meromorphic coproduct of the restricted mode algebra (Proposition \ref{prop: functoriality of meromorphic coproduct}) to the W-algebra coproduct \eqref{eqn: W(infinity) coproduct_integral form}
$\Delta_{\mathsf W}:\mathsf W^{(K)}_{\infty}\to \mathsf W^{(K)}_{\infty}\otimes \mathsf W^{(K)}_{\infty}$, we get the compatibility between the meromorphic coproduct and usual coproduct for the W-algebras:
\begin{equation}\label{cd: compatible W-infinity coproducts}
\Delta_{\mathsf W\otimes \mathsf W}(w)\circ \Delta_{\mathsf W}=(\Delta_{\mathsf W}\otimes \Delta_{\mathsf W})\circ \Delta_{\mathsf W}(w)
\end{equation}

\subsection{Duality isomorphism of \texorpdfstring{$\mathcal W^{(K)}_{\infty}$}{Wk inf}}\label{subsec: duality for W}
The rectangular $\mathcal W_{\infty}$-algebra that we have discussed so far is denoted by $\mathcal W^{(K)}_{0,0,\infty}$ in the literature \cite{gaiotto2022miura}. Previously in Section \ref{subsec: duality of W}, we have discussed the $\widetilde{\mathcal W}^{(K)}_L$, which is denoted by $\mathcal W^{(K)}_{0,L,0}$ in \cite{gaiotto2022miura}. By the construction in Section \ref{subsec: duality of W}, there is a vertex algebra isomorphism $\sigma_L:\mathcal W^{(K)}_L\cong \widetilde{\mathcal W}^{(K)}_L$ such that
\begin{align*}
    \sigma_L(\alpha)=\bar\alpha,\quad \sigma_L(W^{a(r)}_b(z))=(-1)^rU^{b(r)}_{a}(z).
\end{align*}
Such construction is apparently uniform in $L$, thus we obtain the following.
\begin{corollary}\label{cor: W(0,infty,0)}
For every $K\in \mathbb N_{\ge 1}$, there exists a $\mathbb Z$-graded vertex algebra $\widetilde{\mathcal W}^{(K)}_{\infty}$ over the base ring $\mathbb C[\alpha,\lambda]$ with strong generators $U^{a(r)}_b(z),1\le a,b\le K,r=1,2,\cdots$. Moreover there is a vertex algebra isomorphism $\sigma_{\infty}:{\mathcal W}^{(K)}_{\infty}\cong \widetilde{\mathcal W}^{(K)}_{\infty}$ such that
\begin{align*}
    \sigma_{\infty}(\lambda)=\lambda,\quad\sigma_\infty(\alpha)=\bar\alpha:=-\alpha-K,\quad \sigma_\infty(W^{a(r)}_b(z))=(-1)^rU^{b(r)}_{a}(z).
\end{align*}
\end{corollary}
\noindent $\widetilde{\mathcal W}^{(K)}_{\infty}$ is denoted by $\mathcal W^{(K)}_{0,\infty,0}$ in \cite{gaiotto2022miura}.\\

We define the vertex algebra coproduct 
\begin{align}
    \Delta_{\widetilde{\mathcal W}}:=(\sigma_{\infty}\otimes\sigma_{\infty})\circ\Delta_{\mathcal W}\circ \sigma_{\infty}: \widetilde{\mathcal W}^{(K)}_{\infty}\to \widetilde{\mathcal W}^{(K)}_{\infty}\otimes \widetilde{\mathcal W}^{(K)}_{\infty},
\end{align}
$\Delta_{\widetilde{\mathcal W}}$ maps the strong generators $U^{a(r)}_{b}$ by
\begin{align*}
\Delta_{\widetilde{\mathcal W}}(U^{a(r)}_b(z))=\sum_{\substack{(s,t,u)\in\mathbb N^3\\s+t+u=r}} \binom{1\otimes \lambda-t}{u}(-\bar\alpha\partial)^u U^{c(s)}_b(z)\otimes W^{a(t)}_c(z),
\end{align*}
where we set $U^{a(0)}_{b}(z)=\delta^a_b$.\\

$\widetilde{\mathcal W}^{(K)}_{\infty}[\alpha^{-1}]$ has stress-energy operator 
\begin{align*}
    T(z)=-\frac{1}{2\alpha}:U^{a(1)}_{b}U^{b(1)}_{a}:(z)+\frac{\bar\alpha(\lambda-1)}{2\alpha}\partial U^{a(1)}_{a}(z)+\frac{1}{\alpha}U^{a(2)}_{a}(z)
\end{align*}
with central charge 
\begin{align*}
    c=\frac{K\lambda}{\alpha}((\lambda^2-1)\bar\alpha^2-\bar\alpha K-1).
\end{align*}
Moreover, $U^{a(r)}_{b}(z)$ has conformal weight $r$ w.r.t. $T(z)$,
and $U^{a(1)}_{b}(z)$ are primary of spin $1$ w.r.t $T(z)$.\\

We can define the integral form of $\widetilde{\mathcal W}^{(K)}_{\infty}$ over the base ring $\mathbb C[\epsilon_1,\epsilon_2,\mathsf c]$ by setting
\begin{align*}
    \mathds U^{a(r)}_b(z):=\epsilon_1^{r-1}U^{a(r)}_b(z),\quad \mathsf c:=\frac{\lambda}{\epsilon_1},
\end{align*}
then it follows from the isomorphism $\sigma_{\infty}$ and Lemma \ref{lem: integral basis} that structure constants in the basis $\mathds U^{a(r)}_b,r=1,2,\cdots$ are polynomials in $\epsilon_1,\epsilon_2$, and $\mathsf c$. Denote by $\widetilde{\mathsf W}^{(K)}_{\infty}$ the $\mathbb C[\epsilon_1,\epsilon_2,\mathsf c]$-vertex algebra strongly generated by $\mathds U^{a(r)}_b(z)$, and we call it the integral form of $\widetilde{\mathcal W}^{(K)}_{\infty}$. Note that $\sigma_{\infty}$ induces an isomorphism of vertex algebras ${\mathsf W}^{(K)}_{\infty}\cong \widetilde{\mathsf W}^{(K)}_{\infty}$ such that
\begin{align*}
    \sigma_{\infty}(\epsilon_1)=\epsilon_1,\quad \sigma_{\infty}(\epsilon_2)=\epsilon_3,\quad\sigma_{\infty}(\mathsf c)=\mathsf c,\quad \sigma_\infty(\mathds W^{a(r)}_b(z))=(-1)^r\mathds U^{b(r)}_{a}(z).
\end{align*}
Composing the duality transform \eqref{eqn: duality for A} $\sigma:\mathsf A^{(K)}\cong \mathsf A^{(K)}$ with the representation $\Psi_\infty: \mathsf A^{(K)}\to \mathfrak{U}(\mathsf W^{(K)}_\infty)[\epsilon_2^{-1}]$, and then applying the isomorphism $\sigma_\infty:{\mathsf W}^{(K)}_{\infty}\cong \widetilde{\mathsf W}^{(K)}_{\infty}$, we get a new map $$\widetilde\Psi_\infty=\sigma_\infty\circ\Psi_\infty\circ\sigma:\mathsf A^{(K)}\to \mathfrak{U}(\widetilde{\mathsf W}^{(K)}_{\infty})[\epsilon_3^{-1}]$$ which is uniquely determined by
\begin{equation}\label{eqn: Psi tilde}
\boxed{
\begin{aligned}
    &\widetilde\Psi_\infty(\mathsf T_{0,n}(E^a_b))= \mathds U^{a(1)}_{b,n}+\frac{\epsilon_1}{\epsilon_3}\delta^a_b\mathds U^{c(1)}_{c,n},\\
    &\widetilde\Psi_\infty(\mathsf t_{1,n})=-\mathsf L_{n-1}-\frac{\epsilon_1\epsilon_2 n \mathsf c}{2\epsilon_3}\mathds U^{a(1)}_{a,n-1}\\
    &\widetilde\Psi_\infty(\mathsf T_{1,0}(E^a_b))= -\epsilon_1\sum_{m\ge 0}\left(\mathds U^{c(1)}_{b,-m-1}\mathds U^{a(1)}_{c,m}+\frac{\epsilon_1}{\epsilon_3}\delta^a_b\mathds U^{c(1)}_{d,-m-1}\mathds U^{d(1)}_{c,m}\right)+\mathds U^{a(2)}_{b,-1}+\frac{1}{\epsilon_3}\delta^a_b\mathds U^{c(2)}_{c,-1}\\
    &\widetilde\Psi_\infty(\mathsf t_{2,0})=-\frac{1}{\epsilon_3}\left(\widetilde{\mathds V}_{-2}+\epsilon_1\epsilon_3\sum_{n=1}^{\infty}n\:\mathds U^{a(1)}_{b,-n-1}\mathds U^{b(1)}_{a,n-1}+\epsilon_1^2\sum_{n=1}^{\infty}n\: \mathds U^{a(1)}_{a,-n-1}\mathds U^{b(1)}_{b,n-1}\right).
\end{aligned}}
\end{equation}
Here $\widetilde{\mathds V}_{-2}$ is the mode of quasi-primary field $\widetilde{\mathds V}(z)=\sum_{n\in \mathbb Z}\widetilde{\mathds V}_{n}z^{-n-3}$ defined as
\begin{equation}
\begin{split}
\widetilde{\mathds V}(z):=&\frac{\epsilon_1^2}{6}\left(:\mathds U^{a(1)}_b(z)\mathds U^{b(1)}_c(z)\mathds U^{c(1)}_a(z):+:\mathds U^{b(1)}_a(z)\mathds U^{c(1)}_b(z)\mathds U^{a(1)}_c(z):\right)\\
&+\mathds U^{a(3)}_a(z)-\epsilon_1:\mathds U^{a(1)}_b(z)\mathds U^{b(2)}_a(z):.
\end{split}
\end{equation}
Similarly we can define algebra homomorphism
$$\widetilde\Delta_\infty=(\sigma\otimes\sigma_\infty)\circ\Delta_\infty\circ\sigma: \mathsf A^{(K)}\to \mathsf A^{(K)}\widetilde\otimes\mathfrak{U}(\widetilde{\mathsf W}^{(K)}_{\infty})[\epsilon_3^{-1}]$$ which is uniquely determined by
\begin{equation}
\boxed{
\begin{aligned}
&\widetilde\Delta_\infty(\mathsf T_{0,n}(E^a_b))=\square( \mathsf T_{0,n}(E^a_b)),\\
&\widetilde\Delta_\infty(\mathsf t_{1,n})=\square(\mathsf t_{1,n})+\epsilon_1\epsilon_3 n\:\mathsf t_{0,n-1}\otimes \mathsf c,\\
&\widetilde\Delta_\infty(\mathsf T_{1,0}(E^a_b))=\square(\mathsf T_{1,0}(E^a_b))+\epsilon_1\sum_{m=0}^{\infty}\left(\mathsf T_{0,m}(E^c_b)\otimes \mathds U^{a(1)}_{c,-m-1}-\mathsf T_{0,m}(E^a_c)\otimes \mathds U^{c(1)}_{b,-m-1}\right),\\
&\widetilde\Delta_\infty(\mathsf t_{2,0})=\square(\mathsf t_{2,0})-2\epsilon_1\sum_{n=1}^{\infty}n \:\mathsf T_{0,n-1}(E^a_b)\otimes \mathds U^{b(1)}_{a,-n-1},
\end{aligned}}
\end{equation}
where $\square(x):=x\otimes 1+1\otimes \widetilde\Psi_\infty(x)$. Proposition \ref{prop: AW coproduct compatible with WW coproduct} implies that $\widetilde{\Delta}_{\infty}$ is compatible with the vertex algebra coproduct $\Delta_{\widetilde{\mathsf W}}:\widetilde{\mathsf W}^{(K)}_{\infty}\to \widetilde{\mathsf W}^{(K)}_{\infty}\otimes \widetilde{\mathsf W}^{(K)}_{\infty}$ in the sense that 
\begin{align*}
    (\widetilde\Psi_{\infty}\otimes 1)\circ\widetilde\Delta_{\infty}=\Delta_{\widetilde{\mathsf W}}\circ\widetilde\Psi_{\infty}.
\end{align*}

\subsection{An anti-involution of the mode algebra of \texorpdfstring{$\mathsf W^{(K)}_{\infty}$}{Wk inf}}

Consider the anti-involution $\mathfrak{s}_L:\mathfrak{U}({\cal W}^{(K)}_L)\cong \mathfrak{U}({\cal W}^{(K)}_L)$ in \eqref{eqn: anti-involution s_L}, its action on $W^{a(r)}_{b,n}$ is determined by the equation:
\begin{multline*}
    z^{L+1} \left((\alpha\partial)^L+\sum_{r=1}^L (-1)^r(\alpha\partial)^{L-r}W^{b(r)}_{a}(z)\right)z^{L-1}=\\
    (\alpha z^2\partial)^L+\sum_{r=1}^L (-1)^r\sum_{n\in \mathbb Z}\mathfrak s_L\left(W^{a(r)}_{b,n}\right)z^{n+r} (\alpha z^2\partial)^{L-r},
\end{multline*}
so there exist $f_{r,n,i}(L)\in \mathbb C[\alpha]$ which depends on $L$ in a polynomial way and such that
\begin{align*}
    \mathfrak s_L\left(W^{a(r)}_{b,n}\right)=\sum_{i=0}^r f_{r,n,i}(L) \cdot W^{b(r-i)}_{a,-n}, 
\end{align*}
where we set $W^{b(0)}_{a,m}:=\delta^b_a\delta_{m,0}$. Note that $f_{r,n,0}(L)=1$.

We define $\mathfrak{s}_{\infty}:\mathfrak{U}({\cal W}^{(K)}_\infty)\cong \mathfrak{U}({\cal W}^{(K)}_\infty)$ to be the anti-involution uniquely determined by
\begin{align*}
    \mathfrak s_\infty\left(W^{a(r)}_{b,n}\right)=W^{b(r)}_{a,-n}+\sum_{i=1}^r f_{r,n,i}(\lambda) \cdot W^{b(r-i)}_{a,-n},\quad f_{r,n,i}(\lambda) \in \mathbb C[\alpha,\lambda].
\end{align*}
Then obviously we have
\begin{align}\label{eqn: compatible anti-involution}
    \pi_L\circ s_\infty=s_L\circ\pi_L.
\end{align}
\begin{lemma}\label{lem: s_inf}
$\mathfrak{s}_{\infty}$ induces anti-involution on $\mathfrak{U}(\mathsf W^{(K)}_\infty)$, i.e. there exist $\mathsf f_{r,n,i}\in \mathbb C[\epsilon_1,\epsilon_2,\mathsf c]$ such that
\begin{align*}
    \mathfrak s_\infty\left(\mathds W^{a(r)}_{b,n}\right)=\mathds W^{b(r)}_{a,-n}+\sum_{i=1}^r \mathsf f_{r,n,i} \cdot \mathds W^{b(r-i)}_{a,-n}.
\end{align*}
\end{lemma}

\begin{proof}
Let $\tilde W^{a(r)}_{b,n}:=\epsilon_1^rW^{a(r)}_{b,n}$, then $\mathfrak s_\infty\left(\tilde W^{a(r)}_{b,n}\right)$ is determined by
\begin{multline*}
    z^{\epsilon_1\mathsf c+1} \left((\epsilon_3\partial)^{\epsilon_1\mathsf c}+\sum_{r=1}^\infty (-1)^r(\epsilon_3\partial)^{\epsilon_1\mathsf c-r}\tilde W^{b(r)}_{a}(z)\right)z^{\epsilon_1\mathsf c-1}=\\
    (\epsilon_3 z^2\partial)^{\epsilon_1\mathsf c}+\sum_{r=1}^\infty (-1)^r\sum_{n\in \mathbb Z}\mathfrak s_\infty\left(\tilde W^{a(r)}_{b,n}\right)z^{n+r} (\epsilon_3 z^2\partial)^{\epsilon_1\mathsf c-r},
\end{multline*}
then there exists $\tilde f_{r,n,i}\in \mathbb C[\epsilon,\epsilon_2,\mathsf c]$ such that
\begin{align*}
    \mathfrak s_\infty\left(\tilde W^{a(r)}_{b,n}\right)=\tilde W^{b(r)}_{a,-n}+\sum_{i=1}^r \tilde f_{r,n,i} \cdot \tilde W^{b(r-i)}_{a,-n}.
\end{align*}
Moreover, setting $\epsilon_1=0$ implies that 
\begin{align*}
    \sum_{r=1}^\infty (-1)^r z(\epsilon_3\partial)^{-r}z^{-1}\tilde W^{b(r)}_{a}(z)=
    \sum_{r=1}^\infty (-1)^r\sum_{n\in \mathbb Z}\mathfrak s_\infty\left(\tilde W^{a(r)}_{b,n}\right)z^{n+r} (\epsilon_3 z^2\partial)^{-r},
\end{align*}
and we deduce from the above equation that $\tilde f_{r,n,r}\equiv 0\pmod{\epsilon_1}$, in other words $\tilde f_{r,n,r}$ is divisible by $\epsilon_1$ in $\mathbb C[\epsilon,\epsilon_2,\mathsf c]$. Setting
\begin{align*}
    \mathsf f_{r,n,i}=\begin{cases}
        \tilde f_{r,n,i}, & 0< i<r,\\
        \frac{1}{\epsilon_1} \tilde f_{r,n,r}, & i=r,
    \end{cases}
\end{align*}
and we are done.
\end{proof}

\begin{definition}\label{def: Psi_infinity minus}
We define the algebra homomorphism $\Psi^-_{\infty}:\mathsf A^{(K)}\to \mathfrak U(\mathsf W^{(K)}_{\infty})[\epsilon_2^{-1}]$ to be the composition
\begin{align}
    \Psi^-_\infty:=\mathfrak{s}_\infty\circ \Psi_\infty\circ \mathfrak{s}_{\mathsf A},
\end{align}
where $\mathfrak{s}_{\mathsf A}$ is the anti-involution \eqref{eqn: anti-involution s_A} on $\mathsf A^{(K)}$.
\end{definition}

It follows from the definition of $\Psi^-_L$ and \eqref{eqn: compatible anti-involution} that 
\begin{align*}
    \pi_L\circ \Psi^-_\infty=\Psi^-_L.
\end{align*}
Using \eqref{eqn: Psi^- obvious equations} and \eqref{eqn: Psi^-(t[2,2])}, we see that $\Psi^-_\infty$ is uniquely determined by the image of following generators
\begin{equation}\label{eqn: Psi_infinity minus on generators}
\begin{split}
\Psi^-_{\infty}(\mathsf T_{0,n}(E^a_b))=\mathds W^{a(1)}_{b,-n},&\quad \Psi^-_{\infty}(\mathsf t_{0,n})=\frac{1}{\epsilon_2}\mathds W^{a(1)}_{a,-n},\\
\Psi^-_{\infty}(\mathsf t_{1,0})=-\Psi_{\infty}(\mathsf t_{1,2})+\frac{\epsilon_1\epsilon_3\mathsf c}{\epsilon_2}\mathds W^{a(1)}_{a,1},& \quad \Psi^-_{\infty}(\mathsf t_{1,2})=-\Psi_{\infty}(\mathsf t_{1,0})-\frac{\epsilon_1\epsilon_3\mathsf c}{\epsilon_2}\mathds W^{a(1)}_{a,-1},\\
\Psi^-_{\infty}(\mathsf T_{1,1}(E^a_b))=-\Psi_{\infty}(\mathsf T_{1,1}(E^a_b)), &\quad \Psi^-_{\infty}(\mathsf t_{2,2})= \Psi_{\infty}(\mathsf t_{2,2}),
\end{split}
\end{equation}

\section{Meromorphic Coproduct of \texorpdfstring{$\mathsf A^{(K)}$}{Ak}} 
Consider the rational map $\mathbb C^{N_1}_{\mathrm{disj}}\times \mathbb C^{N_2}_{\mathrm{disj}}\dashrightarrow \mathbb C^{N_1+N_2}_{\mathrm{disj}}$ which acts on coordinates by $$(x^{(1)}_1,\cdots,x^{(1)}_{N_1})\times (x^{(2)}_1,\cdots,x^{(2)}_{N_2})\mapsto (x^{(1)}_1,\cdots,x^{(1)}_{N_1},x^{(2)}_1,\cdots,x^{(2)}_{N_2}).$$ This is not a globally-defined map since $x^{(1)}_i$ might collide with $x^{(2)}_j$. Alternatively, one can consider the parametrized version of the above rational map $m:\mathbb C^{N_1}_{\mathrm{disj}}\times \mathbb C^{N_2}_{\mathrm{disj}}\times \mathbb P^1\dashrightarrow \mathbb C^{N_1+N_2}_{\mathrm{disj}}$ sending $(x^{(1)}_1,\cdots,x^{(1)}_{N_1})\times (x^{(2)}_1,\cdots,x^{(2)}_{N_2})\times (w)$ to $(x^{(1)}_1,\cdots,x^{(1)}_{N_1},x^{(2)}_1+w,\cdots,x^{(2)}_{N_2}+w)$, where $w$ is the coordinate on $\mathbb P^1$. Then the non-defined loci for $m$ on $\mathbb C^{N_1}_{\mathrm{disj}}\times \mathbb C^{N_2}_{\mathrm{disj}}\times \mathbb P^1$ is union of hyperplanes $x^{(1)}_i=x^{(2)}_j+w$ and the infinity divisor $w=\infty$. Since the hyperplanes do not intersect with the infinity divisor, we can take the formal neighborhood of $w=\infty$ and localize to get a genuine map
\begin{align*}
    m:\mathbb C^{N_1}_{\mathrm{disj}}\times \mathbb C^{N_2}_{\mathrm{disj}}\times \Spec \mathbb C(\!(w^{-1})\!)\to \mathbb C^{N_1+N_2}_{\mathrm{disj}}.
\end{align*}
It maps the function ring $\mathbb C[x^{(1)}_i,x^{(2)}_j,(x^{(1)}_{i_1}-x^{(1)}_{i_2})^{-1},(x^{(2)}_{j_1}-x^{(2)}_{j_2})^{-1},(x^{(1)}_i-x^{(2)}_j)^{-1}]$ to $ \mathbb C[x^{(1)}_i,x^{(2)}_j,(x^{(1)}_{i_1}-x^{(1)}_{i_2})^{-1},(x^{(2)}_{j_1}-x^{(2)}_{j_2})^{-1}](\!(w^{-1})\!)$ by
\begin{equation}
    \begin{split}
        x^{(1)}_i\mapsto x^{(1)}_i,&\quad x^{(2)}_j\mapsto x^{(2)}_j+w,\\ \frac{1}{x^{(1)}_{i_1}-x^{(1)}_{i_2}}\mapsto \frac{1}{x^{(1)}_{i_1}-x^{(1)}_{i_2}},&\quad
    \frac{1}{x^{(2)}_{j_1}-x^{(2)}_{j_2}}\mapsto \frac{1}{x^{(2)}_{j_1}-x^{(2)}_{j_2}}\\
    \frac{1}{x^{(1)}_i-x^{(2)}_j}\mapsto &-\sum_{n=0}^{\infty} w^{-n-1}(x^{(1)}_i-x^{(2)}_j)^n.
    \end{split}
\end{equation}
We call such map a \textit{meromorphic coproduct}, denoted by $\Delta(w)_{N_1,N_2}$. It is coassociative in the obvious sense, in fact it satisfies a more basic property:
\begin{lemma}\label{lemma: Locality}
Meromorphic coproducts are local in the sense that, if we decompose $N=N_1+N_2+N_3$ into three clusters, then for any $f\in \mathscr O(\mathbb C^N_{\mathrm{disj}})$, two elements
\begin{align*}
    (\Delta(w)_{N_1,N_2}\otimes \mathrm{id})\circ\Delta(z)_{N_1+N_2,N_3}f,\quad (\mathrm{id}\otimes P)\circ(\Delta(z)_{N_1,N_3}\otimes \mathrm{id})\circ\Delta(w)_{N_1+N_3,N_2}f,
\end{align*}
are expansions of the same element in $\mathscr O(\mathbb C^{N_1}_{\mathrm{disj}}\times \mathbb C^{N_2}_{\mathrm{disj}}\times \mathbb C^{N_3}_{\mathrm{disj}})[\![z^{-1},w^{-1},(z-w)^{-1}]\!][z,w]$, where $P:\mathscr O(\mathbb C^{N_3}_{\mathrm{disj}})\otimes \mathscr O(\mathbb C^{N_2}_{\mathrm{disj}})\to \mathscr O(\mathbb C^{N_2}_{\mathrm{disj}})\otimes \mathscr O(\mathbb C^{N_3}_{\mathrm{disj}})$ is the permutation operator.
\end{lemma}

\begin{proof}
After taking two-step meromorphic coproduct, $x^{(1)}_i\mapsto x^{(1)}_i, x^{(2)}_j\mapsto x^{(2)}_j+w,x^{(3)}_k\mapsto x^{(3)}_k+z$, and those $(x_i-x_j)^{-1}$ are mapped accordingly and then expanded in power series. Thus we immediately see that $ \Delta(w)_{N_1,N_2}\circ\Delta(z)_{N_1+N_2,N_3}f$ and $ \Delta(z)_{N_1,N_3}\circ\Delta(w)_{N_1+N_3,N_2}f$ are expansions of the same rational function.
\end{proof}

The meromorphic coproduct can be defined for differential operators as well, i.e. there exists 
\begin{align*}
    \Delta(w)_{N_1,N_2}:D(\mathbb C^{N_1+N_2}_{\mathrm{disj}})\otimes \mathfrak{gl}_K^{\otimes N_1+N_2}\to D(\mathbb C^{N_1}_{\mathrm{disj}})\otimes \mathfrak{gl}_K^{\otimes N_1}\otimes D(\mathbb C^{N_2}_{\mathrm{disj}})\otimes \mathfrak{gl}_K^{\otimes N_2}(\!(w^{-1})\!),
\end{align*}
which also satisfies the locality in the Lemma \ref{lemma: Locality}. Restricted to the image of spherical Cherednik algebras via Dunkl embeddings, we get an algebra homomorphism:
\begin{align}
    \Delta(w)_{N_1,N_2}: \mathrm{S}\mathcal H^{(K)}_{N_1+N_2}\to  \mathrm{S}\mathcal H^{(K)}_{N_1}\otimes  \mathrm{S}\mathcal H^{(K)}_{N_2}(\!(w^{-1})\!),
\end{align}
and the formula for $\Delta(w)_{N_1,N_2}$ on the generators of $\mathrm{S}\mathcal H^{(K)}_{N_1+N_2}$ reads
\begin{equation}\label{eqn: AA coproduct_finite N}
\begin{split}
&\Delta(w)(\rho_{N_1+N_2}(\mathsf T_{0,n}(E^a_b)))=\rho_{N_1}(\mathsf T_{0,n}(E^a_b))\otimes 1+\sum_{m=0}^n \binom{n}{m}w^{n-m} 1\otimes\rho_{N_2}(\mathsf T_{0,n}(E^a_b)), \\
&\Delta(w)(\rho_{N_1+N_2}(\mathsf T_{1,0}(E^a_b))=\rho_{N_1}(\mathsf T_{1,0}(E^a_b)\otimes 1+1\otimes \rho_{N_2}(\mathsf T_{1,0}(E^a_b)\\
&~+\epsilon_1\sum_{m,n\ge 0}\frac{(-1)^m}{w^{n+m+1}}\binom{m+n}{n} (\rho_{N_1}(\mathsf T_{0,n}(E^c_b))\otimes \rho_{N_2}(\mathsf T_{0,m}(E^a_c))-\rho_{N_1}(\mathsf T_{0,n}(E^a_c))\otimes \rho_{N_2}(\mathsf T_{0,m}(E^c_b)))\\
&\Delta(w)(\rho_{N_1+N_2}(\mathsf t_{2,0}))=\rho_{N_1}(\mathsf t_{2,0})\otimes 1+1\otimes \rho_{N_2}(\mathsf t_{2,0}),\\
&~ -2\epsilon_1\sum_{m,n\ge 0}\frac{(m+n+1)!}{m!n!w^{n+m+2}}(-1)^m (\rho_{N_1}(\mathsf T_{0,n}(E^a_b))\otimes \rho_{N_2}(\mathsf T_{0,m}(E^b_a))+\epsilon_1\epsilon_2 \rho_{N_1}(\mathsf t_{0,n})\otimes \rho_{N_2}(\mathsf t_{0,m})),
\end{split}
\end{equation}
Since the formula \eqref{eqn: AA coproduct_finite N} are uniform in $N_1$ and $N_2$, the Corollary \ref{cor: truncation of D} implies that the uniform-in-$N_1,N_2$ formula produces an algebra homomorphism.
\begin{proposition}\label{prop: AA coproduct}
There is an algebra homomorphism $\Delta_{\mathsf A}(w): \mathsf A^{(K)}\to \mathsf A^{(K)}\otimes\mathsf A^{(K)}(\!(w^{-1})\!)$ which map the generators as
\begin{equation}\label{eqn: AA coproduct}
\boxed{
\begin{aligned}
&\Delta_{\mathsf A}(w)(\mathsf T_{0,n}(E^a_b))=\mathsf T_{0,n}(E^a_b)\otimes 1+\sum_{m=0}^n \binom{n}{m}w^{n-m} 1\otimes\mathsf T_{0,m}(E^a_b), \\
&\Delta_{\mathsf A}(w)(\mathsf T_{1,0}(E^a_b))=\square(\mathsf T_{1,0}(E^a_b))\\
&~+\epsilon_1\sum_{m,n\ge 0}\frac{(-1)^m}{w^{n+m+1}}\binom{m+n}{n}(\mathsf T_{0,n}(E^c_b)\otimes \mathsf T_{0,m}(E^a_c)-\mathsf T_{0,n}(E^a_c)\otimes \mathsf T_{0,m}(E^c_b)),\\
&\Delta_{\mathsf A}(w)(\mathsf t_{2,0})=\square(\mathsf t_{2,0})-2\epsilon_1\sum_{m,n\ge 0}\frac{(m+n+1)!}{m!n!w^{n+m+2}}(-1)^m (\mathsf T_{0,n}(E^a_b)\otimes \mathsf T_{0,m}(E^b_a)+\epsilon_1\epsilon_2 \mathsf t_{0,n}\otimes \mathsf t_{0,m}),
    \end{aligned}}
\end{equation}
where $\square(X)=X\otimes 1+1\otimes X$. 
\end{proposition}
We shall call $\Delta_{\mathsf A}(w)$ the meromorphic coproduct on $\mathsf A^{(K)}$.

\subsection{Vertex coalgebra structure on \texorpdfstring{$\mathsf A^{(K)}$}{Ak}}

The locality for the meromorphic coproduct can be put into more general framework called the \textit{vertex coalgebra}. We recall its definition in the Appendix \ref{sec: vertex coalgebras and vertex comodules}. The following theorem generalizes the $K=1$ case in \cite{oh2021twisted}.

\begin{theorem}\label{thm: A^K is vertex coalgebra}
The meromorphic coproduct induces a vertex coalgebra structures on $\mathsf A^{(K)}$ over the base ring $\mathbb C[\epsilon_1,\epsilon_2]$. Moreover $\Psi_{\infty}:\mathsf A^{(K)}\to U_+(\mathsf W^{(K)}_{\infty})[\epsilon_2^{-1}]$ is a vertex coalgebra map.
\end{theorem}

\begin{proof}
Let us define the covacuum $\mathfrak{C}_{\mathsf A}:\mathsf A^{(K)}\to \mathbb C[\epsilon_1,\epsilon_2]$ by mapping on generators $\mathfrak{C}_{\mathsf A}(\mathsf t_{n,m})=\mathfrak{C}_{\mathsf A}(\mathsf T_{n,m}(E^a_b))=0$ and extending it to an algebra map. Then the counit and cocreation axioms are easily checked for \eqref{eqn: AA coproduct}, thus these two axioms are satisfied for all elements in $\mathsf A^{(K)}$ since $\mathfrak{C}_{\mathsf A}\otimes \mathrm{id}$ and $\mathrm{id}\otimes \mathfrak{C}_{\mathsf A}$ are algebra homomorphisms. The locality axiom is a consequence of Lemma \ref{lemma: Locality} treated as uniform-in $N$ equations. It remains to check the translation axiom. 

Note that the operator $T=( \mathfrak{C}_{\mathsf A}\otimes \mathrm{id})\circ\Delta_{-2}:\mathsf A^{(K)}\to \mathsf A^{(K)}$ is a derivation, since we can write $$T=\lim_{w\to 0}\frac{\mathrm{d}}{\mathrm{d}w}(\mathfrak{C}_{\mathsf A}\otimes \mathrm{id})\circ\Delta_{\mathsf A}(w),$$
and the operator $\frac{\mathrm{d}}{\mathrm{d}w}$ is a derivation and $(\mathfrak{C}_{\mathsf A}\otimes \mathrm{id})\Delta(_{\mathsf A}w)$ is algebra homomorphism. Since $T(\mathsf t_{2,0})=0$ and $T(\mathsf T_{0,n}(E^a_b))=n\mathsf T_{0,n-1}(E^a_b)$, we conclude that $T$ is the same as the adjoint action of $\mathsf t_{1,0}$. Using the representation $\rho_N$, the operator $T$ can be written explicitly, in fact $\mathsf t_{1,0}$ is mapped to $\sum_{i=1}^N\partial_{x_i}$ in $(D(\mathbb C^{N}_{\mathrm{disj}})\otimes \gl_K^{\otimes N})^{\mathfrak S_N}$, thus 
\begin{align*}
    &\Delta(w)_{N_1,N_2}\circ T(f(x^{(1)}_i,x^{(2)}_j))-( T\otimes \mathrm{id})\circ\Delta(w)_{N_1,N_2}f(x^{(1)}_i,x^{(2)}_j)\\
    &=\left(\sum_{k=1}^{N_1}\frac{\partial f}{\partial {x^{(1)}_k}}\right)(x^{(1)}_i,x^{(2)}_j+w)+\left(\sum_{k=1}^{N_2}\frac{\partial f}{\partial {x^{(2)}_k}}\right)(x^{(1)}_i,x^{(2)}_j+w)-\sum_{k=1}^{N_1}\partial_{x^{(1)}_k}\left(f(x^{(1)}_i,x^{(2)}_j+w)\right)\\
    &=\frac{\mathrm{d}}{\mathrm{d}w}f(x^{(1)}_i,x^{(2)}_j+w),
\end{align*}
for all functions $f$ on $\mathbb C^{N_1+N_2}_{\mathrm{disj}}$, and this equation extends to hold for differential operators in $(D(\mathbb C^{N}_{\mathrm{disj}})\otimes \gl_K^{\otimes N})^{\mathfrak S_N}$ by linearity. In particular, the translation axiom is satisfied for all $\Delta(w)_{N_1,N_2}$, and it is therefore satisfied for the uniform-in-$N$ coproduct $\Delta_{\mathsf A}(w)$.

Finally, the statement that $\Psi_{\infty}$ is a vertex coalgebra map is checked by direct computation using the formula \eqref{eqn: AA coproduct}.
\end{proof}

\begin{remark}
Theorem \ref{thm: A^K is vertex coalgebra} implies that $\mathsf A^{(K)}$ is a vertex coalgebra object in the category of $\mathbb C[\epsilon_1,\epsilon_2]$-algebras, this is analog to the notion of a bialgebra which is a coalgebra object in the category of algebras.
\end{remark}

\section{The Algebra \texorpdfstring{$\mathsf L^{(K)}$}{Lk} and Coproducts}\label{sec: L^K and coproducts}
Let $\mathcal H^{(K)}_N$ be the extended rational Cherednik algebra defined in subsection \ref{subsec: PBW}, then define the extended trigonometric Cherednik algebra \cite{guay2005cherednik,etingof2010lecture}
\begin{align}
    \mathbb H^{(K)}_N=\mathcal H^{(K)}_N\underset{\mathbb C[x_1,\cdots,x_N]}{\otimes} \mathbb C[x_1^{\pm},\cdots,x_N^{\pm}],
\end{align}
and similarly its spherical subalgebra $\mathrm{S}\mathbb{H}^{(K)}_N=\mathbf{e}\mathbb H^{(K)}_N\mathbf e$, where $\mathbf e=\frac{1}{N!}\sum_{g\in \mathfrak{S}_N}g$ is an idempotent element of group algebra $\mathbb C[\mathfrak{S}_N]$. $\mathbb H^{(K)}_N$ and its spherical subalgebra have Dunkl embeddings
\begin{equation}
    \begin{split}
        \mathbb H^{(K)}_N\hookrightarrow &\mathbb C[\mathfrak{S}_N]\ltimes \left(D(\mathbb C^{\times N}_{\mathrm{disj}})\otimes \gl_K^{\otimes N}\right)[\epsilon_1,\epsilon_2],\\
        \mathrm{S}\mathbb{H}^{(K)}_N&\hookrightarrow \left(D(\mathbb C^{\times N}_{\mathrm{disj}})\otimes \gl_K^{\otimes N}\right)^{\mathfrak{S}_N}[\epsilon_1,\epsilon_2].
    \end{split}
\end{equation}
Consider the algebra embedding $D(\mathbb C^{\times N}_{\mathrm{disj}})\hookrightarrow D(\mathbb C^{ N}_{\mathrm{disj}})(\!(w^{-1})\!)$ given by 
\begin{align}
    x_i\mapsto x_i+w,\quad x_i^{-1}\mapsto \sum_{n=0}^{\infty} w^{-n-1} (-x_i)^n,
\end{align}
this map induces algebra embeddings: $\mathbb H^{(K)}_N\hookrightarrow \mathcal H^{(K)}_N(\!(w^{-1})\!)$ and
\begin{align}\label{eqn: meromorphic embedding_finite N}
    S_N(w):\mathrm{S}\mathbb{H}^{(K)}_N\hookrightarrow \mathrm{S}\mathcal{H}^{(K)}_N (\!(w^{-1})\!).
\end{align}
The restriction of $S_N(w)$ on $\mathrm{S}\mathcal{H}^{(K)}_N$ gives a map $\mathrm{S}\mathcal{H}^{(K)}_N\hookrightarrow \mathrm{S}\mathcal{H}^{(K)}_N [w]$ and its formation is independent of $N$, so we can take its uniform-in-$N$ limit and obtain a $\mathbb C[\epsilon_1,\epsilon_2]$-algebra embedding $S(w):\mathsf A^{(K)}\hookrightarrow\mathsf A^{(K)}[w]$ such that
\begin{equation}
\begin{split}
S(w)(\mathsf T_{0,n}(E^a_b))&=\sum_{m=0}^n \binom{n}{m}w^{n-m}\mathsf T_{0,m}(E^a_b),\\
S(w)(\mathsf t_{2,0})&=\mathsf t_{2,0}.
\end{split}
\end{equation}
In fact one can derive that
\begin{align}
    S(w)=(\mathfrak C_{\mathsf A}\otimes 1)\circ \Delta_{\mathsf A}(w).
\end{align}
The image of $S(w)$ is characterized by the annihilator of the derivation operator $\mathrm{ad}_{\mathsf t_{1,0}}-\partial_w$ acting on $\mathsf A^{(K)}[w]$. In fact, $S(-w)$ transports the zero set of $\mathrm{ad}_{\mathsf t_{1,0}}-\partial_w$ to the zero set of $\partial_w$ acting on $\mathsf A^{(K)}[w]$, which is exactly $\mathsf A^{(K)}$, thus $S(w)(\mathsf A^{(K)})$ equals to the zero set of $\mathrm{ad}_{\mathsf t_{1,0}}-\partial_w$.

\begin{remark}
Let $V$ be a vertex coalgbera with translation operator $T$ (see Appendix \ref{sec: vertex coalgebras and vertex comodules}), then we define $\mathcal L(V)$ to be the set of $\mathbb C^{\times }$-finite annihilators of $T-\partial_z$ in $V(\!(z^{-1})\!)$. We expect that there exists a Lie coalgebra structure on $\mathcal L(V)$, which is the dual notion of mode Lie algebra of a vertex algebra. We intend to call $\mathcal L(V)$ the mode Lie coalgebra of $V$. We also expect that there should be a notion of mode coalgebra $\mathfrak U(V)$ which is built from the universal enveloping coalgebra of $\mathcal L(V)$. If $V$ is a vertex coalgebra object in the category of algebras, then we expect that $\mathfrak U(V)$ is a bialgebra.

For the vertex coalgebra $\mathsf A^{(K)}$, the translation operator $T=\mathrm{ad}_{\mathsf t_{1,0}}$. The previous discussions imply that $S(w)(\mathsf A^{(K)})\subset \mathcal L(\mathsf A^{(K)})$. In the following definition, we introduce a certain subset $\mathsf L^{(K)}$ of $\mathcal L(\mathsf A^{(K)})$ which contains $S(w)(\mathsf A^{(K)})$ and will be shown to have a natural bialgebra structure. We expect that $\mathsf L^{(K)}$ can be naturally identified as a sub-bialgebra of the mode coalgebra (bialgebra) $\mathfrak U(\mathsf A^{(K)})$, if the latter is appropriately defined.
\end{remark}

\begin{definition}\label{def: the algebra L^K}
The algebra $\mathsf L^{(K)}$ is the $\mathbb C[\epsilon_1,\epsilon_2]$-subalgebra of $\mathsf A^{(K)}(\!(w^{-1})\!)$ generated by $S(w)(\mathsf A^{(K)})$ and $\mathsf t_{0,-1}$ defined by
\begin{equation}
\begin{split}
\mathsf t_{0,-1}=\sum_{n=0}^{\infty} (-1)^n w^{-n-1} \mathsf t_{0,n}.
\end{split}
\end{equation}
We introduce the notation $\mathsf T_{0,-1}(X):=[\mathsf t_{0,-1},S(w)(\mathsf T_{1,1}(X))]$ for $X\in\gl_K$, and recursively define 
\begin{align}
\mathsf T_{0,-n}(X)=\frac{1}{1-n}[\mathsf t_{1,0},\mathsf T_{0,-n+1}(X)],\quad \mathsf t_{0,-n}=\frac{1}{1-n}[\mathsf t_{1,0},\mathsf t_{0,-n+1}],
\end{align}
for all $n\in \mathbb N_{>1}$. We still use $S(w)$ to denote the canonical embedding $\mathsf L^{(K)}\hookrightarrow \mathsf A^{(K)}(\!(w^{-1})\!)$.
\end{definition}
Note that $\mathsf L^{(K)}$ is a $\mathbb Z$-graded subalgebra of $\mathsf A^{(K)}(\!(w^{-1})\!)$ with $\deg \mathsf T_{0,-n}(X)=\deg \mathsf t_{0,-n}=-n$. Since $\mathsf t_{0,-1}$ is also annihilated by $\mathrm{ad}_{\mathsf t_{1,0}}-\partial_w$, we see that the whole $S(w)(\mathsf L^{(K)})$ is annihilated by $\mathrm{ad}_{\mathsf t_{1,0}}-\partial_w$. Therefore 
\begin{align}\label{eqn: positive part of L is A}
    S(w)(\mathsf L^{(K)})\cap \mathsf A^{(K)}[w]=S(w)(\mathsf A^{(K)}).
\end{align}

\begin{remark}
Our construction of $S(w):\mathsf L^{(K)}\hookrightarrow \mathsf A^{(K)}(\!(w^{-1})\!)$ resembles similarity to that of the formal shift map from the Yangian double $\mathrm{DY}_{\hbar}(\mathfrak g)$ to the formal power series ring $Y_{\hbar}(\mathfrak g)(\!(z^{-1})\!)$ where $Y_{\hbar}(\mathfrak g)$ is the Yangian of a Kac-Moody Lie algebra $\mathfrak g$, see \cite{wendlandt2022formal,wendlandt2022restricted} for relevant discussions.
\end{remark}

\begin{lemma}\label{lem: rho_N extends to loop Yangian}
The map $\rho_N: \mathsf A^{(K)}\to \mathrm{S}\mathcal{H}^{(K)}_N[\epsilon_2^{-1}]$ extends to a map $\rho_N: \mathsf L^{(K)}\to \mathrm{S}\mathbb{H}^{(K)}_N[\epsilon_2^{-1}]$ such that 
\begin{align}
    \rho_N(\mathsf T_{0,-n}(E^a_b))=\sum_{i=1}^N E^a_{b,i} x_i^{-n},\quad \rho_N(\mathsf t_{0,-n})=\frac{1}{\epsilon_2}\sum_{i=1}^N x_i^{-n}.
\end{align}
Moreover $\ker(\prod_N\rho_N)=0$.
\end{lemma}

\begin{proof}
Extending $\rho_N$ by formal power series in $w^{-1}$, we see that the image of $\mathsf T_{0,-1}(E^a_b)$ is exactly $S_N(w)(\sum_{i=1}^N E^a_{b,i} x_i^{-1})$, thus $\rho_N(\mathsf L^{(K)})\subset S_N(w)(\mathrm{S}\mathbb{H}^{(K)}_N[\epsilon_2^{-1}])$. By Corollary \ref{cor: truncation of D}, the intersection of kernels $\ker(\rho_N)\subset \mathsf L^{(K)}$ is trivial.
\end{proof}

Lemma \ref{lem: rho_N extends to loop Yangian} implies that $\mathsf L^{(K)}$ can be equivalently defined as the $\mathbb C[\epsilon_1,\epsilon_2]$-subalgebra of $\prod_N \mathrm{S}\mathbb{H}^{(K)}_N[\epsilon_2^{-1}]$ generated by $\left(\rho_N(\mathsf A^{(K)})\right)_{N=1}^{\infty}$ and $\frac{1}{\epsilon_2}\left(\sum_{i=1}^N x_i^{-1}\right)_{N=1}^{\infty}$.

\begin{lemma}\label{lem: rho_N for loop Yangian is surjective}
After inverting $\epsilon_2$, the map $\rho_N:\mathsf L^{(K)}[\epsilon_2^{-1}]\to \mathrm{S}\mathbb{H}^{(K)}_N[\epsilon_2^{-1}]$ is surjective.
\end{lemma}

\begin{proof}
Consider the algebra involution $\iota: \mathrm{S}\mathbb{H}^{(K)}_N\cong \mathrm{S}\mathbb{H}^{(K)}_N$ induced by
\begin{align}
    x_i\mapsto x_i^{-1},\;y_i\mapsto -x_iy_ix_i.
\end{align}
We claim that the image of $\rho_N:\mathsf L^{(K)}[\epsilon_2^{-1}]\to \mathrm{S}\mathbb{H}^{(K)}_N[\epsilon_2^{-1}]$ is invariant under the involution $\iota$. In fact, it follows from definition that $\iota(\rho_N(\mathsf T_{0,\pm n}(E^a_b)))=\rho_N(\mathsf T_{0,\mp n}(E^a_b))$, and moreover we have
\begin{align}
    \iota(\rho_N(\mathsf t_{2,0}))=\frac{1}{\epsilon_2}\sum_{i=1}^Ny_i^2=\frac{1}{\epsilon_2}\sum_{i=1}^Nx_iy_ix_i^2y_ix_i=\frac{1}{\epsilon_2}\rho_N(\mathsf T_{\mathbf r}(1)),
\end{align}
where $\mathbf r$ is the word of two letters: $\mathbf r(X,Y)=YXYYXY$, see Proposition \ref{prop: other generators}. It is elementary to see that $\mathrm{S}\mathbb{H}^{(K)}_N$ is generated by elements of the form $\sum_{i=1}^N x_i^m y_i^n E^a_{b,i}$, where $m\in \mathbb Z,n\in \mathbb Z_{\ge 0}$, and those which have $m\ge 0$ are inside the subalgebra $\mathrm{S}\mathcal{H}^{(K)}_N$. Corollary \ref{cor: truncation of D} implies that $\mathrm{S}\mathcal{H}^{(K)}_N$ is contained in the image of $\rho_N$, thus it remains to show that $\sum_{i=1}^N x_i^m y_i^n E^a_{b,i}$ for $m<0$ are in the image of $\rho_N$. Notice that
\begin{align*}
    \iota(\sum_{i=1}^N x_i^m y_i^n E^a_{b,i})=\sum_{i=1}^N x_i^{-m+1} y_ix_i^2y_i\cdots x_i^2 y_ix_i E^a_{b,i}=\rho_N(\mathsf T_{\mathbf r'}(E^a_b)),
\end{align*}
where $\mathbf r'$ is the word of two letters: $\mathbf r'(X,Y)=Y^{-m+1}XY^2X\cdots XY^2XY$, see Proposition \ref{prop: other generators}, thus $\iota(\sum_{i=1}^N x_i^m y_i^n E^a_{b,i})\in \mathrm{im}(\rho_N)$. Since $\mathrm{im}(\rho_N)$ is invariant under the involution $\iota$, we have $\sum_{i=1}^N x_i^m y_i^n E^a_{b,i}\in \mathrm{im}(\rho_N)$, this concludes the proof.
\end{proof}

Combining Lemma \ref{lem: rho_N extends to loop Yangian} and \ref{lem: rho_N for loop Yangian is surjective}, we see that $\mathsf L^{(K)}$ can be regarded as a uniform-in-$N$ algebra of the matrix extended spherical trigonometric Cherednik algebra $\mathrm{S}\mathbb{H}^{(K)}_N$.

A byproduct of the proof of the Lemma \ref{lem: rho_N for loop Yangian is surjective} is that the involution $\iota:x_i\mapsto x_i^{-1},\;y_i\mapsto -x_iy_ix_i$ generalizes to uniform-in-$N$ limit:
\begin{proposition}\label{prop: involution on L}
Let $\mathbf r(x,y)=yxyyxy$, then the element $\mathsf T_{\mathbf r}(1)\in \epsilon_2\cdot\mathsf A^{(K)}$, therefore there exists a unique $\mathbb C[\epsilon_1,\epsilon_2]$-algebra involution $\iota: \mathsf L^{(K)}\cong \mathsf L^{(K)}$ such that $\iota(\mathsf T_{0,n}(X))=\mathsf T_{0,-n}(X)$ and $\iota(\mathsf t_{2,0})=\frac{1}{\epsilon_2}\mathsf T_{\mathbf r}(1)$.
\end{proposition}

\begin{proof}
According to the proof of Lemma \ref{lem: rho_N for loop Yangian is surjective}, we already know that there exists a unique $\mathbb C[\epsilon_1,\epsilon_2^{\pm}]$-algebra involution $\iota: \mathsf L^{(K)}[\epsilon_2^{-1}]\cong \mathsf L^{(K)}[\epsilon_2^{-1}]$ such that $\iota(\mathsf T_{0,n}(X))=\mathsf T_{0,-n}(X)$ and $\iota(\mathsf t_{2,0})=\frac{1}{\epsilon_2}\mathsf T_{\mathbf r}(1)$. It remains to show that $\iota$ preserves the subalgebra $\mathsf L^{(K)}$, which will follow from that $\mathsf T_{\mathbf r}(1)\in \epsilon_2\cdot\mathsf A^{(K)}$. We compute directly that $\iota(\mathsf t_{1,0})=-\mathsf t_{1,2}$ and $\iota(\mathsf t_{2,2})=\mathsf t_{2,2}$, thus $$\iota(\mathsf t_{2,0})=\frac{1}{2}\iota([\mathsf t_{1,0},[\mathsf t_{1,0},\mathsf t_{2,2}]])=\frac{1}{2}[\mathsf t_{1,2},[\mathsf t_{1,2},\mathsf t_{2,2}]]\in \mathsf A^{(K)}.$$ This finishes the proof.
\end{proof}

As a corollary, we see that $\iota(\mathsf t_{n,0})\in \mathsf L^{(K)}$. Since $\epsilon_2\iota(\mathsf t_{n,0})\in \mathsf A^{(K)}$, we conclude that $S(w)(\iota(\mathsf t_{n,0}))\in S(w)(\mathsf L^{(K)})\cap \mathsf A^{(K)}[w]=S(w)(\mathsf A^{(K)})$ by \eqref{eqn: positive part of L is A}, whence $\iota(\mathsf t_{n,0})\in \mathsf A^{(K)}$.\\

Composing the canonical embedding $S(w): \mathsf L^{(K)}\hookrightarrow \mathsf A^{(K)}(\!(w^{-1})\!)$ with the augmentation $\mathfrak C_{\mathsf A}:\mathsf A^{(K)}(\!(w^{-1})\!)\to \mathbb C[\epsilon_1,\epsilon_2](\!(w^{-1})\!)$, we obtain a homomorphism $\mathfrak C_{\mathsf L}:\mathsf L^{(K)}\to \mathbb C[\epsilon_1,\epsilon_2](\!(w^{-1})\!)$, which maps all generators $\mathsf T_{n,m}(E^a_b)$, $\mathsf t_{n,m}$, $\mathsf T_{0,-1}(E^a_b)$ and $\mathsf t_{0,-1}$ to zero, so the image of $\mathfrak C_{\mathsf L}$ is the coefficient ring $\mathbb C[\epsilon_1,\epsilon_2]$. In other words, the augmentation $\mathfrak C_{\mathsf A}$ of $\mathsf A^{(K)}$ extends to an augmentation $\mathfrak C_{\mathsf L}$ of $\mathsf L^{(K)}$.\\

Finally, the duality automorphism $\sigma:\mathsf A^{(K)}\cong\mathsf A^{(K)}$ extends naturally to $\sigma:\mathsf L^{(K)}\cong\mathsf L^{(K)}$ such that $\sigma(\epsilon_2)= \epsilon_3$ and
\begin{equation}
\begin{split}
\sigma(\mathsf t_{n,m})=\mathsf t_{n,m}&,\quad \sigma(\mathsf T_{n,m}(X))= -\mathsf T_{n,m}(X^{\mathrm t})-\epsilon_1\mathrm{tr}(X)\mathsf t_{n,m},
\end{split}
\end{equation}
for all $(n,m)\in \mathbb N\times \mathbb Z$.

\subsection{PBW theorem for \texorpdfstring{$\mathsf L^{(K)}$}{Lk}}
Let us define the following elements in $\mathsf A^{(K)}$: $$\mathbf T_{n,2n}(X)=(-1)^n\iota(\mathsf T_{n,0}(X)),\quad \mathbf t_{n,2n}=(-1)^n\iota(\mathsf t_{n,0}),$$ for all $n\in \mathbb N$, using the involution $\iota$ constructed in Proposition \ref{prop: involution on L}. Then we recursively define
\begin{align*}
    \mathbf T_{n,m}(X):=\frac{1}{m+1}[\mathsf t_{1,0},\mathbf T_{n,m+1}(X)],\quad \mathbf t_{n,m}:=\frac{1}{m+1}[\mathsf t_{1,0},\mathbf t_{n,m+1}],
\end{align*}
for all $0\le m<2n$. And recursively define
\begin{align*}
    \mathbf T_{m,2n}(X):=\frac{1}{m+1}[\mathbf T_{m+1,2n}(X),\mathsf t_{0,1}]&,\quad \mathbf t_{m,2n}:=\frac{1}{m+1}[\mathbf t_{m+1,2n},\mathsf t_{0,1}],\\
    \mathbf T_{m,2n-1}(X):=\frac{1}{m+1}[\mathbf T_{m+1,2n-1}(X),\mathsf t_{0,1}]&,\quad \mathbf t_{m,2n-1}:=\frac{1}{m+1}[\mathbf t_{m+1,2n-1},\mathsf t_{0,1}],
\end{align*}
for all $0\le m<n$.
\begin{lemma}\label{lem: generators mathbf T}
$\mathbf T_{n,0}(X)=\mathsf T_{n,0}(X),\mathbf t_{n,0}=\mathsf t_{n,0}$, and $\mathbf T_{0,n}(X)=\mathsf T_{0,n}(X),\mathbf t_{0,n}=\mathsf t_{0,n}$ for all $n\in \mathbb N$. More generally, 
\begin{align*}
    \mathbf T_{n,m}(X)&\equiv\mathsf T_{n,m}(X)\pmod{V_{n-1}\mathsf D^{(K)}},\\
    \mathbf t_{n,m}&\equiv\mathsf t_{n,m}\pmod{V_{n-1}\mathsf A^{(K)}},
\end{align*}
where $V_{\bullet}\mathsf A^{(K)}$ and $V_{\bullet}\mathsf D^{(K)}$ are the vertical filtrations introduced in Definition \ref{def: filtration_vertical and horizontal}. In particular $\mathbf T_{n,m}(X)$ belongs to the subalgebra $\mathsf D^{(K)}$.
\end{lemma}
\begin{proof}
Let $\mathbf r_n(x,y)=(yxy)^n$, then by the definition of $\iota$, we have $\mathbf T_{n,2n}(X)=\mathsf T_{\mathbf r_n}(X)$ and $\mathbf t_{n,2n}=\frac{1}{\epsilon_2}\mathsf T_{\mathbf r_n}(1)$, thus 
\begin{align*}
    \mathbf T_{n,2n}(X)&\equiv\mathsf T_{n,2n}(X)\pmod{V_{n-1}\mathsf D^{(K)}\cap H_{2n-1}\mathsf D^{(K)}},\\
    \mathbf t_{n,2n}&\equiv\mathsf t_{n,2n}\pmod{V_{n-1}\mathsf A^{(K)}\cap H_{2n-1}\mathsf A^{(K)}}
\end{align*}
by the equation \eqref{eqn: transformation to unsymmetrized basis}. Then the rest of claims follow from Proposition \ref{prop: filtration_vertical and horizontal}.
\end{proof}

\begin{lemma}
\begin{align}\label{eqn: [t(1,2), mathbf T(n,m)]}
    \mathbf T_{n,m}(X)=\frac{1}{m-2n-1}[\mathsf t_{1,2},\mathbf T_{n,m-1}(X)],\quad \mathbf t_{n,m}=\frac{1}{m-2n-1}[\mathsf t_{1,2},\mathbf t_{n,m-1}]
\end{align}
for all $0<m\le 2n$. 
\end{lemma}
\begin{proof}
First of all 
\begin{align*}
    \iota([\mathsf t_{1,2},\mathbf T_{n,2n}(X)])=(-1)^{n+1}[\mathsf t_{1,0},\mathbf T_{n,0}(X)]=0,
\end{align*}
thus $[\mathsf t_{1,2},\mathbf T_{n,2n}(X)]=0$, and similarly $[\mathsf t_{1,2},\mathbf t_{n,2n}]=0$; then 
\begin{align*}
    &\frac{1}{m-2n-1}[\mathsf t_{1,2},\mathbf T_{n,m-1}(X)]=\frac{1}{m(m-2n-1)}[\mathsf t_{1,2},[\mathsf t_{1,0},\mathbf T_{n,m}(X)]]\\
    &=\frac{2}{m(2n+1-m)}[\mathsf t_{1,1},\mathbf T_{n,m}(X)]+\frac{1}{m(m-2n-1)}[\mathsf t_{1,0},[\mathsf t_{1,2},\mathbf T_{n,m}(X)]]\\
    &=\frac{2(m-n)}{m(2n+1-m)}\mathbf T_{n,m}(X)+\frac{m-2n}{m(m-2n-1)}[\mathsf t_{1,0},\mathbf T_{n,m+1}(X)]\\
    &=\mathbf T_{n,m}(X),
\end{align*}
and similarly $\frac{1}{m-2n-1}[\mathsf t_{1,2},\mathbf t_{n,m-1}]=\mathbf t_{n,m}$.
\end{proof}
As a corollary to \eqref{eqn: [t(1,2), mathbf T(n,m)]}, we have
\begin{align}
    \iota(\mathbf T_{n,m}(X))=(-1)^n\mathbf T_{n,2n-m}(X),\quad \iota(\mathbf t_{n,m})=(-1)^n\mathbf t_{n,2n-m}
\end{align}
for all $0\le m\le 2n$. Another corollary to \eqref{eqn: [t(1,2), mathbf T(n,m)]} is that 
\begin{align}
    \mathbf T_{n,1}(X)=\mathsf T_{n,1}(X),\quad \mathbf t_{n,1}=\mathsf t_{n,1}, \quad \forall n\in \mathbb N.
\end{align}

\begin{definition}
For $(n,m)\in \mathbb N\times\mathbb N_{>0}$, we define
\begin{align*}
    \mathbf T_{n,-m}(X):=(-1)^n\iota(\mathbf T_{n,2n+m}(X)),\quad \mathbf t_{n,-m}:=(-1)^n\iota(\mathbf t_{n,2n+m}).
\end{align*}
We define $\mathfrak L^{(K)}$ to be the $\mathbb C[\epsilon_1,\epsilon_2]$-subalgebra of $\mathsf L^{(K)}$ generated by $\{\mathbf T_{n,m}(X)\:|\: X\in \gl_K,(n,m)\in \mathbb N\times \mathbb Z\}$.
\end{definition}

\begin{remark}
Equivalently, $\mathfrak L^{(K)}$ to be the $\mathbb C[\epsilon_1,\epsilon_2]$-subalgebra of $\mathsf L^{(K)}$ generated by $\mathsf D^{(K)}$ and $\iota(\mathsf D^{(K)})$. Note that $\rho_N(\mathsf D^{(K)})= \mathrm{S}\mathcal{H}^{(K)}_N$, and $\mathrm{S}\mathbb{H}^{(K)}_N$ is generated by $\mathrm{S}\mathcal{H}^{(K)}_N$ and $\iota(\mathrm{S}\mathcal{H}^{(K)}_N)$, thus the surjective map $\rho_N:\mathsf L^{(K)}[\epsilon_2^{-1}]\twoheadrightarrow \mathrm{S}\mathbb{H}^{(K)}_N[\epsilon_2^{-1}]$ restricts to a surjective map $\rho_N:\mathfrak L^{(K)}\twoheadrightarrow  \mathrm{S}\mathbb{H}^{(K)}_N$.
\end{remark}

It follows from the Lemma \ref{lem: generators mathbf T} that $\mathbf T_{0,-m}(X)=\mathsf T_{0,-m}(X)$ and $\mathbf t_{0,-m}=\mathsf t_{0,-m}$. Moreover, it follows from the definition that 
\begin{align}
    \mathbf T_{n,-m}(X)=\frac{1}{n+1}[\mathsf t_{0,-1},\mathbf T_{n+1,-m+2}(X)],\quad \mathbf t_{n,-m}=\frac{1}{n+1}[\mathsf t_{0,-1},\mathbf t_{n+1,-m+2}]
\end{align}
for all $m>0$. Since $[\mathbf T_{n,0}(X),\mathsf t_{0,1}]=n\mathbf T_{n-1,0}(X)$ and  $[\mathbf T_{n,1}(X),\mathsf t_{0,1}]=n\mathbf T_{n-1,1}(X)$, it follows that
\begin{align}\label{eqn: [mathbf T(n,m), t(0,1)]}
    [\mathbf T_{n,m}(X),\mathsf t_{0,1}]=n\mathbf T_{n-1,m}(X),\quad [\mathbf t_{n,m},\mathsf t_{0,1}]=n\mathbf t_{n-1,m}
\end{align}
for all $m\le 1$.

\begin{definition}\label{def: vertical filtration on L}
The vertical filtration $0=V_{-1}\mathsf L^{(K)}\subset V_0\mathsf L^{(K)}\subset V_1\mathsf L^{(K)}\subset\cdots$ is an exhaustive increasing filtration induced by setting the degree on generators as
\begin{align}
    \deg_v\epsilon_1=\deg_v\epsilon_2=0,\quad \deg_v\mathbf T_{n,m}(X)=\deg_v\mathbf t_{n,m}=n.
\end{align}
And we define $V_\bullet\mathfrak L^{(K)}:=\mathfrak L^{(K)}\cap V_\bullet\mathsf L^{(K)}$.
\end{definition}

By the Lemma \ref{lem: generators mathbf T} there are inclusions $V_n\mathsf A^{(K)}\subset V_n\mathsf L^{(K)}\cap \mathsf A^{(K)}$, later we will see that the inclusion is actually an equality (Remark \ref{rmk: vertival filtrations on L vs A}). Since $\iota(\mathbf T_{n,m}(X))=(-1)^n \mathbf T_{n,2n-m}(X)$ and $\iota(\mathbf t_{n,m})=(-1)^n \mathbf t_{n,2n-m}$ for all $(n,m)\in \mathbb N\times\mathbb Z$, it follows that $\iota (V_n\mathsf L^{(K)})=V_n\mathsf L^{(K)}$.

The following is analogous to Proposition \ref{prop: filtration_vertical and horizontal}.
\begin{proposition}\label{prop: filtration on L}
The commutators between generators of $\mathsf L^{(K)}$ can be schematically written as
\begin{equation}\label{eqn: schematic commutators in L}
\begin{split}
&[\mathbf T_{n,m}(X),\mathbf T_{p,q}(Y)]=\mathbf T_{n+p,m+q}([X,Y])\pmod{V_{n+p-1}\mathfrak L^{(K)}},\\
&[\mathbf t_{n,m},\mathbf T_{p,q}(X)]=(nq-mp)\mathbf T_{n+p-1,m+q-1}(X)\pmod{V_{n+p-2}\mathfrak L^{(K)}},\\
&[\mathbf t_{n,m},\mathbf t_{p,q}]=(nq-mp)\mathbf t_{n+p-1,m+q-1}\pmod{V_{n+p-2}\mathsf L^{(K)}},
\end{split}
\end{equation}
for all $(n,m,p,q)\in \mathbb N\times \mathbb Z\times \mathbb N\times \mathbb Z$, and all $X,Y\in \gl_K$.
\end{proposition}

\begin{proof}
For $(n,m,p,q)\in\mathbb N^4$, \eqref{eqn: schematic commutators in L} follows from Lemma \ref{lem: generators mathbf T} and Proposition \ref{prop: filtration_vertical and horizontal}. Since $\iota (V_n\mathsf L^{(K)})=V_n\mathsf L^{(K)}$, it follows that \eqref{eqn: schematic commutators in L} holds for $(n,m,p,q)\in\mathbb N^4$ such that $m\le 2n$ and $q\le 2p$. It remains to prove the cases when $m<0$ and $q> 2p$ or $q<0$ and $m> 2n$. These two cases are similar and we elaborate the detail for one of them, namely the case when $q<0$ and $m>2n$. As we has explained  and we proceed by induction as follows. Let us fix $m$ and let $n_0=\lceil m/2\rceil$, then \eqref{eqn: schematic commutators in L} holds for $(n_0,m,p,q)$ such that $q<0$. Assume that \eqref{eqn: schematic commutators in L} holds for a fixed pair $(n,m)$ such that $m\ge 2n-1$ and all $(p,q)$ such that $q<0$, then 
\begin{align*}
    &[\mathbf T_{n-1,m}(X),\mathbf T_{p,q}(Y)]=\frac{-1}{n}[\mathrm{ad}_{\mathsf t_{0,1}}(\mathbf T_{n,m}(X)),\mathbf T_{p,q}(Y)]\\
    &=\frac{1}{n}[\mathbf T_{n,m}(X),\mathrm{ad}_{\mathsf t_{0,1}}(\mathbf T_{p,q}(Y))]-\frac{1}{n}\mathrm{ad}_{\mathsf t_{0,1}}([\mathbf T_{n,m}(X),\mathbf T_{p,q}(Y)])\\
    &=\frac{-p}{n}[\mathbf T_{n,m}(X),\mathbf T_{p-1,q}(Y)]+\frac{n+p}{n}\mathbf T_{n+p-1,m+q}([X,Y])\pmod{V_{n+p-2}\mathsf L^{(K)}}\\
    &=\mathbf T_{n+p-1,m+q}([X,Y])\pmod{V_{n+p-2}\mathsf L^{(K)}},
\end{align*}
by our assumption. Similarly the other two equations in \eqref{eqn: schematic commutators in L} also hold for $(n-1,m)$ and all $(p,q)$ such that $q<0$. Therefore the decreasing induction on $n$ starting from $n_0$ implies that \eqref{eqn: schematic commutators in L} holds for all $(n,m,p,q)$ such that $q<0$ and $m>2n$. This finishes the proof.
\end{proof}

Let us choose a basis $\mathfrak B:=\{X_1,\cdots,X_{K^2-1}\}$ of $\mathfrak{sl}_K$, so that $\mathfrak{B}_+:=\{1\}\cup \mathfrak B$ is a basis of $\mathfrak{gl}_K$. We fix a total order $1\preceq X_1\preceq\cdots\preceq X_{K^2-1}$ on $\mathfrak B_+$. Then we put the dictionary order on the set $\mathfrak{G}(\mathsf L^{(K)}):=\{\mathbf T_{n,m}(X),\mathbf t_{n,m}\:|\: X\in \mathfrak B, (n,m)\in \mathbb N\times\mathbb Z\}$, in other words $\mathbf T_{n,m}(X)\preceq\mathbf T_{n',m'}(X')$ if and only only if $n<n'$ or $n=n'$ and $m<m'$ or $(n,m)=(n',m')$ and $X\preceq X'$. 

We also define the set $\mathfrak{G}(\mathfrak L^{(K)}):=\{\mathbf T_{n,m}(X)\:|\: X\in \mathfrak B_+, (n,m)\in \mathbb N\times\mathbb Z\}$ and put the dictionary order on it.

\begin{definition}\label{def: basis of L}
Define the set of ordered monomials in $\mathfrak{G}(\star)$ as 
\begin{align}
    \mathfrak B(\star):=\{1\}\cup\{\mathcal O_1\cdots\mathcal O_n\:|\: n\in \mathbb N_{>0}, \mathcal O_1\preceq \cdots\preceq\mathcal O_n\in \mathfrak{G}(\star)\},
\end{align}
where $\star=\mathsf L^{(K)}$ or $\mathfrak L^{(K)}$.
\end{definition}

\begin{theorem}\label{thm: PBW for L}
$\mathsf L^{(K)}$ (resp. $\mathfrak L^{(K)}$) is a free $\mathbb C[\epsilon_1,\epsilon_2]$-module with basis $\mathfrak B(\mathsf L^{(K)})$ (resp. $\mathfrak B(\mathfrak L^{(K)})$).
\end{theorem}

\begin{proof}
We first show that $\mathsf L^{(K)}$ is spanned by $\mathfrak B(\mathsf L^{(K)})$. Since $V_0\mathsf L^{(K)}$ is generated by $\mathbf T_{0,n}(X)$ and $\mathbf t_{0,n}$, it follows that $V_0\mathsf L^{(K)}$ is a quotient of the universal enveloping algebra of the loop algebra of $\mathfrak{sl}_K\oplus\gl_1$. In particular $V_0\mathsf L^{(K)}$ is spanned by elements in $\mathfrak B(\mathsf L^{(K)})$ by the PBW theorem for the Lie algebra. Assume that $V_s\mathsf L^{(K)}$ is generated by elements in $\mathfrak B(\mathsf L^{(K)})$, then Proposition \ref{prop: filtration on L} implies that we can reorder any monomials in $\mathfrak{G}(\mathsf L^{(K)})$ with total degree $s+1$ into the non-decreasing order modulo terms in $V_{s}\mathsf L^{(K)}$, therefore $V_{s+1}\mathsf L^{(K)}$ is generated by elements in $\mathfrak B(\mathsf L^{(K)})$. $V_{\bullet}\mathsf L^{(K)}$ is obviously exhaustive, thus $\mathsf L^{(K)}$ is generated by $\mathfrak B(\mathsf L^{(K)})$. 
 
Next we show that there is no nontrivial $\mathbb C[\epsilon_1,\epsilon_2]$-linear relations among elements in $\mathfrak B(\mathsf L^{(K)})$. Since $\mathsf L^{(K)}$ has no $\epsilon_1$-torsion, it suffices to show that the natural map $\mathbb C[\epsilon_2]\cdot\mathfrak B(\mathsf L^{(K)})\to \mathsf L^{(K)}/(\epsilon_1=0)\overset{S(w)}{\longrightarrow}\mathsf A^{(K)}/(\epsilon_1=0)(\!(w^{-1})\!)$ is injective. Note that $\mathsf A^{(K)}/(\epsilon_1=0)$ is isomorphic to $D_{\epsilon_2}(\mathbb C)\otimes\gl_K^{\sim}$ (Corollary \ref{cor: A is DDCA}). 

Consider the shift map $S_D(w):D_{\epsilon_2}(\mathbb C^{\times })\otimes \mathfrak{gl}_K^{\sim}\to D_{\epsilon_2}(\mathbb C)\otimes \mathfrak{gl}_K^{\sim}(\!(w^{-1})\!)$, which induces an embedding $U(D_{\epsilon_2}(\mathbb C^{\times })\otimes \mathfrak{gl}_K^{\sim})\hookrightarrow U(D_{\epsilon_2}(\mathbb C)\otimes \mathfrak{gl}_K^{\sim})(\!(w^{-1})\!)\cong \mathsf A^{(K)}/(\epsilon_1=0)(\!(w^{-1})\!)$. Notice that $\mathsf t_{0,-1}$ is contained in the image of $S_D(w)$, thus the image of $\mathbb C[\epsilon_2]\cdot\mathfrak B(\mathsf L^{(K)})$ in $\mathsf A^{(K)}/(\epsilon_1=0)(\!(w^{-1})\!)$ is contained in the image of $S_D(w)$, and it remains to show that the induced map $\mathbb C[\epsilon_2]\cdot\mathfrak B(\mathsf L^{(K)})\to U(D_{\epsilon_2}(\mathbb C^{\times })\otimes \mathfrak{gl}_K^{\sim})$ is injective. It is straightforward to see that the image of $\mathbf T_{n,m}(X)$ in $U(D_{\epsilon_2}(\mathbb C^{\times })\otimes \mathfrak{gl}_K^{\sim})$ is $x^my^n\otimes X$ modulo $\mathrm{Ord}^{n-1}D_{\epsilon_2}(\mathbb C^{\times })\otimes \mathfrak{gl}_K^{\sim}$, where $\mathrm{Ord}^{\bullet}$ is the filtration by order of differential operators on $D_{\epsilon_2}(\mathbb C^{\times })\otimes \mathfrak{gl}_K^{\sim}$. Similarly $\mathbf t_{n,m}\mapsto x^ny^m/\epsilon_2$ modulo $\mathrm{Ord}^{n-1}D_{\epsilon_2}(\mathbb C^{\times })\otimes \mathfrak{gl}_K^{\sim}$. Therefore the image of $\mathbb C[\epsilon_2]\cdot\mathfrak{G}(\mathsf L^{(K)})$ in $U(D_{\epsilon_2}(\mathbb C^{\times })\otimes \mathfrak{gl}_K^{\sim})$ is exactly the primitive part $D_{\epsilon_2}(\mathbb C^{\times })\otimes \mathfrak{gl}_K^{\sim}$, and the map induces an isomorphism $\mathbb C[\epsilon_2]\cdot\mathfrak{G}(\mathsf L^{(K)})\cong D_{\epsilon_2}(\mathbb C^{\times })\otimes \mathfrak{gl}_K^{\sim}$, whence $\mathbb C[\epsilon_2]\cdot\mathfrak B(\mathsf L^{(K)})\to U(D_{\epsilon_2}(\mathbb C^{\times })\otimes \mathfrak{gl}_K^{\sim})$ is an isomorphism by the PBW theorem for the Lie algebra. This proves that $\mathsf L^{(K)}$ is a free $\mathbb C[\epsilon_1,\epsilon_2]$-module with basis $\mathfrak B(\mathsf L^{(K)})$.

\bigskip Similarly, we show that $\mathfrak L^{(K)}$ is spanned by $\mathfrak B(\mathfrak L^{(K)})$. Since $V_0\mathfrak L^{(K)}$ is generated by $\mathbf T_{0,n}(X)$, it follows that $V_0\mathfrak L^{(K)}$ is a quotient of the universal enveloping algebra of the loop algebra of $\mathfrak{gl}_K$. In particular $V_0\mathfrak L^{(K)}$ is spanned by elements in $\mathfrak B(\mathfrak L^{(K)})$ by the PBW theorem for the Lie algebra. Assume that $V_s\mathfrak L^{(K)}$ is generated by elements in $\mathfrak B(\mathfrak L^{(K)})$, then Proposition \ref{prop: filtration on L} implies that we can reorder any monomials in $\mathfrak{G}(\mathfrak L^{(K)})$ with total degree $s+1$ into the non-decreasing order modulo terms in $V_{s}\mathfrak L^{(K)}$, therefore $V_{s+1}\mathfrak L^{(K)}$ is generated by elements in $\mathfrak B(\mathfrak L^{(K)})$. $V_{\bullet}\mathfrak L^{(K)}$ is obviously exhaustive, thus $\mathfrak L^{(K)}$ is generated by $\mathfrak B(\mathfrak L^{(K)})$. 

Since elements in $\mathfrak B(\mathfrak L^{(K)})$ is obtained from elements in $\mathfrak B(\mathsf L^{(K)})$ by a scaling, and we have shown that there is no nontrivial $\mathbb C[\epsilon_1,\epsilon_2]$-linear relations among elements in $\mathfrak B(\mathsf L^{(K)})$, thus elements in $\mathfrak B(\mathfrak L^{(K)})$ are $\mathbb C[\epsilon_1,\epsilon_2]$-linear independent as well. This proves that $\mathfrak L^{(K)}$ is a free $\mathbb C[\epsilon_1,\epsilon_2]$-module with basis $\mathfrak B(\mathfrak L^{(K)})$.
\end{proof}

\begin{remark}\label{rmk: vertival filtrations on L vs A}
$\mathsf L^{(K)}/(\epsilon_1=0)$ is isomorphic to the universal enveloping algebra of $D_{\epsilon_2}(\mathbb C^{\times})\otimes\gl_K^{\sim}$. The associated graded $\mathrm{gr}_V\mathsf L^{(K)}$ is isomorphic to the universal enveloping algebra of $\mathscr O(\mathbb C\times\mathbb C^{\times})\otimes\gl_K$. Moreover, the associated graded map $\mathrm{gr}_V\mathsf A^{(K)}\to \mathrm{gr}_V\mathsf L^{(K)}$ is identified with the canonical map $U(\mathscr O(\mathbb C^2)\otimes\gl_K)\to U(\mathscr O(\mathbb C\times\mathbb C^{\times})\otimes\gl_K)$, which is injective. In particular $V_n\mathsf A^{(K)}=V_n\mathsf L^{(K)}\cap\mathsf A^{(K)}$.
\end{remark}

\begin{remark}\label{rmk: classical limit of mathfrak L}
$\mathfrak L^{(K)}/(\epsilon_1=0)$ is isomorphic to the universal enveloping algebra of $D_{\epsilon_2}(\mathbb C^{\times})\otimes\gl_K$, and the natural map $\mathfrak L^{(K)}/(\epsilon_1=0)\to \mathsf L^{(K)}/(\epsilon_1=0)$ is identified with the natural inclusion $U(D_{\epsilon_2}(\mathbb C^{\times})\otimes\gl_K)\hookrightarrow U(D_{\epsilon_2}(\mathbb C^{\times})\otimes\gl_K^{\sim})$. In particular, $\mathfrak L^{(K)}/(\epsilon_1=0)\to \mathsf L^{(K)}/(\epsilon_1=0)$ is injective.
\end{remark}

\subsection{Gluing construction of \texorpdfstring{$\mathsf L^{(K)}$}{Lk}}
As we have shown, $\mathsf L^{(K)}$ contains two subalgebras which are isomorphic to $\mathsf A^{(K)}$, the first one is generated by $\{\mathbf T_{n,m}(X),\mathbf t_{n,m}\:|\: X\in \gl_K,(n,m)\in \mathbb N^2\}$, and the second one is the image of the first one under the involution $\iota$, i.e. it is generated by $\{\mathbf T_{n,m}(X),\mathbf t_{n,m}\:|\: X\in \gl_K,(n,m)\in \mathbb N\times \mathbb Z,m\le 2n\}$.

\begin{theorem}\label{thm: gluing construction of L}
$\mathsf L^{(K)}$ is generated over $\mathbb C[\epsilon_1,\epsilon_2]$ by two algebras $\mathsf A^{(K)}_+,\mathsf A^{(K)}_-$, both are isomorphic to $\mathsf A^{(K)}$ (whose generators are indicated by superscripts $+$ or $-$ respectively), with relations
\begin{equation}\label{eqn: gluing relations of L}
\begin{split}
\mathsf t^-_{0,0}= \mathsf t^+_{0,0},\quad\mathsf t^-_{1,2}= -\mathsf t^+_{1,0},&\quad  \mathsf t^-_{1,0}= -\mathsf t^+_{1,2},\quad\mathsf t^-_{2,2}= \mathsf t^+_{2,2},\\
\mathsf T^-_{0,0}(X)= \mathsf T^+_{0,0}(X),&\quad \mathsf T^-_{1,1}(X)= -\mathsf T^+_{1,1}(X),\\
[\mathsf t^-_{0,1}, \mathsf t^+_{0,1}]=0,&\quad [\mathsf t^-_{0,1}, \mathsf T^+_{0,1}(X)]=0,
\end{split}
\end{equation}
for all $X\in \mathfrak{sl}_K$.
\end{theorem}

\begin{proof}
Denote by $\mathsf L^{(K)}_{\mathrm{new}}$ the $\mathbb C[\epsilon_1,\epsilon_2]$-algebra generated by $\mathsf A^{(K)}_+,\mathsf A^{(K)}_-$ with relations \eqref{eqn: gluing relations of L}. It follows from \eqref{eqn: gluing relations of L} that there is a surjective $\mathbb C[\epsilon_1,\epsilon_2]$-algebra map $\mathsf L^{(K)}_{\mathrm{new}}\to \mathsf L^{(K)}$ such that $$\mathbf T^+_{n,m}(X)\mapsto \mathbf T_{n,m}(X),\quad \mathbf t^+_{n,m}\mapsto \mathbf t_{n,m},\quad\mathbf T^-_{n,m}(X)\mapsto (-1)^n\mathbf T_{n,2n-m}(X),\quad \mathbf t^-_{n,m}\mapsto (-1)^n\mathbf t_{n,2n-m},$$ for all $(n,m)\in \mathbb N^2$. It remains to show that this map is injective.

To this end, let us derive more relations from \eqref{eqn: gluing relations of L}. We claim that the following relations hold in $\mathsf L^{(K)}_{\mathrm{new}}$
\begin{align}\label{eqn: gluing relations of L, I}
\mathbf t^-_{n,m}= \iota(\mathbf t^+_{n,m}),\quad \mathbf T^-_{n,m}(X)= \iota(\mathbf T^+_{n,m}(X)),
\end{align}
for all $(n,m)\in \mathbb N^2$ such that $m\le 2n$. The case when $(n,m)=(0,0)$ and $(1,0)$ are covered in \eqref{eqn: gluing relations of L}, and using the adjoin action of $\mathsf t^-_{1,2}$ we obtain The case when $(n,m)=(1,1)$ and $(1,2)$. Moreover, since $[\mathsf t^-_{2,0},\mathsf t^-_{1,2}]=4\mathsf t^-_{2,1}$, it follows that $\mathbf t^-_{2,1}=\iota(\mathbf t^+_{2,1})$. Using the consecutive adjoint action of $\mathbf t^-_{2,1}$ on $\mathbf T^-_{1,0}(X)$ and $\mathbf t^-_{1,0}$, \eqref{eqn: gluing relations of L, I} are obtained for $m=0$ and all $n\in \mathbb N$. Then the consecutive adjoint action of $\mathbf t^-_{1,2}$ on $\mathbf T^-_{n,0}(X)$ and $\mathbf t^-_{n,0}$ implies that \eqref{eqn: gluing relations of L, I} holds for all $(n,m)\in \mathbb N^2$ such that $m\le 2n$.

Using the consecutive adjoint action of $\mathbf t^+_{1,2}$ (which equals to $-\mathsf t^-_{1,0}$) on $\mathsf T^+_{0,1}(X)$ and $\mathsf t^+_{0,1}$, we see that 
\begin{align*}
    [\mathsf t^-_{0,1}, \mathsf t^+_{0,n}]=[\mathsf t^-_{0,1}, \mathsf T^+_{0,n}(X)]=0,
\end{align*}
for all $n\in \mathbb N_{>0}$. Then using the identity $[\mathsf T^-_{1,1}(X),\mathsf t^-_{0,1}]=\mathsf T^-_{0,1}(X)$ together with the relation $\mathsf T^-_{1,1}(X)=-\mathsf T^+_{1,1}(X)$ we obtain the relations
\begin{align*}
[\mathsf T^-_{0,1}(X),\mathsf T^+_{0,n}(Y)]=\mathsf T^+_{0,n-1}([X,Y]),\quad [\mathsf T^-_{0,1}(X),\mathsf t^+_{0,n}]=0,
\end{align*}
for all $n\in \mathbb N_{>0}$. Next, using the consecutive adjoint action of $\mathsf t^-_{1,2}$ (which equals to $-\mathsf t^+_{1,0}$) on $\mathsf T^-_{0,1}(X)$ and $\mathsf t^-_{0,1}$, we obtain the relations
\begin{align}\label{eqn: gluing relations of L, II}
[\mathsf T^-_{0,m}(X),\mathsf T^+_{0,n}(Y)]=
\begin{cases}
\mathsf T^+_{0,n-m}([X,Y]),& n\ge m,\\
\mathsf T^-_{0,m-n}([X,Y]),& n\le m,
\end{cases}
\quad [\mathsf t^-_{0,m},\mathsf T^+_{0,n}(X)]=[\mathsf T^-_{0,m}(X),\mathsf t^+_{0,n}]=0,
\end{align}
for all $(n,m)\in \mathbb N$.

Using the relations $\mathbf T^-_{n,2n}(X)=(-1)^n\mathbf T^+_{n,0}(X)$ and $\mathbf t^-_{n,2n}=(-1)^n\mathbf t^+_{n,0}$, we derive that $[\mathsf t^+_{0,1},\mathbf T^-_{n,2n}(X)]=n\mathbf T^-_{n-1,2n-2}(X)$ and $[\mathsf t^+_{0,1},\mathbf t^-_{n,2n}]=n\mathbf t^-_{n-1,2n-2}$; similarly we derive that $[\mathsf t^+_{0,1},\mathbf T^-_{n,2n-1}(X)]=n\mathbf T^-_{n-1,2n-3}(X)$ and $[\mathsf t^+_{0,1},\mathbf t^-_{n,2n-1}]=n\mathbf t^-_{n-1,2n-3}$. Then the consecutive adjoint action of $\mathsf t^-_{0,1}$ implies the relations
\begin{align}\label{eqn: gluing relations of L, III}
[\mathsf t^+_{0,1},\mathbf T^-_{n,m}(X)]=n\mathbf T^-_{n-1,m-2}(X),\quad [\mathsf t^+_{0,1},\mathbf t^-_{n,m}]=n\mathbf t^-_{n-1,m-2},
\end{align}
for all $(n,m)\in \mathbb N^2$ such that $m\ge 2n-1$.

Now we are ready to prove the injectivity of the natural map $\mathsf L^{(K)}_{\mathrm{new}}\to \mathsf L^{(K)}$. Let us re-define the notation for the generators of $\mathsf L^{(K)}_{\mathrm{new}}$, namely, define $\mathbf T_{n,m}(X):=\mathbf T^+_{n,m}(X)$ and $\mathbf t_{n,m}:=\mathbf t^+_{n,m}$ for $(n,m)\in \mathbb N^2$ and $X\in \mathfrak{gl}_K$; then define $\mathbf T_{n,m}(X):=(-1)^n\mathbf T^-_{n,2n-m}(X)$ and $\mathbf t_{n,m}:=(-1)^n\mathbf t^-_{n,2n-m}$, for $(n,m)\in \mathbb N\times \mathbb Z_{<0}$ and $X\in \mathfrak{gl}_K$; so that $\mathsf L^{(K)}_{\mathrm{new}}\to \mathsf L^{(K)}$ maps $\mathbf T_{n,m}(X)\mapsto\mathbf T_{n,m}(X)$ and $\mathbf t_{n,m}\mapsto\mathbf t_{n,m}$ for all $(n,m)\in \mathbb N^2$. Define a $\mathbb C[\epsilon_1,\epsilon_2]$-module map $\mathbb C[\epsilon_1,\epsilon_2]\cdot \mathfrak{B}( \mathsf L^{(K)}) \to \mathsf L^{(K)}_{\mathrm{new}}$ by sending an ordered monomial $\mathcal O_1\cdots\mathcal O_n$ to the monomial in $\mathsf L^{(K)}_{\mathrm{new}}$ given by the same symbol. 

We claim that $\mathbb C[\epsilon_1,\epsilon_2]\cdot \mathfrak{B}( \mathsf L^{(K)}) \to \mathsf L^{(K)}_{\mathrm{new}}$ is surjective. In fact, $\mathsf L^{(K)}_{\mathrm{new}}$ inherits the vertical filtration $V_{\bullet}\mathsf L^{(K)}_{\mathrm{new}}$ from two subalgebras $\mathsf A^{(K)}_\pm$ such that the filtration degree is given by the Definition \ref{def: vertical filtration on L}. Moreover, the Proposition \ref{prop: filtration on L} holds for $\mathsf L^{(K)}_{\mathrm{new}}$ as well, because all the essential ingredients in the proof of \ref{prop: filtration on L} are 
\begin{itemize}
    \item Proposition \ref{prop: filtration} for $\mathsf A^{(K)}_\pm$,
    \item and the relations \eqref{eqn: gluing relations of L, III} for the recursion step.
\end{itemize}
Note that $V_0\mathsf L^{(K)}_{\mathrm{new}}$ is generated by $\mathbf T_{0,m}(X),\mathbf t_{0,m}$, which satisfy the relations $$[\mathbf T_{0,m}(X),\mathbf T_{0,n}(Y)]=\mathbf T_{0,m+n}([X,Y]),\quad [\mathbf T_{0,m}(X),\mathbf t_{0,n}]=[\mathbf t_{0,m},\mathbf t_{0,n}]=0,$$ by \eqref{eqn: gluing relations of L, II}, so there is a surjective map from the universal enveloping algebra of the loop algebra of $\mathfrak{sl}_K\oplus\gl_1$ to $V_0\mathsf L^{(K)}_{\mathrm{new}}$, in particular $V_0\mathsf L^{(K)}_{\mathrm{new}}$ is contained in the image of $\mathbb C[\epsilon_1,\epsilon_2]\cdot \mathfrak{B}( \mathsf L^{(K)})$ by the PBW theorem for the Lie algebra. Assume that $V_s\mathsf L^{(K)}_{\mathrm{new}}$ is generated by elements in $\mathfrak B(\mathsf L^{(K)})$, then Proposition \ref{prop: filtration on L} implies that we can reorder any monomials in $\mathfrak{G}(\mathsf L^{(K)})$ with total degree $s+1$ into the non-decreasing order modulo terms in $V_{s}\mathsf L^{(K)}_{\mathrm{new}}$, therefore $V_{s+1}\mathsf L^{(K)}_{\mathrm{new}}$ is generated by elements in $\mathfrak B(\mathsf L^{(K)})$. $V_{\bullet}\mathsf L^{(K)}_{\mathrm{new}}$ is obviously exhaustive, thus $\mathsf L^{(K)}_{\mathrm{new}}$ is generated by $\mathfrak B(\mathsf L^{(K)})$.  

Finally, the composition $\mathbb C[\epsilon_1,\epsilon_2]\cdot \mathfrak{B}( \mathsf L^{(K)})\to \mathsf L^{(K)}_{\mathrm{new}}\to \mathsf L^{(K)}$ is an isomorphism by Theorem \ref{thm: PBW for L}, this implies that $\mathsf L^{(K)}_{\mathrm{new}}\to \mathsf L^{(K)}$ is an isomorphism.
\end{proof}

\subsection{The intersections of \texorpdfstring{$\mathsf A^{(K)}_{\pm}$}{two copies of Ak}}

\begin{lemma}\label{lem: shift automorphism of L}
For $\beta\in \mathbb C$, there exists an automorphism $\tau_{\beta}:\mathsf L^{(K)}\cong \mathsf L^{(K)}$ which is uniquely determined by
\begin{align}\label{eqn: shift automorphism of L}
\tau_{\beta}(\mathsf T_{0,n}(X))=\mathsf T_{0,n}(X),\quad \tau_{\beta}(\mathsf t_{0,n})=\mathsf t_{0,n},\quad  \tau_{\beta}(\mathsf t_{2,0})=\mathsf t_{2,0}+2\beta\mathbf t_{1,-1}+\beta^2\mathsf t_{0,-2},
\end{align}
for all $n\in \mathbb Z$ and all $X\in \gl_K$.
\end{lemma}
\begin{proof}
Consider the automorphism $\tau_{\beta,N}$ of $\mathbb{H}^{(K)}_N$ uniquely determined by $\tau_{\beta,N}(y_i)= y_i+\frac{\beta}{x_i}$ and $\tau_{\beta,N}$ fixes $\mathbb C[x_1,\cdots,x_N]$ and $\mathbb C[\mathfrak{S}_N]$ and $\gl_K^{\otimes N}$. The it is easy to see that $\tau_{\beta,N}\circ \rho_N$ agrees with $\rho_N\circ \tau_{\beta}$ on the set of elements in \eqref{eqn: shift automorphism of L}. By Lemma \ref{lem: rho_N extends to loop Yangian}, $\tau_\beta$ uniquely determines an algebra homomorphism; $\tau_\beta$ is an automorphism with inverse $\tau_{-\beta}$.
\end{proof}
Note that $\tau_{\beta}$ is additive in $\beta$, i.e. $\tau_{\beta+\beta'}=\tau_{\beta}\circ\tau_{\beta'}=\tau_{\beta'}\circ\tau_{\beta}$. Moreover, it is easy to see that $\tau_{\beta}$ commutes with the involution, i.e. $\iota\circ\tau_{\beta}=\tau_{\beta}\circ \iota$.\\

Define the intersection $$\mathsf C^{(K)}_{\nu}:=\mathsf A^{(K)}_+\cap \tau_{\nu\epsilon_2}(\mathsf A^{(K)}_-).$$ Notice that $\tau_{-\nu\epsilon_2}(\mathbf T_{n,2n}(X))\in \mathsf A^{(K)}_+$ by direct computation, then it follows that $\mathbf T_{n,2n}(X)\in \tau_{\nu\epsilon_2}(\mathsf A^{(K)}_+)$, taking the involution we get $\mathbf T_{n,0}(X)\in \tau_{\nu\epsilon_2}(\mathsf A^{(K)}_-)$, thus $\mathbf T_{n,0}(X)\in \mathsf C^{(K)}_{\nu}$. Similarly $\mathbf t_{n,0}\in\mathsf C^{(K)}_{\nu}$. Moreover, $\mathsf C^{(K)}_{\nu}$ contains an $\mathfrak{sl}_2$-triple $\{\mathbf t_{1,0},\mathbf t_{1,1}+\frac{\nu\epsilon_2}{2}\mathbf t_{0,0},\mathbf t_{1,2}+\nu \epsilon_2\mathbf t_{0,1} \}$.

\begin{proposition}
$\mathsf C^{(K)}_{\nu}$ is generated by $\mathbf T_{n,0}(X),\mathbf t_{n,0}$ and $\mathbf t_{1,2}+\nu \epsilon_2\mathbf t_{0,1}$, for $n\in \mathbb N$ and $X\in \gl_K$.
\end{proposition}

\begin{proof}
Notice that $\mathsf C^{(K)}_{\nu}/(\epsilon_1=0)\to \mathsf L^{(K)}/(\epsilon_1=0)$ is injective, because if $\epsilon_1\cdot f\in \mathsf C^{(K)}_{\nu}$ then $f\in \mathsf A^{(K)}_+$ and $f\in \tau_{\nu\epsilon_2}(\mathsf A^{(K)}_-)$ thus $f\in \mathsf C^{(K)}_{\nu}$. Therefore it suffices to show that the intersection of $\mathsf A^{(K)}_+/(\epsilon_1=0)$ and $\tau_{\nu\epsilon_2}(\mathsf A^{(K)}_+/(\epsilon_1=0))$ in $\mathsf L^{(K)}/(\epsilon_1=0)$ is generated as a $\mathbb C[\epsilon_2]$-algebra by $\mathbf T_{n,0}(X),\mathbf t_{n,0}$ and $\mathbf t_{1,2}+\nu \epsilon_2\mathbf t_{0,1}$, for $n\in \mathbb N$ and $X\in \gl_K$.

According to Remark \ref{rmk: vertival filtrations on L vs A}, $\mathsf L^{(K)}/(\epsilon_1=0)\cong U(D_{\epsilon_2}(\mathbb C^{\times})\otimes\gl_K^{\sim})$, and $\mathsf A^{(K)}_+/(\epsilon_1=0)$ is the subalgebra $U(D_{\epsilon_2}(\mathbb C)\otimes\gl_K^{\sim})$ and the $\tau_{\nu\epsilon_2}\circ\iota$ is the automorphism induce by $x\mapsto x^{-1},\partial_x\mapsto -x^2\partial_x-(\nu+1)x$. Therefore $\mathsf A^{(K)}_+\cap \mathsf A^{(K)}_-$ is isomorphic to the universal enveloping algebra of the intersection between $D_{\epsilon_2}(\mathbb C)\otimes\gl_K^{\sim}$ and its image under the automorphism $\tau_{\nu\epsilon_2}\circ\iota$, and the intersection is nothing but the global section of $D^{\nu}_{\epsilon_2}(\mathbb P^1)\otimes\gl_K^{\sim}$, where $D^{\nu}_{\epsilon_2}(\mathbb P^1)$ is the sheaf of $\mathcal O_{\mathbb P^1}(-\nu-1)$-twisted\footnote{For non-integral $\nu$ we regard $\mathcal O_{\mathbb P^1}(-\nu-1)$ as a twisted sheaf.} $\epsilon_2$-differential operators on $\mathbb P^1$.

By the theorem of Beilinson-Bernstein, $\Gamma(\mathbb P^1,D^{\nu}(\mathbb P^1))$ is isomorphic to $U(\mathfrak{sl}_2)/(\ker \chi_\nu)$, where $(\ker \chi_\nu)$ is the two-sided ideal of $U(\mathfrak{sl}_2)$ generated by the kernel of the central character $\chi_\nu:Z(U(\mathfrak{sl}_2))\to \mathbb C$ corresponding to the $\mathfrak{sl}_2$-character $\nu\rho$, where $\rho$ is the half of the unique positive root. As an $\mathfrak{sl}_2$-representation via the adjoint action, $U(\mathfrak{sl}_2)/(\ker \chi_\nu)$ is the direct sum of all irreducible representation of odd dimension with multiplicity one; on the other hand, $\mathbf T_{n,0}(X)$ and $\mathbf t_{n,0}$ are lowest weight vector of $\mathfrak{sl}_2$ with spin $-n$, hence $\Gamma(\mathbb P^1,D^{\nu}_{\epsilon_2}(\mathbb P^1)\otimes\gl_K^{\sim})$ is exactly spanned by the $\mathfrak{sl}_2$-action on $\mathbf T_{n,0}(X),\mathbf t_{n,0}$. This finishes the proof. 
\end{proof}

\begin{conjecture}
There exists a $\mathbb C[\epsilon_1,\epsilon_2]$-algebra embedding $\mathsf C^{(K)}_{-1/2}\hookrightarrow \mathsf A^{(K)}$ such that 
\begin{align}\label{eqn: embed C into A}
    \mathbf T_{n,0}(X)\mapsto \frac{1}{2^n}\mathsf T_{2n,0}(X),\quad \mathbf t_{1,2}- \frac{\epsilon_2}{2}\mathbf t_{0,1}\mapsto \frac{1}{2}\mathsf t_{0,2}.
\end{align}
In particular the $\mathfrak{sl}_2$-triple $\{f,h,e\}=\{-\mathbf t_{1,0},2\mathbf t_{1,1}- \frac{\epsilon_2}{2}\mathbf t_{0,0} ,\mathbf t_{1,2}- \frac{\epsilon_2}{2}\mathbf t_{0,1}\}$ is mapped to the $\mathfrak{sl}_2$-triple $\{-\frac{1}{2}\mathsf t_{2,0},\mathsf t_{1,1},\frac{1}{2}\mathsf t_{0,2}\}$.
\end{conjecture}

The conjecture is true when $\epsilon_1=0$. To see this, notice that the $\mathfrak{sl}_2$-triple $\{f,h,e\}=\{-\frac{1}{2}\partial_z^2,z\partial_z+\frac{1}{2},\frac{1}{2}z^2\}$ in $D(\mathbb C)$ generates an algebra map $U(\mathfrak{sl}_2)\to D(\mathbb C)$, and such that the quadratic Casimir $C=\frac{1}{2}h^2+ef+fe$ is mapped to the scalar $-\frac{7}{8}$, which equals to $\chi_{-1/2}(C)$, thus the map factors through $U(\mathfrak{sl}_2)/(\ker \chi_{-1/2})$. This map in turn induces Lie algebra embedding $\left(U_{\epsilon_2}(\mathfrak{sl}_2)/(\ker \chi_{-1/2})\right)\otimes\gl_K^{\sim}\hookrightarrow D_{\epsilon_2}(\mathbb C)\otimes\gl_K^{\sim}$, taking the universal enveloping algebra and we get algebra embedding $\mathsf C^{(K)}_{-1/2}/(\epsilon_1=0)\hookrightarrow \mathsf A^{(K)}/(\epsilon_1=0)$, and the \eqref{eqn: embed C into A} follows from direct computation.

\begin{remark}
It is expected that $\mathsf C^{(K)}_{-1/2}$ should be an algebra of gauge-invariant local observables on a topological defect in the twisted M-theory on $\mathbb R \times T^* \mathbb P^1$, where $\mathsf A^{(K)}_{\pm}$ plays the role of the corresponding algebra on the two coordinate patches of $\mathbb P^1$, and the gluing of two patches amounts to taking intersection of $\mathsf A^{(K)}_+$ and $\tau_{-\epsilon_2/2}(\mathsf A^{(K)}_-)$ inside $\mathsf L^{(K)}$, the algebra corresponding algebra on the $\mathbb C\times \mathbb C^{\times}$. The $T^*\mathbb P^1$ is a resolution of an $A_1$ singularity $\mathbb C\times \mathbb C/\mathbb Z_2$, so it is natural to expect that $\mathsf C^{(K)}_{-1/2}$ is isomorphic to the algebra associated to the latter, which is presented as a subalgebra of $\mathsf A^{(K)}$ in \cite[Section 3]{gaiotto2020twisted}.
\end{remark}

\subsection{Meromorphic coproduct of \texorpdfstring{$\mathsf L^{(K)}$}{Lk}}
Analogous to the construction of $\Delta_{\mathsf A}(w)$, we have a morphism:
\begin{align*}
    m:\mathbb C^{\times N_1}_{\mathrm{disj}}\times \mathbb C^{N_2}_{\mathrm{disj}}\times \Spec \mathbb C(\!(w^{-1})\!)\to (\mathbb C^{\times})^{N_1+N_2}_{\mathrm{disj}},
\end{align*}
which acts on generators of function ring by
\begin{equation}\label{eqn: meromorphic coproduct finite N}
    \begin{split}
        x^{(1)}_i\mapsto x^{(1)}_i+w,\quad x^{(2)}_j\mapsto x^{(2)}_j,&\quad \frac{1}{x^{(1)}_i}\mapsto \frac{1}{x^{(1)}_i} ,\quad \frac{1}{x^{(2)}_j}\mapsto \sum_{n=0}^{\infty} w^{-n-1} (-x^{(2)}_j)^n,\\ \frac{1}{x^{(1)}_{i_1}-x^{(1)}_{i_2}}\mapsto \frac{1}{x^{(1)}_{i_1}-x^{(1)}_{i_2}},&\quad
    \frac{1}{x^{(2)}_{j_1}-x^{(2)}_{j_2}}\mapsto \frac{1}{x^{(2)}_{j_1}-x^{(2)}_{j_2}},\\
    \frac{1}{x^{(1)}_i-x^{(2)}_j}\mapsto & -\sum_{n=0}^{\infty} w^{-n-1}(x^{(1)}_i-x^{(2)}_j)^n.
    \end{split}
\end{equation}
The meromorphic coproduct can be defined for differential operators as well, i.e. there exists 
\begin{align*}
    \Delta(w)_{N_1,N_2}:D(\mathbb C^{\times(N_1+N_2)}_{\mathrm{disj}})\otimes \mathfrak{gl}_K^{\otimes N_1+N_2}\to D(\mathbb C^{\times N_1}_{\mathrm{disj}})\otimes \mathfrak{gl}_K^{\otimes N_1}\otimes D(\mathbb C^{N_2}_{\mathrm{disj}})\otimes \mathfrak{gl}_K^{\otimes N_2}(\!(w^{-1})\!).
\end{align*}
Restricted to the image of spherical Cherednik algebras via Dunkl embeddings, we get an algebra homomorphism:
\begin{align}
    \Delta(w)_{N_1,N_2}: \mathrm{S}\mathbb H^{(K)}_{N_1+N_2}\to  \mathrm{S}\mathbb H^{(K)}_{N_1}\otimes  \mathrm{S}\mathcal H^{(K)}_{N_2}(\!(w^{-1})\!).
\end{align}
By taking the uniform-in-$N$ map, we obtain the following result.
\begin{proposition}
There is an algebra homomorphism $\Delta_{\mathsf L}(w): \mathsf L^{(K)}\to \mathsf L^{(K)}\otimes\mathsf A^{(K)}(\!(w^{-1})\!)$ which map the generators as
\begin{equation}\label{eqn: LL meromorphic coproduct}
\boxed{
\begin{aligned}
&\Delta_{\mathsf L}(w)(\mathsf T_{0,n}(E^a_b))=\mathsf T_{0,n}(E^a_b)\otimes 1+\sum_{m=0}^{\infty} \binom{n}{m} w^{n-m} 1\otimes\mathsf T_{0,m}(E^a_b),\quad (n\in \mathbb Z) \\
&\Delta_{\mathsf L}(w)(\mathsf T_{1,0}(E^a_b))=\square(\mathsf T_{1,0}(E^a_b))\\
&~+\epsilon_1\sum_{m,n\ge 0}\frac{(-1)^m}{w^{n+m+1}}\binom{m+n}{n} (\mathsf T_{0,n}(E^c_b)\otimes \mathsf T_{0,m}(E^a_c)-\mathsf T_{0,n}(E^a_c)\otimes \mathsf T_{0,m}(E^c_b)),\\
&\Delta_{\mathsf L}(w)(\mathsf t_{2,0})=\square(\mathsf t_{2,0})-2\epsilon_1\sum_{m,n\ge 0}\frac{(m+n+1)!}{m!n!w^{n+m+2}}(-1)^m (\mathsf T_{0,n}(E^a_b)\otimes \mathsf T_{0,m}(E^b_a)+\epsilon_1\epsilon_2 \mathsf t_{0,n}\otimes \mathsf t_{0,m}),
\end{aligned}}
\end{equation}
where $\square(X)=X\otimes 1+1\otimes X$. 
\end{proposition}
We shall call $\Delta_{\mathsf L}(w)$ the meromorphic coproduct on $\mathsf L^{(K)}$. It is obvious from the formulae \eqref{eqn: LL meromorphic coproduct} that the restriction of $\Delta_{\mathsf L}(w)$ to $\mathsf A^{(K)}$ equals to the meromorphic coproduct $\Delta_{\mathsf A}(w)$ defined in \eqref{eqn: AA coproduct}.

The locality of the meromorphic coproduct $\Delta_{\mathsf L}(w)$ is captured by the notion of a vertex comodule, which is recalled in Appendix \ref{sec: vertex coalgebras and vertex comodules}.

\begin{theorem}\label{thm: loop Yangian is a vertex comodule}
The pair $(\mathsf{L}^{(K)}, \Delta_{\mathsf L}(w))$ is a vertex comodule of the vertex coalgebra $(\mathsf A^{(K)},\Delta_{\mathsf A}(w),\mathfrak{C}_{\mathsf A})$ over the base ring $\mathbb C[\epsilon_1,\epsilon_2]$.
\end{theorem}

\begin{proof}
The counit axiom is easily checked for generators using \eqref{eqn: LL meromorphic coproduct}. The coassociativity axiom can be checked by direct computation for the generators $\mathsf T_{0,n}(X),\mathsf t_{0,n}$ and $\mathsf t_{2,0}$: $(\mathrm{id}\otimes\Delta_{\mathsf A}(z))\circ \Delta_{\mathsf L}(w)(\mathsf T_{0,n}(X))$ and $(\Delta_{\mathsf L}(w)\otimes \mathrm{id})\circ \Delta_{\mathsf L}(z+w)(\mathsf T_{0,n}(X))$ are the expansions of the same element
\begin{equation*}
\begin{split}
\mathsf T_{0,n}(X)\otimes 1\otimes 1+\sum_{m=0}^{\infty} \binom{n}{m} w^{n-m} 1\otimes\mathsf T_{0,m}(X)\otimes 1+\sum_{m=0}^{\infty} \binom{n}{m} (w+z)^{n-m} 1\otimes 1\otimes\mathsf T_{0,m}(X),
\end{split}
\end{equation*}
a similar equation holds for $\mathsf t_{0,n}$, and $(\mathrm{id}\otimes\Delta_{\mathsf A}(z))\circ \Delta_{\mathsf L}(w)(\mathsf t_{2,0})$ and $(\Delta_{\mathsf L}(w)\otimes \mathrm{id})\circ \Delta_{\mathsf L}(z+w)(\mathsf t_{2,0})$ are the expansions of the same element
\begin{equation*}
\begin{split}
\square(\mathsf t_{2,0})&-2\epsilon_1\sum_{m,n\ge 0}\frac{(m+n+1)!}{m!n!w^{n+m+2}}(-1)^m (\mathsf T_{0,n}(E^a_b)\otimes \mathsf T_{0,m}(E^b_a)\otimes 1+\epsilon_1\epsilon_2 \mathsf t_{0,n}\otimes \mathsf t_{0,m}\otimes 1)\\
&-2\epsilon_1\sum_{m,n\ge 0}\frac{(m+n+1)!}{m!n!z^{n+m+2}}(-1)^m (1\otimes\mathsf T_{0,n}(E^a_b)\otimes \mathsf T_{0,m}(E^b_a)+\epsilon_1\epsilon_2 1\otimes\mathsf t_{0,n}\otimes \mathsf t_{0,m})\\
&-2\epsilon_1\sum_{m,n\ge 0}\frac{(m+n+1)!}{m!n!(w+z)^{n+m+2}}(-1)^m (\mathsf T_{0,n}(E^a_b)\otimes 1\otimes \mathsf T_{0,m}(E^b_a)+\epsilon_1\epsilon_2 \mathsf t_{0,n}\otimes 1\otimes \mathsf t_{0,m}),
\end{split}
\end{equation*}
where $\square(x)=x\otimes 1\otimes 1+1\otimes x\otimes 1+1\otimes 1\otimes x$. Since both $(\mathrm{id}\otimes\Delta_{\mathsf A}(z))\circ \Delta_{\mathsf L}(w)$ and $(\Delta_{\mathsf L}(w)\otimes \mathrm{id})\circ \Delta_{\mathsf L}(z+w)$ are algebra homomorphisms and they are equal on the generators, thus they are equal on the whole algebra $\mathsf L^{(K)}$.
\end{proof}

\subsection{Coproduct of \texorpdfstring{$\mathsf L^{(K)}$}{Lk}}
Consider the natural embedding $1\otimes S(w): \mathsf L^{(K)}\otimes\mathsf L^{(K)}\hookrightarrow \mathsf L^{(K)}\otimes \mathsf A^{(K)}(\!(w^{-1})\!)$. We endow the $\mathbb C[\epsilon_1,\epsilon_2]$-algebras $\mathsf L^{(K)}$ and $\mathsf A^{(K)}$ the natural $\mathbb Z$-gradings induced by the $\mathrm{ad}_{\mathsf t_{1,1}}$ actions, in other words
\begin{align}
    \deg(\mathsf T_{n,m}(E^a_b))=\deg (\mathsf t_{n,m})=m-n,
\end{align}
and we let the degree of the formal parameter $w$ to be $1$, then it is easy to see that $1\otimes S(w)$ preserves the $\mathbb Z$-gradings. Moreover the map $1\otimes S(w)$ extends naturally to the completion (Definition \ref{def: completed tensor product}):
\begin{align}
    1\otimes S(w): \mathsf L^{(K)}\widetilde\otimes \mathsf L^{(K)}\hookrightarrow \mathsf L^{(K)}\otimes \mathsf A^{(K)}(\!(w^{-1})\!).
\end{align}
\begin{theorem}\label{thm: LL coproduct}
There is a $\mathbb C[\epsilon_1,\epsilon_2]$-algebra homomorphism $\Delta_{\mathsf L}:\mathsf L^{(K)}\to \mathsf L^{(K)}\widetilde\otimes \mathsf L^{(K)}$ which acts on generators as
\begin{equation}\label{eqn: LL coproduct}
\boxed{
\begin{aligned}
&\Delta_{\mathsf L}(\mathsf T_{0,n}(E^a_b))=\square(\mathsf T_{0,n}(E^a_b)),\quad (n\in \mathbb Z) \\
&\Delta_{\mathsf L}(\mathsf T_{1,0}(E^a_b))=\square(\mathsf T_{1,0}(E^a_b))+\epsilon_1\sum_{n= 0}^{\infty}(\mathsf T_{0,n}(E^c_b)\otimes \mathsf T_{0,-n-1}(E^a_c)-\mathsf T_{0,n}(E^a_c)\otimes \mathsf T_{0,-n-1}(E^c_b)),\\
&\Delta_{\mathsf L}(\mathsf t_{2,0})=\square(\mathsf t_{2,0})-2\epsilon_1\sum_{n= 0}^{\infty}(n+1)(\mathsf T_{0,n}(E^a_b)\otimes \mathsf T_{0,-n-2}(E^b_a)+\epsilon_1\epsilon_2 \mathsf t_{0,n}\otimes \mathsf t_{0,-n-2}),
    \end{aligned}}
\end{equation}
where $\square(X)=X\otimes 1+1\otimes X$. $\Delta_{\mathsf L}$ is counital: $(\mathfrak C_{\mathsf L}\otimes 1)\circ \Delta_{\mathsf L}=(1\otimes\mathfrak C_{\mathsf L})\circ \Delta_{\mathsf L}=\mathrm{id}$. Moreover $\Delta_{\mathsf L}$ is coassociative, i.e. the image of $(\Delta_{\mathsf L}\otimes 1)\circ \Delta_{\mathsf L}$ is contained in the intersection between $(\mathsf L^{(K)}\widetilde\otimes \mathsf L^{(K)})\widetilde\otimes \mathsf L^{(K)}$ and $\mathsf L^{(K)}\widetilde\otimes \mathsf L^{(K)}\widetilde\otimes \mathsf L^{(K)}$ in the sense of Lemma \ref{lem: compare different completions} and the image of $(1\otimes \Delta_{\mathsf L})\circ \Delta_{\mathsf L}$ is contained in the intersection between $\mathsf L^{(K)}\widetilde\otimes (\mathsf L^{(K)}\widetilde\otimes \mathsf L^{(K)})$ and $\mathsf L^{(K)}\widetilde\otimes \mathsf L^{(K)}\widetilde\otimes \mathsf L^{(K)}$ in the sense of Lemma \ref{lem: compare different completions}, and $(\Delta_{\mathsf L}\otimes 1)\circ \Delta_{\mathsf L}=(1\otimes \Delta_{\mathsf L})\circ \Delta_{\mathsf L}$. 
\end{theorem}

\begin{proof}
Notice that $(1\otimes S(w))\circ \Delta_{\mathsf L}$ agrees with $\Delta_{\mathsf L}(w)$ on the generators $\mathsf T_{0,n}(E^a_b)$ and $\mathsf t_{2,0}$, and that $\Delta_{\mathsf L}(w)$ is an algebra homomorphism, and that $1\otimes S(w)$ is injective, thus $\Delta_{\mathsf L}$ is an algebra homomorphism as well. The coassociativity and counity are checked using the formula \eqref{eqn: LL coproduct} and we omit the detail.
\end{proof}

The finite-$N$ counterpart of $\Delta_{\mathsf L}$ can be constructed in a similar way. The spherical $\gl_K$-extended trigonometric Cherednik algebra $\mathrm{S}\mathbb H^{(K)}_{N}$ has a natural grading such that $\deg x_i=1,\deg y_i=-1$, then the map $S_N(w):\mathrm{S}\mathbb H^{(K)}_{N}\hookrightarrow \mathrm{S}\mathcal H^{(K)}_{N}(\!(w^{-1})\!)$ preserves the grading, therefore we have natural embedding 
\begin{align}
    1\otimes S_N(w):  \mathrm{S}\mathbb H^{(K)}_{N_1}\widetilde\otimes  \mathrm{S}\mathbb H^{(K)}_{N_2}\hookrightarrow  \mathrm{S}\mathbb H^{(K)}_{N_1}\otimes  \mathrm{S}\mathcal H^{(K)}_{N}(\!(w^{-1})\!).
\end{align}
It is easy to see from the formula \eqref{eqn: AA coproduct_finite N} that the image of $\Delta(w)_{N_1,N_2}$ is contained in the image of $1\otimes S_N(w)$, thus we obtain an algebra homomorphism
\begin{align}
    \Delta_{N_1,N_2}: \mathrm{S}\mathbb H^{(K)}_{N_1+N_2}\to  \mathrm{S}\mathbb H^{(K)}_{N_1}\widetilde\otimes  \mathrm{S}\mathbb H^{(K)}_{N_2}.
\end{align}

\begin{proposition}\label{prop: L coproduct compatible with DAHA coproduct}
$\Delta_{\mathsf L}$ is compatible with $\Delta_{N_1,N_2}$ in the sense that 
\begin{equation}
   (\rho_{N_1}\otimes\rho_{N_2})\circ \Delta_{\mathsf L} =\Delta_{N_1,N_2}\circ \rho_{N_1+N_2}
\end{equation}
\end{proposition}

\begin{proof}
This is clear from the construction of $\Delta_{\mathsf L}$, $\Delta_{N_1,N_2}$.
\end{proof}

Notice that the involution $\iota: \mathrm{S}\mathbb H^{(K)}_{N}\cong \mathrm{S}\mathbb H^{(K)}_{N}$ intertwines between $\Delta_{N_1,N_2}$ and $\Delta_{N_2,N_1}$, i.e. 
\begin{equation}
    P\circ(\iota\otimes\iota)\circ \Delta_{N_1,N_2}= \Delta_{N_2,N_1}\circ \iota,
\end{equation}
where $P$ is the operator swapping two tensor components. Taking the uniform-in-$N$ version of the above commutative diagram, we see that the involution $\iota: \mathsf L^{(K)}\cong \mathsf L^{(K)}$ intertwines between $\Delta_{\mathsf L}$ and $\Delta_{\mathsf L}^{\mathrm{op}}$, i.e. 
\begin{equation}
    P\circ(\iota\otimes\iota)\circ \Delta_{\mathsf L}= \Delta_{\mathsf L}\circ \iota.
\end{equation}



\section{The Algebra \texorpdfstring{$\mathsf Y^{(K)}$}{Yk} and Coproducts}\label{sec: Y^K and coproducts}
\begin{definition}\label{def: the algebra Y^(K)}
$\mathsf Y^{(K)}$ is the $\mathbb C[\epsilon_1,\epsilon_2,\mathbf c]$-algebra generated by $\mathsf A^{(K)}_+$ and $\mathsf A^{(K)}_-$, both are isomorphic to $\mathsf A^{(K)}$ (whose generators are indicated by superscripts $+$ or $-$ respectively), with relations
\begin{equation}\label{eqn: gluing relations of Y}
\begin{split}
\mathsf t^-_{0,0}= \mathsf t^+_{0,0},\quad \mathsf t^-_{1,2}= -\mathsf t^+_{1,0}-\epsilon_1\epsilon_2\epsilon_3\mathsf t^-_{0,1}\mathbf c,&\quad \mathsf t^-_{1,0}= -\mathsf t^+_{1,2}+\epsilon_1\epsilon_2\epsilon_3\mathsf t^+_{0,1}\mathbf c,\quad \mathsf t^-_{2,2}=\mathsf t^+_{2,2},\\
\mathsf T^-_{0,0}(X)= \mathsf T^+_{0,0}(X),&\quad \mathsf T^-_{1,1}(X)= -\mathsf T^+_{1,1}(X),\\
[\mathsf t^-_{0,1}, \mathsf t^+_{0,1}]=K\mathbf c,&\quad [\mathsf t^-_{0,1}, \mathsf T^+_{0,1}(X)]=0,
\end{split}
\end{equation}
for all $X\in \mathfrak{sl}_K$. We define $\mathfrak{Y}^{(K)}$ to be the $\mathbb C[\epsilon_1,\epsilon_2,\mathbf c]$-subalgebra of $\mathsf Y^{(K)}$ generated by $\mathsf D^{(K)}_+$ and $\mathsf D^{(K)}_-$,
\end{definition}

The following is a list of immediate consequences derived from the definition of $\mathsf Y^{(K)}$.
\begin{itemize}
    \item[(1)] $\mathsf Y^{(K)}$ is a deformation of $\mathsf L^{(K)}$ over the base ring $\mathbb C[\mathbf c]$ (compare with Theorem \ref{thm: gluing construction of L}), i.e. $\mathsf Y^{(K)}/(\mathbf c=0)\cong \mathsf L^{(K)}$.
    \item[(2)] The duality automorphisms $\sigma: \mathsf A^{(K)}_\pm\cong \mathsf A^{(K)}_\pm$ (Proposition \ref{prop: duality for A}) glue to an automorphism $\sigma: \mathsf Y^{(K)}\cong \mathsf Y^{(K)}$.
    \item[(3)] There is a natural $\mathbb C[\epsilon_1,\epsilon_2]$-algebra involution $\iota: \mathsf Y^{(K)}\cong \mathsf Y^{(K)}$ such that
    \begin{align}
        \iota(\mathsf t^\pm_{n,m})=\mathsf t^\mp_{n,m},\quad \iota(\mathsf T^\pm_{n,m}(X))=\mathsf T^\mp_{n,m}(X),\quad\iota(\mathbf c)=-\mathbf c,
    \end{align}
    for all $X\in \gl_K$.
\end{itemize}

Similar to the discussions in the proof of Theorem \ref{thm: gluing construction of L}, we can re-define the generators of $\mathsf Y^{(K)}$. Namely, define $\mathbf T_{n,m}(X):=\mathbf T^+_{n,m}(X)$ and $\mathbf t_{n,m}:=\mathbf t^+_{n,m}$ (and also $\mathsf T_{n,m}(X):=\mathsf T^+_{n,m}(X)$ and $\mathsf t_{n,m}:=\mathsf t^+_{n,m}$) for $(n,m)\in \mathbb N^2$ and $X\in \mathfrak{gl}_K$; then define $\mathbf T_{n,m}(X):=(-1)^n\mathbf T^-_{n,2n-m}(X)$ and $\mathbf t_{n,m}:=(-1)^n\mathbf t^-_{n,2n-m}$, for $(n,m)\in \mathbb N\times \mathbb Z_{<0}$ and $X\in \mathfrak{gl}_K$; and also define $\mathsf T_{0,n}(X):=\mathbf T_{0,n}(X)$ and $\mathsf t_{0,n}:=\mathbf t_{0,n}$ for all $n\in \mathbb Z_{<0}$ and $X\in \gl_K$. Then the subalgebra $\mathfrak{Y}^{(K)}$ is generated over $\mathbb C[\epsilon_1,\epsilon_2,\mathbf c]$ by $\{\mathbf T_{n,m}(X)\:|\: X\in\gl_K,(n,m)\in\mathbb N\times\mathbb Z\}$.

\begin{lemma}\label{lem: A_+ and t[0,-1] generates Y}
$\mathsf Y^{(K)}$ is generated over $\mathbb C[\epsilon_1,\epsilon_2]$ by $\mathsf A^{(K)}_+$ and $\mathsf t_{0,-1}$.
\end{lemma}
\begin{proof}
$\mathsf A^{(K)}_-$ is generated over $\mathbb C[\epsilon_1,\epsilon_2]$ by $\mathsf t^-_{0,0},\mathsf t^-_{0,1},\mathsf t^-_{1,0},\mathsf t^-_{1,2},\mathsf t^-_{2,2}$ and $\mathsf T^-_{0,0}(X),\mathsf T^-_{1,1}(X)$ for $X\in \mathfrak{sl}_K$. The gluing relations \eqref{eqn: gluing relations of Y} relates the above set of generators except $\mathsf t^-_{0,1}$ to linear combinations of elements of $\mathsf A^{(K)}_+$ and $\mathsf t^-_{0,1}$ and $\mathbf c$. Since $\mathbf c$ can be generated from the commutator between $\mathsf t^\pm _{0,1}$, this proves the lemma.
\end{proof}

Similar inductive argument to that of the proof of Theorem \ref{thm: gluing construction of L} shows that $\mathsf T_{0,n}(X),\mathsf t_{0,n}$ satisfy the affine Lie algebra commutation relations\footnote{We illustrate the derivation of \eqref{eqn: affine Lie algebra} by one example. Using the identity $[\mathsf T^-_{1,1}(X),\mathsf t^-_{0,1}]=\mathsf T^-_{0,1}(X)$, we get $[\mathsf T_{0,-1}(X),\mathsf T_{0,1}(Y)]=[\mathsf t_{0,-1},[\mathsf T_{1,1}(X),\mathsf T_{0,1}(Y)]]-[\mathsf T_{1,1}(X),[\mathsf t_{0,-1},\mathsf T_{0,1}(Y)]]$, which equals to $[\mathsf t_{0,-1},\mathsf T_{1,2}([X,Y])-\frac{\epsilon_3}{2}\mathsf T_{0,1}(\{X,Y\})]=\mathsf T_{0,0}([X,Y])-\epsilon_2\epsilon_3\mathrm{tr}(XY)\mathbf c$.} 
\begin{equation}\label{eqn: affine Lie algebra}
\begin{split}
[\mathsf T_{0,n}(X),\mathsf T_{0,m}(Y)]&=\mathsf T_{0,m+n}([X,Y]])+n\epsilon_2\delta_{n,-m}\kappa_{\epsilon_3,\epsilon_1}(X,Y)\mathbf c,\\
[\mathsf t_{0,n},\mathsf t_{0,m}]&=-nK\delta_{n,-m}\mathbf c,
\end{split}
\end{equation}
where $\kappa_{\epsilon_3,\epsilon_1}$ is the inner product $\kappa_{\epsilon_3,\epsilon_1}(E^a_b,E^c_d)=\epsilon_3\delta^a_d\delta^c_b+\epsilon_1\delta^a_b\delta^c_d$.

\subsection{Map \texorpdfstring{$\mathsf Y^{(K)}$}{Yk} to the mode algebra of \texorpdfstring{$\mathcal W^{(K)}_{\infty}$}{Wk inf}}

\begin{proposition}\label{prop: map Y^K to modes of W(infinity)}
The map $\Psi_{\infty}: \mathsf A^{(K)}\to \mathfrak U(\mathsf W^{(K)}_{\infty})[\epsilon_2^{-1}]$ naturally extends to a $\mathbb C[\epsilon_1,\epsilon_2]$ algebra homomorphism $\Psi_{\infty}: \mathsf Y^{(K)}\to \mathfrak U(\mathsf W^{(K)}_{\infty})[\epsilon_2^{-1}]$ such that
\begin{align*}
    \Psi_{\infty}(\mathsf T_{0,-n}(E^a_b))=\mathds W^{a(1)}_{b,-n},\quad \Psi_{\infty}(\mathsf t_{0,-n})=\frac{1}{\epsilon_2}\mathds W^{a(1)}_{a,-n},\quad \Psi_{\infty}(\mathbf c)=\frac{\mathsf c}{\epsilon_2},
\end{align*}
for all $n\in \mathbb N$.
\end{proposition}

\begin{proof}
Consider the algebra homomorphism $\Psi^-_{\infty}:\mathsf A^{(K)}\to \mathfrak U(\mathsf W^{(K)}_{\infty})[\epsilon_2^{-1}]$ in Definition \ref{def: Psi_infinity minus}. Comparing \eqref{eqn: Psi_infinity minus on generators} with the Definition \ref{def: the algebra Y^(K)}, it is straightforward to see that $\Psi^-_{\infty}:\mathsf A^{(K)}_-\to \mathfrak U(\mathsf W^{(K)}_{\infty})[\epsilon_2^{-1}]$ and $\Psi_{\infty}:\mathsf A^{(K)}_+\to \mathfrak U(\mathsf W^{(K)}_{\infty})[\epsilon_2^{-1}]$ glue via the relations \eqref{eqn: gluing relations of Y} to an algebra homomorphism $\Psi_{\infty}: \mathsf Y^{(K)}\to \mathfrak U(\mathsf W^{(K)}_{\infty})[\epsilon_2^{-1}]$, such that \begin{align*}
    \Psi_{\infty}(\mathsf T_{0,-n}(E^a_b))=\mathds W^{a(1)}_{b,-n},\quad \Psi_{\infty}(\mathsf t_{0,-n})=\frac{1}{\epsilon_2}\mathds W^{a(1)}_{b,-n},\quad \Psi_{\infty}(\mathbf c)=\frac{\mathsf c}{\epsilon_2},
\end{align*}
for $n\in \mathbb N$.
\end{proof}

\begin{proposition}\label{prop: Psi_inf has dense image}
$\Psi_\infty(\mathsf Y^{(K)})$ is dense in $\mathfrak U(\mathsf W^{(K)}_{\infty})[\epsilon_2^{-1}]$. $\Psi_\infty(\mathfrak Y^{(K)})$ is dense in $\mathfrak U(\mathsf W^{(K)}_{\infty})$.
\end{proposition}

\begin{proof}
By the definition of $\Psi_\infty$, $\Psi_\infty(\mathsf Y^{(K)})\cap V_0\mathfrak U(\mathsf W^{(K)}_{\infty})[\epsilon_2^{-1}]$ is dense in $V_0\mathfrak U(\mathsf W^{(K)}_{\infty})[\epsilon_2^{-1}]$, where $V_\bullet \mathfrak U(\mathsf W^{(K)}_{\infty})$ is the vertical filtration defined in Section \ref{subsec: vertical filtration on W_inf}. Suppose that $\Psi_\infty(\mathsf Y^{(K)})\cap V_{n-1}\mathfrak U(\mathsf W^{(K)}_{\infty})[\epsilon_2^{-1}]$ is dense in $V_{n-1}\mathfrak U(\mathsf W^{(K)}_{\infty})[\epsilon_2^{-1}]$, for some $n\in \mathbb Z_{>0}$, then we claim that $\mathds W^{a(n+1)}_{b,m}$ is in the closure of $\Psi_\infty(\mathsf Y^{(K)})$, for all $1\le a,b\le K$ and all $m\in \mathbb Z$. In fact, $\Psi_\infty(\mathsf T_{n,s}(E^a_b))\equiv (-1)^n\mathds W^{a(n+1)}_{b,s-n}\pmod{V_{n-1}}\mathfrak U(\mathsf W^{(K)}_{\infty})$ by Proposition \ref{prop: Psi_inf filtered}, then by the induction hypothesis $\mathds W^{a(n+1)}_{b,m}$ is in the closure of $\Psi_\infty(\mathsf Y^{(K)})$ for all $1\le a,b\le K$ and all $m>-n$. Since $\Psi_\infty(\mathsf Y^{(K)})$ is closed under the involution $\mathfrak s_{\infty}$, $\mathfrak s_\infty(\mathds W^{a(n+1)}_{b,m})$ is also in the closure of $\Psi_\infty(\mathsf Y^{(K)})$. By Lemma \ref{lem: s_inf}, $\mathfrak s_\infty(\mathds W^{a(n+1)}_{b,m})\equiv \mathds W^{b(n+1)}_{a,-m}\pmod {V_{n-1}\mathfrak U(\mathsf W^{(K)}_{\infty})}$, thus $\mathds W^{a(n+1)}_{b,m}$ is in the closure of $\Psi_\infty(\mathsf Y^{(K)})$ for all $1\le a,b\le K$ and all $m\in \mathbb Z$ by the induction hypothesis. Therefore $\Psi_\infty(\mathsf Y^{(K)})\cap V_{n}\mathfrak U(\mathsf W^{(K)}_{\infty})[\epsilon_2^{-1}]$ is dense in $V_{n}\mathfrak U(\mathsf W^{(K)}_{\infty})[\epsilon_2^{-1}]$, and by induction on $n$ we see that $\Psi_\infty(\mathsf Y^{(K)})\cap \bigcup_n V_{n}\mathfrak U(\mathsf W^{(K)}_{\infty})[\epsilon_2^{-1}]$ is dense in $\bigcup_n V_{n}\mathfrak U(\mathsf W^{(K)}_{\infty})[\epsilon_2^{-1}]$. We finish the proof by noticing that $\bigcup_n V_{n}\mathfrak U(\mathsf W^{(K)}_{\infty})[\epsilon_2^{-1}]$ is dense in $\mathfrak U(\mathsf W^{(K)}_{\infty})[\epsilon_2^{-1}]$. The second statement is proven similarly.
\end{proof}

\begin{corollary}
If $K>1$, then $\mathfrak U(\mathsf W^{(K)}_{\infty})[\epsilon_3^{-1}]$ is topologically generated by $\{\mathds W^{a(1)}_{b,n},\mathds W^{a(2)}_{b,n}\:|\: 1\le a,b\le K,n\in \mathbb Z\}$. If $K=1$, then $\mathfrak U(\mathsf W^{(1)}_{\infty})[\epsilon_2^{-1}]$ is topologically generated by $\{\mathds W^{(1)}_{n},\mathds W^{(2)}_{n},\mathds W^{(3)}_{n}\:|\: n\in \mathbb Z\}$.
\end{corollary}

\begin{proof}
If $K>1$, then let us denote by $\mathfrak U'$ the subalgebra of $\mathfrak U(\mathsf W^{(K)}_{\infty})[\epsilon_3^{-1}]$ topologically generated by $\{\mathds W^{a(1)}_{b,n},\mathds W^{a(2)}_{b,n}\:|\: 1\le a,b\le K,n\in \mathbb Z\}$ over the base ring $\mathbb C[\epsilon_1,\epsilon_3^{\pm},\mathsf c]$. We note that $\{\Psi_{\infty}(\mathsf T_{n,m}(X))\:|\: 0\le n,m\le 1, X\in \mathfrak{sl}_K\}$ is contained in $\mathfrak U'$, thus $\Psi_\infty(\mathbb D^{(K)})\subset \mathfrak U'$. Since $\epsilon_3$ is invertible in $\mathfrak U'$, then Corollary \ref{cor: compare two DDCAs} implies that $\Psi_\infty(\mathsf D^{(K)})\subset \mathfrak U'$. Moreover, $\mathfrak U'$ is closed under the involution $\mathfrak s_{\infty}$, thus $\mathfrak s_{\infty}\circ\Psi_\infty(\mathsf D^{(K)})\subset \mathfrak U'$. We finish the proof by using Proposition \ref{prop: Psi_inf has dense image}.

If $K=1$, then let us denote by $\mathfrak U''$ the subalgebra of $\mathfrak U(\mathsf W^{(1)}_{\infty})[\epsilon_2^{-1}]$ topologically generated by $\{\mathds W^{(1)}_{n},\mathds W^{(2)}_{n},\mathds W^{(3)}_{n}\:|\: n\in \mathbb Z\}$ over the base ring $\mathbb C[\epsilon_1,\epsilon_2^{\pm},\mathsf c]$. We note that $\{\Psi_\infty(\mathsf t_{2,0}),\Psi_\infty(\mathsf t_{0,n})\:|\: n\in \mathbb N\}$ is contained in $\mathfrak U''$, thus $\Psi_\infty(\mathsf A^{(1)})\subset \mathfrak U''$. Moreover, $\mathfrak U''$ is closed under the involution $\mathfrak s_{\infty}$, thus $\mathfrak s_{\infty}\circ\Psi_\infty(\mathsf A^{(1)})\subset \mathfrak U''$. We finish the proof by using Proposition \ref{prop: Psi_inf has dense image}.
\end{proof}

\begin{theorem}\label{thm: Psi_infty is injective}
The map $\Psi_{\infty}: \mathsf Y^{(K)}\to \mathfrak U(\mathsf W^{(K)}_{\infty})[\epsilon_2^{-1}]$ in the Proposition \ref{prop: map Y^K to modes of W(infinity)} is injective.
\end{theorem}

The above theorem will be proved together with the PBW theorem for $\mathsf Y^{(K)}$ in the next subsection. We note that the map $\widetilde\Psi_{\infty}: \mathsf A^{(K)}\to \mathfrak U(\widetilde{\mathsf W}^{(K)}_{\infty})[\epsilon_3^{-1}]$ also extends to a $\mathbb C[\epsilon_1,\epsilon_2]$ algebra homomorphism $\widetilde\Psi_{\infty}: \mathsf Y^{(K)}\to \mathfrak U(\widetilde{\mathsf W}^{(K)}_{\infty})[\epsilon_3^{-1}]$ such that
\begin{align*}
    \widetilde\Psi_{\infty}(\mathsf T_{0,-n}(E^a_b))=\mathds U^{a(1)}_{b,-n}+\frac{\epsilon_1}{\epsilon_3}\delta^a_b \mathds U^{c(1)}_{c,-n},\quad \widetilde\Psi_{\infty}(\mathsf t_{0,-n})=-\frac{1}{\epsilon_3}\mathds U^{a(1)}_{a,-n},\quad \widetilde\Psi_{\infty}(\mathbf c)=\frac{\mathsf c}{\epsilon_3},
\end{align*}
for all $n\in \mathbb N$. This is because the duality automorphism \eqref{eqn: duality for A} $\sigma:\mathsf A^{(K)}\cong \mathsf A^{(K)}$ extends to a duality automorphism $\sigma:\mathsf Y^{(K)}\cong \mathsf Y^{(K)}$, and we set 
\begin{align*}
    \widetilde\Psi_{\infty}=\sigma_{\infty}\circ \Psi_{\infty}\circ \sigma,
\end{align*}
where $\sigma_{\infty}:{\mathsf W}^{(K)}_{\infty}\cong \widetilde{\mathsf W}^{(K)}_{\infty}$ is the duality automorphism for the rectangular $\mathcal W_{\infty}$-algebra defined in Section \ref{subsec: duality for W}. By Theorem \ref{thm: Psi_infty is injective}, $\widetilde\Psi_{\infty}$ is also injective.

\subsection{PBW theorem for \texorpdfstring{$\mathsf Y^{(K)}$}{Yk}}
In this subsection, we extend the PBW theorem for $\mathsf L^{(K)}$ (Theorem \ref{thm: PBW for L}) to its deformation $\mathsf Y^{(K)}$. We will show that the ordered-monomial basis $\mathfrak B(\mathsf L^{(K)})$ extends to a basis of $\mathsf Y^{(K)}$ over the ring $\mathbb C[\epsilon_1,\epsilon_2,\mathbf c]$. Along the way we also prove Theorem \ref{thm: Psi_infty is injective}.

We begin with noticing that the shifted vertical filtration $\tilde V_{\bullet}\mathsf A^{(K)}_\pm$ naturally extends to a filtration $\tilde V_{\bullet}\mathsf Y^{(K)}$ by letting $\deg_{\tilde v}\epsilon_1=\deg_{\tilde v}\epsilon_2=\deg_{\tilde v}\mathbf c=0$ and $\deg_{\tilde v}\mathbf T_{n,m}(X)=\deg_{\tilde v}\mathbf t_{n,m}=n+1$. 

\begin{lemma}\label{lem: glue filtrations}
For all $(n,m)\in \mathbb N^2$ and $X\in \gl_K$, we have $\mathbf T^-_{n,m}(X)\equiv (-1)^n\mathbf T_{n,2n-m}(X)\pmod{\tilde V_n\mathfrak Y^{(K)}}$ and $\mathbf t^-_{n,m}\equiv (-1)^n\mathbf t_{n,2n-m}\pmod{\tilde V_n\mathsf Y^{(K)}}$.
\end{lemma}

\begin{proof}
For $m>2n$, there is nothing to prove because they are equal by definition. For $m\le 2n$, we proceed by induction as follows. As the first step, we claim that
\begin{align*}
    \mathbf T^-_{n,0}(X)\equiv (-1)^{n}\mathbf T_{n,2n}(X)\pmod{\tilde V_{n}\mathsf D^{(K)}_+},\quad \mathbf t^-_{n,0}\equiv (-1)^{n}\mathbf t_{n,2n}\pmod{\tilde V_{n}\mathsf A^{(K)}_+}.
\end{align*}
By the gluing relations \eqref{eqn: gluing relations of Y}, the lemma automatically holds for $(n,m)=(0,0)$; moreover $\mathbf t^-_{1,0}\equiv -\mathbf t_{1,2}\pmod{\tilde V_1\mathsf A^{(K)}_+}$ and $\mathbf T^-_{1,1}(X)\equiv -\mathbf T_{1,1}(X)\pmod{\tilde V_1\mathsf D^{(K)}_+}$. Then 
\begin{align*}
    \mathbf T^-_{1,0}(X)=[\mathbf t^-_{1,0},\mathbf T^-_{1,1}(X)]\equiv -\mathbf T_{1,2}(X)\pmod{\tilde V_1\mathsf D^{(K)}_+}.
\end{align*}
Similarly 
\begin{align*}
    \mathbf t^-_{2,1}=2[\mathbf t^-_{1,0},\mathbf t^-_{2,2}]\equiv -\mathbf t_{2,3}\pmod{\tilde V_1\mathsf A^{(K)}_+}.
\end{align*}
It follows that
\begin{align*}
    \mathbf T^-_{n,0}(X)=\frac{(-1)^{n-1}}{(n-1)!}\mathrm{ad}_{\mathbf t^-_{2,1}}^{n-1}(\mathbf T^-_{1,0}(X))\equiv (-1)^{n}\mathbf T_{n,2n}(X)\pmod{\tilde V_{n}\mathsf D^{(K)}_+}
\end{align*}
for all $n\in \mathbb N$. Similarly
\begin{align*}
    \mathbf t^-_{n,0}\equiv (-1)^{n}\mathbf t_{n,2n}\pmod{\tilde V_{n}\mathsf A^{(K)}_+}
\end{align*}
for all $n\in \mathbb N$. This proves the claim.

Let us fix $n\ge 1$ and $0\le m<2n$, and suppose that the lemma holds for all $(k,r)$ such that $k<n$, or $k=n$ and $r\le m$. Then 
\begin{align*}
\mathbf T^-_{n,m+1}(X)&=\frac{1}{m-2n}[\mathbf t^-_{1,2},\mathbf T^-_{n,m}(X)]=\frac{1}{2n-m}[\mathbf t^+_{1,0},\mathbf T^-_{n,m}(X)]+\frac{\epsilon_1\epsilon_2\epsilon_3\mathbf c}{2n-m}[\mathbf t^-_{0,1},\mathbf T^-_{n,m}(X)]\\ &\equiv (-1)^n\mathbf T_{n,2n-m-1}(X)\pmod{\tilde V_{n}\mathsf D^{(K)}_+ + \tilde V_{n}\mathsf D^{(K)}_-}.
\end{align*}
Since $\tilde V_{n}\mathsf A^{(K)}_-$ is spanned by monomials in $\mathbf T^-_{k,r}(X),\mathbf t^-_{k,r}$ for $k<n$, so $\tilde V_{n}\mathsf A^{(K)}_-\subset \tilde V_{n}\mathsf Y^{(K)}$ and $\tilde V_{n}\mathsf A=D^{(K)}_-\subset \tilde V_{n}\mathfrak Y^{(K)}$ by our induction assumption, whence $\mathbf T^-_{n,m+1}(X)\equiv (-1)^n\mathbf T_{n,2n-m-1}(X)\pmod{\tilde V_{n}\mathfrak Y^{(K)}}$. Similarly $\mathbf t^-_{n,m+1}\equiv (-1)^n\mathbf t_{n,2n-m-1}\pmod{\tilde V_{n}\mathsf Y^{(K)}}$. Therefore the lemma holds for $(n,m+1)$. The it follows automatically from induction that the lemma holds for all $(n,m)$.
\end{proof}

We have the following generalization of the Proposition \ref{prop: filtration} to $\mathsf Y^{(K)}$.

\begin{proposition}\label{prop: filtration on Y}
The commutators between generators of $\mathsf Y^{(K)}$ can be schematically written as
\begin{equation}\label{eqn: schematic commutators in Y}
\begin{split}
&[\mathbf T_{n,m}(X),\mathbf T_{p,q}(Y)]=\mathbf T_{n+p,m+q}([X,Y])\pmod{\tilde V_{n+p}\mathfrak Y^{(K)}},\\
&[\mathbf t_{n,m},\mathbf T_{p,q}(X)]=(nq-mp)\mathbf T_{n+p-1,m+q-1}(X)\pmod{\tilde V_{n+p-1}\mathfrak Y^{(K)}},\\
&[\mathbf t_{n,m},\mathbf t_{p,q}]=(nq-mp)\mathbf t_{n+p-1,m+q-1}\pmod{\tilde V_{n+p-1}\mathsf Y^{(K)}},
\end{split}
\end{equation}
for all $(n,m,p,q)\in \mathbb N\times \mathbb Z\times \mathbb N\times \mathbb Z$, and all $X,Y\in \mathfrak{gl}_K$.
\end{proposition}

\begin{proof}
For $(n,m,p,q)\in\mathbb N^4$ or $(n,m,p,q)\in\mathbb N\times \mathbb Z_{<0}\times \mathbb N\times \mathbb Z_{<0}$, \eqref{eqn: schematic commutators in Y} follows from Proposition \ref{prop: filtration_vertical and horizontal} together with Lemma \ref{lem: glue filtrations}. 

Next, the equation \eqref{eqn: schematic commutators in Y} in the cases when $(n,m,p,q)\in \mathbb N\times \mathbb Z_{<0}\times \mathbb N\times \{0\}$ follows from the observation that $\mathbf T_{p,0}(X)\equiv (-1)^p\mathbf T^-_{p,2p}(X)\pmod{\tilde V_p\mathsf D^{(K)}_-}$ and $\mathbf t_{p,0}\equiv (-1)^p\mathbf t^-_{p,2p}\pmod{\tilde V_p\mathsf A^{(K)}_-}$. The same observation together with the gluing relation $\mathbf t^-_{1,0}=-\mathbf t^-_{1,2}+\epsilon_1\epsilon_2\epsilon_3\mathbf t^+_{0,1}\mathbf c$ imply that 
\begin{align*}
    \mathbf T_{p,1}(X)-\frac{\epsilon_1\epsilon_2\epsilon_3}{2}\mathbf T_{p-1,0}(X)&\equiv (-1)^p\mathbf T^-_{p,2p-1}(X)\pmod{\tilde V_p\mathsf D^{(K)}_-},\\
    \mathbf t_{p,1}-\frac{\epsilon_1\epsilon_2\epsilon_3}{2}\mathbf t_{p-1,0}&\equiv (-1)^p\mathbf t^-_{p,2p-1}\pmod{\tilde V_p\mathsf A^{(K)}_-}.
\end{align*}
Hence
\begin{align*}
    [\mathbf t_{0,-1},\mathbf T_{p,1}(X)]&\equiv p\mathbf T_{p-1,-1}(X) \pmod{\tilde V_{p-1}\mathsf D^{(K)}_-},\\
     [\mathbf t_{0,-1},\mathbf t_{p,1}]&\equiv p\mathbf t_{p-1,-1}\pmod{\tilde V_{p-1}\mathsf A^{(K)}_-}.
\end{align*}

Next we claim that for all $(n,m)\in \mathbb N\times\mathbb Z$,
\begin{equation}\label{eqn: schematic [T[n,m],t[0,1]]}
\begin{split}
&[\mathbf T_{n,m}(X),\mathbf t_{0,1}]=n\mathbf T_{n-1,m}(X)\pmod{\tilde V_{n-1}\mathfrak Y^{(K)}},\\
&[\mathbf t_{n,m},\mathbf t_{0,1}]=n\mathbf t_{n-1,m}\pmod{\tilde V_{n-1}\mathsf Y^{(K)}}.
\end{split}
\end{equation}
We prove this claim by increasing induction on $n$. The $n=0$ case is implied by \eqref{eqn: affine Lie algebra}. Assume that \eqref{eqn: schematic [T[n,m],t[0,1]]} is true for all $n<n_0$ and all $m\in \mathbb Z$, then 
\begin{align*}
[\mathbf T_{n_0,-1}(X),\mathbf t_{0,1}]&\equiv \frac{1}{n_0+1}[[\mathbf t_{0,-1},\mathbf T_{n_0+1,1}(X)],\mathbf t_{0,1}]=[\mathbf t_{0,-1},\mathbf T_{n_0,1}(X)]\\
&\equiv n_0 \mathbf T_{n_0-1,-1}(X) \pmod{\tilde V_{n_0-1}\mathfrak Y^{(K)}},
\end{align*}
and similarly $[\mathbf t_{n_0,-1},\mathbf t_{0,1}]\equiv n_0 \mathbf t_{n_0-1,-1} \pmod{\tilde V_{n_0-1}\mathsf Y^{(K)}}$, i.e. \eqref{eqn: schematic [T[n,m],t[0,1]]} is true for all $n=n_0$ and $m=-1$. Note that for $m<0$, there are relations
\begin{align*}
    [\mathbf t_{0,-1},\mathbf T_{n_0,m}(X)]&\equiv n_0\mathbf T_{n_0,m-1}(X)\pmod{\tilde V_{n_0-1}\mathsf D^{(K)}_-},\\
    [\mathbf t_{0,-1},\mathbf t_{n_0,m}]&\equiv n_0\mathbf t_{n_0,m-1}\pmod{\tilde V_{n_0-1}\mathsf A^{(K)}_-},
\end{align*}
by the Proposition \ref{prop: filtration_vertical and horizontal}, so we can use the decreasing induction on $m$ staring from $m=-1$ to prove \eqref{eqn: schematic [T[n,m],t[0,1]]} for $n=n_0$ and all $m<0$. The \eqref{eqn: schematic [T[n,m],t[0,1]]} for $m>0$ is covered in Proposition \ref{prop: filtration_vertical and horizontal}, and this finishes the proof of the claim.

Then it can be deduced from \eqref{eqn: schematic [T[n,m],t[0,1]]} that for all $(n,m)\in \mathbb N\times\mathbb Z$,
\begin{equation}\label{eqn: schematic [T[n,m],t[1,2]]}
\begin{split}
&[\mathbf T_{n,m}(X),\mathbf t_{1,2}]=(2n-m)\mathbf T_{n,m+1}(X)\pmod{\tilde V_{n}\mathfrak Y^{(K)}},\\
&[\mathbf t_{n,m},\mathbf t_{1,2}]=(2n-m)\mathbf t_{n,m+1}\pmod{\tilde V_{n}\mathsf Y^{(K)}}.
\end{split}
\end{equation}
The cases when $m\ge 0$ is covered in Proposition \ref{prop: filtration_vertical and horizontal}. For the cases when $m<0$, the gluing relation $\mathbf t^-_{1,0}=-\mathbf t^-_{1,2}+\epsilon_1\epsilon_2\epsilon_3\mathbf t^+_{0,1}\mathbf c$ implies that 
\begin{align*}
[\mathbf T_{n,m}(X),\mathbf t_{1,2}]&=(-1)^{n-1}[\mathbf T^-_{n,2n-m}(X),\mathbf t^-_{1,0}]+\epsilon_1\epsilon_2\epsilon_3\mathbf c[\mathbf T_{n,m}(X),\mathbf t_{0,1}]\\
&\equiv (-1)^{n} (2n-m)\mathbf T^-_{n,2n-m-1}(X)+n\epsilon_1\epsilon_2\epsilon_3\mathbf c \mathbf T_{n-1,m}(X)\pmod{\tilde V_{n}\mathfrak Y^{(K)}}\\
&\equiv (2n-m)\mathbf T_{n,m+1}(X)\pmod{\tilde V_{n}\mathfrak Y^{(K)}},
\end{align*}
and similarly $[\mathbf t_{n,m},\mathbf t_{1,2}]=(2n-m)\mathbf t_{n,m+1}\pmod{\tilde V_{n}\mathsf Y^{(K)}}$. 

Using \eqref{eqn: schematic [T[n,m],t[1,2]]}, we can prove \eqref{eqn: schematic commutators in Y} for all $(n,m,p,q)\in \mathbb N\times \mathbb Z_{<0}\times \mathbb N\times\mathbb N$ such that $q\le 2p$, by increasing induction on $q$ starting from $q=0$ which has been proven in the previous step. In fact, 
\begin{align*}
&[\mathbf T_{n,m}(X),\mathbf T_{p,q}(Y)]=\frac{1}{q-1-2p}[\mathbf T_{n,m}(X),[\mathbf t_{1,2},\mathbf T_{p,q-1}(Y)]]\\
&=\frac{1}{q-1-2p}[\mathbf t_{1,2},[\mathbf T_{n,m}(X),\mathbf T_{p,q-1}(Y)]]-\frac{1}{q-1-2p}[[\mathbf t_{1,2},\mathbf T_{n,m}(X)],\mathbf T_{p,q-1}(Y)]\\
&\equiv \frac{1}{q-1-2p}[\mathbf t_{1,2},\mathbf T_{n+p,m+q-1}([X,Y])]-\frac{m-2n}{q-1-2p}[\mathbf T_{n,m+1}(X),\mathbf T_{p,q-1}(Y)] \pmod{\tilde V_{n+p}\mathfrak Y^{(K)}}\\
&\equiv \mathbf T_{n+p,m+q}([X,Y])\pmod{\tilde V_{n+p}\mathfrak Y^{(K)}},
\end{align*}
and the other two equations in \eqref{eqn: schematic commutators in Y} are deduced similarly.

Finally, the remaining cases are those $(n,m,p,q)\in \mathbb N\times \mathbb Z_{<0}\times \mathbb N\times\mathbb N$ such that $q> 2p$. This is proven by decreasing induction on $p$ using \eqref{eqn: schematic [T[n,m],t[0,1]]}, starting from $q=2p$ or $q=2p-1$ which is proven in the previous step. In fact, 
\begin{align*}
&[\mathbf T_{n,m}(X),\mathbf T_{p,q}(Y)]=\frac{-1}{p+1}[\mathbf T_{n,m}(X),[\mathbf t_{0,1},\mathbf T_{p+1,q}(Y)]]\\
&=\frac{-1}{p+1}[\mathbf t_{0,1},[\mathbf T_{n,m}(X),\mathbf T_{p+1,q}(Y)]]+\frac{1}{p+1}[[\mathbf t_{0,1},\mathbf T_{n,m}(X)],\mathbf T_{p+1,q}(Y)]\\
&\equiv \frac{-1}{p+1}[\mathbf t_{0,1},\mathbf T_{n+p+1,m+q}([X,Y])]-\frac{n}{p+1}[\mathbf T_{n-1,m}(X),\mathbf T_{p+1,q}(Y)] \pmod{\tilde V_{n+p}\mathfrak Y^{(K)}}\\
&\equiv \mathbf T_{n+p,m+q}([X,Y])\pmod{\tilde V_{n+p}\mathfrak Y^{(K)}},
\end{align*}
and the other two equations in \eqref{eqn: schematic commutators in Y} are deduced similarly. This finishes the proof of the Proposition.
\end{proof}

\begin{proposition}\label{prop: t[m,n] preserves Y^K}
For all $(n,m)\in \mathbb N\times\mathbb Z$, the adjoint action of $\mathbf t_{n,m}$ preserves the subalgebra $\mathfrak Y^{(K)}$.
\end{proposition}
\begin{proof}
This follows from the second equation in \eqref{eqn: schematic commutators in Y}.
\end{proof}

Note that there is a natural $\mathbb C[\epsilon_1,\epsilon_2,\mathbf c]$-module map $\mathbb C[\epsilon_1,\epsilon_2,\mathbf c]\cdot \mathfrak B(\mathsf L^{(K)})\to \mathsf Y^{(K)}$. We claim that this map is surjective, which can be proved by the following inductive argument. $\tilde V_0\mathsf Y^{(K)}$ is generated as a $\mathbb C[\epsilon_1,\epsilon_2,\mathbf c]$-module by $1$. Assume that $\tilde V_s\mathsf Y^{(K)}$ is generated by elements in $\mathfrak B(\mathsf L^{(K)})$, then Proposition \ref{prop: filtration on Y} implies that we can reorder any monomials in $\mathfrak{G}(\mathsf L^{(K)})$ with total degree $s+1$ into the non-decreasing order modulo terms in $\tilde V_{s}\mathsf Y^{(K)}$, therefore $\tilde V_{s+1}\mathsf Y^{(K)}$ is generated by elements in $\mathfrak B(\mathsf L^{(K)})$.

Similarly there is a natural surjective $\mathbb C[\epsilon_1,\epsilon_2,\mathbf c]$-module map $\mathbb C[\epsilon_1,\epsilon_2,\mathbf c]\cdot \mathfrak B(\mathfrak L^{(K)})\twoheadrightarrow \mathfrak Y^{(K)}$.

\begin{theorem}\label{thm: PBW for Y}
The map $\mathbb C[\epsilon_1,\epsilon_2,\mathbf c]\cdot \mathfrak B(\mathsf L^{(K)})\to \mathsf Y^{(K)}$ is an isomorphism. In particular $\mathsf Y^{(K)}$ is a free module over $\mathbb C[\epsilon_1,\epsilon_2,\mathbf c]$.
\end{theorem}

\begin{proof}[Proof of Theorem \ref{thm: Psi_infty is injective} and Theorem \ref{thm: PBW for Y}]
Notice that both Theorem \ref{thm: Psi_infty is injective} and Theorem \ref{thm: PBW for Y} follow if we can show that the composition of $\mathbb C[\epsilon_1,\epsilon_2,\mathbf c]$-module maps
\begin{align*}
    \mathbb C[\epsilon_1,\epsilon_2,\mathbf c]\cdot \mathfrak B(\mathsf L^{(K)})\longrightarrow \mathsf Y^{(K)}\overset{\Psi_\infty}{\longrightarrow} \mathfrak U(\mathsf W^{(K)}_{\infty})[\epsilon_2^{-1}]
\end{align*}
is injective. Since both $\mathbb C[\epsilon_1,\epsilon_2,\mathbf c]\cdot \mathfrak B(\mathsf L^{(K)})$ and $\mathfrak U(\mathsf W^{(K)}_{\infty})[\epsilon_2^{-1}]$ are torsion-free over the base ring $\mathbb C[\epsilon_1,\epsilon_2,\mathbf c]$, so it suffices to show that $$\mathbb C[\epsilon_2]\cdot \mathfrak B(\mathsf L^{(K)})\to \mathfrak U(\mathsf W^{(K)}_{\infty}/(\epsilon_1=\mathsf c=0))[\epsilon_2^{-1}]$$ is injective. By the PBW theorem for $\mathsf L^{(K)}$, the natural map $\mathbb C[\epsilon_1,\epsilon_2]\cdot \mathfrak B(\mathsf L^{(K)})\to \mathsf Y^{(K)}/(\mathbf c=0)\cong \mathsf L^{(K)}$ is an isomorphism, therefore it is enough to show that the map $$\Psi_{\infty}:\mathsf L^{(K)}/(\epsilon_1=0)\to \mathfrak U(\mathsf W^{(K)}_{\infty}/(\epsilon_1=\mathsf c=0))[\epsilon_2^{-1}]$$ is injective. 

Using Proposition \ref{prop: linear degeneration}, we see that $\mathfrak U(\mathsf W^{(K)}_{\infty}/(\epsilon_1=\mathsf c=0))$ is the degree-wise completion of the universal enveloping algebra $U(D_{\epsilon_2}(\mathbb C^{\times})\otimes\gl_K)$. On the other hand, $\mathsf L^{(K)}/(\epsilon_1=0)$ is isomorphic to the universal enveloping algebra $U(D_{\epsilon_2}(\mathbb C^{\times})\otimes\gl_K^{\sim})$ according to the Remark \ref{rmk: vertival filtrations on L vs A}. Moreover, it is shown in the proof of Proposition \ref{prop: linear degeneration} that $\Psi_\infty$ maps $E^a_bx^{m}(\epsilon_2\partial_x)^n\in D_{\epsilon_2}(\mathbb C^{\times})\otimes\gl_K^{\sim}$ to $\mathds W^{a(n+1)}_{b,-m-n}+\epsilon_2\cdot(\text{linear combination of $\mathds W^{(i)}_{-m-n}$ for $i\le n$})\in\mathfrak U(\mathsf W^{(K)}_{\infty}/(\epsilon_1=\mathsf c=0))$, for all $(m,n)\in \mathbb Z\times\mathbb N$. In particular, $\Psi_\infty$ maps $D_{\epsilon_2}(\mathbb C^{\times})\otimes\gl_K^{\sim}$ injectively into $D_{\epsilon_2}(\mathbb C^{\times})\otimes\gl_K[\epsilon_2^{-1}]$, thus $\Psi_{\infty}:\mathsf L^{(K)}/(\epsilon_1=0)\to \mathfrak U(\mathsf W^{(K)}_{\infty}/(\epsilon_1=\mathsf c=0))[\epsilon_2^{-1}]$ is injective. This proves Theorem \ref{thm: Psi_infty is injective} and Theorem \ref{thm: PBW for Y}.
\end{proof}

\begin{corollary}\label{cor: PBW for mathfrak Y}
The map $\mathbb C[\epsilon_1,\epsilon_2,\mathbf c]\cdot \mathfrak B(\mathfrak L^{(K)})\to \mathfrak Y^{(K)}$ is an isomorphism. In particular $\mathfrak Y^{(K)}$ is a free module over $\mathbb C[\epsilon_1,\epsilon_2,\mathbf c]$. Moreover, the specialization $\mathfrak Y^{(K)}/(\mathbf c)\to \mathsf Y^{(K)}/(\mathbf c)\cong\mathsf L^{(K)}$ is injective and it induces an isomorphism $\mathfrak Y^{(K)}/(\mathbf c)\cong \mathfrak L^{(K)}$.
\end{corollary}

\begin{proof}
The elements in $\mathfrak B(\mathfrak L^{(K)})$ are rescaling of elements in $\mathfrak B(\mathsf L^{(K)})$, and we have just shown that the latter are $\mathbb C[\epsilon_1,\epsilon_2,\mathbf c]$-linearly independent in $\mathsf Y^{(K)}$, thus elements in $\mathfrak B(\mathfrak L^{(K)})$ are $\mathbb C[\epsilon_1,\epsilon_2,\mathbf c]$-linearly independent in $\mathfrak Y^{(K)}$. According to Theorem \ref{thm: PBW for L}, $\mathfrak B(\mathfrak L^{(K)})$ is a $\mathbb C[\epsilon_1,\epsilon_2]$-basis of $\mathfrak L^{(K)}$, so the composition $\mathbb C[\epsilon_1,\epsilon_2]\cdot \mathfrak B(\mathfrak L^{(K)})\to \mathfrak Y^{(K)}/(\mathbf c)\to \mathsf Y^{(K)}/(\mathbf c)\cong\mathsf L^{(K)}$ is an isomorphism, thus $\mathfrak Y^{(K)}/(\mathbf c)\cong \mathfrak L^{(K)}$.
\end{proof}

\begin{remark}
Another corollary to the Theorem \ref{thm: Psi_infty is injective} is that the subalgebra $\mathsf Y^{(K)}\subset \mathfrak{U}(\mathsf W^{(K)}_\infty)[\epsilon_2^{-1}]$ is invariant under the automorphism $\boldsymbol{\tau}_\beta:\mathfrak{U}(\mathsf W^{(K)}_\infty)\cong \mathfrak{U}(\mathsf W^{(K)}_\infty)$ in \eqref{eqn: shift automorphism of W_integral basis}. It is straightforward to compute that
\begin{equation}\label{eqn: shift automorphism of Y}
\begin{split}
&\boldsymbol{\tau}_\beta(\mathsf T_{0,n}(X))=\mathsf T_{0,n}(X)+\epsilon_2\beta\delta_{n,0}\mathrm{tr}(X)\mathbf c,\quad \boldsymbol{\tau}_\beta(\mathsf t_{0,n})=\mathsf t_{0,n}+\beta\delta_{n,0} K\mathbf c\\
&\boldsymbol{\tau}_\beta(\mathsf t_{1,m})=\mathsf t_{1,m}+\beta\mathsf t_{0,m-1}+\delta_{m,1}(\beta+\epsilon_1\epsilon_2\epsilon_3\mathbf c)\beta\mathbf c,\\
&\boldsymbol{\tau}_\beta(\mathsf T_{1,0}(X))=\mathsf T_{1,0}(X)+\beta\mathsf T_{0,-1}(X),\\
&\boldsymbol{\tau}_\beta(\mathsf t_{2,0})=\mathsf t_{2,0}+2\beta\mathbf t_{1,-1}+(\beta^2-2\beta\epsilon_1\epsilon_2\epsilon_3\mathbf c)\mathsf t_{0,-2},
\end{split}
\end{equation}
where $n\in \mathbb Z$ and $m\in \mathbb N$.
In particular, modulo $\mathbf c$, the automorphism $\boldsymbol{\tau}_\beta:\mathsf Y^{(K)}\cong \mathsf Y^{(K)}$ agrees with the automorphism ${\tau}_\beta:\mathsf L^{(K)}\cong \mathsf L^{(K)}$ defined in \eqref{eqn: shift automorphism of L}.
\end{remark}

\subsection{Coproduct of \texorpdfstring{$\mathsf Y^{(K)}$}{Yk}}
Obviously $\mathsf Y^{(K)}$ inherits a $\mathbb Z$-grading from $\mathsf A^{(K)}_\pm$ such that $\deg \mathsf T_{n,m}(X)=\deg \mathsf t_{n,m}=m-n$. Therefore we have completed tensor product $\mathsf Y^{(K)}\widetilde\otimes \mathsf Y^{(K)}$.

Straightforward computation shows that $\Delta_{\mathsf W}$ maps the subalgebra $\mathsf Y^{(K)}\subset \mathfrak U(\mathsf W^{(K)}_{\infty})[\epsilon_2^{-1}]$ to the subalgebra $\mathsf Y^{(K)}\widetilde\otimes \mathsf Y^{(K)}\subset \mathfrak U(\mathsf W^{(K)}_{\infty})\widetilde\otimes \mathfrak U(\mathsf W^{(K)}_{\infty})[\epsilon_2^{-1}]$.

\begin{definition}\label{def: YY coproduct}
We define the $\mathbb C[\epsilon_1,\epsilon_2]$-algebra homomorphism $\Delta_{\mathsf Y}:\mathsf Y^{(K)}\to \mathsf Y^{(K)}\widetilde\otimes \mathsf Y^{(K)}$ by the restriction of $\Delta_{\mathsf W}$ to $\mathsf Y^{(K)}$.
\end{definition}

$\Delta_{\mathsf Y}$ is uniquely determined by its action on a set of generators as follows
\begin{equation}\label{eqn: YY coproduct}
\boxed{
\begin{aligned}
&\Delta_{\mathsf Y}(\mathsf T_{0,n}(X))=\square(\mathsf T_{0,n}(X)),\;X\in \gl_K,\quad \Delta_{\mathsf Y}(\mathsf t_{0,n})=\square(\mathsf t_{0,n}),\quad (n\in \mathbb Z) \\
&\Delta_{\mathsf Y}(\mathsf T_{1,0}(E^a_b))=\square(\mathsf T_{1,0}(E^a_b))+\epsilon_1\sum_{n= 0}^{\infty}(\mathsf T_{0,n}(E^c_b)\otimes \mathsf T_{0,-n-1}(E^a_c)-\mathsf T_{0,n}(E^a_c)\otimes \mathsf T_{0,-n-1}(E^c_b)),\\
&\Delta_{\mathsf Y}(\mathsf t_{2,0})=\square(\mathsf t_{2,0})-2\epsilon_1\sum_{n= 0}^{\infty}(n+1)(\mathsf T_{0,n}(E^a_b)\otimes \mathsf T_{0,-n-2}(E^b_a)+\epsilon_1\epsilon_2 \mathsf t_{0,n}\otimes \mathsf t_{0,-n-2}),\\
&\Delta_{\mathsf Y}(\mathbf c)=\square(\mathbf c),
\end{aligned}}
\end{equation}
where $\square(X)=X\otimes 1+1\otimes X$. 

$\Delta_{\mathsf Y}$ is coassociative, i.e. the image of $(\Delta_{\mathsf Y}\otimes 1)\circ \Delta_{\mathsf Y}$ is contained in the intersection between $(\mathsf Y^{(K)}\widetilde\otimes \mathsf Y^{(K)})\widetilde\otimes \mathsf Y^{(K)}$ and $\mathsf Y^{(K)}\widetilde\otimes (\mathsf Y^{(K)}\widetilde\otimes \mathsf Y^{(K)})$ in the sense of Lemma \ref{lem: compare different completions} and $(\Delta_{\mathsf Y}\otimes 1)\circ \Delta_{\mathsf Y}=(1\otimes \Delta_{\mathsf Y})\circ \Delta_{\mathsf Y}$. 

Moreover, the composition $\mathsf Y^{(K)}\to \mathsf Y^{(K)}/(\mathbf c=0)\cong \mathsf L^{(K)}$ with the augmentation $\mathfrak{C}_{\mathsf L}$ gives an augmentation $\mathfrak C_{\mathsf Y}:\mathsf Y^{(K)}\to \mathbb C[\epsilon_1,\epsilon_2]$ which is a counit for $\Delta_{\mathsf Y}$, i.e. $(\mathfrak C_{\mathsf Y}\otimes 1)\circ \Delta_{\mathsf Y}=\mathrm{id}=( 1\otimes \mathfrak C_{\mathsf Y})\circ \Delta_{\mathsf Y}$.

\bigskip On the other hand, it is easy to see that the map $ \Delta_{\infty}:\mathsf A^{(K)}\to\mathsf A^{(K)}\widetilde\otimes \mathfrak U(\mathsf W^{(K)}_\infty)[\epsilon_2^{-1}]$ in \eqref{eqn: AW(infinity) coproduct} factors through the subalgebra $\mathsf A^{(K)}\widetilde\otimes \mathsf Y^{(K)}$, so $\Delta_{\infty}$ is induced from a mixed coproduct map
\begin{align}\label{eqn: mixed AY coproduct}
    \Delta_{\mathsf A,\mathsf Y}: \mathsf A^{(K)}\to\mathsf A^{(K)}\widetilde\otimes \mathsf Y^{(K)}.
\end{align}
It follows from Proposition \ref{prop: AW coproduct compatible with WW coproduct} that $\Delta_{\mathsf A,\mathsf Y}$ is the restriction of $\Delta_{\mathsf Y}$ to the subalgebra $\mathsf A^{(K)}$.

The co-associativity of $\Delta_{\mathsf Y}$ together with the compatibility \eqref{eqn: D_Y compatible with mixed coproduct} implies that $\mathsf A^{(K)}$ is a comodule of $\mathsf Y^{(K)}$, i.e. the two ways to map $\mathsf A^{(K)} \to \mathsf A^{(K)} \widetilde{\otimes}\mathsf Y^{(K)}\widetilde{\otimes}\mathsf Y^{(K)}$ agree:
\begin{equation}
(\Delta_{\mathsf A,\mathsf Y} \otimes 1 ) \circ \Delta_{\mathsf A,\mathsf Y} = (1 \otimes \Delta_{\mathsf Y}) \circ \Delta_{\mathsf A,\mathsf Y}.
\end{equation}

\begin{proposition}\label{prop: Y^K is a bialgebra}
The coproduct $\Delta_{\mathsf Y}$ and the counit $\mathfrak C_{\mathsf Y}$ make $\mathsf Y^{(K)}$ a bialgebra, and the mixed coproduct $\Delta_{\mathsf A,\mathsf Y}$ makes $\mathsf A^{(K)}$ a comodule of $\mathsf Y^{(K)}$. 
\end{proposition}

With some patience, one can compute the formula of coproduct for more elements, for example:
\begin{align*}
\Delta_{\mathsf Y}(\mathsf t_{1,n})=\square(\mathsf t_{1,n})+n\epsilon_1\epsilon_2\epsilon_3\mathsf t_{0,n-1}\otimes \mathbf c,
\end{align*}
and
\begin{align*}
\Delta_{\mathsf Y}(\mathsf t_{2,n})&=\square(\mathsf t_{2,n})+n\epsilon_1\epsilon_2\epsilon_3\mathsf t_{1,n-1}\otimes \mathbf c+\frac{n(n-1)}{2}(\epsilon_1\epsilon_2\epsilon_3)^2\mathsf t_{0,n-2}\otimes \mathbf c^2\\
&-\epsilon_1\sum_{m= 0}^{\infty}(n+2m+2)(\mathsf T_{0,m+n}(E^a_b)\otimes \mathsf T_{0,-m-2}(E^b_a)+\epsilon_1\epsilon_2 \mathsf t_{0,m+n}\otimes \mathsf t_{0,-m-2})\\
&-\epsilon_1\sum_{m= 0}^{n-1}\frac{(m+1)(m+2)}{n+1}(\mathsf T_{0,m}(E^a_b)\otimes \mathsf T_{0,-m-2+n}(E^b_a)+\epsilon_1\epsilon_2 \mathsf t_{0,m}\otimes \mathsf t_{0,-m-2+n}).
\end{align*}

\subsection{Meromorphic coproduct of \texorpdfstring{$\mathsf Y^{(K)}$}{Yk}}

\begin{proposition}\label{prop: Y^K is vertex comodule}
The meromorphic coproduct $\Delta_{\mathsf A}(w): \mathsf A^{(K)}\to \mathsf A^{(K)}\otimes\mathsf A^{(K)}(\!(w^{-1})\!)$ extends to a $\mathbb C[\epsilon_1,\epsilon_2]$-algebra homomorphism $\Delta_{\mathsf Y}(w): \mathsf Y^{(K)}\to \mathsf Y^{(K)}\otimes \mathsf A^{(K)}(\!(w^{-1})\!)$ such that
\begin{equation}\label{eqn: YA meromorphic coproduct}
\begin{split}
\Delta_{\mathsf Y}(w)(\mathsf T_{0,-n}(E^a_b))&=\mathsf T_{0,-n}(E^a_b)\otimes 1+\sum_{m=0}^{\infty} \binom{-n}{m} w^{-n-m} 1\otimes\mathsf T_{0,m}(E^a_b),\\
\Delta_{\mathsf Y}(w)(\mathbf c)&=\mathbf c\otimes 1.
\end{split}
\end{equation}
The pair $(\mathsf{Y}^{(K)}, \Delta_{\mathsf Y}(w))$ is a vertex comodule of the vertex coalgebra $(\mathsf A^{(K)},\Delta_{\mathsf A}(w),\mathfrak{C}_{\mathsf A})$ over the base ring $\mathbb C[\epsilon_1,\epsilon_2]$. Moreover, the natural map $\Psi_{\infty}:\mathsf Y^{(K)}\hookrightarrow U(\mathsf W^{(K)}_{\infty})[\epsilon_3^{-1}]$ intertwines between their vertex comodule structures with respect to vertex coalgebras $\mathsf A^{(K)}$ and $U_+(\mathsf W^{(K)}_{\infty})[\epsilon_3^{-1}]$.
\end{proposition}
\begin{proof}
Consider the natural $\mathbb Z$-graded $\mathbb C[\epsilon_1,\epsilon_2]$-algebra map $1\otimes S(w): \mathsf Y^{(K)}\otimes \mathsf L^{(K)}\to \mathsf Y^{(K)}\otimes \mathsf A^{(K)}(\!(w^{-1})\!)$, it is easy to to that it extends uniquely to a $\mathbb C[\epsilon_1,\epsilon_2]$-algebra map $1\otimes S(w): \mathsf Y^{(K)}\widetilde\otimes \mathsf L^{(K)}\to \mathsf Y^{(K)}\otimes \mathsf A^{(K)}(\!(w^{-1})\!)$, thus we define $\Delta_{\mathsf Y}(w)$ to be the composition of $\Delta_{\mathsf Y}$ followed by the projection $\mathsf Y^{(K)}\widetilde\otimes \mathsf Y^{(K)}\twoheadrightarrow \mathsf Y^{(K)}\widetilde\otimes \mathsf L^{(K)}$ and then applying $1\otimes S(w)$. The equation \eqref{eqn: YA meromorphic coproduct} follows from equation \eqref{eqn: YY coproduct}.

For the statement on vertex comodule structures, the counit axiom is the result of counital property of $\Delta_{\mathsf Y}$, and the coassociativity can be checked on the generators directly, the computation is similar to that of Theorem \ref{thm: loop Yangian is a vertex comodule} and we omit it.
\end{proof}

\subsection{A bimodule of \texorpdfstring{$\mathsf Y^{(K)}$}{Yk}}\label{subsec: bimodule}
In the Lemma \ref{lem: rho_N extends to loop Yangian} we have constructed the truncation $\rho_n: \mathsf L^{(K)}\twoheadrightarrow \mathrm{S}\mathbb H^{(K)}_n$, we will keep using the same notation $\rho_n$ for its precomposition with the projection $\mathsf Y^{(K)}\twoheadrightarrow \mathsf L^{(K)}$, and we also identify the spherical Cherednik algebra $\mathrm{S}\mathbb H^{(K)}_n$ with its image in $D(\mathbb C^{\times n}_{\mathrm{disj}})\otimes \gl_K^{\otimes n}$ under the Dunkl embedding.

We observe that the image of $$(\rho_n\otimes \widetilde\Psi_\infty)\circ \Delta_{\mathsf Y}: \mathsf Y^{(K)}\to \left( D(\mathbb C^{\times n}_{\mathrm{disj}})\widetilde\otimes U(\widetilde{\mathsf W}^{(K)}_\infty)\right)\otimes \gl_K^{\otimes n}[\epsilon_3^{-1}]$$
is contained in the subalgebra $D(\mathbb C^{\times n}_{\mathrm{disj}};\widetilde{\mathsf W}^{(K)}_\infty)\otimes  \gl_K^{\otimes n}[\epsilon_3^{-1}]$ (see the Definition \ref{def: differential operators valued in restricted modes} and Proposition \ref{prop: D(V) is a sub of completed tensor product}). In fact, the images of $\mathsf T_{0,n}(X), \mathsf t_{0,n}$ and $\mathbf c$ under the map $(\rho_n\otimes \widetilde\Psi_\infty)\circ \Delta_{\mathsf Y}$ are in the usual tensor product $D(\mathbb C^{\times n}_{\mathrm{disj}})\otimes U(\widetilde{\mathsf W}^{(K)}_\infty)\otimes \gl_K^{\otimes n}[\epsilon_3^{-1}]$. $D(\mathbb C^{\times n}_{\mathrm{disj}})\otimes U(\widetilde{\mathsf W}{(K)}_\infty)\otimes \gl_K^{\otimes n}$ is obviously contained in $D(\mathbb C^{\times n}_{\mathrm{disj}};\widetilde{\mathsf W}^{(K)}_\infty)\otimes  \gl_K^{\otimes n}$. For the generator $\mathsf t_{2,0}$, we can rewrite $(\rho_n\otimes \widetilde\Psi_\infty)\circ \Delta_{\mathsf Y}(\mathsf t_{2,0})$ as
\begin{align*}
    \rho_n(\mathsf t_{2,0})\otimes 1+1\otimes \widetilde\Psi_\infty(\mathsf t_{2,0})-2\epsilon_1\sum_{i=1}^n\mathcal O\left(U^{a(1)}_{b,-1}|0\rangle;\frac{1}{(x_i-z_1)^2}\right)E^b_a,
\end{align*}
so $(\rho_n\otimes \widetilde\Psi_\infty)\circ \Delta_{\mathsf Y}(\mathsf t_{2,0})$ is also in the subalgebra $D(\mathbb C^{\times n}_{\mathrm{disj}};\widetilde{\mathsf W}^{(K)}_\infty)\otimes  \gl_K^{\otimes n}[\epsilon_3^{-1}]$, thus the image of $(\rho_n\otimes \widetilde\Psi_\infty)\circ \Delta_{\mathsf Y}$ is contained in $D(\mathbb C^{\times n}_{\mathrm{disj}};\widetilde{\mathsf W}^{(K)}_\infty)\otimes  \gl_K^{\otimes n}[\epsilon_3^{-1}]$.

According to Proposition \ref{prop: D(V) bimodule}, the space $$\Hom_{\mathbb C[\epsilon_1,\epsilon_2]}(\widetilde{\mathsf W}^{(K)}_\infty,D(\mathbb C^{\times n}_{\mathrm{disj}})\widetilde\otimes \widetilde{\mathsf W}^{(K)}_\infty)\otimes  \gl_K^{\otimes n}$$
admits a $D(\mathbb C^{\times n}_{\mathrm{disj}};\widetilde{\mathsf W}^{(K)}_\infty)\otimes  \gl_K^{\otimes n}$ bimodule structure, and by precomposing with $(\rho_n\otimes \widetilde\Psi_\infty)\circ \Delta_{\mathsf Y}$, it is endowed with a $\mathsf Y^{(K)}$ bimodule structure. More precisely, the left action of $f\in \mathsf Y^{(K)}$ is the left multiplication on $D(\mathbb C^{\times n}_{\mathrm{disj}})\widetilde\otimes \widetilde{\mathsf W}^{(K)}_\infty$ by $(\rho_n\otimes \widetilde\Psi_\infty)\circ \Delta_{\mathsf Y}(f)$, and the the right action of $f\in \mathsf Y^{(K)}$ is the precomposing every map in the Hom space with the action of $(\rho_n\otimes \widetilde\Psi_\infty)\circ \Delta_{\mathsf Y}^{\mathrm{op}}(f)$ on $\widetilde{\mathsf W}^{(K)}_\infty$. Note that the right action uses the opposite coproduct because the order of expansion is opposite to that of the left action.

We can generalize the differential operator to pseudodifferential symbols, namely for any $\mathbb Z$-graded vertex algebra $\mathcal V$ we define a $\mathcal V$-valued pseudodifferential symbol in $n$ variables to be the linear combination of following terms
\begin{align*}
    D=\sum_{r\in \mathbb Z_{\ge 0}}\partial_{x_1}^{\mu_1-r_1}\cdots\partial_{x_n}^{\mu_n-r_n}\cdot D_r,\quad D_r\in \mathscr O(\mathbb C^{\times n}_{\mathrm{disj}})\widetilde\otimes \mathcal V,
\end{align*}
where $\mu_1,\cdots,\mu_n$ are formal variables. Let $\Psi\mathrm{DS}_n(\mathcal V)$ be the linear space of $\mathcal V$-valued pseudodifferential symbol in $n$ variables, then the Hom space
\begin{align}
    \mathsf M^{(K)}_n:=\Hom_{\mathbb C[\epsilon_1,\epsilon_2]}(\widetilde{\mathsf W}^{(K)}_\infty,\Psi\mathrm{DS}_n(\widetilde{\mathsf W}^{(K)}_\infty) )\otimes\gl_K^{\otimes n}
\end{align}
possesses a natural $\mathsf Y^{(K)}$ bimodule structure.

\begin{example}
The pseudodifferential Miura operator
\begin{align}
    \mathcal L_1(x):=(\epsilon_2\partial_x)^{\epsilon_1\mathsf c}+\epsilon_1\sum_{r\ge 1}(\epsilon_2\partial_x)^{\epsilon_1\mathsf c-r}\cdot \mathds U^{a(r)}_{b}(x)E^b_a
\end{align}
is an element in $\mathsf M^{(K)}_1$ such that its image under the truncation map $\pi_L:\widetilde{\mathsf W}^{(K)}_\infty\to \widetilde{\mathcal W}^{(K)}_L$ is $\epsilon_1^L\mathcal L_1^L(x)$, where $\mathcal L_1^L(x)$ is the Miura operator defined in \ref{eqn: finite L Miura operator}. More generally
\begin{align}
    \mathcal L_n(x_1,\cdots,x_n):=\mathcal L_1(x_1)\cdots\mathcal L_1(x_n)
\end{align}
is an element in $\mathsf M^{(K)}_n$.
\end{example}

\section{Miura Operators as Intertwiners}\label{sec: Miura Operators as Intertwiners}

In the end of last section, we introduce a bimodule $\mathsf M^{(K)}_n$ of the algebra $\mathsf Y^{(K)}$, and introduce the pseudodifferential Miura operator $\mathcal L_n(x_1,\cdots,x_n)\in\mathsf M^{(K)}_n$. The main result of this section is the following.
\begin{theorem}\label{thm: PsiDS Miura as intertwiner}
The pseudodifferential Miura operator $\mathcal L_n(x_1,\cdots,x_n)$ in the $\mathsf Y^{(K)}$ bimodule $\mathsf M^{(K)}_n$ intertwines between the left and the right $\mathsf Y^{(K)}$-actions, i.e. for all $f\in \mathsf Y^{(K)}$,
\begin{align}\label{eqn: infinity intertwiner}
    (\rho_n\otimes \widetilde\Psi_\infty)\circ \Delta_{\mathsf Y}(f)\cdot\mathcal L_n(x_1,\cdots,x_n)=\mathcal L_n(x_1,\cdots,x_n)\cdot(\rho_n\otimes \widetilde\Psi_\infty)\circ \Delta_{\mathsf Y}^{\mathrm{op}}(f).
\end{align}
\end{theorem}

Notice that $\widetilde\Psi_\infty$ maps $\mathsf A^{(K)}$ to the non-negative modes of $\widetilde{\mathsf W}^{(K)}_\infty$, thus for $f\in \mathsf A^{(K)}$ we have $\widetilde\Psi_\infty(f)|0\rangle=\mathfrak{C}_{\mathsf A}(f)|0\rangle$, whence
\begin{align*}
    (\rho_n\otimes \widetilde\Psi_\infty)\circ\Delta^{\mathrm{op}}_{\mathsf Y}(f)|0\rangle=\rho_n(f)\otimes|0\rangle.
\end{align*} 
$\mathsf N^{(K)}_n:= \Psi\mathrm{DS}_n(\widetilde{\mathsf W}^{(K)}_\infty) \otimes\gl_K^{\otimes n}$ possesses a natural $\mathsf A^{(K)}$ bimodule structure, such that the left action is $\mathsf A^{(K)}\ni f\mapsto $ left multiplication by $(\rho_n\otimes \widetilde\Psi_\infty)\circ \Delta_{\mathsf Y}(f)$ and the right action is $\mathsf A^{(K)}\ni f\mapsto $ right multiplication by $\rho_n(f)$. The above computation together with Theorem \ref{thm: PsiDS Miura as intertwiner} implies the following. 
\begin{corollary}
The pseudodifferential Miura operator acting on vacuum $\mathcal L_n(x_1,\cdots,x_n)|0\rangle$ is an element of $\mathsf N^{(K)}_n$ and it intertwines between left and right $\mathsf A^{(K)}$ actions, i.e. for all $f\in \mathsf A^{(K)}$,
\begin{align}\label{eqn: infinity intertwiner_DDCA}
    (\rho_n\otimes \widetilde\Psi_\infty)\circ \Delta_{\mathsf Y}(f)\cdot\mathcal L_n(x_1,\cdots,x_n)|0\rangle=\mathcal L_n(x_1,\cdots,x_n)|0\rangle\cdot\rho_n(f).
\end{align}
\end{corollary}

To prove Theorem \ref{thm: PsiDS Miura as intertwiner}, notice that the difference between the left-hand-side and right-hand-side of \eqref{eqn: infinity intertwiner} can be written as
\begin{align*}
\sum_{r\in \mathbb Z^n} (\epsilon_2\partial_{x_1})^{\epsilon_1\mathsf c-r_1}\cdots (\epsilon_2\partial_{x_n})^{\epsilon_1\mathsf c-r_n}\cdot F_r,
\end{align*}
$F_r$ is an element in $\mathscr O(\mathbb C^{\times n}_{\mathrm{disj}})\widetilde\otimes \widetilde{\mathsf W}^{(K)}_\infty$ which can be written as $F_r=\sum_{m\in \mathbb Z_{\ge 0}} g_{r,m}\otimes v_{r,m}$ such that $v_{r,m}$ is a homogeneous element in $\widetilde{\mathsf W}^{(K)}_\infty$ of degree $-m$ and $g_{r,m}$ is a homogeneous function on $\mathbb C^{\times n}_{\mathrm{disj}}$ of degree $m+\deg f-\sum_{i=1}^n r_i$. To show \eqref{eqn: infinity intertwiner} holds, it is equivalent to showing that $v_{r,m}$ vanishes for all $r,m$. Since $v_{r,m}$ depends on $\epsilon_1,\epsilon_2,\mathsf c$ in a polynomial way, then it is enough to show that \eqref{eqn: infinity intertwiner} holds after applying the truncation $1\otimes \pi_L$ on both sides, for all $L\in \mathbb Z_{\ge 1}$. In other words, we need to prove that for all $f\in \mathsf Y^{(K)}$,
\begin{align}\label{eqn: finite L intertwiner}
    (\rho_n\otimes \widetilde\Psi_L)\circ \Delta_{\mathsf Y}(f)\cdot\mathcal L^L_n(x_1,\cdots,x_n)=\mathcal L^L_n(x_1,\cdots,x_n)\cdot(\rho_n\otimes \widetilde\Psi_L)\circ \Delta_{\mathsf Y}^{\mathrm{op}}(f).
\end{align}

\subsection{The case of elementary Miura operator}
Consider the elementary Miura operator\footnote{To match with the brane setting, the elementary Miura operator defined here is denoted by $\mathcal L^{0,1,0}_{0,1,0}(x)$ in \cite{gaiotto2022miura}.}
\begin{align}
    \mathcal L^1_1 (x)=\bar\alpha\partial_x+J^a_b(x)E^b_a,
\end{align}
where $\bar\alpha=\epsilon_2/\epsilon_1$ and $J^a_b(x)$ is the field of affine Kac-Moody vertex algebra $V^{\kappa_{\bar\alpha}}(\mathfrak{gl}_K)$ and $E^b_a$ is the elementary matrix of $\mathfrak{gl}_K$.
\begin{lemma}\label{lemma: Miura as intertwiner A}
The elementary Miura operator $\mathcal L^1_1 (x)$ intertwines the left and right $\mathsf Y^{(K)}$ actions, i.e. for all $f\in \mathsf Y^{(K)}$,
\begin{align}\label{eqn: elementary intertwiner}
    (\rho_1\otimes \widetilde\Psi_1)\circ \Delta_{\mathsf Y}(f)\cdot\mathcal L^1_1(x)=\mathcal L^1_1(x)\cdot(\rho_1\otimes \widetilde\Psi_1)\circ \Delta_{\mathsf Y}^{\mathrm{op}}(f).
\end{align}
\end{lemma}

\begin{proof}
We need to check \eqref{eqn: elementary intertwiner} for a set of generators in $\mathsf A^{(K)}$. The set of generators that we choose depends on $K$:
\begin{itemize}
    \item For $K=1$, we choose $\mathsf t_{2,0}$ and $\mathsf T_{0,n}(1)$.
    \item For $K>1$, we choose $\mathsf T_{1,0}(E^1_2)$ and $\mathsf T_{0,n}(E^a_b)$.
\end{itemize}
In both of the cases, these subsets of element generate $\mathsf A^{(K)}$ after localizing $\epsilon_2$ and $\epsilon_3$ (see Corollary \ref{cor: compare two DDCAs}), and by the flatness of $\mathsf A^{(K)}$ and $\widetilde{\mathcal W}^{(K)}_L$ over the base $\mathbb C[\epsilon_1,\epsilon_2]$, it suffices to prove \eqref{eqn: elementary intertwiner} after localization. This justifies our choice of generators. 

For $\mathsf T_{0,n}(E^a_b)$, it is easy to compute the commutator between the left and right action, and we omit the details. It remains to compute for $\mathsf t_{2,0}$ when $K=1$, and for $\mathsf T_{1,0}(E^1_2)$ when $K>1$.\\

When $K=1$, 
\begin{align*}
    (\rho_1\otimes \widetilde\Psi_1)\circ \Delta_{\mathsf Y}(\mathsf t_{2,0})&=\epsilon_2\partial_x^2-\frac{\epsilon_1}{3\alpha}\sum_{k,l\in \mathbb Z}:J_{-k-l-2}J_kJ_l:+\frac{\epsilon_2}{\alpha}\sum_{n\ge 1}nJ_{-n-1}J_{n-1}-2\epsilon_1\partial J(x)_+,\\
    (\rho_1\otimes \widetilde\Psi_1)\circ \Delta^{\mathrm{op}}_{\mathsf Y}(\mathsf t_{2,0})&=\epsilon_2\partial_x^2-\frac{\epsilon_1}{3\alpha}\sum_{k,l\in \mathbb Z}:J_{-k-l-2}J_kJ_l:+\frac{\epsilon_2}{\alpha}\sum_{n\ge 1}nJ_{-n-1}J_{n-1}+2\epsilon_1\partial J(x)_-.
\end{align*}
where $J(z)$ is the Heisenberg field with OPE $J(z)J(w)\sim\frac{-\alpha}{(z-w)^2}$. The difference between the left action and the right action of $\mathsf t_{2,0}$ on $\mathcal L^1_1(x)$ is
\begin{align*}
&\epsilon_2[\partial_x^2,J(x)]+2\epsilon_1:\partial J(x)J(x):+\epsilon_2(\partial^2J(x)_- -\partial^2J(x)_+)\\
&-2\epsilon_2 \partial J(x)_+\cdot \partial_x-2\epsilon_2  \partial_x\cdot \partial J(x)_- -2\epsilon_1:\partial J(x)J(x):
\end{align*}
which vanishes by direct computation.\\

When $K>1$, 
\begin{align*}
    (\rho_1\otimes \widetilde\Psi_1)\circ \Delta_{\mathsf Y}(\mathsf T_{1,0}(E^1_2))&=\epsilon_2E^1_2\partial_x-\epsilon_1\sum_{m\ge 0}J^c_{2,-m-1}J^1_{c,m}+\epsilon_1[J^a_b(x)_+E^b_a,E^1_2],\\
    (\rho_1\otimes \widetilde\Psi_1)\circ \Delta^{\mathrm{op}}_{\mathsf Y}(\mathsf T_{1,0}(E^1_2))&=\epsilon_2E^1_2\partial_x-\epsilon_1\sum_{m\ge 0}J^c_{2,-m-1}J^1_{c,m}+\epsilon_1[E^1_2,J^a_b(x)_-E^b_a].
\end{align*}
where $J^a_b(z)$ is the affine Kac-Moody field with OPE $J^a_b(z)J^c_d(w)\sim\frac{\bar\alpha\delta^a_d\delta^c_b+\delta^a_b\delta^c_d}{(z-w)^2}$. The difference between the left action and the right action of $\mathsf T_{1,0}(E^1_2)$ on $\mathcal L^1_1(x)$ is
\begin{align*}
&[\epsilon_2E^1_2\partial_{x}-\epsilon_1\sum_{m\ge 0}J^{c}_{2,-m-1}J^{1}_{c,m},\mathcal L^1_1(x)]+\epsilon_2[J^a_b(x)_+E^b_a,E^1_2]\cdot\partial_x+\epsilon_2\partial_x\cdot[J^a_b(x)_-E^b_a,E^1_2]\\
&+\epsilon_1 [J^a_b(x)_+E^b_a,E^1_2]\cdot J^c_d(x)E^d_c+ \epsilon_1 J^c_d(x)E^d_c\cdot [J^a_b(x)_-E^b_a,E^1_2].
\end{align*}
The above expression can be further expanded to
\begin{align*}
&\epsilon_2(\partial J^a_2(x)_+E^1_a+\partial J^1_b(x)_-E^b_2)-\epsilon_1(J^a_2(x)_+J^1_b(x)E^b_a-J^c_2(x)_+J^a_c(x)E^1_a)\\
&-\epsilon_1(J^c_b(x)J^1_c(x)_-E^b_2-J^a_2(x)J^1_b(x)_-E^b_a)|0\rangle-\epsilon_2(\partial J^a_2(x)_+E^1_a+\partial J^1_b(x)_-E^b_2)-\epsilon_1\partial J^1_2(x)\\
&+\epsilon_1(J^1_c(x)_+J^a_2(x)E^c_a-J^c_2(x)_+J^a_c(x)E^1_a-J^1_b(x)J^c_2(x)_-E^b_c+J^a_b(x)J^1_a(x)_-E^b_2)\\
=&\epsilon_1 [J^1_c(x)_{+},J^a_2(x)_+]E^c_a-\epsilon_1 [J^1_c(x)_{-},J^a_2(x)_-]E^c_a -\epsilon_1 \partial J^1_2(x),
\end{align*}
which vanishes by direct computation.
\end{proof}

\subsection{Proof of Theorem \ref{thm: PsiDS Miura as intertwiner}}

As we have explained, it remains to prove \eqref{eqn: finite L intertwiner} for every $f\in \mathsf Y^{(K)}$. Notice that \eqref{eqn: finite L intertwiner} can be reformulated as
\begin{align}\label{eqn: finite L intertwiner_reform}
    (\rho_n\otimes \widetilde\Psi_L)\circ \Delta_{\mathsf Y}(f)\cdot\mathcal L^L_n(x_1,\cdots,x_n)=\mathcal L^L_n(x_1,\cdots,x_n)\cdot( \widetilde\Psi_L\otimes\rho_n)\circ \Delta_{\mathsf Y}(f).
\end{align}
Let us bootstrap this equation from the elementary case \eqref{eqn: elementary intertwiner}. First, we consider the case when $L=1$ and $n>1$. Using the compatibility between the coproducts (Proposition \ref{prop: L coproduct compatible with DAHA coproduct}) we have
\begin{align*}
    \rho_n=(\rho_1\otimes\cdots\otimes\rho_1)\circ\Delta_{\mathsf Y}^{n-1},
\end{align*}
this equation should be understood as expanding rational functions in the order $|x_1|<\cdots<|x_n|$. Under such ordering, the Miura operator $\mathcal L^1_n(x_1,\cdots,x_n)$ should be expanded as $\mathcal L^1_1(x_n)\cdots \mathcal L^1_1(x_1)$. Therefore we have
\begin{align*}
    &(\rho_n\otimes \widetilde\Psi_1)\circ \Delta_{\mathsf Y}(f)\cdot\mathcal L^1_n(x_1,\cdots,x_n)=(\rho_1\otimes\cdots\otimes\rho_1\otimes \widetilde\Psi_1)\circ \Delta_{\mathsf Y}^n(f)\cdot\mathcal L^1_1(x_n)\cdots \mathcal L^1_1(x_1)\\
    &=\mathcal L^1_1(x_n)(\rho_1\otimes\cdots\otimes\widetilde\Psi_1\otimes \rho_1)\circ \Delta_{\mathsf Y}^n(f)\cdot\mathcal L^1_1(x_{n-1})\cdots \mathcal L^1_1(x_1)\\
    &=\mathcal L^1_1(x_n)\cdots \mathcal L^1_1(x_1)\cdot(\widetilde\Psi_1\otimes\rho_1\otimes\cdots\otimes\rho_1)\circ \Delta_{\mathsf Y}^n(f)\\
    &=\mathcal L^1_n(x_1,\cdots,x_n)\cdot( \widetilde\Psi_1\otimes\rho_n)\circ \Delta_{\mathsf Y}(f).
\end{align*}
This proves \eqref{eqn: finite L intertwiner_reform} in the case $L=1$. For general case, we proceed similarly using the compatibility between affine Yangian coproduct and W-algebra coproducts. In fact, by the definition of $\Delta_{\mathsf Y}$ we have
\begin{align*}
    \widetilde\Psi_L=(\widetilde\Psi_1\otimes\cdots\otimes\widetilde\Psi_1)\circ\Delta_{\mathsf Y}^{L-1}.
\end{align*}
Then
\begin{align*}
    &(\rho_n\otimes \widetilde\Psi_L)\circ \Delta_{\mathsf Y}(f)\cdot\mathcal L^L_n(x_1,\cdots,x_n)\\
    &=(\rho_n\otimes\widetilde\Psi_1\otimes\cdots\otimes \widetilde\Psi_1)\circ \Delta_{\mathsf Y}^L(f)\cdot\mathcal L^1_n(x_1,\cdots,x_n)^{[1]}\cdots \mathcal L^1_n(x_1,\cdots,x_n)^{[L]}\\
    &=\mathcal L^1_n(x_1,\cdots,x_n)^{[1]}(\widetilde\Psi_1\otimes\rho_1\otimes\widetilde\Psi_1\otimes\cdots\otimes \widetilde\Psi_1)\circ \Delta_{\mathsf Y}^L(f)\cdot\mathcal L^1_n(x_1,\cdots,x_n)^{[2]}\cdots \mathcal L^1_n(x_1,\cdots,x_n)^{[L]}\\
    &=\mathcal L^1_n(x_1,\cdots,x_n)^{[1]}\cdots \mathcal L^1_n(x_1,\cdots,x_n)^{[L]}\cdot(\widetilde\Psi_1\otimes\cdots\otimes \widetilde\Psi_1\otimes\rho_n)\circ \Delta_{\mathsf Y}^L(f)\\
    &=\mathcal L^L_n(x_1,\cdots,x_n)\cdot( \widetilde\Psi_L\otimes\rho_n)\circ \Delta_{\mathsf Y}(f).
\end{align*}
This proves \eqref{eqn: finite L intertwiner_reform} for all $n$ and $L$, whence finishes the proof of Theorem \ref{thm: PsiDS Miura as intertwiner}.

\subsection{Application: the correlators of Miura operators}

\begin{corollary}\label{cor: correlator of Miura}
The correlator $\langle 0|\mathcal L_n(x_1,\cdots,x_n)|0\rangle$ equals to
\begin{align}
       (y_1\cdots y_n\mathbf e)^{\epsilon_1\mathsf c}\in \Psi\mathrm{DS}_n(\gl_K^{\otimes n})
\end{align}
where $y_i$ are elements of the Cherednik algebra $\mathbb H^{(K)}_n$, and $y_1\cdots y_n\mathbf e$ is regarded as an element of $(D_{\epsilon_2}(\mathbb C^{\times n}_{\mathrm{disj}})\otimes \gl_K^{\otimes n})^{\mathfrak{S}_n}$ via the Dunkl representation.
\end{corollary}

\begin{proof}
It suffices to show that $\langle 0|\mathcal L^L_n(x_1,\cdots,x_n)|0\rangle=\epsilon_1^{-nL}(y_1\cdots y_n\mathbf e)^{L}$ for all $L\in \mathbb Z_{\ge 1}$. Since it is obvious from the definition of Miura operator that $\langle 0|\mathcal L^L_n(x_1,\cdots,x_n)|0\rangle=(\langle 0|\mathcal L^1_n(x_1,\cdots,x_n)|0\rangle)^L$, it is enough to prove the case $L=1$.

Let $\mathcal G_n(x_1,\cdots,x_n)=\langle 0|\mathcal L^1_n(x_1,\cdots,x_n)|0\rangle$ and $\mathcal F_n(x_1,\cdots,x_n)=\epsilon_1^{-n}y_1y_2\cdots y_n\mathbf e$ and let $r_n=\mathcal G_n-\mathcal F_n\in (D_{\epsilon_2}(\mathbb C^{\times n}_{\mathrm{disj}})\otimes \gl_K^{\otimes n})^{\mathfrak{S}_n}$, we will show by induction on $n$ that $r_n=0$. The $n=1$ case is obvious. Now assume that $r_i=0$ for all $i<n$, then we claim that 
\begin{align*}
    [r_n,\rho_n(\mathsf t_{0,m})]=0\text{ for all }m\in \mathbb Z_{\ge 0}.
\end{align*}
To prove the claim, it suffices to regard $r_n$ as an element in $D_{\epsilon_2}(\mathbb C^{\times n}_{\mathrm{disj}})\otimes \gl_K^{\otimes n}$ and show that $[r_n,x_i]=0$ for all $1\le i\le n$. By symmetry, we only need to show that $[r_n,x_n]=0$. It is easy to see that $[\mathcal G_n,x_n]=\bar\alpha\mathcal G_{n-1}(x_1,\cdots,x_{n-1})$, thus we need to show that $[\mathcal F_n,x_n]=\bar\alpha\mathcal F_{n-1}(x_1,\cdots,x_{n-1})$.

Let $\bar{y}_1,\cdots,\bar{y}_{n-1}$ be the (differential operators part) generators of $\mathbb{H}_{n-1}^{(K)}$, then the image of $\epsilon_1^{-n}y_1y_2\cdots y_n\mathbf e$ in the Dunkl representation can be written as
\begin{align}
    \left(\frac{\bar{y}_1}{\epsilon_1}+\frac{1}{x_1-x_n}s^x_{1,n}\right)\cdots \left(\frac{\bar{y}_{n-1}}{\epsilon_1}+\frac{1}{x_{n-1}-x_n}s^x_{n-1,n}\right)\left(\bar\alpha\partial_{x_n}+\sum_{i=1}^{n-1}\frac{1}{x_n-x_i}s^x_{n,i}\right)[s_{i,j}^x\mapsto \Omega_{i,j}],
\end{align}
where $\bar{y}_1,\cdots ,\bar{y}_{n-1}$ are identified with their corresponding Dunkl operators, and the last term $[s_{i,j}^x\mapsto \Omega_{i,j}]$ means whenever a permutation operator $s_{i,j}^x$ moves to the rightmost side, it becomes the quadratic Casimir $\Omega_{i,j}$. It is easy to see from the above expression that $\mathcal F_n$ can be schematically written as
\begin{align}
    \mathcal F_n(x_1,\cdots,x_n)=\bar\alpha\mathcal F_{n-1}(x_1,\cdots,x_{n-1})\partial_{x_n}+\text{diff. ops. which do not involve }\partial_{x_n}.
\end{align}
Hence $[\mathcal F_n,x_n]=\bar\alpha\mathcal F_{n-1}(x_1,\cdots,x_{n-1})$ and our previous claim is proven.

Next we cap the equation \eqref{eqn: infinity intertwiner_DDCA} with covacuum $\langle 0|$, and take $f=\mathsf t_{2,0}$, and get:
\begin{equation}
    \begin{split}
        \langle 0|\mathcal L_n(x_1,\cdots,x_n)|0\rangle \rho_n(\mathsf t_{2,0})&=\rho_n(\mathsf t_{2,0})\langle 0|\mathcal L_n(x_1,\cdots,x_n)|0\rangle\\
        +\langle 0|\widetilde\Psi_1(\mathsf t_{2,0})\mathcal L_n(x_1,\cdots,x_n)|0\rangle
    &- 2\epsilon_1\langle 0|\sum_{i=1}^n\partial J^a_b(x_i)_{+}E^b_{a,i}\mathcal L_n(x_1,\cdots,x_n)|0\rangle.
    \end{split}
\end{equation}
Since $\deg \widetilde\Psi_1(\mathsf t_{2,0})=-2$ and $\langle 0|J^a_b(x_i)_{+}=0$, only the first term on the right hand side of the above equation is nonzero. In other word $[\rho_n(\mathsf t_{2,0}),\mathcal G_n]=0$. On the other hand, $[\rho_n(\mathsf t_{2,0}),y_1\cdots y_n\mathbf e]=0$ in $\mathrm{S}\mathbb H_{n}^{(K)}$, thus $[\rho_n(\mathsf t_{2,0}),\mathcal F_n]=0$ and whence $[\rho_n(\mathsf t_{2,0}),r_n]=0$.

Finally, combining the two commutation relations $[\rho_n(\mathsf t_{2,0}),r_n]=0$ and $[\rho_n(\mathsf t_{0,2}),r_n]=0$, we see that  $[\rho_n(\mathsf t_{1,1}),r_n]=0$. However, $\rho_n(\mathsf t_{1,1})$ is the sum of the Euler vector field $\sum_{i=1}^nx_i\partial_{x_i}$ and a scalar $n/2$, so $[\rho_n(\mathsf t_{1,1}),r_n]=-nr_n$ because $r_n$ has conformal dimension $n$. This forces $r_n$ to be zero, which proves our induction step.
\end{proof}

Let $u\in \mathbb C$, then the correlation function
\begin{align*}
    \mathcal T_u(x_1,\cdots,x_n):=\langle 0|(\bar\alpha\partial_{x_1}+u+J^a_b(x_1)E^b_{a,1})\cdots(\bar\alpha\partial_{x_n}+u+J^a_b(x_n)E^b_{a,n})|0\rangle
\end{align*}
can be written as
\begin{align}
    \mathcal T_u(x_1,\cdots,x_n)=\prod_{i=1}^n(u+\epsilon_1^{-1}y_i)\mathbf e\in \mathrm{S}\mathbb H_{n}^{(K)}
\end{align}
via the Dunkl representation of the right hand side. This can be deduced from Corollary \ref{cor: correlator of Miura} as follows:
\begin{align*}
\mathcal T_u(x_1,\cdots,x_n)&=e^{-\bar\alpha^{-1}u(x_1+\cdots+x_n)}\mathcal T_0(x_1,\cdots,x_n)e^{\bar\alpha^{-1}u(x_1+\cdots+x_n)}\\
&=\frac{1}{\epsilon_1^{n}}e^{-\bar\alpha^{-1}u(x_1+\cdots+x_n)}y_1y_2\cdots y_ne^{\bar\alpha^{-1}u(x_1+\cdots+x_n)}\mathbf e\\
&=\frac{1}{\epsilon_1^{n}}(y_1+\epsilon_1 u)(y_2+\epsilon_1 u)\cdots (y_n+\epsilon_1 u)\mathbf e.
\end{align*}
Then it follows that the operators $\mathcal T_u(x_1,\cdots,x_n)$ commute with each other, i.e. for $u,v\in \mathbb C$, $[\mathcal T_u,\mathcal T_v]=0$. In other words, $\mathcal T_u(x_1,\cdots,x_n)$ is the generating function of $n$ commuting differential operators of degrees $1,2,\cdots,n$. The degree one and two differential operators in this hierarchy are
\begin{align}
    \sum_{i=1}^n\bar\alpha\partial_{x_{i}}\quad\text{ and }\quad \sum_{i<j}^n \left[\bar\alpha^2\partial_{x_i}\partial_{x_j}+\frac{\bar\alpha \Omega_{ij}+1}{(x_i-x_j)^2}\right].
\end{align}
And we recognize that $\mathcal T_u(x_1,\cdots,x_n)$ is the generating function of higher Calogero-Sutherland Hamiltonians of $n$-particle system with $\mathfrak{gl}_K$ internal symmetry on a plane. In the case of $K=1$, Feynman showed the commutativity $[\mathcal T_u,\mathcal T_v]=0$ by elementary method, see \cite{polychronakos2019feynman}.\\

Similarly, we can also look at Miura operators on a cylinder coordinate $\theta=\log(x)$:
\begin{equation}
\begin{split}
\bar\alpha \partial_\theta+\sum_{m\in \mathbb Z}J^a_{b,m}E^b_a e^{-m\theta}=x\mathcal L^1_1(x).
\end{split}
\end{equation}
The generating function of their correlators
\begin{align}
\mathbb T_u(\theta_1,\cdots,\theta_n):=\langle 0|(\bar\alpha \partial_{\theta_1}+u+J^a_{b}(\theta_1)E^b_{a,1})\cdots (\bar\alpha \partial_{\theta_n}+u+J^a_{b}(\theta_n)E^b_{a,n})|0\rangle
\end{align}
is expected to be related to higher Calogero-Sutherland Hamiltonians of $n$-particle system with $\mathfrak{gl}_K$ internal symmetry on a cylinder \cite{polychronakos2019feynman}. This is indeed the case. Recall the definition of Cherednik operator \cite{bernard1993yang}.
\begin{definition}\label{def: Cherednik operators}
The Cherednik operator $\mathcal D_i,i=1,\cdots,N$ in the extended trignometric Cherednik algebra $\mathbb H^{(K)}_N$ is defined as\footnote{In the literature \cite{bernard1993yang}, the rescaled Cherednik operator $\epsilon_2^{-1}\mathcal D_i$ is denoted by $\widehat{D}_i$, and $D_i$ in \emph{loc. cit.} is the operator $\epsilon_2^{-1}x_iy_i$ in this paper. Note that the parameter $\lambda$ in \cite{bernard1993yang} matches with our $\bar\alpha^{-1}$.}
    \begin{align}
        \mathcal D_i=x_iy_i-\epsilon_1\sum_{j<i}s_{ij}\Omega_{ij}.
    \end{align}
One can easily verify that Cherednik operators commute with each other, i.e. $[\mathcal D_i,\mathcal D_j]=0$. 
\end{definition}

\begin{corollary}\label{cor: correlator of Miura_trigonometric}
The correlator between Miura operators on a cylinder can be written as
\begin{align}\label{eqn: correlator of Miura_trigonometric}
    \mathbb T_u(\theta_1,\cdots,\theta_n)=\prod_{i=1}^n(u+\epsilon_1^{-1}\mathcal D_i)\mathbf e
\end{align}
via the Dunkl representation of the right hand side. 
\end{corollary}

\begin{proof}
The case when $u=0$ follows from Corollary \ref{cor: correlator of Miura} together with the following equation in $\mathbb H_{n}^{(K)}$:
\begin{align}\label{eqn: Cherednik operators identity}
    \prod_{i=1}^nx_i\prod_{j=1}^n y_j=\prod_{i=1}^n\mathcal D_i.
\end{align}
Equation \eqref{eqn: Cherednik operators identity} can be shown by direct computation:
\begin{align*}
    &\prod_{i=1}^nx_i\prod_{j=1}^n y_j=\prod_{i=1}^{n-1}x_i\prod_{j=1}^{n-1} y_j x_ny_n+\prod_{i=1}^{n-1}x_i\left[x_n,\prod_{j=1}^{n-1} y_j \right]y_n\\
    &~=\prod_{i=1}^{n-1}x_i\prod_{j=1}^{n-1} y_j x_ny_n-\epsilon_1\prod_{i=1}^{n-1}x_i\left(\sum_{j=1}^{n-1}y_1\cdots\hat{y}_j\cdots y_{n-1}s_{jn}\Omega_{jn}\right)y_n\\
    &~=\prod_{i=1}^{n-1}x_i\prod_{j=1}^{n-1} y_j x_ny_n-\epsilon_1\prod_{i=1}^{n-1}x_i\left(\sum_{j=1}^{n-1}y_1\cdots y_i\cdots y_{n-1}s_{jn}\Omega_{jn}\right)\\
    &~=\prod_{i=1}^{n-1}x_i\prod_{j=1}^{n-1} y_j \mathcal D_n\\
    &~=\prod_{i=1}^{n-2}x_i\prod_{j=1}^{n-2} y_j \mathcal D_{n-1}\mathcal D_n=\cdots =\prod_{i=1}^n\mathcal D_i.
\end{align*}
The general case can be derived from the $u=0$ case by the automorphism $\partial_{x_i}\mapsto \partial_{x_i}+\bar\alpha^{-1}u /x_i$ of $(D_{\epsilon_2}(\mathbb C^{\times n}_{\mathrm{disj}})\otimes \gl_K^{\otimes n})^{\mathfrak{S}_n}$.
Dunkl representation intertwines this automorphism with the automorphism $y_i\mapsto y_i+\epsilon_1 u/x_i$ of $\mathrm{S}\mathbb{H}_{n}^{(K)}$, hence $\mathcal D_i$ is mapped to $\mathcal D_i+\epsilon_1u$, whence the equation \eqref{eqn: correlator of Miura_trigonometric} follows.
\end{proof}

Then it follows from \eqref{eqn: correlator of Miura_trigonometric} that the operators $\mathbb T_u(\theta_1,\cdots,\theta_n)$ commute with each other, i.e. for $u,v\in \mathbb C$, $[\mathbb T_u,\mathbb T_v]=0$. In other words, $\mathbb T_u(\theta_1,\cdots,\theta_n)$ is the generating function of $n$ commuting differential operators of degrees $1,2,\cdots,n$. The degree one and two differential operators in this hierarchy are
\begin{align}
    \sum_{i=1}^n\bar\alpha\partial_{\theta_{i}}\quad\text{ and }\quad \sum_{i<j}^n \left[\bar\alpha^2\partial_{\theta_i}\partial_{\theta_j}+\frac{\bar\alpha \Omega_{ij}+1}{4\sinh^2 \left(\frac{\theta_i-\theta_j}{2}\right)}\right].
\end{align}
And we recognize that $\mathbb T_{u}(z_1,\cdots,z_n)$ is the generating function of higher Calogero-Sutherland Hamiltonian of $n$-particle system with $\mathfrak{gl}_K$ internal symmetry on a cylinder. The $K=1$ case was discussed in \cite{prochazka2019instanton,polychronakos2019feynman}.

In \cite{bernard1993yang}, it was conjectured that the coefficients of the expansion $\prod_{i=1}^n(u+\epsilon_2^{-1}\mathcal D_i)\mathbf e=\sum_{p=0}^n C_p u^{n-p}$ are given by the following explicit formula
\begin{equation}
    \begin{split}
        C_p=\sum_{i_1<\cdots<i_p}\sum_{I\sqcup J=\{i_1,\cdots,i_p\}}F_J\prod_{i\in I}x_i\partial_{x_i},
    \end{split}
\end{equation}
with
\begin{equation}
    F_J=\sum_{\sqcup\{j_k,j_k'\}=J}\prod_{k}\left((\bar\alpha^{-1}\Omega_{j_k,j_k'}+\bar\alpha^{-2})\frac{x_{j_k}x_{j_k'}}{(x_{j_k}-x_{j_k'})^2}\right).
\end{equation}
Compare the expansions $\bar\alpha^{-n}\mathbb T_{\bar\alpha u}=\prod_{i=1}^n(u+\epsilon_2^{-1}\mathcal D_i)\mathbf e=\sum_{p=0}^n C_p u^{n-p}$ with respect to $u$, we can easily see that this conjecture is true when $K=1$, in this case $F_J$ is the Wick contraction formula of the correlator
\begin{align}
    \bar\alpha^{-|J|}\bigg\langle \prod_{i\in J}x_j\phi(x_j)\bigg\rangle,
\end{align}
where $\phi(x)$ is the field of Heisenberg algebra: $\phi(x)\phi(y)\sim \frac{-\alpha}{(x-y)^2}$. However, when $K>1$, the conjecture is not true. To correct it, we need to modify the $F_J$ to be the correlator
\begin{align}
    \bar\alpha^{-|J|}\bigg\langle \prod_{j\in J}x_j J^a_b(x_j)E^b_{a,j}\bigg\rangle,
\end{align}
which is not given by a simple Wick contraction formula \cite[Equation (54)]{gaberdiel2000axiomatic}.\\

As a corollary of the correlator $\langle 0|\mathcal L^L_n(x_1,\cdots,x_n)|0\rangle$ being an element in the spherical Cherednik algebra, we see that all the coefficients of the Miura operator $\mathcal L^L_n(x_1,\cdots,x_n)$ are elements in the spherical Cherednik algebra. 
\begin{proposition}\label{prop: Miura operator coeff in Cherednik alg}
$\mathcal L^L_n(x_1,\cdots,x_n)$ is an element in
\begin{align}
    \Hom_{\mathbb C[\epsilon_1,\epsilon_2]}(\widetilde{\mathcal W}^{(K)}_L,\mathrm{S}\mathbb{H}_{n}^{(K)}\widetilde{\otimes}\widetilde{\mathcal W}^{(K)}_L),
\end{align}
where $\mathrm{S}\mathbb{H}_{n}^{(K)}$ is identified with a subspace of $(D_{\epsilon_2}(\mathbb C^{\times n}_{\mathrm{disj}})\otimes \gl_K^{\otimes n})^{\mathfrak{S}_n}$ via the Dunkl embedding.
\end{proposition}

\begin{proof}
Since $\mathcal L^L_n(x_1,\cdots,x_n)=\mathcal L^1_n(x_1,\cdots,x_n)^{[1]}\cdots \mathcal L^1_n(x_1,\cdots,x_n)^{[L]}$, it suffices to show that $\mathcal L^1_n(x_1,\cdots,x_n)$ is an element in $\Hom_{\mathbb C[\epsilon_1,\epsilon_2]}(\widetilde{\mathcal W}^{(K)}_1,\mathrm{S}\mathbb{H}_{n}^{(K)}\widetilde{\otimes}\widetilde{\mathcal W}^{(K)}_1)$. Consider the correlator
\begin{align*}
    \mathcal G_{r,n,s}&=\langle 0|\mathcal L^1_r(z_1,\cdots,z_r)\mathcal L^1_n(x_1,\cdots,x_n)\mathcal L^1_s(w_1,\cdots,w_s)|0\rangle\\
    &=\langle 0|\mathcal L^1_{r+n+s}(z_1,\cdots,z_r,x_1,\cdots,x_n,w_1,\cdots,w_s)|0\rangle,
\end{align*}
then expand $\mathcal G_{r,n,s}$ in the region $|z_i|>|x_j|>|w_k|$ for all $1\le i\le r,1\le j\le n,1\le k\le s$, then the expansion takes value in 
\begin{align*}
\mathrm{S}\mathbb{H}_{r}^{(K)}\widetilde{\otimes}\mathrm{S}\mathbb{H}_{n}^{(K)}\widetilde{\otimes}\mathrm{S}\mathbb{H}_{s}^{(K)}.
\end{align*}
By singling out terms which do not involve $\partial_{z_i}$ or $\partial_{w_j}$ for any $1\le i\le r,1\le j\le s$, we see that 
the correlation functions $\langle 0|J^{a_1}_{b_1}(z_1)\cdots J^{a_r}_{b_r}(z_r)\mathcal L_n(x_1,\cdots,x_n)J^{c_1}_{d_1}(w_1)\cdots J^{c_s}_{d_s}(w_s)|0\rangle$ are in the space $\mathscr O(\mathbb C^{\times r}_{\mathrm{disj}}) \widetilde{\otimes}\mathrm{S}\mathbb{H}_{n}^{(K)}\widetilde{\otimes} \mathscr O(\mathbb C^{\times s}_{\mathrm{disj}})$ for all $a_1,b_1,\cdots,a_r,b_r,c_1,d_1,\cdots,c_s,d_s$. Taking the Fourier modes in terms of variables $z_1,\cdots,z_s,w_1,\cdots, w_s$, we see that for any pair of vectors $|v_1\rangle$ in the vacuum module and $\langle v_2|$ in the dual vacuum module, we have
\begin{align*}
    \langle v_2|\mathcal L_n(x_1,\cdots,x_n)|v_1\rangle\in \mathrm{S}\mathbb{H}_{n}^{(K)}.
\end{align*}
This implies that $\mathcal L^1_n(x_1,\cdots,x_n)\in\Hom_{\mathbb C[\epsilon_1,\epsilon_2]}(\widetilde{\mathcal W}^{(K)}_1,\mathrm{S}\mathbb{H}_{n}^{(K)}\widetilde{\otimes}\widetilde{\mathcal W}^{(K)}_1)$, from which the general case follows.
\end{proof}

\section{Compare \texorpdfstring{$\mathsf Y^{(K)}$}{Yk} with Affine Yangian of Type \texorpdfstring{$A_{K-1}$}{A(k-1)}}
In this section, we compare our algebra $\mathsf Y^{(K)}$ with a more conventional algebra, known as the affine Yangian algebra of type $A_{K-1}$. We begin with a brief review of the definition and the basic properties of affine Yangians of type $A_{K-1}$. We basically follow the literature \cite{ueda2019affine,ueda2022affine,guay2007affine,bershtein2019homomorphisms,tsymbaliuk2017affine,kodera2015affine}.

For a number $n\in \mathbb N_{>0}$, we use the notation $[n]:=\{0,1,\cdots,n-1\}$ and regard it as mod $n$ residues.
\begin{definition}
The affine Yangian $\mathbb{Y}^{(K)}$ is defined as $\mathbb C[\epsilon_1,\epsilon_2]$-algebra generated by $\{X^+_{i,r},X^-_{i,r},H_{i,r}\}_{i\in [K],r\in \mathbb Z_{\ge 0}}$ with relations specified as follows. Set $[x,y]=xy-yx$ and $\{x,y\}=xy+yx$ and $\epsilon_3=-K\epsilon_1-\epsilon_2$. Let $\{a_{ij}\}_{i,j\in [K]}$ be the Cartan matrix of type $A_{K-1}^{(1)}$.
\begin{equation}\label{eqn: Y0}
    [H_{i,r},H_{j,s}]=0,\quad [X^+_{i,r},X^-_{j,s}]=\delta_{ij}H_{i,r+s},\tag{Y0}
\end{equation}
and we specify \eqref{eqn: Y1}-\eqref{eqn: Y4} in the cases $K>2$ or $K=2$ or $K=1$ separately.\\
$\bullet$ Case $K>2$. Let $\{m_{ij}\}_{i,j\in [K]}$ be the following matrix
\begin{equation}
    m_{ij}=\begin{cases}
        1,& (i,j)=(0,1)\text{ or }(K-1,0)\\
        -1,& (i,j)=(1,0)\text{ or }(0,K-1)\\
        0,& \text{otherwise}.
    \end{cases}
\end{equation}
Then the relations \eqref{eqn: Y1}-\eqref{eqn: Y4} are:
\begin{equation}\label{eqn: Y1}
    [H_{i,r+1},X^{\pm}_{j,s}]-[H_{i,r},X^{\pm}_{j,s+1}]=\pm\frac{a_{ij}}{2}\epsilon_1\{H_{i,r},X^{\pm}_{j,s}\}+\frac{m_{ij}}{4}(\epsilon_2-\epsilon_3)[H_{i,r},X^{\pm}_{j,s}],\tag{Y1}
\end{equation}
\begin{equation}\label{eqn: Y2}
    [X^{\pm}_{i,r+1},X^{\pm}_{j,s}]-[X^{\pm}_{i,r},X^{\pm}_{j,s+1}]=\pm\frac{a_{ij}}{2}\epsilon_1\{X^{\pm}_{i,r},X^{\pm}_{j,s}\}+\frac{m_{ij}}{4}(\epsilon_2-\epsilon_3)[X^{\pm}_{i,r},X^{\pm}_{j,s}],\tag{Y2}
\end{equation}
\begin{equation}\label{eqn: Y3}
    [H_{i,0},X^{\pm}_{j,r}]=\pm a_{ij}X^{\pm}_{j,r},\tag{Y3}
\end{equation}
\begin{equation}\label{eqn: Y4}
    \underset{r_1,r_2}{\sym}\:[X^{\pm}_{i,r_1},[X^{\pm}_{i,r_2},X^{\pm}_{i\pm 1,s}]]=0, \text{ and }[X^{\pm}_{i,r},X^{\pm}_{j,s}]=0\text{ if }a_{ij}=0.\tag{Y4}
\end{equation}
$\bullet$ Case $K=2$. Then the relations \eqref{eqn: Y1_2}-\eqref{eqn: Y4_2} are:
\begin{equation}\label{eqn: Y1_2}
\begin{cases}
[H_{i,r+1},X^{\pm}_{j,s}]-[H_{i,r},X^{\pm}_{j,s+1}]=\pm\epsilon_1\{H_{i,r},X^{\pm}_{j,s}\}, & i=j\\
\begin{aligned}
&[H_{i,r+2},X^{\pm}_{j,s}]-2[H_{i,r+1},X^{\pm}_{j,s+1}]+[H_{i,r},X^{\pm}_{j,s+2}]=\\
&-\frac{\epsilon_2\epsilon_3}{4}[H_{i,r},X^{\pm}_{j,s}]\mp \epsilon_1(\{H_{i,r+1},X^{\pm}_{j,s}\}-\{H_{i,r},X^{\pm}_{j,s+1}\})
\end{aligned}, & i\neq j
\end{cases}\tag{Y1}
\end{equation}

\begin{equation}\label{eqn: Y2_2}
\begin{cases}
[X^{\pm}_{i,r+1},X^{\pm}_{j,s}]-[X^{\pm}_{i,r},X^{\pm}_{j,s+1}]=\pm\epsilon_1\{X^{\pm}_{i,r},X^{\pm}_{j,s}\}, & i=j\\
\begin{aligned}
&[X^{\pm}_{i,r+2},X^{\pm}_{j,s}]-2[X^{\pm}_{i,r+1},X^{\pm}_{j,s+1}]+[H_{i,r},X^{\pm}_{j,s+2}]=\\
&-\frac{\epsilon_2\epsilon_3}{4}[X^{\pm}_{i,r},X^{\pm}_{j,s}]\mp \epsilon_1(\{X^{\pm}_{i,r+1},X^{\pm}_{j,s}\}-\{X^{\pm}_{i,r},X^{\pm}_{j,s+1}\})
\end{aligned}, & i\neq j
\end{cases}\tag{Y2}
\end{equation}

\begin{equation}\label{eqn: Y3_2}
    [H_{i,0},X^{\pm}_{j,r}]=\pm a_{ij}X^{\pm}_{j,r},\quad [H_{i,1},X^{\pm}_{i+1,r}]=\mp (2X^{\pm}_{i+1,r+1}+\epsilon_1\{H_{i,0},X^{\pm}_{i+1,r}\}),\tag{Y3}
\end{equation}

\begin{equation}\label{eqn: Y4_2}
    \underset{r_1,r_2,r_3}{\sym}\:[X^{\pm}_{i,r_1},[X^{\pm}_{i,r_2},[X^{\pm}_{i,r_3},X^{\pm}_{i+ 1,s}]]]=0.\tag{Y4}
\end{equation}
$\bullet$ Case $K=1$. Then the relations \eqref{eqn: Y1'}-\eqref{eqn: Y4'} are:
\begin{equation}\label{eqn: Y1'}
\begin{split}
    &[H_{0,r+3},X^{\pm}_{0,s}]-3[H_{0,r+2},X^{\pm}_{0,s+1}]+3[H_{0,r+1},X^{\pm}_{0,s+2}]-[H_{0,r},X^{\pm}_{0,s+3}]=\\
    &-\sigma_2([H_{0,r+1},X^{\pm}_{0,s}]-[H_{0,r},X^{\pm}_{0,s+1}])\pm \sigma_3\{H_{0,r},X^{\pm}_{0,s}\},
\end{split}\tag{Y1}
\end{equation}
\begin{equation}\label{eqn: Y2'}
    \begin{split}
    &[X^{\pm}_{0,r+3},X^{\pm}_{0,s}]-3[X^{\pm}_{0,r+2},X^{\pm}_{0,s+1}]+3[X^{\pm}_{0,r+1},X^{\pm}_{0,s+2}]-[X^{\pm}_{0,r},X^{\pm}_{0,s+3}]=\\
    &-\sigma_2([X^{\pm}_{0,r+1},X^{\pm}_{0,s}]-[X^{\pm}_{0,r},X^{\pm}_{0,s+1}])\pm \sigma_3\{X^{\pm}_{0,r},X^{\pm}_{0,s}\},
\end{split}\tag{Y2}
\end{equation}
\begin{equation}\label{eqn: Y3'}
    [H_{0,0},X^{\pm}_{0,r}]=[H_{0,1},X^{\pm}_{0,r}]=0,\quad [H_{0,2},X^{\pm}_{0,r}]=\pm 2 X^{\pm}_{0,r},\tag{Y3}
\end{equation}
\begin{equation}\label{eqn: Y4'}
    \underset{r_1,r_2,r_3}{\sym}\:[X^{\pm}_{0,r_1},[X^{\pm}_{0,r_2},X^{\pm}_{0,r_3+1}]]=0,\tag{Y4}
\end{equation}
where we set $\sigma_2=\epsilon_1\epsilon_2+\epsilon_2\epsilon_3+\epsilon_3\epsilon_1$ and $\sigma_3=\epsilon_1\epsilon_2\epsilon_3$.
\end{definition}

\begin{remark}\label{rmk: compare different presentations}
Our presentation of affine Yangian for $K>2$ is aligned with \cite{guay2007affine}, such that the parameters $(\lambda,\beta)$ in \cite{guay2007affine} are related to our $\epsilon_1,\epsilon_2,\epsilon_3$ by
\begin{align*}
    \epsilon_1=\lambda,\quad \frac{\epsilon_2-\epsilon_3}{2}=\lambda-2\beta.
\end{align*}
A slightly different presentation is given in \cite{kodera2022coproduct} and it is related to ours by a re-parametrization as follows. Let $x^{\pm}_{i,r},h_{i,r}$ denote the generators in \cite[Definition 6.1]{kodera2022coproduct}. We set $\epsilon_1=\hbar$ and $\epsilon_2=\epsilon$, and set
\begin{align*}
X^\pm_i(z):=\sum_{r\in \mathbb N}X^{\pm}_{i,r}z^{-r-1},&\quad H_i(z):=1+\sum_{r\in \mathbb N}H_{i,r}z^{-r-1},\\
    x^\pm_i(z):=\sum_{r\in \mathbb N}x^{\pm}_{i,r}z^{-r-1},&\quad h_i(z):=1+\sum_{r\in \mathbb N}h_{i,r}z^{-r-1},
\end{align*}
then the isomorphism is given by 
\begin{align*}
    X^{\pm}_{0}(z)&=x^{\pm}_0\left(z-\frac{2\epsilon+K\hbar}{4}\right),& H_{0}(z)&=h_0\left(z-\frac{2\epsilon+K\hbar}{4}\right),\\
    X^{\pm}_{i}(z)&=x^{\pm}_i(z),& H_{i}(z)&=h_{i}(z),
\end{align*}
for $1\le i\le K-1$.

Another presentation of $\mathbb Y^{(K)}$ shows up in the context of quiver Yangian \cite{li2020quiver,galakhov2020quiver,galakhov2021shifted,bao2022note}.
\end{remark}

\begin{definition}
The loop Yangian $\mathbb{L}^{(K)}$ is the quotient of $\mathbb{Y}^{(K)}$ by the central element $\mathfrak c=\sum_{i\in [K]}H_{i,0}$.
\end{definition}

\begin{proposition}[{\cite[Proposition 3.1]{guay2007affine}}]\label{prop: minimal generators}
If $K>2$, then $\mathbb Y^{(K)}$ is generated by $\{X^{\pm}_{i,r},H_{i,r}\}_{i\in [K],r\in \{0,1\}}$ subject to relations \eqref{eqn: Y0}-\eqref{eqn: Y4} restricted to $r=0,1$.
\end{proposition}

When $K>2$, Guay proved a PBW theorem for $\mathbb{Y}^{(K)}$ \cite[6.1]{guay2007affine}, in particular $\mathbb{Y}^{(K)}$ is a free $\mathbb C[\epsilon_1,\epsilon_2]$-module in the cases $K>2$.

We observe that $X^{\pm}_{i,0},H_{i,0}$ for $0\le i\le K-1$ generates a copy of $U(\widehat{\mathfrak{sl}}_K)$ inside $\mathbb{Y}^{(K)}$. Another observation is that $X^{\pm}_{i,r},H_{i,r}$ for $i\neq 0, r\in \mathbb N$ generates a subalgebra inside $\mathbb{Y}^{(K)}$ which is isomorphic to the Yangian $Y_{\epsilon_1}({\mathfrak{sl}}_K)$ \cite[6.1]{guay2007affine}. $Y_{\epsilon_1}({\mathfrak{sl}}_K)$ has another presentation with generators $x,J(x)$ for $x\in {\mathfrak{sl}}_K$ \cite{chari1995guide}, and $J(x)$ are related to $X^{\pm}_{i,r},H_{i,r}$ by the following
\begin{equation}\label{eqn: J generators}
    \begin{split}
    J(X_i^{\pm})=X^{\pm}_{i,1}+\epsilon_1\omega_i^{\pm},\;&\text{ where }\omega_i^{\pm}=\pm\frac{1}{4}\sum_{\alpha\in \Delta^+}\{[X^{\pm}_i,X_{\pm \alpha}],X_{\mp \alpha}\}-\frac{1}{4}\{X^{\pm}_i,H_i\}\\
        J(H_i)=H_{i,1}+\epsilon_1\nu_i,\;&\text{ where }\nu_i=\frac{1}{4}\sum_{\alpha\in \Delta^+}(\alpha,\alpha_i)\{X_{ \alpha},X_{-\alpha}\}-\frac{1}{2}H_i^2.
    \end{split}
\end{equation}
\begin{definition}\label{def: J generators and modes}
    Assume that $K>2$, then for $X\in \mathfrak{sl}_K$ we denote by $J(X)$ the $J$-generator of the subalgebra $Y_{\epsilon_1}({\mathfrak{sl}}_K)\subset \mathbb Y^{(K)}$ determined by \eqref{eqn: J generators}, and denote by $X\otimes u^n$ the elements in the subalgebra $U(\widehat{\mathfrak{sl}}_K)\subset \mathbb Y^{(K)}$ determined by 
    \begin{equation}
        \begin{split}
            X^{\pm}_i\otimes u^0=X^{\pm}_{i,0}&, \quad H_i\otimes u^0=H_{i,0}, \\
            E^K_1\otimes u=X^{+}_{0,0}&,\quad E^1_K\otimes u^{-1}=X^{-}_{0,0}.
        \end{split}
    \end{equation}
\end{definition}

\subsection{Compare loop Yangian with \texorpdfstring{$\mathsf L^{(K)}$}{Lk}}
Recall that we introduce an algebra $\mathsf L^{(K)}$ in Definition \ref{def: the algebra L^K}. In this subsection we show that it is closely related to the loop Yangian $\mathbb L^{(K)}$, at least when $K\neq 2$.

\subsubsection*{The case $K>2$}

When $K>2$, Guay constructed a family of representations of $\mathbb{L}^{(K)}$ by mapping it to the matrix extended trigonometric Cherednik algebra \cite{guay2005cherednik}, and we recall his construction here. Fix $N\in \mathbb N$, and let $\mathbb H^{(K)}_N$ (resp. $\mathrm{S}\mathbb{H}^{(K)}_N$) be the extended trigonometric Cherednik algebra (resp. its spherical subalgebra) defined in Section \ref{sec: L^K and coproducts}.

\begin{proposition}[{\cite[Theorem 5.2]{guay2005cherednik},\cite[Corollary 7.2]{guay2007affine}}]\label{prop: Schur-Weyl for affine Yangian}
Assume $K>2$, then for every $N\in \mathbb Z_{\ge 1}$ there exists an algebra homomorphism $\rho'_N :\mathbb{L}^{(K)}\to \mathrm{S}\mathbb{H}^{(K)}_N$ which maps the generators as:
\begin{equation}\label{eqn: rep in trig Cherednik}
    \begin{split}
        \rho'_N(X\otimes u^n)= \sum_{j=1}^N X_jx_j^n \mathbf e,\quad
        \rho'_N(J(X))= \sum_{j=1}^N X_{j}\mathcal Y_j \mathbf e, \quad
        \rho'_N(X^-_{0,1})= \sum_{j=1}^N E^K_{1,j}y_j \mathbf e
    \end{split}
\end{equation}
where $X\in\mathfrak{sl}_K$, and $\mathcal Y_i=\frac{1}{2}\{x_i,y_i\}\in \mathbb H_{N}^{(K)}$, and $J(X)$ is defined in equation \eqref{eqn: J generators}. Moreover, the intersection of kernels of $\rho'_N$ for all $N$ is trivial, i.e. $\ker(\prod_N \rho'_N)=0$.
\end{proposition}

Combining the above proposition with our Lemma \ref{lem: rho_N extends to loop Yangian}, we conclude that when $K>2$ there is a $\mathbb C[\epsilon_1,\epsilon_2]$-algebra embedding $\mathbb L^{(K)}\hookrightarrow \mathsf L^{(K)}$ given by 
\begin{equation}\label{eqn: loop Yangian into L^(K)}
    \begin{split}
        X\otimes u^n\mapsto \mathsf T_{0,n}(X),\quad
        J(X)\mapsto\mathsf T_{1,1}(X), \quad
        X^-_{0,1}\mapsto \mathsf T_{1,0}(E^K_{1}).
    \end{split}
\end{equation}
Note that the image of $\mathbb L^{(K)}$ generates $\mathsf L^{(K)}$ after localizing to $\mathbb C[\epsilon_1,\epsilon_2^\pm,\epsilon_3^\pm]$, by the Corollary \ref{cor: compare two DDCAs}.

\subsubsection*{The case $K=1$}

When $K=1$, it is known \cite{gaiotto2022miura,tsymbaliuk2017affine,schiffmann2013cherednik} that $\mathbb L^{(1)}$ embeds into the $\prod_N \mathrm{S}\mathbb{H}^{(1)}_N$ via the representations $\rho'_N: \mathbb L^{(1)}\to \mathrm{S}\mathbb{H}^{(1)}_N$ which is uniquely determined by 
\begin{equation}
\begin{split}
&\rho'_N\left(\mathrm{ad}_{X^+_{0,1}}^{n-1}X^+_{0,0}\right)=\frac{(n-1)!}{\epsilon_2}\sum_{i=1}^N x^n_i,\quad \rho'_N(X^-_{0,0})=-\frac{1}{\epsilon_2}\sum_{i=1}^N x^{-1}_i,\\
&\rho'_N(H_{0,1})= \frac{N}{\epsilon_2},\quad \rho'_N([X^-_{0,1},X^-_{0,2}])=\frac{1}{\epsilon_2}\sum_{i=1}^N y_i^2.
\end{split}
\end{equation}
Compare with $\rho_N: \mathsf L^{(1)}\to \mathrm{S}\mathbb{H}^{(1)}_N$, and we obtain an algebra embedding $\mathsf L^{(1)}\hookrightarrow \mathbb L^{(1)}$ such that 
\begin{align}\label{eqn: map L to loop Yangian}
    \mathsf t_{0,0}\mapsto H_{0,1},\quad\mathsf t_{0,-1}\mapsto -X^{-}_{0,0},\quad \mathsf t_{0,n}\mapsto \frac{1}{(n-1)!}\mathrm{ad}_{X^+_{0,1}}^{n-1}X^+_{0,0},\quad\mathsf t_{2,0}\mapsto [X^-_{0,1},X^-_{0,2}].
\end{align}
It is easy to see that the image of $\mathsf L^{(1)}$ generates $\mathbb L^{(1)}$, thus $\mathsf L^{(1)}$ is isomorphic to $\mathbb L^{(1)}$.

The following proposition summarize the discussions in this subsection.

\begin{proposition}
If $K=1$, then $\mathbb L^{(1)}\cong \mathsf L^{(1)}$. If $K>2$, then there is an embedding $\mathbb L^{(K)}\hookrightarrow \mathsf L^{(K)}$ such that it becomes an isomorphism after localizing to $\mathbb C[\epsilon_1,\epsilon_2^\pm,\epsilon_3^\pm]$.
\end{proposition}

\subsection{Imaginary 1-shifted affine Yangian algebras of type \texorpdfstring{$A_{K-1}$}{A(k-1)}}
Let $\lambda$ be a \textit{dominant} coweight of the root system $A^{(1)}_{K-1}$, we define the shifted affine Yangian \footnote{See also \cite{braverman2016coulomb} for the discussion of shifted finite Yangians.} $\mathbb Y^{(K)}_\lambda$ as the subalgebra of $\mathbb Y^{(K)}$ generated by $\{X^+_{i,r},H_{i,r}\}_{i\in [K],r\in \mathbb Z_{\ge 0}}$ and $\{X^-_{i,r}\}_{i\in [K],r\in \mathbb Z_{\ge \langle\lambda,\alpha_i\rangle}}$.

We only focus on the case $\mathbb Y^{(K)}_\delta$ where $\delta$ is the imaginary fundamental coweight, i.e. $\langle\delta,\alpha_i\rangle=\delta_{i0}$. Then $\mathbb Y^{(K)}_\delta$ is generated by $\{X^{\pm}_{i,r},H_{i,r}, X^+_{0,r}\}_{i\neq 0,r\in \mathbb Z_{\ge 0}}$ and $\{X^-_{0,r}\}_{i\in [K],r\in \mathbb Z_{>0}}$. Using the notation in Definition \ref{def: J generators and modes}, $\mathbb Y^{(K)}_\delta$ is generated by the subalgebra $Y_{\epsilon_1}(\mathfrak{sl}_K)$ and non-negative modes $\mathfrak{sl}_K[u]\subset U(\widehat{\mathfrak{sl}}_K)$ and $X^-_{0,1}$. $\mathbb Y^{(K)}_\delta$ was discussed in \cite[Definition 3.5]{guay2007affine} under the name deformed double current algebra.

\begin{lemma}\label{lem: shifted Yangian embeds into loop Yangian}
Assume that $K\neq 2$, then the projection $\mathbb Y^{(K)}_\delta\to \mathbb{L}^{(K)}$ is injective.
\end{lemma}

\begin{proof}
If $K>2$, then the lemma is a direct consequence of PBW theorem for $\mathbb Y^{(K)}$ and $\mathbb Y^{(K)}_\delta$ proved in \cite[Section 7]{guay2007affine}. More precisely, if we give $\mathbb Y^{(K)}$ (resp. $\mathbb Y^{(K)}_\delta$) a filtration by letting $\deg X^{\pm}_{i,r}=\deg H_{i,r}=r$, then it is shown in \emph{loc. cit.} that there are canonical isomorphism $U(\mathfrak{st}_K[t_1,t_2])\cong \mathrm{gr}\:\mathbb Y^{(K)}_\delta$ and $U(\mathfrak{st}_K[u^{\pm},v])\cong \mathrm{gr}\:\mathbb Y^{(K)}$ where $t_1=u,t_2=u^{-1}v$, where $\mathfrak{st}_K[t_1,t_2]$ is the universal central extension of the Lie algebra $\mathfrak{sl}_K[t_1,t_2]$ and similarly for $\mathfrak{st}_K[u^{\pm},v]$. It is shown in \cite[Théorème 1.7]{kassel1982extensions} (see also \cite[Theorem 3.1]{bloch2006dilogarithm}) that for a commutative algebra $A$ the Lie algebra homology $H_2(\mathfrak{sl}_K(A))$ is isomorphic to the second Hochschild homology $ HC_2(A)\cong \Omega^1(A)/dA$, and that the universal central extension of $\mathfrak{sl}_K(A)$ is given by the $2$-cocycle 
\begin{align}
    x\wedge y\mapsto \mathrm{Tr}(xdy)
\end{align}
from $\wedge^2 \mathfrak{sl}_K(A)$ to $\Omega^1(A)/dA$. The element $\mathbf c\in \mathfrak{st}_K[u^{\pm},v]$ corresponds to the one-form $u^{-1}du$, which is not an element of $\Omega^1(\mathbb C[u,u^{-1}v])$, so $\mathbf c\notin \mathfrak{st}_K[t_1,t_2]$. Then it follows that the composition $\mathfrak{st}_K[t_1,t_2]\hookrightarrow \mathfrak{st}_K[u^{\pm},v]\twoheadrightarrow \mathfrak{st}_K[u^{\pm},v]/\mathbf c$ is injective, so $\mathrm{gr}\:\mathbb Y^{(K)}_\delta\to \mathrm{gr}\:\mathbb{L}^{(K)}$ is injective, thus $\mathbb Y^{(K)}_\delta\to \mathbb{L}^{(K)}$ is injective.

The case $K=1$ follows from the Proposition \ref{prop: DDCA=shifted Yangian}.
\end{proof}

\begin{proposition}\label{prop: DDCA=shifted Yangian}
If $K=1$, then $\mathbb Y^{(1)}_\delta\cong \mathsf A^{(1)}$. If $K>2$, then $\mathbb Y^{(K)}_\delta\cong \mathbb D^{(K)}$.
\end{proposition}

\begin{proof}
The case $K=1$ is known in the literature (see for example \cite{gaiotto2022miura}), and the isomorphism is given by an explicit map between generators:
\begin{align}
    \mathsf t_{0,0}\mapsto H_{0,1},\quad\mathsf t_{2,0}\mapsto [X^-_{0,1},X^-_{0,2}],\quad \mathsf t_{0,n}\mapsto \frac{1}{(n-1)!}\mathrm{ad}_{X^+_{0,1}}^{n-1}X^+_{0,0}.
\end{align}
For the cases $K>2$, we claim that the following map of generators gives rise to an algebra isomorphism between $ \mathbb D^{(K)}$ and $\mathbb Y^{(K)}_\delta$
\begin{equation}\label{eqn: embedding into affine Yangian}
    \begin{split}
        \mathsf T_{0,n}(X)\mapsto X\otimes u^n,\quad \mathsf T_{1,1}(X)\mapsto J(X),\quad \mathsf T_{1,0}(E^K_1)\mapsto X^-_{0,1},
    \end{split}
\end{equation}
where $X\in\mathfrak{sl}_K$ (see Definition \ref{def: J generators and modes}). In fact, compare \eqref{eqn: rep in trig Cherednik} with Lemma \ref{lem: map rho_N} and we find that equations  $\rho_N(\mathsf T_{0,n}(X))=\rho'_N(X\otimes u^n), \rho_N(\mathsf T_{1,1}(X))=\rho'_N( J(X)),\rho_N(\mathsf T_{1,0}(E^K_1))=\rho'_N( X^-_{0,1})$ hold for all $N$, therefore the claim follows from Corollary \ref{cor: A is DDCA} and Proposition \ref{prop: Schur-Weyl for affine Yangian} and Lemma \ref{lem: shifted Yangian embeds into loop Yangian}.
\end{proof}

\subsection{Map from the affine Yangian to \texorpdfstring{$\mathsf Y^{(K)}$}{Yk}}\label{subsec: map from affine Yangian to Y}

If $K\neq 2$, then it is known that there exists algebra homomorphism $\Psi'_L: \mathbb Y^{(K)}\to \mathfrak{U}(\mathcal W^{(K)}_L)[\bar\alpha^{-1}]$ for every $L\in \mathbb N_{>0}$.

The case when $K=1$ is worked out by Schiffmann-Vasserot in \cite{schiffmann2013cherednik}, see also \cite{gaiotto2022miura,tsymbaliuk2017affine}. Explicitly, $\Psi'_L$ is uniquely determined by
\begin{equation}\label{eqn: map affine Yangian to W_K=1}
\begin{split}
&\Psi'_L\left(\mathrm{ad}_{X^+_{0,1}}^{n-1}X^+_{0,0}\right)=\frac{(n-1)!}{\epsilon_2}W^{(1)}_n,\quad \Psi'_L(X^-_{0,0})=-\frac{1}{\epsilon_2}W^{(1)}_{-1},\\
&\Psi'_L(H_{0,1})=\frac{1}{\epsilon_2} W^{(1)}_0,\quad \Psi'_L(H_{0,0})= \frac{L}{\epsilon_1\epsilon_2},\quad \Psi'_L([X^-_{0,1},X^-_{0,2}])=\Psi_L(\mathsf t_{2,0}).
\end{split}
\end{equation}
The maps $\Psi'_L$ promote uniquely to a map $\Psi'_\infty: \mathbb Y^{(1)}\to \mathfrak{U}(\mathsf W^{(1)}_{\infty})[\epsilon_2^{-1}]$, and the image of $\Psi'_\infty$ is contained in the image of $\Psi_\infty:\mathsf Y^{(1)}\to \mathfrak{U}(\mathsf W^{(1)}_{\infty})[\epsilon_2^{-1}]$, thus it induces a $\mathbb C[\epsilon_1,\epsilon_2]$-algbera map $f: \mathbb Y^{(1)}\to \mathsf Y^{(1)}$ which is uniquely determined by
\begin{equation}\label{eqn: isom Y^1}
\begin{split}
&f\left(\mathrm{ad}_{X^+_{0,1}}^{n-1}X^+_{0,0}\right)=(n-1)!\mathsf t_{0,n},\quad f(X^-_{0,0})=-\mathsf t_{0,-1},\\
&f(H_{0,1})=\mathsf t_{0,0},\quad f(H_{0,0})= \mathbf c,\quad f([X^-_{0,1},X^-_{0,2}])=\mathsf t_{2,0}.
\end{split}
\end{equation}
The image of $f$ generates $\mathsf Y^{(1)}$ by Lemma \ref{lem: A_+ and t[0,-1] generates Y}, so $f$ is surjective. On the other hand, modulo $H_{0,0}$, $f$ agrees with the inverse of the isomorphism $\mathsf L^{(1)}\cong \mathbb L^{(1)}$ in \eqref{eqn: map L to loop Yangian}, this implies that $f$ is also injective because $H_{0,0}$ has no zero-divisor in $\mathbb Y^{(1)}$ (by \cite[Proposition 1.36]{schiffmann2013cherednik}, the central term $H_{0,0}$ is added freely so that $\mathbb Y^{(1)}\cong \mathbb L^{(1)}\otimes \mathbb C[H_{0,0}]$ as a vector space), and $\mathbf c$ has no zero-divisor in $\mathsf Y^{(1)}$ (by the Theorem \ref{thm: PBW for Y}). Therefore $f$ is an isomorphism.\\

The case when $K>2$ is worked out in Kodera-Ueda \cite{kodera2022coproduct}. In \emph{loc. cit.}, a homomorphism $\Phi_L: \mathbb Y^{(K)}\to \mathfrak{U}(\mathcal W^{(K)}_L)$ is found for every $L\in \mathbb N_{>0}$. Recall the automorphism $\eta_{\beta}^{\otimes L}$ defined in Remark \ref{rmk: shift automorphism of W}, we define $$\Psi'_L:=\eta_{\frac{K}{4}+(\frac{1}{2}-L)\alpha}^{\otimes L}\circ \Phi_L,$$ which is uniquely determined by
\begin{equation}\label{eqn: map affine Yangian to W_K>2}
\begin{split}
\Psi'_L(X^+_{i,0})=\begin{cases}
    W^{i(1)}_{i+1,0},& i\neq 0\\
    W^{K(1)}_{1,1}, & i=0
\end{cases},&
\quad 
\Psi'_L(X^-_{i,0})=\begin{cases}
    W^{i+1(1)}_{i,0},& i\neq 0\\
    W^{1(1)}_{K,-1}, & i=0
\end{cases},\\
\Psi'_L(X^-_{0,1})=\Psi_L(\mathsf T_{1,0}(E^K_1)),&\quad \Psi'_L(J(X))= \Psi_L(\mathsf T_{1,1}(X)).
\end{split}
\end{equation}
The maps $\Psi'_L$ promote uniquely to a map $\Psi'_\infty: \mathbb Y^{(K)}\to \mathfrak{U}(\mathsf W^{(K)}_{\infty})$, and the image of $\Psi'_\infty$ is contained in the image of $\Psi_\infty:\mathsf Y^{(K)}\to \mathfrak{U}(\mathsf W^{(K)}_{\infty})[\epsilon_2^{-1}]$, thus it induces a $\mathbb C[\epsilon_1,\epsilon_2]$-algbera map $g: \mathbb Y^{(K)}\to \mathsf Y^{(K)}$ which is uniquely determined by
\begin{equation}\label{eqn: isom Y^K}
\begin{split}
&g(X^+_{i,0})=\begin{cases}
    \mathsf T_{0,0}(E^i_{i+1}),& i\neq 0\\
    \mathsf T_{0,1}(E^K_{1}), & i=0
\end{cases},
\quad 
g(X^-_{i,0})=\begin{cases}
    \mathsf T_{0,0}(E^{i+1}_i),& i\neq 0\\
    \mathsf T_{0,-1}(E^1_{K}), & i=0
\end{cases},\\
&g(X^-_{0,1})=\mathsf T_{1,0}(E^K_1),\quad g(J(X))= \mathsf T_{1,1}(X), \quad g(\mathfrak{c})=\epsilon_2\epsilon_3\mathbf c.
\end{split}
\end{equation}
Moreover, modulo $\mathfrak c$, $g$ agrees with the embedding $\mathbb L^{(K)}\hookrightarrow \mathsf L^{(K)}$ in \eqref{eqn: map L to loop Yangian}, this implies that $g$ is also injective because $\mathfrak c$ has no zero-divisor in $\mathbb Y^{(K)}$ (by the PBW theorem for $\mathbb Y^{(K)}$ \cite{guay2007affine}) and $\epsilon_2\epsilon_3\mathbf c$ has no zero-divisor in $\mathsf Y^{(K)}$ (by the Theorem \ref{thm: PBW for Y}). Finally, if we localize to $\mathbb C[\epsilon_1,\epsilon_2^\pm,\epsilon_3^\pm]$, then the image of $g$ generates $\mathsf Y^{(K)}$, by Corollary \ref{cor: compare two DDCAs}.

The above discussions are summarized in the following theorem.

\begin{theorem}\label{thm: compare affine Yangian with Y}
If $K=1$, then \eqref{eqn: isom Y^1} induces an isomorphism $\mathbb Y^{(1)}\cong \mathsf Y^{(1)}$. If $K>2$, then \eqref{eqn: isom Y^K} induces an embedding $\mathbb Y^{(K)}\hookrightarrow \mathsf Y^{(K)}$ such that it becomes an isomorphism after localizing to $\mathbb C[\epsilon_1,\epsilon_2^\pm,\epsilon_3^\pm]$.
\end{theorem}

\subsection{Compare coproducts}
When $K=1$, Schiffmann-Vasserot show that there exists a coproduct $\Delta:\mathbb Y^{(1)}\to \mathbb Y^{(1)}\widehat{\otimes}\mathbb Y^{(1)}$ \cite[Theorem 7.9]{schiffmann2013cherednik}, moreover $\Delta$ is compatible with the W-algebra coproduct $\Delta_{L_1,L_2}:\mathcal W^{(1)}_{L_1+L_2}\to \mathcal W^{(1)}_{L_1}\otimes\mathcal W^{(1)}_{L_2}$ \cite[Section 8.9]{schiffmann2013cherednik} in the sense that
\begin{align*}
    \Delta_{L_1,L_2}\circ \Psi'_{L_1+L_2}=(\Psi'_{L_1}\otimes\Psi'_{L_2})\circ \Delta,
\end{align*}
where $\Psi'_L$ is the map in \eqref{eqn: map affine Yangian to W_K=1}. The above compatibility promotes to an $L\to \infty$ version:
\begin{align*}
    \Delta_{\mathsf W}\circ \Psi'_{\infty}=(\Psi'_{\infty}\otimes \Psi'_{\infty})\circ \Delta.
\end{align*}
Since $\Psi'_{\infty}$ induces the isomorphism $\mathbb Y^{(1)}\cong \mathsf Y^{(1)}$ in the Theorem \ref{thm: compare affine Yangian with Y}, thus $\Delta$ agrees with our coproduct $\Delta_{\mathsf Y}$.\\

When $K>2$, a coproduct for $\mathbb Y^{(K)}$ was found by Guay in \cite{guay2007affine}, see also \cite{guay2018coproduct}. A variation of Guay's formulation was presented by Kodera-Ueda in \cite{kodera2022coproduct}, the difference is that \cite{kodera2022coproduct} uses a completed tensor product which is different from \cite{guay2018coproduct}, and the corproduct in \cite{kodera2022coproduct} is opposite to that of \cite{guay2007affine,guay2018coproduct}. We shall compare the coproduct in \cite{kodera2022coproduct} with our coproduct \eqref{eqn: YY coproduct}. We first recall their coproduct as follows. 

\begin{proposition}[{\cite[Theorem 7.1, Corollary 10.2]{kodera2022coproduct}}]
Let $K>2$, then there exists a coproduct $\Delta:\mathbb{Y}^{(K)}\to \mathbb{Y}^{(K)}\widetilde{\otimes} \mathbb{Y}^{(K)}$ such that
\begin{align}\label{eqn: affine Yangian coproduct vs W coproduct}
    (\eta^{\otimes L_1}_{-\alpha L_2}\otimes 1)\circ\Delta_{L_1,L_2}\circ \Phi_{L_1+L_2}=(\Phi_{L_1}\otimes\Phi_{L_2})\circ \Delta,
\end{align}
where $\Phi_L:\mathbb Y^{(K)}\to \mathfrak{U}(\mathcal W^{(K)}_L)$ is the map defined in \cite[Definition 9.1]{kodera2022coproduct}, and $\eta$ is the shift automorphism of $\mathfrak{U}(\mathcal W^{(K)}_L)$ defined in Remark \ref{rmk: shift automorphism of W}.
\end{proposition}

Apparently $\Delta$ does not agree with our $\Delta_{\mathsf Y}$ since the latter is compatible with W-algebra coproduct by the Definition \ref{def: YY coproduct}, whereas the former is compatible with a twisted coproduct of W-algebras. Nevertheless these two coproducts are closely related by the automorphism $\boldsymbol{\tau}$ in \eqref{eqn: shift automorphism of Y}, more precisely we have the following.

\begin{lemma}
The subalgebra $\mathbb{Y}^{(K)}\subset \mathsf{Y}^{(K)}$ is invariant under the automorphism $\boldsymbol{\tau}_{\beta}:\mathsf{Y}^{(K)}\cong \mathsf{Y}^{(K)}$
\end{lemma}

\begin{proof}
Using \eqref{eqn: shift automorphism of Y}, it is straightforward to compute that
\begin{align}
    \boldsymbol{\tau}_{\beta}(X^{\pm}_{i,0})=X^{\pm}_{i,0},\quad \boldsymbol{\tau}_{\beta}(X^{\pm}_{i,1})=X^{\pm}_{i,1}+\beta X^{\pm}_{i,0},
\end{align}
for all $0\le i\le K-1$. Since $\{X^{\pm}_{i,r}\}_{i\in [K],r=0,1}$ generates $\mathbb Y^{(K)}$, the lemma follows from the the above equation.
\end{proof}

\begin{remark}
For all $K\in \mathbb N_{>0}$, one can define an automorphism $\tau_{\beta}:\mathbb{Y}^{(K)}\cong \mathbb{Y}^{(K)}$ by letting
\begin{align}
    \tau_{\beta}(X^{\pm}_{i}(z))=X^{\pm}_{i}(z-\beta),\quad \tau_{\beta}(H_{i}(z))=H_{i}(z-\beta),
\end{align}
where $X^{\pm}_{i}(z):=\sum_{r\in \mathbb N}X^{\pm}_{i,r}z^{-r-1}$ and $H_{i}(z):=1+\sum_{r\in \mathbb N}H_{i,r}z^{-r-1}$ are generating functions. When $K\neq 2$, $\tau_{\beta}$ agrees with the restriction of $\boldsymbol{\tau}_{\beta}$ to the subalgebra $\mathbb{Y}^{(K)}\subset \mathsf{Y}^{(K)}$.
\end{remark}

Let us set 
\begin{align}
    \Delta':=(\mathrm{id}\otimes \boldsymbol{\tau}_{-\mathfrak{c}\otimes 1})\circ \Delta,
\end{align}
then we have the following.
\begin{proposition}
The coproduct $\Delta_{\mathsf Y}:\mathsf Y^{(K)}\to \mathsf Y^{(K)}\widehat{\otimes}\mathsf Y^{(K)}$ maps the subalgebra $\mathbb Y^{(K)}$ to $\mathbb Y^{(K)}\widehat{\otimes}\mathbb Y^{(K)}$. Moreover, the restriction of $\Delta_{\mathsf Y}$ to $\mathbb Y^{(K)}$ agrees with $\Delta'$.
\end{proposition}

\begin{proof}
It suffices to show that $\Delta'$ is compatible with the maps $\Psi'_L:\mathbb Y^{(K)}\to \mathfrak{U}(\mathcal W^{(K)}_L)$ in \eqref{eqn: map affine Yangian to W_K>2}. Notice that
\begin{align*}
    \Psi'_L(\mathfrak c)=\alpha L,
\end{align*}
then we have
\begin{align*}
    \Delta_{L_1,L_2}\circ \Psi'_{L_1+L_2}&=(1\otimes \eta_{-L_1\alpha}^{\otimes L_2})\circ(\Psi'_{L_1}\otimes \Psi'_{L_2})\circ \Delta\\
    &=(1\otimes \Psi'_{L_2})\circ (1\otimes \boldsymbol{\tau}_{-\alpha L_1})\circ (\Psi'_{L_1}\otimes 1)\circ \Delta\\
    &=(\Psi'_{L_1}\otimes \Psi'_{L_2})\circ (1\otimes \boldsymbol{\tau}u_{-\mathfrak c\otimes 1})\circ \Delta\\
    &=(\Psi'_{L_1}\otimes \Psi'_{L_2})\circ \Delta'.
\end{align*}
Here in the first line of the above equation we have used \eqref{eqn: affine Yangian coproduct vs W coproduct}. The above compatibility promotes to an $L\to \infty$ version:
\begin{align*}
    \Delta_{\mathsf W}\circ \Psi'_{\infty}=(\Psi'_{\infty}\otimes \Psi'_{\infty})\circ \Delta'.
\end{align*}
Since $\Psi'_{\infty}$ induces the embedding $\mathbb Y^{(K)}\hookrightarrow \mathsf Y^{(K)}$ in the Theorem \ref{thm: compare affine Yangian with Y}, thus $\Delta'$ agrees with our coproduct $\Delta_{\mathsf Y}$. In particular $\Delta_{\mathsf Y}$ maps the subalgebra $\mathbb Y^{(K)}$ to $\mathbb Y^{(K)}\widehat{\otimes}\mathbb Y^{(K)}$. This finishes the proof.
\end{proof}

\subsection{Duality automorphism of the affine Yangian} 
To conclude this section, we note that the duality automorphism $\sigma:\mathsf Y^{(K)}\cong \mathsf Y^{(K)}$, which is obtained by gluing two duality automorphisms $\sigma: \mathsf A^{(K)}_{\pm}\cong \mathsf A^{(K)}_{\pm}$ subject to the relations \eqref{eqn: gluing relations of Y}, maps the subalgebra $\mathbb Y^{(K)}$ to itself. This can be shown by the following straightforward computation of its action on as set of generators of $\mathbb Y^{(K)}$:
\begin{align}
\sigma(\epsilon_1)=\epsilon_1,\quad\sigma(\epsilon_2)=\epsilon_3,\quad\sigma(X\otimes u^n)= -X^{\mathrm{t}}\otimes u^n,\quad
       \sigma(J(X))= -J(X^{\mathrm{t}}),
\end{align}
where we have used the notation on the elements of $\mathbb Y^{(K)}$ introduced in Definition \ref{def: J generators and modes}. Therefore the restriction of $\sigma$ to $\mathbb Y^{(K)}$ is an algebra involution $\mathbb Y^{(K)}\cong \mathbb Y^{(K)}$.\\

Such duality automorphism on $\mathbb Y^{(K)}$ is closely related to the reflection symmetry of the $A^{(1)}_{K-1}$ Dynkin diagram. More precisely, let $\iota:\mathbb Y^{(K)}\cong \mathbb Y^{(K)}$ be the Dynkin diagram reflection automorphism defined as follows
\begin{align*}
\iota(\epsilon_1)=\epsilon_1,\quad\iota(\epsilon_2)=\epsilon_3,\quad \iota(X^{\pm}_{i,r})=-X^{\pm}_{K-i,r},\quad \iota(H_{i,r})=H_{K-i,r},
\end{align*}
for all $i\in [K]$. Then direct computation shows that 
\begin{align}
    \sigma=\mathrm{Ad}(w_0)\circ \iota,
\end{align}
where $w_0\in \mathrm{SL}_K$ is the longest element in the Weyl group of $\mathrm{SL}_K$, which acts naturally on $\mathbb Y^{(K)}$ by integrating the infinitesimal adjoint action of $X\otimes u^0$, where $X\in \mathfrak{sl}_K$.

\section*{Acknowledgement}
We thank Kevin Costello, Ryosuke Kodera, Si Li, and Hiraku Nakajima for discussions. Research at Perimeter Institute is supported by the Government of Canada through Industry Canada and by the Province of Ontario through the Ministry of Research and Innovation. Kavli IPMU is supported by World Premier International Research Center Initiative (WPI), MEXT, Japan.

\appendix

\section{Quantum ADHM Quiver Variety}\label{sec: quantum ADHM quiver variety}
In this appendix we study the Calogero representation of the quantum ADHM quiver variety. The ADHM quiver is the following:
\begin{center}
\begin{tikzpicture}[x={(2cm,0cm)}, y={(0cm,2cm)}, baseline=0cm]
  \node[draw,circle,fill=white] (Gauge) at (0,0) {$N$};
  \node[draw,rectangle,fill=white] (Framing) at (1,0) {$K$};
  \node (Z) at (-.5,0) {\scriptsize $X$};
  \node (Zdag) at (-.73,0) {\scriptsize $Y$};
 \draw[->] (Gauge.15) -- (Framing.160) node[midway,above] {\scriptsize $I$};
 \draw[<-] (Gauge.345) -- (Framing.197) node[midway,below] {\scriptsize $J$};

  \draw[->,looseness=5] (Gauge.210) to[out=210,in=150] (Gauge.150);
  \draw[<-,looseness=6] (Gauge.240) to[out=210,in=150] (Gauge.120);

\end{tikzpicture}
\end{center}
The operators $\{X^i_j,Y^i_j,I^a_i,J^j_a\:|\: 1\le i,j\le N,1\le a\le K\}$, depicted in the above quiver diagram, satisfy the following commutation relations:
\begin{equation}\label{eqn: X,Y,I,J relations}
\begin{aligned}
[X^i_j,Y^k_l]=\epsilon_1\delta^i_l\delta^k_j,\quad [J^j_a,I^b_i]=\epsilon_1\delta^j_i\delta_a^b.
\end{aligned}
\end{equation}
Denote by $R_{\epsilon_1}$ the $\mathbb C[\epsilon_1]$ algebra generated by $(X,Y,I,J)$ with relations \eqref{eqn: X,Y,I,J relations}.
\begin{definition}
We define the quantum Nakajima quiver variety associated to the ADHM quiver $\mathscr O_{\epsilon_1}(\mathcal M_{\epsilon_2}(N,K))$ to be the $\mathbb C[\epsilon_1,\epsilon_2]$-algebra 
\begin{align*}
    \mathscr O_{\epsilon_1}(\mathcal M_{\epsilon_2}(N,K)):=\left(R_{\epsilon_1}[\epsilon_2]/R_{\epsilon_1}[\epsilon_2]\cdot \mu_{\epsilon_2}(\mathfrak{gl}_N)\right)^{\GL_N}.
\end{align*}
Here $\mu_{\epsilon_2}:\mathfrak{gl}_N\to R_{\epsilon_1}[\epsilon_2]$ is the Lie algebra map
\begin{align*}
    \mu_{\epsilon_2}(E^i_j)=:{[X,Y]^i_j}:+I^a_jJ^i_a-\epsilon_2\delta^i_j,
\end{align*}
where the normal ordering convention is such that $Y$ is to the left of $X$ and that $I$ to the left of $J$.
\end{definition}
Note that $\left(R_{\epsilon_1}[\epsilon_2]\cdot \mu_{\epsilon_2}(\mathfrak{gl}_N)\right)^{\GL_N}$ is a two-sided ideal in $\left(R_{\epsilon_1}[\epsilon_2]\right)^{\GL_N}$, so $\mathscr O_{\epsilon_1}(\mathcal M_{\epsilon_2}(N,K))$ is an algebra.

\smallskip There is also a sheaf-theoretic version of the quantum Nakajima quiver variety, recalled as follows.

Let us choose a nontrivial stability $\theta\neq 0$, then the stable moduli space $\mathcal{M}^{\theta}_{\epsilon_2}(N,K)$ is smooth over the base $\Spec \mathbb C[\epsilon_2]$. Moreover the natural projection $p:\mathcal{M}^{\theta}_{\epsilon_2}(N,K)\to \mathcal{M}_{\epsilon_2}(N,K)$ is projective, and for all $\lambda\in \mathbb C$, $p_{\lambda}:\mathcal{M}^{\theta}_{\epsilon_2=\lambda}(N,K)\to \mathcal{M}_{\epsilon_2=\lambda}(N,K)$ is a symplectic resolution. 
By the construction in \cite{losev2012isomorphisms}, there is a sheaf of flat $\mathbb C[\![\epsilon_1]\!]$-algebras on $\mathcal{M}^{\theta}_{\epsilon_2}(N,K)$, denote by $\widetilde{\mathcal O}_{\mathcal{M}^{\theta}_{\epsilon_2}(N,K)}$, such that $\widetilde{\mathcal O}_{\mathcal{M}^{\theta}_{\epsilon_2}(N,K)}/(\epsilon_1)$ is the structure sheaf ${\mathcal O}_{\mathcal{M}^{\theta}_{\epsilon_2}(N,K)}$. By \cite{losev2012isomorphisms} there is a natural $\mathbb C[\epsilon_1,\epsilon_2]$-algebra homomorphism $$\mathscr O_{\epsilon_1}(\mathcal M_{\epsilon_2}(N,K))\to \Gamma\left(\mathcal{M}^{\theta}_{\epsilon_2}(N,K),\widetilde{\mathcal O}_{\mathcal{M}^{\theta}_{\epsilon_2}(N,K)}\right).$$ Since the ADHM quiver satisfies the flatness condition of Crawley-Boevey \cite[Theorem 1.1]{crawley2001geometry}, i.e. the moment map $\mu$ is flat, so its quantization commutes with reduction \cite[Lemma 3.3.1]{losev2012isomorphisms}, therefore we can apply \cite[Lemma 4.2.4]{losev2012isomorphisms} to the ADHM quiver and obtain the following.
\begin{proposition}\label{Proposition_Sheaf version of quantization}
Under the above homomorphism $\mathscr O_{\epsilon_1}(\mathcal M_{\epsilon_2}(N,K))$ is identified with $\mathbb C^{\times}$-finite elements in $\Gamma\left(\mathcal{M}^{\theta}_{\epsilon_2}(N,K),\widetilde{\mathcal O}_{\mathcal{M}^{\theta}_{\epsilon_2}(N,K)}\right)$, where $\mathbb C^{\times}$ acts on quiver path generators with weight one and on $\epsilon_1,\epsilon_2$ with weight two.
\end{proposition}

The flatness of moment map and the quantization commutes with reduction property imply the following.

\begin{lemma}
$\mathscr O_{\epsilon_1}(\mathcal M_{\epsilon_2}(N,K))$ is a flat over the base ring $\mathbb C[\epsilon_1,\epsilon_2]$, and 
\begin{equation*}
    \begin{split}
        \mathscr O_{\epsilon_1}(\mathcal M_{\epsilon_2}(N,K))/(\epsilon_1)&\cong \mathscr O(\mathcal M_{\epsilon_2}(N,K)).
    \end{split}
\end{equation*}
\end{lemma}
The goal of the rest of this appendix is to present the generators and some of the relations of $\mathscr O_{\epsilon_1}(\mathcal M_{\epsilon_2}(N,K))$.

\begin{definition}
Define the algebra $\mathscr O_{\epsilon_1^{\pm}}(\mathcal M_{\epsilon_2}(N,K))$ be $\mathscr O_{\epsilon_1}(\mathcal M_{\epsilon_2}(N,K))[\epsilon_1^{-1}]$, and we define the following elements in $\mathscr O_{\epsilon_1^{\pm}}(\mathcal M_{\epsilon_2}(N,K))$
\begin{align*}
    e^a_{b;n,m}=\frac{1}{\epsilon_1}I^a\sym (X^nY^m)J_b,\quad t_{n,m}=\frac{1}{\epsilon_1}\mathrm{Tr}(\sym (X^nY^m)).
\end{align*}
The above elements generate $\mathscr O_{\epsilon_1^{\pm}}(\mathcal M_{\epsilon_2}(N,K))$ over the base ring $\mathbb C[\epsilon_1^{\pm},\epsilon_2]$, and the following relations are easily derived from definition.
\end{definition}

\begin{lemma}\label{Lemma_Basic Commutation Relation}
$t_{0,0}=N/\epsilon_1$, in particular it is central. The trace of $e^a_{b;n,m}$ is related to $t_{n,m}$ by
\begin{align}
    e^a_{a;n,m}=\epsilon_2t_{n,m}.
\end{align}
$e^a_{b;0,0}$ acts on $\mathscr O_{\epsilon_1^{\pm}}(\mathcal M_{\epsilon_2}(N,K))$ as generators of $\mathfrak{gl}_K$, namely
\begin{align}
    [e^a_{b;0,0},e^c_{d;n,m}]=\delta^c_b e^a_{d;n,m}-\delta^a_d e^c_{b;n,m}.
\end{align}
The linear span of $t_{2,0},t_{1,1},t_{0,2}$ is an $\mathfrak{sl}_2$-triple such that the span of $e^a_{b;n,m}$ with fixed $n+m$ is an irreducible representation of spin $\frac{n+m}{2}$, i.e.
\begin{align}
    [t_{2,0},e^a_{b;n,m}]=2m e^a_{b;n+1,m-1},\;[t_{1,1},e^a_{b;n,m}]=(m-n)e^a_{b;n,m},\;[t_{0,2},e^a_{b;n,m}]=-2n e^a_{b;n-1,m+1},
\end{align}
and the action of $t_{1,0}$ and $t_{0,1}$ lower the spin by $\frac{1}{2}$
\begin{align}
    [t_{1,0},e^a_{b;n,m}]=m e^a_{b;n,m-1},\; [t_{0,1},e^a_{b;n,m}]=n e^a_{b;n-1,m}.
\end{align}
And moreover
\begin{align}
    [e^a_{b;n,0},t_{2,1}]=n e^a_{b;n+1,0}.
\end{align}
\end{lemma}

The next two commutation relations are harder and the purpose if the rest of this subsection is to derive them using the Calogero representation technique developed in the previous subsection.

\begin{proposition}\label{Proposition_The Key Commutation Relation}
Let $\epsilon_3=-K\epsilon_1-\epsilon_2$, then
\begin{equation}\label{eqn: [e[1,0],e[0,n]]}
    \begin{split}
        [e^a_{b;1,0},e^c_{d;0,n}]&=\delta^c_b e^a_{d;1,n}-\delta^a_d e^c_{b;1,n}-\frac{\epsilon_3 n}{2}\left(\delta^c_b e^a_{d;0,n-1}+\delta^a_d e^c_{b;0,n-1}\right)-n\epsilon_1\delta^c_d e^a_{b;0,n-1}\\
&-\epsilon_1\sum_{m=0}^{n-1}\frac{m+1}{n+1}\delta^a_d e^c_{f;0,m}e^f_{b;0,n-1-m}-\epsilon_1\sum_{m=0}^{n-1}\frac{n-m}{n+1}\delta^c_b e^a_{f;0,m}e^f_{d;0,n-1-m}\\
&+\epsilon_1\sum_{m=0}^{n-1}e^a_{d;0,m}e^c_{b;0,n-1-m}
    \end{split}
\end{equation}

\begin{equation}\label{eqn: [t[3,0],t[0,n]]}
    \begin{split}
        [t_{3,0},t_{0,n}]=~&3n t_{2,n-1}+\frac{n(n-1)(n-2)}{4}(\epsilon_1^2-\epsilon_2\epsilon_3) t_{0,n-3}\\
&-\frac{3\epsilon_1}{2}\sum_{m=0}^{n-3}(m+1)(n-2-m)(e^a_{c;0,m}e^c_{a;0,n-3-m}+\epsilon_1\epsilon_2t_{0,m}t_{0,n-3-m}).
    \end{split}
\end{equation}

\end{proposition}

As an initial step towards the proof of Proposition \ref{Proposition_The Key Commutation Relation}, we notice that \eqref{eqn: [t[3,0],t[0,n]]} is equivalent to the following two equations:
\begin{align}\label{eqn_[t[2,1],t[0,n]]}
    [t_{2,1},t_{0,n}]=2nt_{1,n},
\end{align}
\begin{equation}\label{eqn_[t[2,1],t[1,n]]}
    \begin{split}
   [t_{2,1},t_{1,n}]=&(2n-1)t_{2,n}+\frac{n(n-1)}{4}(\epsilon_2\epsilon_3-\epsilon_1^2)t_{0,n-2}+\\
    &+\frac{3\epsilon_1}{2}\sum_{m=0}^{n-2}\frac{(m+1)(n-1-m)}{n+1}(e^a_{c;0,m}e^c_{a;0,n-2-m}+\epsilon_1\epsilon_2t_{0,m}t_{0,n-2-m}).
    \end{split}
\end{equation}
Indeed, $[t_{3,0},t_{0,n}]=\frac{1}{2}[[t_{2,0},t_{2,1}],t_{0,n}]=\frac{1}{2}[t_{2,0},[t_{2,1},t_{0,n}]]-n[t_{2,1},t_{1,n-1}]$.\\

\subsection{Calogero representation}
A physicist reader may be familiar with the Calogero representation as a standard manipulation employed to study non-singlet sectors in matrix quantum mechanics \cite{gibbons1984generalisation,krichever1994spin}. The following discussions hold

Let $R(N,K)$ be the affine space of representations of the following ``half'' of ADHM quiver:
\begin{center}
\begin{tikzpicture}[x={(2cm,0cm)}, y={(0cm,2cm)}, baseline=0cm]
  \node[draw,circle,fill=white] (Gauge) at (0,0) {$N$};
  \node[draw,rectangle,fill=white] (Framing) at (1,0) {$K$};
  \node (Zdag) at (-.73,0) {\scriptsize $Y$};
 \draw[->] (Gauge.0) -- (Framing.180) node[midway,above] {\scriptsize $I$};

  \draw[<-,looseness=6] (Gauge.240) to[out=210,in=150] (Gauge.120);

\end{tikzpicture}
\end{center}
Denote by $R^s(N,K)$ be the stable locus of $R(N,K)$, i.e. the open subset of $R(N,K)$ consisting of $(Y,I)$ such that if $V\subset \mathrm{Ker}(I)$ and $Y(V)\subset V$ then $V=0$. There is an open embedding
\begin{align*}
    T^*{R}^{s}(N,K)\sslash _{\epsilon_2} \GL(\mathbf v) \hookrightarrow \mathcal{M}^{\theta}_{\epsilon_2}(N,K),
\end{align*}
where the stability condition $\theta=-1$. Note that $ T^*{R}^{s}(N,K)\sslash _{\epsilon_2} \GL(\mathbf v)$ is the $\epsilon_2$-twisted cotangent bundle $T^*_{\epsilon_2}\mathcal N(N,K)$ of the quotient variety
\begin{align}\label{eqn: stable moduli of quiver}
    \mathcal N(N,K)={R}^{s}(N,K)/ \GL(\mathbf v).
\end{align}
Here $\epsilon_2$-twisted cotangent bundle is defined as the affine bundle over $ \mathcal N(N,K)\times \Spec \mathbb C[\epsilon_2]$ whose underlying vector bundle is $T^*\mathcal N(N,K)\times \Spec \mathbb C[\epsilon_2]$ and is determined by the class $\epsilon_2\cdot c_1(\mathcal O(1))\in \mathrm{H}^1(\Omega^1_{\mathcal N(N,K)})$, where $\mathcal O(1)$ is the tautological line bundle on $\mathcal N(N,K)$.

\bigskip After passing to quantization, $\widetilde{\mathcal O}_{\mathcal{M}^{\theta}_{\epsilon_2}(N,K)}|_{T^*_{\epsilon_2}\mathcal N(N,K)}$ is naturally identified with $\epsilon_1$-adic completion of the sheaf of $\epsilon_2$-twisted $\epsilon_1$-differential operators on $\mathcal N(N,K)$. Since $T^*_{\epsilon_2}\mathcal N(N,K)$ is open and dense in $\mathcal{M}^{\theta}_{\epsilon_2}(N,K)$, composing the embedding in Proposition \ref{Proposition_Sheaf version of quantization} with the restriction map $\Gamma\left(\mathcal{M}^{\theta}_{\epsilon_2}(N,K),\widetilde{\mathcal O}_{\mathcal{M}^{\theta}_{\epsilon_2}(N,K)}\right)\hookrightarrow \Gamma\left(T^*_{\epsilon_2}\mathcal N(N,K),\widetilde{\mathcal O}_{\mathcal{M}^{\theta}_{\epsilon_2}(N,K)}\right)$, we obtain an embedding between $\mathbb C[\epsilon_1,\epsilon_2]$-algebras
\begin{align}\label{eqn: Calogero embedding}
    \mathscr O_{\epsilon_1}(\mathcal M_{\epsilon_2}(N,K))\hookrightarrow D^{\epsilon_2}_{\epsilon_1}(\mathcal N(N,K)),
\end{align}
where the right-hand-side is the ring of $\epsilon_2$-twisted $\epsilon_1$-differential operators on $\mathcal N(N,K)$. We call such embedding a Calogero representation of $\mathscr O_{\epsilon_1}(\mathcal M_{\epsilon_2}(N,K))$.

\bigskip Concretely, $\mathcal N(N,K)$ is isomorphic to the Quot scheme $\mathrm{Quot}^K_N$ parametrizing length-$N$ quotients of $\mathcal O_{\mathbb C}^{\oplus K}$. The Hilbert-Chow map $\mathrm{Quot}^K_N\to \sym^N(\mathbb C)$ sends a quotient sheaf to the cycle class corresponding to the sheaf, in quiver language this maps $(I,Y)$ to the spectrum of $Y$. Restricted to the open locus where spectrum of $Y$ are distinct, $\mathrm{Quot}^K_N$ is isomorphic to product of $N$ copies of $\mathbb P^{K-1}$ fibered over the base $\sym^N(\mathbb C)_{\mathrm{disj}}$. 

\bigskip The previous discussion shows that the Calogero representation \eqref{eqn: Calogero embedding} embeds the quantum ADHM quiver variety into the ring of $\mathcal O(1)^{\otimes \epsilon_2}$-twisted $\epsilon_1$-differential operators on the Quot scheme, where $\mathcal O(1)$ is the tautological line bundle on $\mathrm{Quot}^K_N$. Setting $\epsilon_1$ to be invertible, i.e. tensoring with $\mathbb C[\epsilon^{\pm}]$, we can rescale the differential operators, then the target of Calogero representation \eqref{eqn: Calogero embedding} becomes $D^{\epsilon_2/\epsilon_1}(\mathrm{Quot}^K_N)$, i.e. ring of $\mathcal O(1)^{\otimes \frac{\epsilon_2}{\epsilon_1}}$-twisted differential operators on $\mathrm{Quot}^K_N$. The next step is to write down the explicit formula of the image of generators of $\mathscr O_{\epsilon_1^{\pm}}(\mathcal M_{\epsilon_2}(N,K))$ as differential operators on the aforementioned open locus.

\bigskip Obviously the flavor symmetry $\GL_K$ acts on $\mathrm{Quot}^K_N$ and this action is compatible with the natural $\GL_K$ action on $\mathbb P^{K-1}$. The infinitesimal action gives rise to a Lie algebra map $\mathfrak{gl}_K\to D^{\epsilon_2/\epsilon_1}(\mathbb P^{K-1})$ and we define $E^a_{b}$ as the image of the generator $e^a_b\in \mathfrak{gl}_K$ under this map.

\begin{lemma}\label{lem: BB map}
$E^a_b$ satisfy the the following equations:
\begin{align}\label{eqn: BB relations}
    E^a_a=\frac{\epsilon_2}{\epsilon_1},\quad E^a_cE^c_b=-\frac{\epsilon_1+\epsilon_3}{\epsilon_1}E^a_b,\quad E^a_bE^b_a=-\frac{(\epsilon_1+\epsilon_3)\epsilon_2}{\epsilon_1^2}.
\end{align}
\end{lemma}

\begin{proof}
$D^{\epsilon_2/\epsilon_1}(\mathbb P^{K-1})$ is obtained from $D(\mathbb C^K)$ by performing quantum Hamiltonian reduction by $\mathbb C^{\times}$. Namely, let $u^a,a\in \{1,\cdots,K\}$ be the coordinates on $\mathbb C^K$ and let $v_a$ be the differential operators on $\mathbb C^K$, i.e.
\begin{align*}
    [v_a,u^b]=\delta^b_a,
\end{align*}
then $E^a_b:=u^av_b$ satisfy the $\gl_K$ commutation relations, which gives a Lie algebra map $\gl_K\to D(\mathbb C^K)$. $D^{\epsilon_2/\epsilon_1}(\mathbb P^{K-1})$ is the $\mathbb C^{\times}$-invariant subalgebra of $D(\mathbb C^K)$ generated by $E^a_b:=u^av_b$ modulo the right ideal generated by $u^av_a-\frac{\epsilon_2}{\epsilon_1}$. Therefore the equation $E^a_a=\frac{\epsilon_2}{\epsilon_1}$ follows from construction, and 
\begin{align*}
    E^a_cE^c_b=u^av_cu^cv_b=-u^av_b+u^av_bv_cu^c=(K-1+\frac{\epsilon_2}{\epsilon_1})E^a_b=-\frac{\epsilon_1+\epsilon_3}{\epsilon_1}E^a_b.
\end{align*}
The last equation in \eqref{eqn: BB relations} follows from the first and the second.
\end{proof}

Now we can write down the explicit formula of the Calogero representation.
\begin{lemma}\label{lemma: t[2,0] and e[0,n] in Calogero}
Composing the Calogero representation $ \mathscr O_{\epsilon_1}(\mathcal M_{\epsilon_2}(N,K))\hookrightarrow D^{\epsilon_2/\epsilon_1}(\mathrm{Quot}^K_N)$ with the restriction map $D^{\epsilon_2/\epsilon_1}(\mathrm{Quot}^K_N)\hookrightarrow D^{\epsilon_2/\epsilon_1}(\mathbb P^{K-1}\times\cdots\times \mathbb P^{K-1}\times \mathbb C^N_{\mathrm{disj}})$, then the elements $t_{2,0}$ and $e^a_{b;0,n}$ are mapped to
\begin{align}
    t_{2,0}\mapsto \epsilon_1\sum_{i=1}^N\Delta^{-1}\partial_{y_i}^2\Delta-2\sum_{i<j}^N\frac{\epsilon_1 \Omega_{ij}+\epsilon_2}{(y_i-y_j)^2},\qquad e^a_{b;0,n}\mapsto \sum_{i=1}^N E^a_{b,i} y_i^n,
\end{align}
where $(y_1,\cdots,y_N)$ is the coordinate on $\mathbb C^N$, $E^a_{b,i}$ is the $E^a_{b}$ for the $i$-th $\mathbb P^{K-1}$, $\Omega_{ij}=E^a_{b,i}E^b_{a,j}$ is the quadratic Casimir of $ij$ sites, $\Delta$ is the Vandermonde factor
\begin{align*}
    \Delta=\prod_{i>j}^N(y_i-y_j).
\end{align*}
\end{lemma}

\begin{proof}
We diagonalize $Y=H\mathrm{diag}(y_1,\cdots,y_N)H^{-1}$, and define 
\begin{align}
    u^a_i=(IH)^a_i,\quad v^i_a=(H^{-1}J)^i_a.
\end{align}
The commutators between $u$ and $v$ are
\begin{align}
    [v^{i}_a,u^b_j]=\epsilon_1 \delta^b_a\delta^i_j.
\end{align}
Therefore $u^a_i$ are the projective coordinates on the $i$-th $\mathbb P^{K-1}$ and $v^{i}_a$ are the differential operators on it. According to the Lemma \ref{lem: BB map}, $e^a_b\mapsto E^a_{b,i}:=\frac{1}{\epsilon_1}u^a_iv^i_b$ (do not sum over $i$) gives the Lie algebra map from $\gl_K$ to the differential operators on the $i$-th copy of $\mathbb P^{K-1}$.

From the diagonalization $Y=HDH^{-1}$ where $D$ is the diagonal matrix, we read out the tangent map $dY=[dH\cdot H^{-1},Y]+HdD H^{-1}$, and in the dual basis the above equation becomes
\begin{align}
    \partial_{H^i_j}=\frac{1}{\epsilon_1}(H^{-1}:[Y,X]:)^j_i,\quad \partial_{y_i}=\frac{1}{\epsilon_1}:(H^{-1}XH)^i_i:,
\end{align}
here we use the identification $X^i_j=\epsilon_1\partial_{Y^j_i}$, and the normal ordering such that $X$ is always at the right-hand-side of $H$ and $Y$. Let $\overline{X}^i_j=:(H^{-1}XH)^i_j:$, then 
\begin{align}\label{eq: transform of X}
    \overline{X}^i_j=\begin{cases}
    \frac{\epsilon_1}{y_i-y_j}:(\partial_H\cdot H)^i_j:, & i\neq j\\
    \epsilon_1 \partial_{y_i}, & i=j,
    \end{cases}
\end{align}
and the quantum moment map equation becomes
\begin{align}
    u^a_iv^j_a=\begin{cases}
    -\epsilon_1 :(\partial_H\cdot H)^i_j:,& i\neq j\\
    \epsilon_2,& i=j.
    \end{cases}
\end{align}
Thus the image of $e^a_{b;0,n}=\frac{1}{\epsilon_1}I^aY^nJ_b$ is
\begin{align*}
    \frac{1}{\epsilon_1}\sum_{i=1}^Nu^a_iy_i^nv^i_b=\sum_{i=1}^N E^a_{b,i} y_i^n.
\end{align*}
To compute the image of $t_{2,0}=\frac{1}{\epsilon_1}{X}^i_j{X}^j_i$, let us write
\begin{equation}\label{eq: bar X^2}
\begin{split}
     \overline{X}^i_j\overline{X}^j_i&={X}^i_j{X}^j_i+(H^{-1})^i_lH^m_j[X^l_m,(H^{-1})^j_kH^n_i]X^k_n\\
    &={X}^i_j{X}^j_i-(H^{-1})^i_l[X^l_m,H^m_p](H^{-1})^p_kH^n_iX^k_n+(H^{-1})^i_l[X^l_k,H^n_i]X^k_n\\
    &={X}^i_j{X}^j_i-(H^{-1})^i_l[X^l_m,H^m_p]\overline{X}^p_i+(H^{-1})^i_l[X^l_k,H^n_i]X^k_n.
\end{split}
\end{equation}
Using the relation $X^i_j=:(H\overline{X}H^{-1})^i_j:$ and \eqref{eq: transform of X}, we find 
\begin{align}\label{eq: [X,H]}
    [X^a_b,H^c_d]=\sum_{e\neq d}^N\frac{\epsilon_1}{y_d-y_e}H^a_dH^c_e(H^{-1})^e_b.
\end{align}
Plug \eqref{eq: [X,H]} into \eqref{eq: bar X^2} and we get $\overline{X}^i_j\overline{X}^j_i={X}^i_j{X}^j_i-2\epsilon_1^2\sum_{i\neq j}^N\frac{1}{y_i-y_j}\partial_{y_i}$. Therefore the image of $t_{2,0}$ is
\begin{align*}
    \epsilon_1\sum_{i=1}^N\partial_{y_i}^2+2\epsilon_1\sum_{i\neq j}^N\frac{1}{y_i-y_j}\partial_{y_i}-\frac{2}{\epsilon_1}\sum_{i<j}^N\frac{ u^a_iv^j_a u^b_jv^i_b}{(y_i-y_j)^2}= \epsilon_1\sum_{i=1}^N\Delta^{-1}\partial_{y_i}^2\Delta-\frac{2}{\epsilon_1}\sum_{i<j}^N\frac{ u^a_iv^j_a u^b_jv^i_b}{(y_i-y_j)^2}.
\end{align*}
It is straightforward to compute that $\frac{1}{\epsilon_1}u^a_iv^j_a u^b_jv^i_b=\epsilon_1 \Omega_{ij}+\epsilon_2$, this proves our claim.
\end{proof}

From now on, we conjugate the Calogero representation by the Vandermonde factor $\Delta$, so that
\begin{align}\label{eq: image of t[2,0] and e[0,n]}
    t_{2,0}\mapsto \epsilon_1\sum_{i=1}^N\partial_{y_i}^2-2\sum_{i<j}^N\frac{\epsilon_1 \Omega_{ij}+\epsilon_2}{(y_i-y_j)^2},\qquad e^a_{b;0,n}\mapsto \sum_{i=1}^N E^a_{b,i} y_i^n.
\end{align}
The conjugation is an algebra isomorphism so it will not change the relations. Using \eqref{eq: image of t[2,0] and e[0,n]}, we can derive the formula for more generators.
\begin{align}
    t_{0,n}\mapsto\frac{1}{\epsilon_1}\sum_{i=1}^N y_i^{n}, \qquad t_{1,n}\mapsto \sum_{i=1}^N \left(\frac{n}{2}y_i^{n-1}+y_i^n\partial_{y_i}\right),
\end{align}

\begin{align}
    e^a_{b;1,n}\mapsto\epsilon_1 \sum_{i=1}^N E^a_{b,i}\left(\frac{n}{2}y_i^{n-1}+y_i^n\partial_{y_i}\right)+\epsilon_1\sum_{i<j}^N\frac{y_i^{n+1}-y_j^{n+1}}{n+1}\frac{E^a_{c,i}E^c_{b,j}-E^a_{c,j}E^c_{b,i}}{(y_i-y_j)^2}
\end{align}

\begin{align}
    t_{2,n}\mapsto \epsilon_1 \sum_{i=1}^N\left(\frac{n(n-1)}{4}y_i^{n-2}+ny_i^{n-1}\partial_{y_i}+y_i^n\partial_{y_i}^2\right)-\frac{2}{n+1}\sum_{i<j}^N\frac{y_i^{n+1}-y_j^{n+1}}{(y_i-y_j)^3}(\epsilon_1 \Omega_{ij}+\epsilon_2).
\end{align}

\noindent We can compute more relations in the Calogero representation.
\begin{equation}\label{eqn_e[0,m]e[0,n]}
    \begin{split}
        e^a_{c;0,m}e^c_{b;0,n}&=-\frac{\epsilon_1+\epsilon_3}{\epsilon_1}e^a_{b;0,m+n}+\sum_{i<j}^N y_i^my_j^nE^a_{c,i}E^c_{b,j}+ y_i^ny_j^mE^c_{b,i}E^a_{c,j},\\
    e^a_{b;0,m}e^b_{a;0,n}&=-\frac{(\epsilon_1+\epsilon_3)\epsilon_2}{\epsilon_1}t_{0,m+n}+\sum_{i<j}^N (y_i^my_j^n+y_i^ny_j^m)\Omega_{ij}\\
    &=-\frac{\epsilon_2\epsilon_3}{\epsilon_1}t_{0,m+n}-\epsilon_1\epsilon_2t_{0,m}t_{0,n}+\frac{1}{\epsilon_1}\sum_{i<j}^N (y_i^my_j^n+y_i^ny_j^m)(\epsilon_1\Omega_{ij}+\epsilon_2).
    \end{split}
\end{equation}

\begin{proof}[Proof of Equation \eqref{eqn: [e[1,0],e[0,n]]}]
The left hand side of \eqref{eqn: [e[1,0],e[0,n]]} can be written as
\begin{align*}
    &[e^a_{b;1,0},e^c_{d;0,n}]=\epsilon_1 \sum_{i=1}^N [E^a_{b,i}\partial_{y_i},E^c_{d,i}y_i^{n-1}]+\epsilon_1\sum_{i<j}^N[\frac{E^a_{f,i}E^f_{b,j}-E^a_{f,j}E^f_{b,i}}{y_i-y_j},E^c_{d,i}y_i^n+E^c_{d,j}y_j^n]\\
    & =\epsilon_1\sum_{i=1}^N\left([E^a_{b,i},E^c_{d,i}](ny_i^{n-1}+y_i^{n}\partial_{y_i})+nE^c_{d,i}E^a_{b,i}y_i^n\right)+\epsilon_1\sum_{i<j}^N\frac{y_i^n-y_j^n}{y_i-y_j}(E^a_{d,i}E^c_{b,j}+E^a_{d,j}E^c_{b,i})\\
    &~ -\epsilon_1\delta^a_d\sum_{i<j}^N\frac{E^c_{f,i}E^f_{b,j}y_i^n-E^c_{f,j}E^f_{b,i}y_j^n}{y_i-y_j}-\epsilon_1\delta^c_b\sum_{i<j}^N\frac{E^f_{d,i}E^a_{f,j}y_i^n-E^f_{d,j}E^a_{f,i}y_j^n}{y_i-y_j}\\
    & =\delta^c_b e^a_{d;1,n}-\delta^a_d e^c_{b;1,n}+\frac{\epsilon_3 n}{2}\left(\delta^c_b e^a_{d;0,n-1}+\delta^a_d e^c_{b;0,n-1}\right)-\epsilon_1 n \delta^a_d e^c_{b;0,n-1}\\
    &~ -\epsilon_1\delta^a_d\sum_{m=0}^{n-1}\frac{m+1}{n+1} e^c_{f;0,m}e^f_{b;0,n-1-m}-\epsilon_1\sum_{m=0}^{n-1}\frac{n-m}{n+1}\delta^c_b e^a_{f;0,m}e^f_{d;0,n-1-m}\\
    &~ +\epsilon_1 n \sum_{i=1}^N(E^c_{d,i}E^a_{b,i}-E^a_{d,i}E^c_{b,i})y_i^{n-1}+\epsilon_1\sum_{m=0}^{n-1}e^a_{d;0,m}e^c_{b;0,n-1-m}.
\end{align*}
Use the identity $E^c_{d,i}E^a_{b,i}-E^a_{d,i}E^c_{b,i}=\delta^a_dE^c_{b,i}-\delta^c_dE^a_{b,i}$, we get the right hand side of \eqref{eqn: [e[1,0],e[0,n]]}.
\end{proof}

\begin{proof}[Proof of Equation \eqref{eqn_[t[2,1],t[0,n]]}]
In the Calogero representation we have
\begin{align*}
    [t_{2,1},t_{0,n}]=\sum_{i=1}^N[\partial_{y_i}+y_i\partial_{y_i}^2,y_i^n]=\sum_{i=1}^N \left(n^2y_i^{n-1}+2ny_i^n\partial_{y_i}\right)=2nt_{1,n}.
\end{align*}
\end{proof}

\begin{proof}[Proof of Equation \eqref{eqn_[t[2,1],t[1,n]]}]
The left hand side of \eqref{eqn_[t[2,1],t[1,n]]} can be written as
\begin{align*}
    [t_{2,1},t_{1,n}]=&\epsilon_1\sum_{i=1}^N[\partial_{y_i}+y_i\partial_{y_i}^2,\frac{n}{2}y_i^{n-1}+y_i^n\partial_{y_i}]-\sum_{i<j}^N(\epsilon_1 \Omega_{ij}+\epsilon_2)[\frac{y_i+y_j}{(y_i-y_j)^2},y_i^n\partial_{y_i}+y_j^n\partial_{y_j}]\\
    =&\epsilon_1 \sum_{i=1}^N\left(\frac{n(n-1)^2}{2}y_i^{n-2}+n(2n-1)y_i^{n-1}\partial_{y_i}+(2n-1) y_i^{n}\partial_{y_i}^2\right)\\
    &+\sum_{i<j}^N(\epsilon_1 \Omega_{ij}+\epsilon_2)\frac{3(y_iy_j^n-y_i^ny_j)-(y_i^{n+1}-y_j^{n+1})}{(y_i-y_j)^3},
\end{align*}
And the relevant summations that we encounter in the right hand side of  \eqref{eqn_[t[2,1],t[1,n]]} can be written as
\begin{equation}\label{eqn_sum ee}
    \begin{split}
        \frac{\epsilon_1}{2}&\sum_{m=0}^{n-2}(m+1)(n-1-m)(e^a_{c;0,m}e^c_{a;0,n-2-m}+\epsilon_1\epsilon_2t_{0,m}t_{0,n-2-m})\\
    &=-\frac{(n+1)n(n-1)}{12}\epsilon_2\epsilon_3t_{0,n-2}\\
    &~ +\sum_{i<j}^N\frac{(n-1)(y_i^{n+1}-y_j^{n+1})+(n+1)(y_iy_j^n-y_i^ny_j)}{(y_i-y_j)^3}(\epsilon_1\Omega_{ij}+\epsilon_2).
    \end{split}
\end{equation}
Now we can see that two sides of \eqref{eqn_[t[2,1],t[1,n]]} agree by direct computation using \eqref{eqn_sum ee}.
\end{proof}

To conclude this appendix, we remark that $\mathscr O_{\epsilon_1}(\mathcal M_{\epsilon_2}(N,K))$ is closely related to the spherical $\gl_K$-extended Cherednik algebra $\mathrm{S}\mathcal H^{(K)}_N$ when $\epsilon_1=\epsilon_2$. In fact, we have the following.

\begin{proposition}
There is a surjective $\mathbb C[\epsilon^\pm]$-algebra map $\mathscr O_{\epsilon^{\pm}}(\mathcal M_{\epsilon}(N,K))\twoheadrightarrow \mathrm{S}\mathcal H^{(K)}_N[\epsilon^\pm]/(\epsilon_1=\epsilon_2=\epsilon)$ which is determined by
\begin{align*}
    e^a_{b;n,m}\mapsto \rho_N(\mathsf T_{n,m}(E^a_b)).
\end{align*}
\end{proposition}

\begin{proof}
First of all, we note that $D^{1}(\mathbb P^{K-1}$ naturally acts on $\Gamma(\mathbb P^{K-1},\mathcal O(1))\cong \mathbb C^K$, such that $E^a_b$ is mapped to elementary matrix corresponding to $a$-th row and $b$-th column. I.e. there is an algebra map $D^{1}(\mathbb P^{K-1})\to \mathrm{End}(\mathbb C^K)$. Composing the embedding $ \mathscr O_{\epsilon}(\mathcal M_{\epsilon}(N,K))\hookrightarrow D^{1}(\mathbb P^{K-1}\times\cdots\times \mathbb P^{K-1}\times \mathbb C^N_{\mathrm{disj}})$ (c.f. Lemma \ref{lemma: t[2,0] and e[0,n] in Calogero}) with the algebra map $D^{1}(\mathbb P^{K-1})\to \mathrm{End}(\mathbb C^K)$, we obtain a $\mathbb C[\epsilon]$ algebra map $\mathscr O_{\epsilon}(\mathcal M_{\epsilon}(N,K))\to D(\mathbb C^N_{\mathrm{disj}})\otimes \gl_K^{\otimes N}$, which is determined by
\begin{align*}
    t_{2,0}\mapsto \epsilon\sum_{i=1}^N\partial_{y_i}^2-2\sum_{i<j}^N\frac{\epsilon \Omega_{ij}+\epsilon}{(y_i-y_j)^2},\qquad e^a_{b;0,n}\mapsto \sum_{i=1}^N E^a_{b,i} y_i^n.
\end{align*}
Here $E^a_{b,i}$ are elementary matrix in the $i$-th copy of $\gl_K$. Compare with the Dunkl embedding, we see that $t_{2,0}$ is mapped to $\rho_N(\mathsf t_{2,0})$ and $e^a_{b;n,m}$ is mapped to $\rho_N(\mathsf T_{n,m}(E^a_b))$, thus the assignment $e^a_{b;n,m}\mapsto \rho_N(\mathsf T_{n,m}(E^a_b))$ uniquely determines a $\mathbb C[\epsilon]$-algebra map $\mathscr O_{\epsilon^{\pm}}(\mathcal M_{\epsilon}(N,K))\twoheadrightarrow \mathrm{S}\mathcal H^{(K)}_N[\epsilon^\pm]/(\epsilon_1=\epsilon_2=\epsilon)$. This map is surjective because $\{\rho_N(\mathsf T_{n,m}(E^a_b)\:|\: 1\le a,b\le K, (n,m)\in \mathbb N^2\}$ generates $\mathrm{S}\mathcal H^{(K)}_N$.
\end{proof}

\section{Computation of \texorpdfstring{$W^{(1)}W^{(n)}$}{W1Wn} OPE in \texorpdfstring{$\mathcal W^{(K)}_{\infty}$}{Wk inf}}\label{app: W1Wn OPE}

In this appendix, we show that the OPE between $W^{a(1)}_a$ and $W^{b(n)}_{c}$ in $\mathcal W^{(K)}_{\infty}$ is the following:
\begin{equation}\label{W^1W^n OPE}
\begin{split}
W^{a(1)}_b(z)W^{c(n)}_{d}(w)\sim \:& \sum_{i=0}^{n-1}\frac{(\lambda-i)\cdots(\lambda-n+1)\alpha^{n-1-i}}{(z-w)^{n+1-i}}(\alpha\delta^c_bW^{a(i)}_{d}(w)+\delta^a_b W^{c(i)}_{d}(w))\\
&+\frac{\delta^c_bW^{a(n+1)}_{d}(w)-\delta^a_dW^{c(n+1)}_{b}(w)}{z-w}.
\end{split}
\end{equation}
We prove \eqref{W^1W^n OPE} by proving its counterparts in $\mathcal W^{(K)}_{L}$ for all $L$. Recall that the matrix-valued Miura operator is defined as
\begin{align*}
    \mathcal L^L(z)=(\alpha\partial_z-J(z)^{[1]})\cdots(\alpha\partial_z-J(z)^{[L]})=\sum_{i=0}^L(-1)^i(\alpha\partial_z)^{L-i}\cdot W^{(i)}(z),
\end{align*}
where $J(z)^{[i]}$ is the $i$-th copy of affine Kac-Moody $V^{\alpha,1}(\gl_K)$, and we set $W^{a(0)}_{b}(z)=\delta^a_b$. Then we have OPE
\begin{align*}
    J^a_b(z)^{[i]}\mathcal L^L(w)^c_d\sim\: & -\mathcal L^{i-1}(w)^c_b\frac{\alpha}{(z-w)^2}\mathcal L^{L-i}(w)^a_d-\mathcal L^{i-1}(w)^c_e\frac{\delta^a_b}{(z-w)^2}\mathcal L^{L-i}(w)^e_d\\
    &-\mathcal L^{i-1}(w)^c_b\frac{J^a_e(w)^{[i]}}{z-w}\mathcal L^{L-i}(w)^e_d+\mathcal L^{i-1}(w)^c_f\frac{J^f_b(w)^{[i]}}{z-w}\mathcal L^{L-i}(w)^a_d
\end{align*}
Write $J(w)^{[i]}=\alpha\partial_w-(\alpha\partial_w-J(w)^{[i]})$ in the second line, and we get
\begin{align*}
    J^a_b(z)^{[i]}\mathcal L^L(w)^c_d\sim\: & -\mathcal L^{i-1}(w)^c_b\frac{\alpha}{(z-w)^2}\mathcal L^{L-i}(w)^a_d-\mathcal L^{i-1}(w)^c_e\frac{\delta^a_b}{(z-w)^2}\mathcal L^{L-i}(w)^e_d\\
    &-\mathcal L^{i-1}(w)^c_b\frac{\alpha}{z-w}\partial_w\cdot\mathcal L^{L-i}(w)^a_d+\mathcal L^{i-1}(w)^c_b\partial_w\cdot\frac{\alpha}{z-w}\mathcal L^{L-i}(w)^a_d\\
    &+\mathcal L^{i-1}(w)^c_b\frac{1}{z-w}\mathcal L^{L-i+1}(w)^a_d-\mathcal L^{i}(w)^c_b\frac{1}{z-w}\mathcal L^{L-i}(w)^a_d\\
    =& \mathcal L^{i-1}(w)^c_b\frac{1}{z-w}\mathcal L^{L-i+1}(w)^a_d-\mathcal L^{i-1}(w)^c_e\frac{\delta^a_b}{(z-w)^2}\mathcal L^{L-i}(w)^e_d -\mathcal L^{i}(w)^c_b\frac{1}{z-w}\mathcal L^{L-i}(w)^a_d.
\end{align*}
Summing over $i$ from $1$ to $L$, we get
\begin{align*}
    W^{a(1)}_{b}(z)\mathcal L^L(w)^c_d\sim\:  \frac{\delta^c_b}{z-w}\mathcal L^{L}(w)^a_d-\mathcal L^{L}(w)^c_b\frac{\delta^a_d}{z-w}-\delta^a_b\sum_{i=1}^L\mathcal L^{i-1}(w)^c_e\frac{1}{(z-w)^2}\mathcal L^{L-i}(w)^e_d.
\end{align*}
We can move $\frac{\delta^c_b}{z-w}$ to the right-hand-side of $\mathcal L^{L}(w)^a_d$ at a cost of $w$-derivatives on $1/(z-w)$ and get
\begin{align}\label{W1L OPE}
    W^{a(1)}_{b}(z)\mathcal L^L(w)^c_d\sim\: \mathcal L^{L}(w)^a_d\frac{\delta^c_b}{z-w}-\mathcal L^{L}(w)^c_b\frac{\delta^a_d}{z-w}+[\frac{\delta^c_b}{z-w},\mathcal L^{L}(w)^a_d]+[\frac{\delta^a_b/\alpha}{z-w},\mathcal L^{L}(w)^c_d].
\end{align}
The commutator between Miura operator and $1/(z-w)$ can be computed using \eqref{eq: PDS commutator}:
\begin{align*}
    [\frac{1}{z-w},\mathcal L^{L}(w)]&=\sum_{i=0}^L(-1)^i[\frac{1}{z-w},(\alpha\partial_w)^{L-i}]\cdot W^{(i)}(w)\\
    &=\sum_{i=0}^L \sum_{s=1}^{L-i}(-1)^{i+s}\binom{L-i}{s}(\alpha\partial_w)^{L-i-s}\cdot  \frac{\alpha^s s! W^{(i)}(w)}{(z-w)^{s+1}}\\
    &=\sum_{n=0}^L (-1)^{n}(\alpha\partial_w)^{L-n}\cdot \sum_{i=0}^{n-1}\frac{(L-i)!}{(L-n)!} \frac{\alpha^{n-i}  W^{(i)}(w)}{(z-w)^{n+1-i}}.
\end{align*}
Plug the above commutator into \eqref{W1L OPE} and extract the coefficient of $(\alpha\partial_w)^{L-n}$, and we get \eqref{W^1W^n OPE}.


\section{Completed Tensor Product of Graded Algebras}\label{sec: Completion of Tensor Product}
In this appendix, we define a version of completed tensor product of graded algebras that we use in the body of this paper. We fix a commutative base ring $\mathds k$ throughout this section.

Let $R_i=\bigoplus_{j\in \mathbb Z}R^j_i$, $i=1,2$ be two $\mathbb Z$-graded algebras over the base ring $\mathds k$, where $R^j_i$ is the $j$-th homogeneous component of $R_i$, then we define a completed tensor product as follows.
\begin{definition}\label{def: completed tensor product}
The completed tensor product $R_1\widetilde\otimes R_2$ is the $\mathbb Z$-graded $\mathds k$-algebra
\begin{align}
    \bigoplus_{d\in \mathbb Z}\left(\underset{\substack{\longrightarrow\\M}}{\lim}\prod_{n\in \mathbb Z_{\ge -M}}R_1^{d+n}\otimes R^{-n}_2\right).
\end{align}
It is easy to see that degreewise multiplication of $R_1\otimes R_2$ naturally extends to the $R_1\widetilde\otimes R_2$ therefore the latter inherits a natural $\mathbb Z$-graded $\mathds k$-algebra structure.
\end{definition}

\begin{example}
Let $R_1=\mathds k[x_1^{\pm}]$ and $R_2=\mathds k[x_2^{\pm}]$, and we grade them by letting $\deg x_1=\deg x_2=1$, then $R_1\widetilde\otimes R_2=\mathds k(\!(\frac{x_1}{x_2})\!)[x_1^{\pm},x_2^{\pm}]$.
\end{example}
The next example shows that the completed tensor product is not associative.
\begin{example}
Let $R_i=\mathds k[x_i^{\pm}]\;(i=1,2,3)$, and we grade them by letting $\deg x_i=1$, then $(R_1\widetilde\otimes R_2)\widetilde\otimes R_3\ncong R_1\widetilde\otimes (R_2\widetilde\otimes R_3)$. In fact, $\sum_{n=0}^{\infty}x_1^{-n}x_2^{2n}x_3^{-n}\in (R_1\widetilde\otimes R_2)\widetilde\otimes R_3$, whereas $\sum_{n=0}^{\infty}x_1^{-n}x_2^{2n}x_3^{-n}\notin R_1\widetilde\otimes (R_2\widetilde\otimes R_3)$.
\end{example}
In the body of the paper we need to discuss coassociativity of coproduct, to remedy the issue of nonassociativity of completed tensor product, we introduce the completed tensor product of three (and more) $\mathbb Z$-graded $\mathds k$-algebras, which turns out to be the relevant object to coassociativity of coproduct.

\begin{definition}\label{def: completed tensor product_general}
Let $R_i=\bigoplus_{j\in \mathbb Z}R^j_i, \;(i=1,2,\cdots,s)$, be $\mathbb Z$-graded $\mathds k$-algebras, where $R^j_i$ is the $j$-th homogeneous component of $R_i$. The completed tensor product $R_1\widetilde\otimes R_2\widetilde\otimes \cdots\widetilde\otimes R_s$ is the $\mathbb Z$-graded $\mathds k$-algebra whose degree $d$ component is
\begin{align}
     \underset{\substack{\longrightarrow\\M}}{\lim}\prod_{(n_1,n_2,\cdots,n_{s-1})\in \mathbb Z^{s-1}_{\ge -M}} R_1^{d+n_1}\otimes R_2^{n_2-n_1}\otimes\cdots \otimes R_{s-1}^{n_{s-1}-n_{s-2}}\otimes R_s^{-n_{s-1}}.
\end{align}
It is easy to see that degreewise multiplication of $R_1\otimes R_2\otimes\cdots\otimes R_s$ naturally extends to the $R_1\widetilde\otimes R_2\widetilde\otimes \cdots\widetilde\otimes R_s$ therefore the latter inherits a natural $\mathbb Z$-graded $\mathds k$-algebra structure.
\end{definition}

\begin{lemma}\label{lem: compare different completions}
Let $R_i,(i=1,2,3)$ be $\mathbb Z$-graded $\mathds k$-algebras, then $R_1\widetilde\otimes R_2\widetilde\otimes R_3$, $(R_1\widetilde\otimes R_2)\widetilde\otimes R_3$ and $R_1\widetilde\otimes (R_2\widetilde\otimes R_3)$ are $\mathbb Z$-graded $\mathds k$-submodules of
\begin{align}
    (R_1\otimes R_2\otimes R_3)^{\wedge}:=\bigoplus_{d\in \mathbb Z}\prod_{n_1+n_2+n_3=d}R_1^{n_1}\otimes R_2^{n_2}\otimes R_3^{n_3}.
\end{align}
Moreover, the ring structures on $R_1\widetilde\otimes R_2\widetilde\otimes R_3$, $(R_1\widetilde\otimes R_2)\widetilde\otimes R_3$ and $R_1\widetilde\otimes (R_2\widetilde\otimes R_3)$ are compatible in the sense that the intersection between any pair of them is a sub-algebra of both of the algebras in the pair.
\end{lemma}

\begin{proof}
The first statement follows directly from the definition. The second statement is also easy to see: the ring structures on $R_1\widetilde\otimes R_2\widetilde\otimes R_3$, $(R_1\widetilde\otimes R_2)\widetilde\otimes R_3$ and $R_1\widetilde\otimes (R_2\widetilde\otimes R_3)$ are all natural extension of the ring structure on $R_1\otimes R_2\otimes R_3$, where we take component-wise multiplication and allow infinitely many terms, the results are finite because in each of the three algebras there is a choice of direction to which the homogeneous degree is allowed to approach infinity, thus the intersection of any pair of the three algebras inherit such ring structure.
\end{proof}

\begin{example}
In  general there is no inclusion relation between any pair of $R_1\widetilde\otimes R_2\widetilde\otimes R_3$, $(R_1\widetilde\otimes R_2)\widetilde\otimes R_3$ and $R_1\widetilde\otimes (R_2\widetilde\otimes R_3)$. For example let $R_i=\mathds k[x_i^{\pm},y_i^{\pm}],\;(i=1,2,3)$ and we grade them by letting $\deg x_i=\deg y_i=1$, then 
\begin{itemize}
    \item $\sum_{n=0}^{\infty}x_1^{-n}x_2^{2n}x_3^{-n}\in (R_1\widetilde\otimes R_2)\widetilde\otimes R_3$, but it is not an element of $ R_1\widetilde\otimes (R_2\widetilde\otimes R_3)$ or $ R_1\widetilde\otimes R_2\widetilde\otimes R_3$.
    \item $\sum_{n=0}^{\infty}x_1^{n}x_2^{-2n}x_3^{n}\in R_1\widetilde\otimes (R_2\widetilde\otimes R_3)$, but it is not an element of $(R_1\widetilde\otimes R_2)\widetilde\otimes R_3$ or $R_1\widetilde\otimes R_2\widetilde\otimes R_3$.
    \item $\sum_{m=0}^{\infty}\sum_{n=0}^{\infty}x_1^{m-n}y_1^nx_2^{n-m}x_3^{-n-m}y_3^m\in R_1\widetilde\otimes R_2\widetilde\otimes R_3$, but it is not an element of $(R_1\widetilde\otimes R_2)\widetilde\otimes R_3$ or $R_1\widetilde\otimes (R_2\widetilde\otimes R_3)$.
\end{itemize}
\end{example}

\section{Vertex Coalgebras and Vertex Comodules}\label{sec: vertex coalgebras and vertex comodules}
Keith Hubbard defined the vertex coalgebras and vertex comodules \cite{hubbard2009vertex}, which are natural dual notion to vertex algebras and vertex modules. In this appendix, We recall the definitions in \cite{hubbard2009vertex} with mild modifications. 

\begin{definition}
A vertex coalgebra is a vector space $V$ together with linear maps
\begin{itemize}
    \item Coproduct $\Delta(w) :V\to V\otimes V(\!(w^{-1})\!)$, and write $\Delta(w)(v)=\sum_{n\in \mathbb Z}\Delta_{n}(v)w^{-n-1}$,
    \item Covacuum $\mathfrak{C}:V\to \mathbb C$,
\end{itemize}
satisfying the following axioms:
\begin{itemize}
    \item[(1)] Counit: $\forall v\in V$, $$(\mathrm{id}\otimes \mathfrak{C})\circ\Delta(w)(v)=v.$$
    \item[(2)] Cocreation: $\forall v\in V$, $$(\mathfrak{C}\otimes\mathrm{id})\circ\Delta(w)(v)\in V[w],\text{ and }\lim_{w\to 0}(\mathfrak{C}\otimes\mathrm{id})\circ\Delta(w)(v)=\mathrm{id}.$$
    \item[(3)] Translation: let $T=(\mathfrak{C}\otimes\mathrm{id})\circ\Delta_{-2}$, then
    $$\frac{\mathrm{d}}{\mathrm{d}w}\Delta(w)=\Delta(w)\circ T-(T\otimes \mathrm{id})\circ\Delta(w).$$
    \item[(4)] Locality: $\forall v\in V$, the following two elements
\begin{align*}
    (\Delta(w)\otimes \mathrm{id})\circ\Delta(z)(v),\quad(\mathrm{id}\otimes P)\circ( \Delta(z)\otimes\mathrm{id})\circ\Delta(w)(v),
\end{align*}
are expansions of the same element in $(V\otimes V\otimes V)[\![z^{-1},w^{-1},(z-w)^{-1}]\!][z,w]$, where $P:V\otimes V\to V\otimes V$ is the operator that swaps two tensor component, i.e. $P(v_1\otimes v_2)=v_2\otimes v_1$.
\end{itemize}
We can similarly define vertex coalgebra over a base ring $\mathds k$.
\end{definition}

Similarly one can define vertex comodule of a vertex coalgebra \cite{hubbard2009vertex}.

\begin{definition}
We say that $(M,\Delta_M(w))$ is a vertex comodule of a vertex coalgebra $(V,\Delta_V(w),\mathfrak{C})$ if $\Delta_M(w): M\to  M\otimes V(\!(w^{-1})\!)$ is a linear map which satisfies the axioms:
\begin{itemize}
    \item[(1)] Counit: $$( \mathrm{id}\otimes \mathfrak{C})\circ \Delta_M(w)=\mathrm{id}.$$
    \item[(2)] Coassociativity: $\forall m\in M$, two elements 
\begin{align*}
    (\mathrm{id}\otimes\Delta_V(z))\circ \Delta_M(w)(m),\quad ( \Delta_M(w)\otimes \mathrm{id})\circ \Delta_M(z+w)(m),
\end{align*}
are expansions of the same element in $(M\otimes V\otimes V)[\![z^{-1},w^{-1},(z+w)^{-1}]\!][z,w]$.
\end{itemize}
We can similarly define vertex comodule of a vertex coalgebra over a base ring $\mathds k$.
\end{definition}

\section{Restricted Mode Algebra of a Vertex Algebra}\label{sec: Restricted Mode Algebra}
In this appendix we define a modified version of the mode algebra of a vertex algebra, which behaves better than the usual (topologically completed) mode algebra when discussing meromorphic coproducts. 

Throughout this appendix, we fix a vertex algebra $\mathcal V$ with the translation operator $T$. For simplicity, we shall assume that the $\mathcal V$ is $\mathbb Z$-graded, i.e. $\mathcal V=\oplus_{d\in \mathbb Z}\mathcal V^d$ such that $\deg |0\rangle=0, \deg T=-1$, and for homogeneous elements $A,B\in \mathcal V$, $\deg A_{(n)}B=\deg A+\deg B+n+1$. Our convention is opposite to the familiar one in the vertex algebra literature (for example \cite{arakawa2007representation}), because this convention matches with our grading convetion for the affine Yangian $\mathsf Y^{(K)}$ in the body of the paper.

Consider the vector space $\mathcal V\otimes \mathbb C[t,t^{-1}]$. For element $A\in \mathcal V$ we write $A_{[n]}:=A\otimes t^n$. The mode Lie algebra $\mathfrak{L}(\mathcal V)$ is defined as the vector space $\mathcal V\otimes \mathbb C[t,t^{-1}]/\mathrm{Im}(T+\partial_{t})$, equipped with the Lie bracket:
\begin{align}\label{eqn: Lie bracket}
    [A_{[m]},B_{[n]}]=\sum_{k\ge 0}\binom{m}{k}(A_{(k)}B)_{[m+n-k]}.
\end{align}
If we set $\deg A_{[n]}=\deg A+n+1$, then the Lie bracket \eqref{eqn: Lie bracket} preserves the grading so $\mathfrak{L}(\mathcal V)$ is a graded Lie algebra and its universal enveloping algebra is a graded algebra $U(\mathfrak{L}(\mathcal V))=\bigoplus_{d\in \mathbb Z}U(\mathfrak{L}(\mathcal V))^d$. Define a topology on the universal enveloping algebra $U(\mathfrak{L}(\mathcal V))$ by requiring the open neighborhood of zero to be the left ideals $I_N$ generated by $U(\mathfrak{L}(\mathcal V))_{d>N}$, and define the completed universal enveloping algebra $\widehat{U}(\mathfrak{L}(\mathcal V))$ to be the degree-wise completion:
\begin{align*}
    \widehat{U}(\mathfrak{L}(\mathcal V))=\bigoplus_{d\in \mathbb Z}\underset{\substack{\longleftarrow\\ N}}{\lim}\:U(\mathfrak{L}(\mathcal V))^d/I_N\cap U(\mathfrak{L}(\mathcal V))^d.
\end{align*}
For $A\in \mathcal V$, set $A[z]=\sum_{n\in \mathbb Z}A_{[n]}z^{-n-1}$, then the Lie bracket \eqref{eqn: Lie bracket} is equivalent to
\begin{align}
    \oint_{|x|>|z|} (x-z)^kA[x]B[z]\frac{dx}{2\pi i } -\oint_{|z|>|x|}(x-z)^kB[z]A[x]\frac{dx}{2\pi i }=(A_{(k)}B)[z],\; \forall k\ge 0.
\end{align}
Define the normal ordered product $:A[z]B[z]:$ as 
\begin{align}
    :A[z]B[z]:=\oint_{|x|>|z|}\frac{1}{x-z}A[x]B[z]\frac{dx}{2\pi i } -\oint_{|z|>|x|}\frac{1}{x-z}B[z]A[x]\frac{dx}{2\pi i }.
\end{align}
The mode algebra $\mathfrak{U}(\mathcal V)$ is defined as the quotient of $\widehat{U}(\mathfrak{L}(\mathcal V))$ by the two-sided ideal topologically generated by $|0\rangle _{[-1]}-1$ and Fourier coefficients of
\begin{align}
    :A[z]B[z]:-(A_{(-1)}B)[z].
\end{align}
We shall define another version of mode algebra which is built of physical (multi-local) operators on the torus $\mathbb C^{\times}$.

Let $F(\mathcal V)$ be the $\mathbb C$-linear space with basis $\mathcal O(A_1,\cdots,A_m;f)$, where $A_i$ are chosen from a basis of $\mathcal V$, and $f$ is chosen from a basis of $\mathbb C[z_i^{\pm 1},(z_i-z_j)^{-1}\:|\: 1\le i,j\le m]$. Formally speaking, 
\begin{align}
    F(\mathcal V)=\bigoplus_{n> 0}\mathcal V^{\otimes n}\otimes \mathscr O(\mathbb C^{\times n}_{\mathrm{disj}}).
\end{align}
Define the multiplication on $F(\mathcal V)$ by
\begin{align}\label{eqn: multiplication on U(V)}
    \mathcal O(A_1,\cdots,A_m;f)\cdot \mathcal O(B_1,\cdots,B_n;g)=\mathcal O(A_1,\cdots,A_m, B_1,\cdots,B_n;fg),
\end{align}
so that $F(\mathcal V)$ acquires a non-unital associative algebra structure.
\begin{definition}\label{def: restricted mode alg}
The restricted mode algebra $U(\mathcal V)$ is the quotient of $F(\mathcal V)$ by the linear space spanned by
\begin{equation}\label{eqn: U(V) unity relation}
    \mathcal O(A_1,\cdots,A_{i-1},|0\rangle,A_{i+1},\cdots,A_m;f)-\mathcal O(A_1,\cdots,A_{i-1},A_{i+1},\cdots,A_m;\underset{z_i\to 0}{\Res}f),
\end{equation}
\begin{equation}\label{eqn: U(V) translation relation}
    \mathcal O(A_1,\cdots,TA_i,\cdots,A_m;f)+\mathcal O(A_1,\cdots,A_i,\cdots,A_m;\partial_{z_i}f),
\end{equation}
\begin{equation}\label{eqn: U(V) commutation relation}
    \begin{split}
        &\mathcal O(A_1,\cdots,A_{i-1},A,B,A_{i+2},\cdots,A_{n};g(z_1,\cdots, z_{i},z_{i+1},\cdots,z_n))\\
        -&\mathcal O(A_1,\cdots,A_{i-1},B,A,A_{i+2},\cdots,A_{n};g(z_1,\cdots,z_{i+1},z_i,\cdots,z_n))\\
        -&\sum_{k\in \mathbb Z}\mathcal O(A_1,\cdots,A_{i-1},A_{(k)}B,A_{i+2},\cdots, A_n;f_k),
    \end{split}
\end{equation}
for all $k\ge -1$ and $g\in \mathbb C[z_i^{\pm 1},(z_i-z_j)^{-1}\:|\: 1\le i,j\le n]$, here $f_k$ are the coefficients of $(z_i-z_{i+1})^k$ in the expansion
\begin{align*}
    g(z_1,\cdots,z_i,z_{i+1},\cdots,z_n)=\sum_{k\in \mathbb Z} (z_i-z_{i+1})^kf_k(z_1,\cdots,z_{i-1},z_{i+1},\cdots,z_n),
\end{align*}
where the RHS is in the space $\mathbb C[z_{k}^{\pm 1},(z_k-z_l)^{-1}\:|\: k,l\neq i](\!(z_i-z_{i+1})\!)$, i.e. expanded in the limit $z_i-z_{i+1}\to 0$, in particular the power series is bounded from below. Note that the summation in \eqref{eqn: U(V) commutation relation} is bounded because $A_{(k)}B$ vanishes for the sufficiently large $k$. 
\end{definition}
$U(\mathcal V)$ is unital since 
\begin{align*}
    \mathcal O(|0\rangle;x^{-1})\mathcal O(A_1,\cdots,A_n;f)=\mathcal O(A_1,\cdots,A_n;f)\mathcal O(|0\rangle;x^{-1})=\mathcal O(A_1,\cdots,A_n;f),
\end{align*}
which is a consequence of \eqref{eqn: U(V) unity relation}.

We also define a linear subspace $U_+(\mathcal V)\subset U(\mathcal V)$ spanned by the identity $\mathcal O(|0\rangle;z_1^{-1})$ and those $\mathcal O(A_1,\cdots,A_m;f)$ such that $f\in \mathbb C[z_i,(z_i-z_j)^{-1}\:|\: 1\le i,j\le m]$. It is easy to see that $U_+(\mathcal V)$ is a subalgebra, and we call it the positive restricted mode algebra.

Note that $U(\mathcal V)$ inherits a $\mathbb Z$-grading from $\mathcal V$ such that $\deg \mathcal O(A_1,\cdots,A_n;f)$ for a homogeneous function $f$ is $\deg A_1+\cdots+\deg A_n+\deg f+n$. This makes $U_+(\mathcal V)$ a $\mathbb Z$-graded subalgebra.

There is a $\mathbb Z$-graded algebra homomorphism $U(\mathcal V)\to \mathfrak{U}(\mathcal V)$ given by
\begin{align}
     \mathcal O(A_1,\cdots,A_m;f)\mapsto\oint_{|z_1|>\cdots>|z_m|}f(z_1,\cdots,z_m)A_1[z_1]\cdots A_m[z_m]\prod_{j=1}^m\frac{d z_j}{2\pi i },
\end{align}
For example $\mathcal O(A;z_1^n)\mapsto A_{[n]}$ and
\begin{align*}
    \mathcal O\left(A,B;\frac{z_1}{z_1-z_2}\right)\mapsto \sum_{n=0}^{\infty}A_{[-n]}B_{[n]}\;\text{ and }\;\mathcal O\left(A,B;\frac{z_1z_2}{(z_1-z_2)^{2}}\right)\mapsto \sum_{n=1}^{\infty}n A_{[-n]}B_{[n]}.
\end{align*} 
The map $U(\mathcal V)\to \mathfrak{U}(\mathcal V)$ is  injective under some technical assumption on $\mathcal V$, one occurrence is the following

\begin{proposition}\label{prop: U(V) is a sub of mode algebra}
 Assume that $\mathcal V$ has a Hamiltonian $H$, and an increasing filtration $F$ such that $\mathrm{gr}_F\mathcal V$ is commutative, and an $H$-invariant subspace $U$ of $\mathcal V$ such that its image $\overline U$ in $\mathrm{gr}_F\mathcal V$ generate a PBW basis of $\mathrm{gr}_F\mathcal V$. Then $U(\mathcal V)\to \mathfrak{U}(\mathcal V)$ is injective.   
\end{proposition}

\begin{proof}[Sketch of Proof]
Consider the linear map
\begin{align*}
    \rho:\bigoplus_{n\ge 1}(\overline U^{\otimes n}\otimes &\mathbb C[z_i^{\pm 1},(z_i-z_j)^{-1}\:|\: 1\le i,j\le n])^{\mathfrak{S}_n}\to \mathrm{gr}_FU(\mathcal V)\\
    A_1\otimes \cdots\otimes &A_n\otimes f\mapsto \frac{1}{n!}\sum_{\sigma\in\mathfrak{S}_n}\mathcal O(A_{\sigma(1)},\cdots,A_{\sigma(n)};f(z_{\sigma(1)},\cdots,z_{\sigma(n)})),
\end{align*}
where the filtration on $U(\mathcal V)$ is defined by requiring $F_pU(\mathcal V)$ to be spanned by those $\mathcal O(A_1,\cdots,A_n;f)$ such that $A_i\in F_{p_i}\mathcal V$ and $p_1+\cdots+p_n\le p$. $\rho$ is surjective because $\mathrm{gr}_F\mathcal V$ is strongly generated by $\overline U$. We claim that $\rho$ is also injective. The composition of $\rho$ with the natural map $\mathrm{gr}_FU(\mathcal V)\to \mathrm{gr}_F\mathfrak U_N(\mathcal V)$ factors though the natural map 
\begin{align*}
    \pi_N:\bigoplus_{n\ge 1}(\overline U^{\otimes n}\otimes &\mathbb C[z_i^{\pm 1},(z_i-z_j)^{-1}\:|\: 1\le i,j\le n])^{\mathfrak{S}_n}\to \mathbb S_N(\overline U),
\end{align*}
defined by summing over all permutations $\sigma\in \mathfrak{S}_n$ the expansion of a function $f\in \mathbb C[z_i^{\pm 1},(z_i-z_j)^{-1}\:|\: 1\le i,j\le n]$ in the order $|z_{\sigma(1)}|>\cdots >|z_{\sigma(n)}|$, cut-off at order $\le N$ for all variables, then divide by $n!$. For example, 
\begin{align*}
    \pi_N(A\otimes B\otimes \frac{1}{z_1-z_2})=\frac{1}{2}\sum_{\substack{m,n\le N\\m+n=-1}}A_{[m]}B_{[n]}.
\end{align*}
Since $\mathbb S_N(\overline U)\to \mathrm{gr}_F\mathfrak U_N(\mathcal V)$ is isomorphism \cite[Theorem 3.14.1]{arakawa2007representation}, to show that $\rho$ is injective, it suffices to show that $\underset{\substack{\longleftarrow\\N}}{\lim}\:\pi_N$ is injective. To this end, we fix a basis of $\overline U$ and it is enough to show that if summing over all permutations $\sigma\in \mathfrak{S}_n$ the expansion of a function $f\in \mathbb C[z_i^{\pm 1},(z_i-z_j)^{-1}\:|\: 1\le i,j\le n]$ in the order $|z_{\sigma(1)}|>\cdots >|z_{\sigma(n)}|$, cut-off at order $\le N$ for all variables, vanishes, then $f$ is identically zero. In fact, the sum over expansion being trivial implies that the expansion of $f$ in arbitrary order $|z_{\sigma(1)}|>\cdots >|z_{\sigma(n)}|$ is bounded, i.e. $f$ is in the subspace $\mathbb C[z_i^{\pm 1}\:|\: 1\le i\le n]$, thus $\pi_N(f)=f$ for sufficiently large $N$, therefore $f=0$. This concludes the proof.
\end{proof}

\begin{remark}
    The technical assumption in Proposition \ref{prop: U(V) is a sub of mode algebra} is satisfied for a wide range of vertex algebras, including the rectangular W-algebra $\mathcal W^{(K)}_L$ \cite{arakawa2007representation}.
\end{remark}

The linear map $\mathfrak{L}(\mathcal V)\to U(\mathcal V)$ by sending $A_{[n]}$ to $\mathcal O(A;x^n)$ preserves the Lie bracket, and we will denote $\mathcal O(A;x^n)$ by $A_{[n]}$. Moreover, we set
\begin{align}
    \mathcal O(\emptyset;f)=f\in \mathbb C.
\end{align}

\subsection{Meromorphic coproduct of \texorpdfstring{$U(\mathcal V)$}{U(V)}}\label{subsec: meromorphic coproduct of mode algebra}
Recall that there is a morphism:
\begin{align*}
    \mathbb C^{\times I}_{\mathrm{disj}}\times \mathbb C^{ J}_{\mathrm{disj}}\times \Spec \mathbb C(\!(w^{-1})\!)\to \mathbb C^{\times (I\sqcup J)}_{\mathrm{disj}},
\end{align*}
which is induced from 
\begin{align*}
    z_i\mapsto z_i,\; i\in I,\quad z_j\mapsto z_j+w,\;j\in J,
\end{align*}
and we write the corresponding map on the function ring as:
\begin{align}
    \Delta_{IJ}(f)=\sum_{s\in \mathbb Z} \Delta_{IJ}^{(s)}(f)\otimes \prescript{}{I}\Delta_{J}^{(s)}(f) w^{-s}.
\end{align}
The RHS is an abbreviation form, for a fixed $s$ there is a finite sum of terms $\Delta_{IJ}^{(s)}(f)_i\otimes \prescript{}{I}\Delta_{J}^{(s)}(f)_i$, indexed by $i$, but to simplify the notation, we omit $i$ and only keep the form of the summand.

We define a linear map $\Delta_{\mathcal V}(w):F(\mathcal V)\to {U}(\mathcal V)\otimes {U}_{+}(\mathcal V) (\!(w^{-1})\!)$ by
\begin{equation}\label{eqn: meromorphic coproduct of U(V)}
\boxed{
    \begin{aligned}
        &\Delta_{\mathcal V}(w)(\mathcal O(A_1,\cdots,A_n;f))=\\
        &\sum_{\substack{I=(i_1,\cdots,i_m)\\J= (j_1,\cdots,j_{n-m})\\ I\sqcup J\in\mathrm{shuffle}(1,\cdots,n)}}\sum_{s\in \mathbb Z}\mathcal O(A_{i_1},\cdots,A_{i_m};\Delta^{(s)}_{IJ}(f))\otimes \mathcal O(A_{j_1},\cdots,A_{j_{n-m}};\prescript{}{I}\Delta^{(s)}_{J}(f))w^{-s},
    \end{aligned}
}
\end{equation}
Here $I$ or $J$ can be empty set. For example,
\begin{align*}
    \Delta_{\mathcal V}(w)(A_{[n]})=\Delta_{\mathcal V}(w)(\mathcal O(A;z_1^n))=A_{[n]}\otimes 1+\sum_{s\ge 0}\binom{n}{s}  w^{n-s} 1\otimes A_{[s]},
\end{align*}
\begin{align*}
   \Delta_{\mathcal V}(w)(\mathcal O(A,B;(z_1-z_2)^{-1}))=\square(\mathcal O(A,B;(z_1-z_2)^{-1}))+\\
   +\sum_{m,n= 0}^{\infty} (-1)^m \binom{m+n}{n}  w^{-m-n-1} (B_{[n]}\otimes A_{[m]}-A_{[n]}\otimes B_{[m]}).
\end{align*}

\begin{lemma}\label{lem: meromorphic coproduct for restricted modes}
The linear map $\Delta_{\mathcal V}(w)$ factors through the restricted mode algebra $U(\mathcal V)$. Moreover $\Delta_{\mathcal V}(w):U(\mathcal V)\to {U}(\mathcal V)\otimes {U}_{+}(\mathcal V) (\!(w^{-1})\!)$ is unital and associative, i.e. it is an algebra homomorphism.
\end{lemma}
\begin{proof}
The first statement is about checking relations \eqref{eqn: U(V) translation relation} and \eqref{eqn: U(V) commutation relation}. Equation \eqref{eqn: U(V) translation relation} follows from the identity $\partial_{z_i}\Delta_{IJ}(f)=\Delta_{IJ}(\partial_{z_i}f)$. Equation \eqref{eqn: U(V) commutation relation} follows from the identity:
\begin{align*}
    \Delta_{IJ}((z_i-z_{i+1})^k&f_k(z_1,\cdots,z_{i-1},z_{i+1},\cdots, z_n))=\\
    &\begin{cases}
        (z_i-z_{i+1})^k\Delta_{IJ}(f_k), & i,i+1\in I \text{ or }i,i+1\in J,\\
        (z_i-z_{i+1}-w)^k\Delta_{IJ}(f_k), & i\in I\text{ and }i+1\in J,\\
        (z_i+w-z_{i+1})^k\Delta_{IJ}(f_k), & i\in J\text{ and }i+1\in I.
    \end{cases}
\end{align*}
Note that the cases when $i$ and $i+1$ are in the different index groups cancel in the difference $\mathcal O(\cdots A,B\cdots;g(\cdots z_{i},z_{i+1}\cdots))
        -\mathcal O(\cdots B,A\cdots;g(\cdots z_{i+1},z_i\cdots))$.

The second statement is formal: the associativity of $\Delta_{\mathcal V}(w)$ comes from the associativity of $\Delta_{IJ}$ and the definition of $\Delta_{\mathcal V}(w)$; $\Delta_{\mathcal V}(w)$ is unital since 
\begin{align*}
    \Delta_{\mathcal V}(w)(|0\rangle_{[-1]})=|0\rangle_{[-1]}\otimes 1+ \sum_{s\ge 0} (-1)^s w^{-s-1} 1\otimes |0\rangle_{[s]}=|0\rangle_{[-1]}\otimes 1=1\otimes 1.
\end{align*}
\end{proof}

\begin{remark}
    It is easy to see that $\Delta_{\mathcal V}(w)$ maps $U_+(\mathcal V)$ to ${U}_{+}(\mathcal V)\otimes {U}_+(\mathcal V) (\!(w^{-1})\!)$.
\end{remark}

The following proposition is obvious from the construction of the meromorphic coproduct.
\begin{proposition}\label{prop: functoriality of meromorphic coproduct}
The vertex algebra meromorphic coproduct $\Delta_{\mathcal V}(w)$ is functorial, it commutes with the vertex algebra morphism $\varphi:\mathcal V_1\to \mathcal V_2$, i.e. 
\begin{equation}
    (\varphi\otimes\varphi)\circ \Delta_{\mathcal V_1}(w)=\Delta_{\mathcal V_2}(w)\circ \varphi
\end{equation}
\end{proposition}

\subsection{Vertex coalgebra from a vertex algebra}
Consider the linear map $\mathfrak{C}_{\mathcal V}:U_+(\mathcal V)\to \mathbb C$ sending $\mathcal O(A_1,\cdots,A_n;f)$ to $0$ for $f\in \mathbb C[z_i,(z_i-z_j)^{-1}\:|\: 1\le i,j\le n]$ and sending the identity $\mathcal O(|0\rangle;x^{-1})$ to $1$. $\mathfrak{C}_{\mathcal V}$ is an algebra homomorphism.

\begin{proposition}\label{prop: vertex coalgebra from vertex algebra}
    $(U_+(\mathcal V),\Delta_{\mathcal V}(w),\mathfrak{C}_{\mathcal V})$ is a vertex coalgebra, and $(U(\mathcal V),\Delta_{\mathcal V})$ is a vertex comodule of $U_+(\mathcal V)$.
\end{proposition}
\begin{proof}
Counit axiom of $\Delta_{\mathcal V}(w)$ is checked as follows:
\begin{align*}
    (\mathrm{id}\otimes\mathfrak{C}_{\mathcal V})\circ \Delta_{\mathcal V}(w)(\mathcal O(A_1,\cdots,A_n;f))=\sum_{s\in \mathbb Z} \mathcal O(A_1,\cdots,A_n;\Delta_{(1\cdots n)\emptyset}^{(s)}(f))w^{-s}=\mathcal O(A_1,\cdots,A_n;f).
\end{align*}
Here $f$ is allowed to have poles at $z_i=0$, i.e. the counit axiom holds for both $U(\mathcal V)$ and $U_+(\mathcal V)$. 

The cocreation axiom for $U_+(\mathcal V)$ is checked similarly:
\begin{align*}
    (\mathfrak{C}_{\mathcal V}\otimes\mathrm{id})\circ \Delta_{\mathcal V}(w)(\mathcal O(A_1,\cdots,A_n;f))=\sum_{s\in \mathbb Z} \mathcal O(A_1,\cdots,A_n;\prescript{}{\emptyset}\Delta_{(1\cdots n)}^{(s)}(f))w^{-s},
\end{align*}
where $f\in \mathbb C[z_i,(z_i-z_j)^{-1}\:|\: 1\le i,j\le n]$. Since $\prescript{}{\emptyset}\Delta_{(1\cdots n)}^{(s)}(f)$ is zero for $s>0$, and $\prescript{}{\emptyset}\Delta_{(1\cdots n)}^{(0)}(f)=f$, the cocreation axiom holds.

The translation operator $D=(\mathfrak{C}_{\mathcal V}\otimes\mathrm{id})\circ \Delta_{\mathcal V,-2}$ for $U_+(\mathcal V)$ is the sum of derivatives:
\begin{equation}
    \begin{split}
        &D(\mathcal O(A_1,\cdots,A_n;f))=\frac{\mathrm{d}}{\mathrm{d}w}\bigg|_{w=0}\mathcal O(A_1,\cdots,A_n;f(z_1+w,\cdots,z_n+w))\\
        &~=\sum_{i=1}^n \mathcal O(A_1,\cdots,A_n;\partial_{z_i}f)=-\sum_{i=1}^n \mathcal O(A_1,\cdots,TA_i,\cdots,A_n;f).
    \end{split}
\end{equation}
So we have
\begin{align*}
    &(\Delta_{\mathcal V}(w)\circ D-(D\otimes\mathrm{id})\circ \Delta_{\mathcal V}(w))(\mathcal O(A_1,\cdots,A_n;f))\\
    &~=\sum_{IJ}\sum_{s\in \mathbb Z}\sum_{l=1}^{n-m}\mathcal O(A_{i_1},\cdots,A_{i_m};\Delta^{(s)}_{IJ}(f))\otimes \mathcal O(A_{j_1},\cdots,A_{j_{n-m}};\partial_{z_{j_l}}\prescript{}{I}\Delta^{(s)}_{J}(f))w^{-s}.
\end{align*}
Use the identity $\sum_{l=1}^m \partial_{z_{i_l}}\Delta_{IJ}(f)=\frac{\mathrm{d}}{\mathrm{d}w}\Delta_{IJ}(f)$, and we get the translation axiom:
\begin{align*}
    (\Delta_{\mathcal V}(w)\circ D-(D\otimes\mathrm{id})\circ \Delta_{\mathcal V}(w))(\mathcal O(A_1,\cdots,A_n;f))=\frac{\mathrm{d}}{\mathrm{d}w}\Delta_{\mathcal V}(w)(\mathcal O(A_1,\cdots,A_n;f)).
\end{align*}
The proof of the locality for $U_+(\mathcal V)$ is formal. In fact we have
\begin{align*}
    &(\Delta_{\mathcal V}(w)\otimes\mathrm{id})\circ \Delta_{\mathcal V}(z)(\mathcal O(A_1,\cdots,A_n;f))=\\
    &\sum_{I\sqcup J\sqcup K\in\mathrm{shuffle}(1,\cdots,n)}\sum_{t\in \mathbb Z}\sum_{s\in \mathbb Z}\mathcal O(A_I;\Delta^{(s)}_{IJK}(f))\otimes \mathcal O(A_J;\prescript{}{I}\Delta^{(t)}_{JK}(f))\otimes \mathcal O(A_K;\prescript{}{IJ}\Delta^{(s,t)}_{K}(f)) z^{-s}w^{-t},
\end{align*}
where $A_I$ is the abbreviation of $A_{i_1},\cdots,A_{i_{|I|}}$, and
\begin{align*}
    \Delta_{IJK}(f)=\sum_{t\in \mathbb Z}\sum_{s\in \mathbb Z}\Delta^{(s)}_{IJK}(f)\otimes \prescript{}{I}\Delta^{(t)}_{JK}(f)\otimes \prescript{}{IJ}\Delta^{(s,t)}_{K}(f) z^{-s} w^{-t}
\end{align*}
is the expansion of $f(z_I,z_J+w,z_K+z)$ in $\mathscr O(\mathbb C^I_{\mathrm{disj}}\times \mathbb C^J_{\mathrm{disj}}\times \mathbb C^K_{\mathrm{disj}})(\!(z^{-1})\!)(\!(w^{-1})\!)$. And similarly
\begin{align*}
    &(\mathrm{id}\otimes P)\circ(\Delta_{\mathcal V}(z)\otimes \mathrm{id})\circ\Delta_{\mathcal V}(w)(\mathcal O(A_1,\cdots,A_n;f))=\\
    &\sum_{I\sqcup J\sqcup K\in\mathrm{shuffle}(1,\cdots,n)}\sum_{s\in \mathbb Z}\sum_{t\in \mathbb Z}\mathcal O(A_I;\widetilde\Delta^{(s)}_{IJK}(f))\otimes \mathcal O(A_J;\prescript{}{I}{\widetilde{\Delta}}^{(t)}_{JK}(f))\otimes \mathcal O(A_K;\prescript{}{IJ}{\widetilde{\Delta}}^{(s,t)}_{K}(f)) w^{-t}z^{-s},
\end{align*}
where
\begin{align*}
    \widetilde\Delta_{IJK}(f)=\sum_{s\in \mathbb Z}\sum_{t\in \mathbb Z}{\widetilde{\Delta}}^{(s)}_{IJK}(f)\otimes \prescript{}{I}{\widetilde{\Delta}}^{(t)}_{JK}(f)\otimes \prescript{}{IJ}{\widetilde{\Delta}}^{(s,t)}_{K}(f) w^{-t} z^{-s} 
\end{align*}
is the expansion of $f(z_I,z_J+w,z_K+z)$ in $\mathscr O(\mathbb C^I_{\mathrm{disj}}\times \mathbb C^J_{\mathrm{disj}}\times \mathbb C^K_{\mathrm{disj}})(\!(w^{-1})\!)(\!(z^{-1})\!)$. Therefore $(\Delta_{\mathcal V}(w)\otimes\mathrm{id})\circ \Delta_{\mathcal V}(z)(\mathcal O(A_1,\cdots,A_n;f))$ and $(\mathrm{id}\otimes P)\circ(\Delta_{\mathcal V}(z)\otimes \mathrm{id})\circ\Delta_{\mathcal V}(w)(\mathcal O(A_1,\cdots,A_n;f))$ are the expansions of the same element in $U_{+}(\mathcal V)^{\otimes 3}[\![z^{-1},w^{-1},(z-w)^{-1}]\!][z,w]$. This prove the locality axiom for $U_+(\mathcal V)$. The coassociativity for $U(\mathcal V)$ is proven similarly, and we shall omit it.
\end{proof}

\subsection{Ring of differential operators valued in restricted modes}
In this subsection, we keep using the notation $\mathcal V$ to denote a $\mathbb Z$-graded vertex algebra. We fix a finite index set $I$.

Let $F(I;\mathcal V)$ be the $\mathbb C$-linear space 
\begin{align}
    F(I;\mathcal V)=\bigoplus_{n> 0}\mathcal V^{\otimes n}\otimes D_I(\mathbb C^{\times (I\sqcup\{1,\cdots,n\})}_{\mathrm{disj}}),
\end{align}
where $D_I(\mathbb C^{\times (I\sqcup\{1,\cdots,n\})}_{\mathrm{disj}})$ is the subalgebra of the ring of differential operators $D(\mathbb C^{\times (I\sqcup\{1,\cdots,n\})}_{\mathrm{disj}})$ such that only the derivatives in the index set $I$ appears. Our convention for coordinates on $\mathbb C^{\times (I\sqcup\{1,\cdots,n\})}$ is such that for index $\alpha\in I$ we use $x_{\alpha}$ and for index $i\in \{1,\cdots,n\}$ we use $z_i$. Define the multiplication on $F(I;\mathcal V)$ by
\begin{align}\label{eqn: multiplication on D(V)}
    \mathcal O(A_1,\cdots,A_m;f)\cdot \mathcal O(B_1,\cdots,B_n;g)=\mathcal O(A_1,\cdots,A_m, B_1,\cdots,B_n;fg),
\end{align}
so that $F(I;\mathcal V)$ acquires a non-unital associative algebra structure.
\begin{definition}\label{def: differential operators valued in restricted modes}
The ring of differential operators on $\mathbb C^{\times I}_{\mathrm{disj}}$ valued in restricted modes of $\mathcal V$, denoted by $D(\mathbb C^{\times I}_{\mathrm{disj}};\mathcal V)$ is the quotient of $F(I;\mathcal V)$ by the linear space spanned by
\begin{equation}\label{eqn: D(V) unity relation}
    \mathcal O(A_1,\cdots,A_{i-1},|0\rangle,A_{i+1},\cdots,A_m;f)-\mathcal O(A_1,\cdots,A_{i-1},A_{i+1},\cdots,A_m;\underset{z_i\to 0}{\Res}f),
\end{equation}
\begin{equation}\label{eqn: D(V) translation relation}
    \mathcal O(A_1,\cdots,TA_i,\cdots,A_m;f)+\mathcal O(A_1,\cdots,A_i,\cdots,A_m;\partial_{z_i}f),
\end{equation}
\begin{equation}\label{eqn: D(V) commutation relation}
    \begin{split}
        &\mathcal O(A_1,\cdots,A_{i-1},A,B,A_{i+2},\cdots,A_{n};g(z_1,\cdots, z_{i},z_{i+1},\cdots,z_n))\\
        -&\mathcal O(A_1,\cdots,A_{i-1},B,A,A_{i+2},\cdots,A_{n};g(z_1,\cdots,z_{i+1},z_i,\cdots,z_n))\\
        -&\sum_{k\in \mathbb Z}\mathcal O(A_1,\cdots,A_{i-1},A_{(k)}B,A_{i+2},\cdots, A_n;f_k),
    \end{split}
\end{equation}
for all $k\ge -1$ and $g\in D_I(\mathbb C^{\times (I\sqcup\{1,\cdots,n\})}_{\mathrm{disj}})$, here $f_k$ are the coefficients of $(z_i-z_{i+1})^k$ in the expansion
\begin{align*}
    g(z_1,\cdots,z_i,z_{i+1},\cdots,z_n)=\sum_{k\in \mathbb Z} (z_i-z_{i+1})^kf_k(z_1,\cdots,z_{i-1},z_{i+1},\cdots,z_n),
\end{align*}
where the RHS is in the space $D_I(\mathbb C^{\times (I\sqcup\{1,\cdots,i-1,i+1,\cdots,n\})}_{\mathrm{disj}})(\!(z_i-z_{i+1})\!)$, i.e. expanded in the limit $z_i-z_{i+1}\to 0$, in particular the power series is bounded from below. The summation in \eqref{eqn: D(V) commutation relation} is bounded because $A_{(k)}B$ vanishes for the sufficiently large $k$. 
\end{definition}

$D(\mathbb C^{\times I}_{\mathrm{disj}};\mathcal V)$ is unital since 
\begin{align*}
    \mathcal O(|0\rangle;z^{-1})\mathcal O(A_1,\cdots,A_n;f)=\mathcal O(A_1,\cdots,A_n;f)\mathcal O(|0\rangle;z^{-1})=\mathcal O(A_1,\cdots,A_n;f),
\end{align*}
which is a consequence of \eqref{eqn: D(V) unity relation}. 

Note that $D(\mathbb C^{\times I}_{\mathrm{disj}};\mathcal V)$ inherits a $\mathbb Z$-grading from $\mathcal V$ and the ring of differential operators such that $\deg \mathcal O(A_1,\cdots,A_n;f)$ for a homogeneous differential operator $f$ is $\deg A_1+\cdots+\deg A_n+\deg f+n$, where the degree of a differential operator is such that $\deg x_{\alpha}=\deg z_i=1,\deg\partial_{x_{\alpha}}=-1$ for all $\alpha\in I,i\in \{1,\cdots,n\}$. 

There is a $\mathbb Z$-graded algebra homomorphism $D(\mathbb C^{\times I}_{\mathrm{disj}};\mathcal V)\to D(\mathbb C^{\times I}_{\mathrm{disj}})\widetilde\otimes\mathfrak{U}(\mathcal V)$ given by
\begin{align}\label{eqn: from D(V) to completed tensor product}
     \mathcal O(A_1,\cdots,A_m;f)\mapsto\oint_{|z_1|>\cdots>|z_m|>|x_{\alpha}|}f(z_1,\cdots,z_m)A_1[z_1]\cdots A_m[z_m]\prod_{j=1}^m\frac{d z_j}{2\pi i },
\end{align}
For example
\begin{align*}
    \mathcal O\left(A;\frac{1}{z_1-x_\alpha}\partial_{x_{\beta}}\right)&\mapsto \sum_{n=0}^{\infty}x_{\alpha}^n\partial_{x_{\beta}}\otimes A_{[-n-1]},\\
    \mathcal O\left(A,B;\frac{z_1z_2}{(z_1-z_2)^2(z_1-x_{\alpha})}\right)&\mapsto \sum_{m=0}^{\infty}x_{\alpha}^m\otimes\left(\sum_{n=1}^{\infty}n A_{[-n-m-1]}B_{[n]}\right).
\end{align*}

\begin{proposition}\label{prop: D(V) is a sub of completed tensor product}
Assume that $\mathcal V$ satisfies the technical assumptions in the Proposition \ref{prop: U(V) is a sub of mode algebra}, then $D(\mathbb C^{\times I}_{\mathrm{disj}};\mathcal V)\to D(\mathbb C^{\times I}_{\mathrm{disj}})\widetilde\otimes\mathfrak{U}(\mathcal V)$ is injective.   
\end{proposition}

The proof is analogous to that of Proposition \ref{prop: U(V) is a sub of mode algebra} and we omit it.\\

Instead of expanding an element in $D_I(\mathbb C^{\times (I\sqcup\{1,\cdots,m\})}_{\mathrm{disj}})$ in the order $|z_1|>\cdots>|z_m|>|x_{\alpha}|$, we can also expand it in the order $|x_{\alpha}|>|z_1|>\cdots>|z_m|$ and get a $\mathbb Z$-graded algebra homomorphism $D(\mathbb C^{\times I}_{\mathrm{disj}};\mathcal V)\to \mathfrak{U}(\mathcal V)\widetilde\otimes D(\mathbb C^{\times I}_{\mathrm{disj}})$ given by
\begin{align}\label{eqn: from D(V) to completed tensor product_opposite}
     \mathcal O(A_1,\cdots,A_m;f)\mapsto\oint_{|x_{\alpha}|>|z_1|>\cdots>|z_m|}f(z_1,\cdots,z_m)A_1[z_1]\cdots A_m[z_m]\prod_{j=1}^m\frac{d z_j}{2\pi i }.
\end{align}
If $\mathcal V$ satisfies the technical assumptions in the Proposition \ref{prop: U(V) is a sub of mode algebra}, then $D(\mathbb C^{\times I}_{\mathrm{disj}};\mathcal V)\to \mathfrak{U}(\mathcal V)\widetilde\otimes D(\mathbb C^{\times I}_{\mathrm{disj}})$ is injective.\\

We also define a linear subspace $D_+(\mathbb C^{I}_{\mathrm{disj}};\mathcal V)\subset D(\mathbb C^{\times I}_{\mathrm{disj}};\mathcal V)$ spanned by the identity $\mathcal O(|0\rangle;z_1^{-1})$ and those $\mathcal O(A_1,\cdots,A_m;f)$ such that $f\in \mathscr O(\mathbb C^{I\sqcup\{1,\cdots,m\}}_{\mathrm{disj}})$. It is easy to see that $D_+(\mathbb C^{I}_{\mathrm{disj}};\mathcal V)$ is a subalgebra, and we call it the ring of differential operators on $\mathbb C^{I}_{\mathrm{disj}}$ valued in positive restricted modes of $\mathcal V$.

Consider the morphism:
\begin{align*}
    \mathbb C^{\times (I_1\sqcup J_1)}_{\mathrm{disj}}\times \mathbb C^{ I_2\sqcup J_2}_{\mathrm{disj}}\times \Spec \mathbb C(\!(w^{-1})\!)\to \mathbb C^{\times (I_1\sqcup J_1\sqcup I_2\sqcup J_2)}_{\mathrm{disj}},
\end{align*}
which is induced from 
\begin{align*}
    x_{\alpha}\mapsto x_{\alpha},(\alpha\in I_1),\quad z_j\mapsto z_j,\; (i\in J_1),\quad x_{\beta}\mapsto x_{\beta}+w,(\beta\in I_2),\quad z_k\mapsto z_k+w,\;j\in J_2.
\end{align*}
The map between the spaces induces an algebra map between ring of differential operators
\begin{align*}
    \Delta^{I_1I_2}_{J_1J_2}:D_{I_1\sqcup I_2}(\mathbb C^{\times (I_1\sqcup J_1\sqcup I_2\sqcup J_2)}_{\mathrm{disj}})\to D_{I_1}(\mathbb C^{\times (I_1\sqcup J_1)}_{\mathrm{disj}})\otimes D_{I_2}(\mathbb C^{ I_2\sqcup J_2}_{\mathrm{disj}})(\!(w^{-1})\!)
\end{align*}
and we write the corresponding map on the ring of differential operators as:
\begin{align}
    \Delta^{I_1I_2}_{J_1J_2}(f)=\sum_{s\in \mathbb Z} \Delta^{I_1I_2(s)}_{J_1J_2}(f)\otimes \prescript{I_1}{J_1}\Delta_{J_2}^{I_2(s)}(f) w^{-s}.
\end{align}
Using this ``point-splitting'' trick, we define a linear map $\Delta_{I_1I_2,\mathcal V}(w):F_{I_1\sqcup I_2}(\mathcal V)\to D(\mathbb C^{\times I_1}_{\mathrm{disj}};\mathcal V)\otimes D_{+}(\mathbb C^{I_2}_{\mathrm{disj}};\mathcal V) (\!(w^{-1})\!)$ by
\begin{equation}\label{eqn: meromorphic coproduct of D(V)}
\boxed{
    \begin{aligned}
        &\Delta_{I_1I_2,\mathcal V}(w)(\mathcal O(A_1,\cdots,A_n;f))=\\
        &\sum_{\substack{J_1=(j_1,\cdots,j_m)\\J_2= (k_1,\cdots,k_{n-m})\\ J_1\sqcup J_2\in\mathrm{shuffle}(1,\cdots,n)}}\sum_{s\in \mathbb Z}\mathcal O(A_{j_1},\cdots,A_{j_m};\Delta^{I_1I_2(s)}_{J_1J_2}(f))\otimes \mathcal O(A_{k_1},\cdots,A_{k_{n-m}};\prescript{I_1}{J_1}\Delta_{J_2}^{I_2(s)}(f))w^{-s},
    \end{aligned}
}
\end{equation}
Here $J_1$ or $J_2$ can be empty set.

\begin{lemma}\label{lem: meromorphic coproduct for diff ops}
The linear map $\Delta_{I_1I_2,\mathcal V}(w)$ factors through the restricted mode algebra $D(\mathbb C^{\times (I_1\sqcup I_2)}_{\mathrm{disj}};\mathcal V)$, and it gives rise to an algebra homomorphism $\Delta_{I_1I_2,\mathcal V}(w):D(\mathbb C^{\times (I_1\sqcup I_2)}_{\mathrm{disj}};\mathcal V)\to D(\mathbb C^{\times I_1}_{\mathrm{disj}};\mathcal V)\otimes D_{+}(\mathbb C^{I_2}_{\mathrm{disj}};\mathcal V) (\!(w^{-1})\!)$.
\end{lemma}

The proof of Lemma \ref{lem: meromorphic coproduct for diff ops} is almost the same as that of Lemma \ref{lem: meromorphic coproduct for restricted modes} and we omit the details. A special case of Lemma \ref{lem: meromorphic coproduct for diff ops} is when $I_2$ is the empty set, then $D_{+}(\mathbb C^{\emptyset}_{\mathrm{disj}};\mathcal V)=U_+(\mathcal V)$ by definition. It is then easy to check that $\Delta_{I\emptyset,\mathcal V}(w)$ endows $D(\mathbb C^{\times I}_{\mathrm{disj}};\mathcal V)$ with a vertex comodule structure with respect to the vertex coalgebra $U_+(\mathcal V)$. \\

Finally, consider the vector space $\Hom(\mathcal V,D(\mathbb C^{\times I}_{\mathrm{disj}})\widetilde\otimes \mathcal V)$, then it possesses a natural left module structure of $D(\mathbb C^{\times I}_{\mathrm{disj}})\widetilde\otimes\mathfrak{U}(\mathcal V)$, of which the action is given by the composition of a map in $\Hom(\mathcal V,D(\mathbb C^{\times I}_{\mathrm{disj}})\widetilde\otimes \mathcal V)$ with the action of $D(\mathbb C^{\times I}_{\mathrm{disj}})\widetilde\otimes\mathfrak{U}(\mathcal V)$. On the other hand, it also possesses a natural right module structure of $\mathfrak{U}(\mathcal V)\widetilde\otimes D(\mathbb C^{\times I}_{\mathrm{disj}})$, of which the action is given by the precomposition of a map $\Hom(\mathcal V,D(\mathbb C^{\times I}_{\mathrm{disj}})\widetilde\otimes \mathcal V)$ with the action of $\mathfrak{U}(\mathcal V)\widetilde\otimes D(\mathbb C^{\times I}_{\mathrm{disj}})$ on $\mathcal V$. Note that for all $x\in \mathfrak{U}(\mathcal V)\widetilde\otimes D(\mathbb C^{\times I}_{\mathrm{disj}})$ and all $v\in \mathcal V$, the action $x\cdot v$ belongs to the usual tensor product space $D(\mathbb C^{\times I}_{\mathrm{disj}})\otimes \mathcal V$. We summarize the above discussions as follows.

\begin{proposition}\label{prop: D(V) bimodule}
The vectors space $\Hom(\mathcal V,D(\mathbb C^{\times I}_{\mathrm{disj}})\widetilde\otimes \mathcal V)$ possesses a natural $D(\mathbb C^{\times I}_{\mathrm{disj}};\mathcal V)$ bimodule structure.
\end{proposition}

\bibliographystyle{unsrt}
\bibliography{Bib}

\begin{thebibliography}{10}

\bibitem{gaiotto2022miura}
Davide Gaiotto and Miroslav Rap{\v{c}}{\'a}k.
\newblock {Miura operators, degenerate fields and the M2-M5 intersection}.
\newblock {\em Journal of High Energy Physics}, 2022(1):1--80, 2022.

\bibitem{costello2017holography}
Kevin Costello.
\newblock {Holography and Koszul duality: the example of the $ M2 $ brane}.
\newblock {\em preprint arXiv:1705.02500}, 2017.

\bibitem{gaiotto2019aspects}
Davide Gaiotto and Jihwan Oh.
\newblock {Aspects of $\Omega$-deformed M-theory}.
\newblock {\em preprint arXiv:1907.06495}, 2019.

\bibitem{oh2021feynman}
Jihwan Oh and Yehao Zhou.
\newblock {Feynman diagrams and $\Omega$-deformed M-theory}.
\newblock {\em SciPost Physics}, 10(2):029, 2021.

\bibitem{oh2021twisted}
Jihwan Oh and Yehao Zhou.
\newblock Twisted holography of defect fusions.
\newblock {\em SciPost Physics}, 10(5):105, 2021.

\bibitem{ueda2019affine}
Mamoru Ueda.
\newblock {Affine Super Yangian}.
\newblock {\em preprint arXiv:1911.06666}, 2019.

\bibitem{ueda2022affine}
Mamoru Ueda.
\newblock {Affine super Yangians and rectangular W-superalgebras}.
\newblock {\em Journal of Mathematical Physics}, 63(5):051701, 2022.

\bibitem{rapvcak2020extensions}
Miroslav Rap{\v{c}}{\'a}k.
\newblock {On extensions of $\mathfrak{gl}\widehat{(m|n)}$ Kac-Moody algebras
  and Calabi-Yau singularities}.
\newblock {\em Journal of High Energy Physics}, 2020(1):1--35, 2020.

\bibitem{gaberdiel2018twin}
Matthias~R Gaberdiel, Wei Li, and Cheng Peng.
\newblock {Twin-plane-partitions and $\mathcal N=2$ affine Yangian}.
\newblock {\em Journal of High Energy Physics}, 2018(11):1--62, 2018.

\bibitem{guay2017deformed}
Nicolas Guay and Yaping Yang.
\newblock {On deformed double current algebras for simple Lie algebras}.
\newblock {\em Mathematical Research Letters}, 24(5):1307--1384, 2017.

\bibitem{linshaw2021universal}
Andrew~R Linshaw.
\newblock {Universal two-parameter $\mathcal W_{\infty}$-algebra and vertex
  algebras of type $\mathcal W(2, 3,\cdots, N)$}.
\newblock {\em Compositio Mathematica}, 157(1):12--82, 2021.

\bibitem{costello2021factorization}
Kevin Costello and Owen Gwilliam.
\newblock {\em Factorization algebras in quantum field theory}, volume~2.
\newblock Cambridge University Press, 2021.

\bibitem{butson2020equivariant}
Dylan Butson.
\newblock {Equivariant localization in factorization homology and applications
  in mathematical physics II: Gauge theory applications}.
\newblock {\em preprint arXiv:2011.14978}, 2020.

\bibitem{guay2005cherednik}
Nicholas Guay.
\newblock {Cherednik algebras and Yangians}.
\newblock {\em International Mathematics Research Notices},
  2005(57):3551--3593, 2005.

\bibitem{etingof2002symplectic}
Pavel Etingof and Victor Ginzburg.
\newblock {Symplectic reflection algebras, Calogero-Moser space, and deformed
  Harish-Chandra homomorphism}.
\newblock {\em Inventiones mathematicae}, 147:243--348, 2002.

\bibitem{opdam2000lecture}
Eric~M Opdam.
\newblock {\em Lecture notes on Dunkl operators for real and complex reflection
  groups}.
\newblock Mathematical Society of Japan, 2000.

\bibitem{etingof2010lecture}
Pavel Etingof and Xiaoguang Ma.
\newblock {Lecture notes on Cherednik algebras}.
\newblock {\em preprint arXiv:1001.0432}, 2010.

\bibitem{kalinov2021deformed}
Daniil Kalinov.
\newblock {Deformed Double Current Algebras via Deligne Categories}.
\newblock {\em preprint arXiv:2101.08317}, 2021.

\bibitem{hu2023quantum}
Sen Hu, Si~Li, Dongheng Ye, and Yehao Zhou.
\newblock {Quantum Algebra of Chern-Simons Matrix Model and Large $ N $ Limit}.
\newblock {\em preprint arXiv:2308.14046}, 2023.

\bibitem{etingof2023new}
Pavel Etingof, Daniil Kalinov, and Eric Rains.
\newblock {New realizations of deformed double current algebras and Deligne
  categories}.
\newblock {\em Transformation Groups}, 28(1):185--239, 2023.

\bibitem{guay2007affine}
Nicolas Guay.
\newblock {Affine Yangians and deformed double current algebras in type A}.
\newblock {\em Advances in Mathematics}, 211(2):436--484, 2007.

\bibitem{kodera2018quantized}
Ryosuke Kodera and Hiraku Nakajima.
\newblock {Quantized Coulomb branches of Jordan quiver gauge theories and
  cyclotomic rational Cherednik algebras}.
\newblock In {\em Proc. Symp. Pure Math}, volume~98, page~49, 2018.

\bibitem{rapcak2020cohomological}
Miroslav Rap{\v{c}}{\'a}k, Yan Soibelman, Yaping Yang, and Gufang Zhao.
\newblock {Cohomological Hall algebras and perverse coherent sheaves on toric
  Calabi-Yau 3-folds}.
\newblock {\em preprint arXiv:2007.13365}, 2020.

\bibitem{creutzig2022trialities}
Thomas Creutzig and Andrew~R Linshaw.
\newblock {Trialities of $\mathcal{W}$-algebras}.
\newblock {\em Cambridge Journal of Mathematics}, 10(1), 2022.

\bibitem{schiffmann2013cherednik}
Olivier Schiffmann and Eric Vasserot.
\newblock {Cherednik algebras, W-algebras and the equivariant cohomology of the
  moduli space of instantons on $\mathbb A^2$}.
\newblock {\em Publications math{\'e}matiques de l'IH{\'E}S}, 118(1):213--342,
  2013.

\bibitem{moosavian2021towards}
Seyed~Faroogh Moosavian and Yehao Zhou.
\newblock {Towards the Finite-$ N $ Twisted Holography from the Geometry of
  Phase Space}.
\newblock {\em preprint arXiv:2111.06876}, 2021.

\bibitem{bernard1993yang}
Denis Bernard, M~Gaudin, FDM Haldane, and V~Pasquier.
\newblock {Yang-Baxter equation in spin chains with long range interactions}.
\newblock {\em preprint hep-th/9301084}, 1993.

\bibitem{kodera2022coproduct}
Ryosuke Kodera and Mamoru Ueda.
\newblock {Coproduct for affine Yangians and parabolic induction for
  rectangular W-algebras}.
\newblock {\em Letters in Mathematical Physics}, 112(1):1--37, 2022.

\bibitem{prochazka2015exploring}
Tom{\'a}{\v{s}} Proch{\'a}zka.
\newblock {Exploring $\mathcal W_{\infty}$ in the quadratic basis}.
\newblock {\em Journal of High Energy Physics}, 2015(9):1--63, 2015.

\bibitem{eberhardt2019matrix}
Lorenz Eberhardt and Tom{\'a}{\v{s}} Proch{\'a}zka.
\newblock {The matrix-extended $\mathcal{W}_{\infty}$ algebra}.
\newblock {\em Journal of High Energy Physics}, 2019(12):1--35, 2019.

\bibitem{gaiotto2019vertex}
Davide Gaiotto and Miroslav Rap{\v{c}}{\'a}k.
\newblock Vertex algebras at the corner.
\newblock {\em Journal of High Energy Physics}, 2019(1):1--88, 2019.

\bibitem{arakawa2017explicit}
Tomoyuki Arakawa and Alexander Molev.
\newblock {Explicit generators in rectangular affine W-algebras of type A}.
\newblock {\em Letters in Mathematical Physics}, 107(1):47--59, 2017.

\bibitem{arakawa2007representation}
Tomoyuki Arakawa.
\newblock {Representation theory of $\mathcal W$-algebras}.
\newblock {\em Inventiones mathematicae}, 169(2):219--320, 2007.

\bibitem{neguct2022deformed}
Andrei Negu{\c{t}}.
\newblock {Deformed W-algebras in type A for rectangular nilpotent}.
\newblock {\em Communications in Mathematical Physics}, 389(1):153--195, 2022.

\bibitem{Costello:2016nkh}
Kevin Costello.
\newblock M-theory in the omega-background and 5-dimensional non-commutative
  gauge theory.
\newblock {\em preprint arXiv:1610.04144}, 2016.

\bibitem{ben2010symmetry}
David Ben-Zvi and Thomas Nevins.
\newblock {$\mathcal W$-Symmetry of the Ad{\`e}lic Grassmannian}.
\newblock {\em Communications in Mathematical Physics}, 293(1):185, 2010.

\bibitem{wendlandt2022formal}
Curtis Wendlandt.
\newblock {The formal shift operator on the Yangian double}.
\newblock {\em International Mathematics Research Notices},
  2022(14):10952--11010, 2022.

\bibitem{wendlandt2022restricted}
Curtis Wendlandt.
\newblock {The restricted quantum double of the Yangian}.
\newblock {\em preprint arXiv:2204.00983}, 2022.

\bibitem{gaiotto2020twisted}
Davide Gaiotto and Jacob Abajian.
\newblock {Twisted M2 brane holography and sphere correlation functions}.
\newblock {\em preprint arXiv:2004.13810}, 2020.

\bibitem{polychronakos2019feynman}
Alexios~P Polychronakos.
\newblock {Feynman’s proof of the commutativity of the Calogero integrals of
  motion}.
\newblock {\em Annals of Physics}, 403:145--151, 2019.

\bibitem{prochazka2019instanton}
Tom{\'a}{\v{s}} Proch{\'a}zka.
\newblock {Instanton R-matrix and $\mathcal W$-symmetry}.
\newblock {\em Journal of High Energy Physics}, 2019(12):1--58, 2019.

\bibitem{gaberdiel2000axiomatic}
Matthias~R Gaberdiel and Peter Goddard.
\newblock Axiomatic conformal field theory.
\newblock {\em Communications in Mathematical Physics}, 209(3):549--594, 2000.

\bibitem{bershtein2019homomorphisms}
Mikhail Bershtein and Alexander Tsymbaliuk.
\newblock {Homomorphisms between different quantum toroidal and affine Yangian
  algebras}.
\newblock {\em Journal of Pure and Applied Algebra}, 223(2):867--899, 2019.

\bibitem{tsymbaliuk2017affine}
Alexander Tsymbaliuk.
\newblock {The affine Yangian of $\mathfrak{gl}_1$ revisited}.
\newblock {\em Advances in Mathematics}, 304:583--645, 2017.

\bibitem{kodera2015affine}
Ryosuke Kodera.
\newblock {Affine Yangian action on the Fock space}.
\newblock {\em preprint arXiv:1506.01246}, 2015.

\bibitem{li2020quiver}
Wei Li and Masahito Yamazaki.
\newblock {Quiver Yangian from crystal melting}.
\newblock {\em Journal of High Energy Physics}, 2020(11):1--127, 2020.

\bibitem{galakhov2020quiver}
Dmitry Galakhov and Masahito Yamazaki.
\newblock {Quiver Yangian and Supersymmetric Quantum Mechanics}.
\newblock {\em preprint arXiv:2008.07006}, 2020.

\bibitem{galakhov2021shifted}
Dmitry Galakhov, Wei Li, and Masahito Yamazaki.
\newblock {Shifted quiver Yangians and representations from BPS crystals}.
\newblock {\em Journal of High Energy Physics}, 2021(8):1--75, 2021.

\bibitem{bao2022note}
Jiakang Bao.
\newblock {A note on quiver Yangians and $\mathcal R$-matrices}.
\newblock {\em Journal of High Energy Physics}, 2022(8):1--46, 2022.

\bibitem{chari1995guide}
Vyjayanthi Chari and Andrew Pressley.
\newblock {\em A guide to quantum groups}.
\newblock Cambridge university press, 1995.

\bibitem{braverman2016coulomb}
Alexander Braverman, Michael Finkelberg, and Hiraku Nakajima.
\newblock {Coulomb branches of 3d $\mathcal N=4$ quiver gauge theories and
  slices in the affine Grassmannian (with appendices by Alexander Braverman,
  Michael Finkelberg, Joel Kamnitzer, Ryosuke Kodera, Hiraku Nakajima, Ben
  Webster, and Alex Weekes)}.
\newblock {\em preprint arXiv:1604.03625}, 2016.

\bibitem{kassel1982extensions}
Christian Kassel and Jean-Louis Loday.
\newblock {Extensions centrales d'alg{\`e}bres de Lie}.
\newblock In {\em {Annales de l'institut Fourier}}, volume~32, pages 119--142,
  1982.

\bibitem{bloch2006dilogarithm}
Spencer Bloch.
\newblock {The dilogarithm and extensions of Lie algebras}.
\newblock In {\em Algebraic K-Theory Evanston 1980: Proceedings of the
  Conference Held at Northwestern University Evanston, March 24--27, 1980},
  pages 1--23. Springer, 2006.

\bibitem{guay2018coproduct}
Nicolas Guay, Hiraku Nakajima, and Curtis Wendlandt.
\newblock {Coproduct for Yangians of affine Kac--Moody algebras}.
\newblock {\em Advances in Mathematics}, 338:865--911, 2018.

\bibitem{losev2012isomorphisms}
Ivan Losev.
\newblock Isomorphisms of quantizations via quantization of resolutions.
\newblock {\em Advances in Mathematics}, 231(3-4):1216--1270, 2012.

\bibitem{crawley2001geometry}
William Crawley-Boevey.
\newblock Geometry of the moment map for representations of quivers.
\newblock {\em Compositio Mathematica}, 126(3):257--293, 2001.

\bibitem{gibbons1984generalisation}
John Gibbons and Theo Hermsen.
\newblock {A generalisation of the Calogero-Moser system}.
\newblock {\em Physica D Nonlinear Phenomena}, 11(3):337--348, 1984.

\bibitem{krichever1994spin}
I~Krichever, O~Babelon, E~Billey, and M~Talon.
\newblock {Spin generalization of the Calogero-Moser system and the matrix KP
  equation}.
\newblock {\em preprint hep-th/9411160}, 1994.

\bibitem{hubbard2009vertex}
Keith Hubbard.
\newblock Vertex coalgebras, comodules, cocommutativity and coassociativity.
\newblock {\em Journal of Pure and Applied Algebra}, 213(1):109--126, 2009.

\end{thebibliography}

\end{document}